\documentclass[custombib, a4paper,12pt,times,print,xindex]{PhDThesisPSnPDF}

\ifsetCustomMargin
  \RequirePackage[left=37mm,right=30mm,top=35mm,bottom=30mm]{geometry}
  \setFancyHdr % To apply fancy header after geometry package is loaded
\fi

\ifsetCustomFont
  \RequirePackage{helvet}
\fi

\RequirePackage[labelsep=space,tableposition=top]{caption}
\usepackage{subcaption}

\usepackage{booktabs} % For professional looking tables
\usepackage{multirow}
\usepackage{multicol}
\usepackage{amsfonts}
\usepackage{amsmath}
\usepackage{mathtools}
\usepackage{amssymb}
\usepackage{siunitx} % use this package module for SI units
\usepackage{amsthm}
\usepackage{bm}

\ifuseCustomBib
\RequirePackage[backend=bibtex, style=ieee, citestyle=numeric-verb, sorting=nty, natbib=true]{biblatex}
\nocite{*}
\fi

\newtheorem{definition}{Definition}[section]
\newtheorem{proposition}[definition]{Proposition}
\newtheorem{example}[definition]{Exemple}
\newtheorem{corollary}[definition]{Corollary}
\newtheorem{theorem}[definition]{Theorem}
\newtheorem{lemma}[definition]{Lemma}
\newtheorem{remark}[definition]{Remark}
\newtheorem{notation}[definition]{Notation}
\newtheorem{scholium}[definition]{Scholium}

\usepackage[T1]{fontenc}
\usepackage[english]{babel}
\usepackage[latin1]{inputenx}
\usepackage[pdftex]{graphicx}           % usamos arquivos pdf/png como figuras
\usepackage{float}
\usepackage{setspace}                   % espa�amento flex�vel
\usepackage{indentfirst}                % indenta��o do primeiro par�grafo
\usepackage{makeidx}                    % �ndice remissivo
\usepackage[nottoc]{tocbibind}          % acrescentamos a bibliografia/indice/conteudo no Table of Contents
\usepackage{courier}                    % usa o Adobe Courier no lugar de Computer Modern Typewriter
\usepackage{type1cm}                    % fontes realmente escal�veis
\usepackage{listings}                   % para formatar c�digo-fonte (ex. em Java)
\usepackage{titletoc}
\usepackage{enumerate}
\usepackage[all]{xy}
\usepackage{quotchap}
\usepackage{url}
\usepackage{hyphenat}
\usepackage{fancyhdr}
\usepackage{xcolor}
\usepackage[utf8]{inputenc}
\usepackage{csquotes}
\usepackage{stmaryrd}
\usepackage[nomain]{glossaries}
\usepackage{bbm}
\usepackage{epigraph}

\usepackage{helvet}

\usepackage{lmodern}
\usepackage{titlesec}
\usepackage{microtype}
\usepackage{tikz}

\definecolor{myblue}{RGB}{0,0,0} %{0,82,155}

\titleformat{\chapter}[display]
  {\normalfont\bfseries\color{myblue}}
  {\filleft%
    \begin{tikzpicture}
    \node[
      outer sep=0pt,
      text width=2.5cm,
      minimum height=3cm,
      fill=myblue,
      font=\color{white}\fontsize{80}{90}\selectfont,
      align=center
      ] (num) {\thechapter};
    \node[
      rotate=90,
      anchor=south,
      font=\color{black}\Large\normalfont
      ] at ([xshift=-5pt]num.west) {\textls[180]{\textsc{\chaptertitlename}}};  
    \end{tikzpicture}%
  }
  {10pt}
  {\titlerule[2.5pt]\vskip3pt\titlerule\vskip4pt\Large\centering\sffamily}

\newcommand{\calgebra}{\mathfrak{A}}
\newcommand{\nalgebra}{\mathfrak{M}}
\newcommand{\algebra}{\mathfrak{S}}
\newcommand{\nanalytic}{\mathfrak{M}_\mathcal{A}}
\newcommand{\ex}{\mathfrak{E}}
\newcommand{\hilbert}{\mathcal{H}}
\newcommand{\iu}{i\mkern1mu}
\newcommand{\Dom}[1]{\mathcal{D}\left(#1\right)}
\newcommand{\ie}{\textit{i}.\textit{e}.}
\newcommand{\eg}{\textit{e}.\textit{g}.}
\newcommand{\strip}[1]{\mathcal{D}_{#1}}
\newcommand{\ip}[2]{\left\langle #1 , #2 \right \rangle}

\newcommand{\cchull}[1]{\overline{conv}\left(#1\right)}
\newcommand{\chull}[1]{conv\left(#1\right)}
\makeatletter
\newcommand*{\defeq}{\mathrel{\rlap{%
                     \raisebox{0.3ex}{$\m@th\cdot$}}%
                     \raisebox{-0.3ex}{$\m@th\cdot$}}%
                     =}
\makeatother

\DeclareMathOperator{\supp}{supp}
\DeclareMathOperator{\Ker}{Ker}
\renewcommand{\ker}[1]{\operatorname{Ker}\left(#1\right)}
\DeclareMathOperator{\Ran}{Ran}
\newcommand{\ran}[1]{\operatorname{Ran}\left(#1\right)}
\newcommand{\Tr}[0]{\mathrm{Tr}}
\let\Re\relax
\newcommand{\Re}[1]{\operatorname{Re}\left(#1\right)}
\let\Im\relax
\newcommand{\Im}[1]{\operatorname{Im}\left(#1\right)}
\newcommand{\sgn}[1]{\operatorname{sgn}\left(#1\right)}

\title{Noncommutative $L_p$-Spaces \\and\\ Perturbations of KMS States}
\author{Ricardo Correa da Silva}
\advisor{Prof. Dr. Jo\~ao Carlos Alves Barata}
\coadvisor{Prof. Dr. Christian Dieter J\"akel}

\dept{Institute of Physics}

\university{University of S\~ao Paulo}
\crest{\includegraphics[width=0.3\textwidth]{brasao_usp_cor}}

\renewcommand{\submissiontext}{This thesis is submitted for the degree of}

\degreetitle{Doctor of Sciences}

\college{}

\subject{Operator Algebras} 

\ifdefineAbstract
\fi

\newglossary{symbols}{sym}{sbl}{List of Symbols}
\newcolumntype{L}[1]{>{\let\\\arraybackslash\hspace{0pt}}p{#1}}
\newcolumntype{C}[1]{>{\centering\let\newline\\\arraybackslash\hspace{0pt}}p{#1}}
\newcolumntype{R}[1]{>{\raggedleft\let\newline\\\arraybackslash\hspace{0pt}}p{#1}}

\newglossarystyle{tabx3col}{%
	{\begin{longtable}{L{0.2\textwidth}L{0.6\textwidth}R{0.2\textwidth}}}%
		{\end{longtable}}%
	\renewcommand*{\glsgroupheading}[1]{}%
}
\glossarystyle{tabx3col}

\newglossaryentry{balgebra}{type=symbols,name=$\mathfrak{S}$,
	description={Banach algebra}}

\newglossaryentry{calgebra}{type=symbols,name=$\calgebra$,
	description={$C^\ast$-algebra}}

\newglossaryentry{nalgebra}{type=symbols,name=$\nalgebra$,
	description={von Neumann algebra}}

\newglossaryentry{projections}{type=symbols,name=$\nalgebra_p$,
	description={set of all orthogonal projections in $\nalgebra$}}

\newglossaryentry{nanalytic}{type=symbols,name=$\nanalytic$,
	description={set of analytic elements of a von Neumann algebra with respect to a one parameter automorphism group}}

\newglossaryentry{extr}{type=symbols,name=$\mathcal{E}(X)$,
	description={set of extremal points of $X$}}

\newglossaryentry{exp}{type=symbols,name=$\ex_p$,
	description={set of $p$-exponentiable elements}}

\newglossaryentry{hilbert}{type=symbols,name=$\hilbert$,
	description={Hilbert space}}

\newglossaryentry{strip}{type=symbols,name=$\strip{\gamma}$,
	description={the strip in the complex plane with \mbox{$0<\sgn{\gamma}\Im{z}< |\gamma|$}}}

\newglossaryentry{chull}{type=symbols,name=$\chull{X}$,
	description={convex hull of the set $X$}}

\newglossaryentry{weight}{type=symbols,name={$\phi$,$\psi$},
	description={weights on a $C^\ast$-algebra}}
	
\newglossaryentry{state}{type=symbols,name=$\omega$,
	description={states on a $C^\ast$-algebra}}
	
\newglossaryentry{trace}{type=symbols,name=$\tau$,
	description={trace on a $C^\ast$-algebra}}
	
\newglossaryentry{modaut}{type=symbols,name={$\tau^\phi$,$\tau_\Phi$},
	description={the modular automorphism group with respect to the weight $\phi$ or to the vector $\Phi$, respectively}}

\newglossaryentry{modDelta}{type=symbols,name=$\Delta_\Phi$,
	description={the modular operator with respect to the cyclic and separating vector $\Phi$}}
	
\newglossaryentry{modJ}{type=symbols,name=$J_\Phi$,
	description={the modular conjugation operator with respect to the cyclic and separating vector $\Phi$}}
	
\newglossaryentry{cones}{type=symbols,name=$V_\Phi^\alpha$,
	description={the cones}}

\newglossaryentry{Nphi}{type=symbols,name=$\mathfrak{N}_\Phi$,
	description={}}

\newglossaryentry{Fphi}{type=symbols,name=$\mathfrak{F}_\Phi$,
	description={}}

\newglossaryentry{Mphi}{type=symbols,name=$\mathfrak{M}_\Phi$,
	description={}}

\newglossaryentry{N2phi}{type=symbols,name=$N_\Phi$,
	description={}}

\newglossaryentry{fourierT}{type=symbols,name=$\hat{f}$,
	description={the Fourier transform of f}}

\newglossaryentry{affil}{type=symbols,name=$\nalgebra_\eta$,
	description={set of the operators affiliated to $\nalgebra$}}

\newglossaryentry{measurable}{type=symbols,name=$\nalgebra_\tau$,
	description={set of all $\tau$-measurable operators affiliated to $\nalgebra$}}

\newglossaryentry{posimeasurable}{type=symbols,name=$\nalgebra_\tau^+$,
	description={set of all positive $\tau$-measurable operators affiliated to $\nalgebra$}}

\newglossaryentry{realnumbers}{type=symbols,name=$\mathbb{R}_+$,
	description={set of the positive real numbers, which includes zero}}
	
\newglossaryentry{realextendedline}{type=symbols,name=$\overline{\mathbb{R}}$,
	description={the extended real line}}	

\makeglossaries

\begin{document}

\frontmatter

%\begin{titlepage}
%  \maketitle
%\end{titlepage}

% flatex input: [Titlepage/titleport.tex]
% ******************************************************************************
% ********************** MODELO DE FOLHA DE ROSTO - IFUSP **********************
% ******************************************************************************
% ******************************************************************************
%Alguns editores de LATEX podem ter problemas em compilar arquivos com nomes que estejam em outros idiomas %diferentes do %inglês e tenham caracteres especiais, como "ç" e a "~". Neste caso, renomeie os arquivos sem os %caracteres.
%\Some LATEX editors may present problems to compile files whose names are in other languages and have special %characters, %like "ç" and "~". In this case, rename the files without the characters.
%
%
%%%%%%%%%%%%%%%%%%%%%%%%%%%%%%%%%%%%%%%%%%%%%%%%%%%%%%%%%%%%%%%%%%%%%%%%%%%%%%%%%%%%
%Configuraçções de página e pacotes utilizados pelo template \ Page setup and packages used by the template

%\documentclass[a4paper,12pt,oneside]{book}
%\usepackage{graphicx}
%\usepackage[a4paper,top=3.5cm,left=3cm,right=3cm,bottom=2.5cm]{geometry}
%\usepackage{times} %pacote da fonte Times New Roman

\begin{titlepage}
	\newgeometry{a4paper,top=3.5cm,left=3cm,right=3cm,bottom=2.5cm}
	\pagestyle{empty}
	\begin{center}
		
		%%%%%%%%%%%%%%%%%%%%%%%%%%%%%%%%%%%%%%%%%%%%%%%%%%%%%%%%%%%%%%%%%%%%%%%%%%%%%%%%%%%%%    
		%\Título da Tese \ Thesis title
		{\fontsize{16}{16} \selectfont Universidade de S\~ao Paulo \\}
		\vspace{0.1cm}
		{\fontsize{16}{16} \selectfont Instituto de F\'{i}sica}
		\vspace{2.5cm}
		
		{\fontsize{22}{22}\selectfont Espa{\c c}os $L_p$ N\~ao-Comutativos \\ \vspace{0.2cm} e \\ \vspace{0.2cm} Perturba{\c c}\~oes de Estados KMS\par}
		\vspace{1.2cm}
		
		%%%%%%%%%%%%%%%%%%%%%%%%%%%%%%%%%%%%%%%%%%%%%%%%%%%%%%%%%%%%%%%%%%%%%%%%%%%%%%%%%%%%%    
		%\Nome do Autor \ Author's name
		
		{\fontsize{18}{18}\selectfont Ricardo Correa da Silva\par}
		
		\vspace{1.5cm}
		
	\end{center}
	
	%%%%%%%%%%%%%%%%%%%%%%%%%%%%%%%%%%%%%%%%%%%%%%%%%%%%%%%%%%%%%%%%%%%%%%%%%%%%%%%%%%%%%%    
	%\Orientador e coorientador (se existir) \ Supervisor and co-supervisor (if there is one)
	\leftskip 6cm
	\begin{flushright}	
		\leftskip 6cm
		Orientador: Prof. Dr. Jo\~ao Carlos Alves Barata
		\leftskip 6cm
		%Se não houve coorientador, comente a linha abaixo \ If there is no co-advisor, comment line below
		Co-orientador: Prof. Dr. Christian Dieter J\"akel
	\end{flushright}	
	
	\vspace{0.8cm}    
	
	%%%%%%%%%%%%%%%%%%%%%%%%%%%%%%%%%%%%%%%%%%%%%%%%%%%%%%%%%%%%%%%%%%%%%%%%%%%%%%%%%%%%%    
	% Grau Acadêmico \ Degree
	
	\par
	\leftskip 6cm
	\noindent {Tese de doutorado apresentada ao Instituto de F\'{i}sica da Universidade de S\~{a}o Paulo, como requisito parcial para a obten\c{c}\~{a}o do t\'{i}tulo de Doutor em Ci\^{e}ncias.}
	\par
	\leftskip 0cm
	\vskip 1.5cm
	
	%%%%%%%%%%%%%%%%%%%%%%%%%%%%%%%%%%%%%%%%%%%%%%%%%%%%%%%%%%%%%%%%%%%%%%%%%%%%%%%%%%%%    
	% Banca Examinadora -- Primeiro nome é do presidente ou do presidente da banca \ Examining committee -- The first name must be the supervisor's name or the examination committee president's name.
	
	\noindent Banca Examinadora: \\
	\noindent Prof. Dr. Jo\~ao Carlos Alves Barata - Orientador (IFUSP)\\
	Prof. Dr. C\'esar Rog\'erio de Oliveira (UFSCar)\\
	Prof. Dr. Domingos Humberto Urbano Marchetti (IFUSP)\\
	Prof. Dr Jean-Bernard Bru (BCAM)\\
	Prof. Dr. Walter Alberto de Siqueira Pedra (IFUSP)\\
	\vspace{1cm}

	%%%%%%%%%%%%%%%%%%%%%%%%%%%%%%%%%%%%%%%%%%%%%%%%%%%%%%%%%%%%%%%%%%%%%%%%%%%%%%%%%%%%    
	%Data \ Date
	\centering
	{S\~ao Paulo \\  2018}
	\clearpage
	
		\pagebreak
	
	\vspace*{240pt}
	\centering
	{\fontfamily{\sfdefault}\selectfont
		\fontsize{14pt}{16pt}{\bfseries{FICHA CATALOGR\'AFICA\\
				Preparada pelo Servi\c{c}o de Biblioteca e Informa\c{c}\~ao\\
				do Instituto de F\'isica da Universidade de S\~ao Paulo}}}\\
	\vspace{0.2cm}
	\fbox{\begin{minipage}{14cm}
			\vspace{12pt}
			
			{\fontfamily{\sfdefault}\selectfont
				\hspace{0.8cm}Silva, Ricardo Correa da
				
				\vspace{12pt}	
				
				\hspace{0.8cm}Espa\c{c}os Lp n\~ao-comutativos e perturba\c{c}\~oes de estados KMS /
				Noncommutative Lp-spaces and perturbations of KMS states.  S\~ao Paulo, 2018.
				
				\vspace{12pt}
				
				\hspace{0.8cm} Tese (Doutorado) - Universidade de S\~ao Paulo. Instituto de F\'isica. Depto. de F\'isica Matem\'atica.
				
				\vspace{12pt}
				
				\hspace{0.8cm}Orientador: Prof. Dr. Jo\~ao Carlos Alves Barata
				
				\hspace{0.8cm}\'Area de Concentra\c{c}\~ao: Matem\'atica Aplicada.
				
				\vspace{12pt}
				
				\hspace{0.8cm} Unitermos: 1. F\'isica matem\'atica; 2. Mec\^anica estat\'istica; 3. Mec\^anica qu\^antica; 4. M\'etodos matem\'aticos da f\'isica.
				
				\vspace{36pt}
				
				USP/IF/SBI-066/2018}
		\end{minipage}
	}
	
\end{titlepage}

\begin{titlepage}
	\newgeometry{a4paper,top=3.5cm,left=3cm,right=3cm,bottom=2.5cm}
	\pagestyle{empty}
	\begin{center}
		
		%%%%%%%%%%%%%%%%%%%%%%%%%%%%%%%%%%%%%%%%%%%%%%%%%%%%%%%%%%%%%%%%%%%%%%%%%%%%%%%%%%%%%    
		%\Título da Tese \ Thesis title
		{\fontsize{16}{16} \selectfont University of S\~ao Paulo \\}
		\vspace{0.1cm}
		{\fontsize{16}{16} \selectfont Physics Institute}
		\vspace{2.5cm}
		
		{\fontsize{22}{22}\selectfont Noncommutative $L_p$-Spaces \\ \vspace{0.2cm} and \\ \vspace{0.2cm} Perturbations of KMS States \par}
		\vspace{1.2cm}
		
		%%%%%%%%%%%%%%%%%%%%%%%%%%%%%%%%%%%%%%%%%%%%%%%%%%%%%%%%%%%%%%%%%%%%%%%%%%%%%%%%%%%%%    
		%\Nome do Autor \ Author's name
		
		{\fontsize{18}{18}\selectfont Ricardo Correa da Silva\par}
		
		\vspace{1.5cm}
		
	\end{center}
	
	%%%%%%%%%%%%%%%%%%%%%%%%%%%%%%%%%%%%%%%%%%%%%%%%%%%%%%%%%%%%%%%%%%%%%%%%%%%%%%%%%%%%%%    
	%\Orientador e coorientador (se existir) \ Supervisor and co-supervisor (if there is one)
	\leftskip 6cm
	\begin{flushright}	
		\leftskip 6cm
		Supervisor: Prof. Dr. Jo\~ao Carlos Alves Barata\\
		%Se não houver coorientador, comente a linha abaixo \ If there is no co-supervisor, comment below
		Co-supervisor: Prof. Dr. Christian Dieter J\"akel
	\end{flushright}	
	
	\vspace{0.8cm}    
	
	%%%%%%%%%%%%%%%%%%%%%%%%%%%%%%%%%%%%%%%%%%%%%%%%%%%%%%%%%%%%%%%%%%%%%%%%%%%%%%%%%%%%%    
	% Grau Acadêmico \ Degree
	
	\par
	\leftskip 6cm
	\noindent {Thesis submitted to the Physics Institute of the University of S\~ao Paulo in partial fulfillment of the requirements for the degree of Doctor of Science.}
	\par
	\leftskip 0cm
	\vskip 1.5cm
	
	%%%%%%%%%%%%%%%%%%%%%%%%%%%%%%%%%%%%%%%%%%%%%%%%%%%%%%%%%%%%%%%%%%%%%%%%%%%%%%%%%%%%    
	% Banca Examinadora -- Primeiro nome é do presidente ou do presidente da banca \ Examining committee -- The first name must be the supervisor's name or the examination committee president's name.
	
	\noindent Examining Committee: \\
	\noindent 	Prof. Dr. Jo\~ao Carlos Alves Barata - Supervisor (IFUSP)\\
	Prof. Dr. C\'esar Rog\'erio de Oliveira (UFSCar)\\
	Prof. Dr. Domingos Humberto Urbano Marchetti (IFUSP)\\
	Prof. Dr Jean-Bernard Bru (BCAM)\\
	Prof. Dr. Walter Alberto de Siqueira Pedra (IFUSP)\\
	\\
	\vspace{1cm}

	%%%%%%%%%%%%%%%%%%%%%%%%%%%%%%%%%%%%%%%%%%%%%%%%%%%%%%%%%%%%%%%%%%%%%%%%%%%%%%%%%%%%    
	%Data \ Date
	\centering
	{S\~ao Paulo \\  2018}
	\clearpage

\end{titlepage}

\iffalse
% ******************************* Thesis Title Page ***************************

\begin{titlepage}
    \begin{center}
    	
       \fontsize{16}{0}\selectfont
       University of S\~ao Paulo\\
       Physics Institute\\
        
        \vspace*{1cm}
        
        \fontsize{22}{0}\selectfont
        \textbf{Noncommutative $L_p$-Spaces \\ \vspace{0.2cm} and \\ \vspace{0.2cm} Perturbations of KMS States}

        \vspace{0.5cm} 
%        Thesis Subtitle
        
        \vspace{1.5cm}
        
        \fontsize{12}{0}\selectfont
        \textbf{Ricardo Correa da Silva}
                
        \vspace{3cm}
        
        \begin{flushright}
        \begin{minipage}{8.5cm}
        Supervisor: Prof. Dr. Jo\~ao Carlos Alves Barata\\
        Co-supervisor: Prof. Dr. Christian Dieter J\"akel\\
        \end{minipage}
		\end{flushright}
		
        \vspace{3cm}
        
        \noindent Examining Committee: \\
		Prof. Dr. Jo\~ao Carlos Alves Barata - Supervisor (IFUSP)\\
		Prof. Dr. C\'esar Rog\'erio de Oliveira (UFSCar)\\
		Prof(a). Dr(a). Domingos Humberto Urbano Marchetti (IFUSP)\\
		Prof. Dr Jean-Bernard Bru (BCAM)\\
		Prof. Dr. Walter Alberto de Siqueira Pedra (IFUSP)\\
        
        \begin{flushright}
        \begin{minipage}{8.5cm}     
		Thesis submitted to the Physics Institute of the University of S\~ao Paulo in partial fulfillment of the requirements for the degree of Doctor of Science.
        
        \end{minipage}
        \end{flushright}
        
		\vfill
		
		S\~ao Paulo, 2018
        
    \end{center}
\end{titlepage}
\fi
\restoregeometry
% flatex input end: [Titlepage/title.tex]

%\end{titlepage}
% flatex input: [Acknowledgement/acknowledgement.tex]
% ************************** Thesis Acknowledgements **************************

\begin{acknowledgements}      

I would like to express my gratitude to my advisors Prof. Dr. Christian J\"akel and Prof. Dr. Jo\~ao Carlos Alves Barata for trusting me such a wonderful research project, for their guidance during its development, and for all valuable advices over the last years.

I also want to thank Prof. Dr. Jean-Bernard Bru, Prof. Dr. Walter de Siqueira Pedra, and the Basque Center of Applied Mathematics the opportunity of visiting this fantastic research centre where I found a wonderful academic environment that also helped me improve some ideas on this thesis.

To all my friends, whose names I will not write here to avoid the risk of forgetting one, for all those fruitful discussions that ended up on a chalk board, those enjoyable moments around a coffee table, and for still having faith in me even when I didn't, I can't thank you enough.

To my family, the most reliable and supportive people in my life, I would never be able to finish this work without you.

Finally, thanks to CAPES for the financial support.

\end{acknowledgements}

% flatex input end: [Acknowledgement/acknowledgement.tex]

%\end{titlepage}
% flatex input: [Dedication/dedication.tex]
% ******************************* Thesis Dedidcation ********************************
\thispagestyle{empty}
\setlength\epigraphwidth{10cm}
\setlength\epigraphrule{0pt}
\epigraph{\vspace{15cm} \textit{A mathematician is a person who can find analogies between theorems; a better mathematician is one who can see analogies between proofs and the best mathematician can notice analogies between theories. One can imagine that the ultimate mathematician is one who can see analogies between analogies.}}{\textit{Stefan Banach}}

% flatex input end: [Dedication/dedication.tex]

%\end{titlepage}
% flatex input: [Abstract/abstract.tex]
% ************************** Thesis Abstract *****************************
% Use `abstract' as an option in the document class to print only the titlepage and the abstract.
\begin{abstract}
We extend the theory of perturbations of KMS states to some class of unbounded perturbations using noncommutative $L_p$-spaces. We also prove certain stability of the domain of the Modular Operator associated to a $\|\cdot\|_p$-continuous state. This allows us to define an analytic multiple-time KMS condition and to obtain its analyticity together with some bounds to its norm. The main results are Theorem \ref{TR0}, Theorem \ref{TR1} and Corollary \ref{CR1}.

Apart from that, this work contains a detailed review, with minor contributions due to the author, starting with the description of $C^\ast$-algebras and von Neumann algebras followed by weights and representations, a whole chapter is devoted to the study of KMS states and its physical interpretation as the states of thermal equilibrium, then the Tomita-Takesaki Modular Theory is presented, furthermore, we study analytical properties of the modular operator automorphism group, positive cones and bounded perturbations of states, and finally we start presenting multiple versions of noncommutative $L_p$-spaces.

\vspace{1cm}

\end{abstract}
% flatex input end: [Abstract/abstract.tex]

%\end{titlepage}
% flatex input: [Abstract/resumo.tex]
% ************************** Thesis Abstract *****************************
% Use `abstract' as an option in the document class to print only the titlepage and the abstract.
\begin{resumo}
	
Apresentamos uma extens\~ao da teoria de perturba\c{c}\~oes de estados KMS para uma classe de operadores ilimitados atrav\'es dos espa\c{c}os $L_p$ n\~ao-comutativos. Al\'em disso, provamos certa estabilidade do dom\'inio do Operador Modular de um estado \mbox{$\|\cdot\|_p$-cont\'inuo}, o que nos permite escrever a condi\c{c}\~ao KMS para tempos m\'ultiplos e obter sua analiticidade junto com majorantes para sua norma. Os principais resultados s\~ao o Teorema \ref{TR0}, o Teorema \ref{TR1} e o Corol\'ario \ref{CR1}.

Al\'em disso, nesse trabalho fazemos uma detalhada revis\~ao, com contribui{\c c}\~oes menores devidas ao autor, come{\c c}ando com uma descri{\c c}\~ao de \'algebras $C^\ast$ e \'algebras de von Neumann, seguida por pesos e representa{\c c}\~oes, um cap\'itulo inteiro \'e dedicado ao estudo de estados KMS e sua interpreta{\c c}\~ao como estados de equil\'ibrio t\'ermico, depois apresentamos a Teoria Modular de Tomita-Takesaki, al\'em disso, estudamos as propriedades de analiticidade do grupo de automorfismo modular, cones positivos e perturba{\c c}\~oes de estados e finalmente, come{\c c}amos a apresentar m\'ultiplas vers\~oes dos espa{\c c}os $L_p$ n\~ao-comutativos.

\vspace{1cm}
\noindent \textbf{Palavras-chave:} estados KMS, espa\c{c}os $L_p$ n\~ao-comutativos, perturba\c{c}\~oes ilimitadas.
\end{resumo}

% flatex input end: [Abstract/resumo.tex]

%\end{titlepage}

% ****************** Adding TOC and List of Figures *********************

\tableofcontents

%\listoffigures

%\listoftables

% \printnomencl[space] space can be set as 2em between symbol and description
%\printnomencl[3em]

\printnomencl

% *************************** Main Matter *****************************
\mainmatter

\renewcommand\thechapter{P}
% flatex input: [Introduction/introduction.tex]
\chapter{Preamble}

% ************************* Introduction *************************

\section{Motives and Objectives}

\textit{Why this thesis?} The answer to this question demands a discussion about the meaning of several mathematical structures and their historical origin.

The first appearance of the term ``von Neumann algebra'' occurred in Dixmier's book ``Les Alg\'ebres d'Op\'erateurs dans l'Espace Hilbertien'' in 1957, following a suggestion by Dieudonn\'e, but the first appearance of such structure goes back to von Neumann article of 1930 \cite{Neumann1930} where he proved his fundamental Double Commutant Theorem.

His studies in quantum mechanics began when von Neumann started as Hilbert's assistant in 1926, Hilbert was very interested in mathematical formalization of physics, as can be seen in his sixth problem. He gave lectures on mathematical foundations of Quantum Mechanics that year, von Neumann attended the lectures and published a article in 1927 based on them, his first article on the subject. Hilbert strongly influenced von Neumann to work in both: formalization of Quantum Mechanics and Mathematical Logic. It is important to mention that it was von Neumann who gave the abstract definition of Hilbert spaces as known today and it was he who formulated the eigenvalue problem for self-adjoint operators in therms of spectral measures, solving the problems encountered in Dirac's work. 

Next, he started the study of a structure called, at that time, rings of operators. It was a natural consequence of von Neumann work on rigorous mathematical formalization of Quantum Mechanics, since rings of operators are closely related to observables as seen in Definition \ref{DCCA} and \ref{DCvNA}.

After that, von Neumann algebras were left out, reappearing only five years later in a series of articles by Murray and von Neumann \cite{Murray1936}, \cite{Murray1937}, \cite{Neumann1940}, and \cite{Murray1943}, giving rise to the classification theory of von Neumann algebras. Interestingly, at first just the type I von Neumann algebras seemed to be of physical interest while type III were almost mystical.  In fact, it wasn't even known if such an algebra existed, as its existence was proven only in 1940 in \cite{Neumann1940}. Nowadays, however, type III von Neumann algebras became the type with most physical significance, especially to Quantum Statistical Mechanics and Quantum Field Theory.

The core of classification theory of von Neumann algebras are the projections, as will be presented in Theorem \ref{projectionsLattice}, which form a complete lattice. This property was one of the clues used by Murray to propose an order relation for the set of projections, and to prove a Cantor-Schr\"oder-Bernstein type theorem.

Probably the more powerful tool in $C^\ast$-algebras is the functional calculus, a result of ideas and efforts of Stone and von Neumann started in an article in 1937 treating the subject (just for von Neumann algebras at that time). This new ideas also allow to treat unbounded operator, which are vital to the study of Quantum Mechanics. The ultimate result on this direction was made by Gelfand and Naimark in 1943 where they proved Theorem \ref{TGN} using the ideas developed by Stone, Murray and von Neumann in the above mentioned articles.

At that point the theory had slightly moved away from physics, but the work of Haag, Hugenholtz and Winnink \cite{haag67} changed that. Problems in quantum statistical mechanics were treated with a density matrix, \index{density matrix} but the limitations of this approach are evident: density matrices express the expectation value of an observable trough a trace, $\langle A \rangle=Tr(\rho A)$ where $\rho$ is the density matrix. Since this trace must be normalized, it only makes sense if the matrix has a $\ell_1$ sequence on its diagonal (remember we are supposing that the Hilbert space is separable), but then the condition of maximizing entropy in equilibrium fails, since the Lagrange multiplier method gives us a uniform probability. Another problem is that it is known that, in a box, the spectrum of an operator like the Hamiltonian typically appears in a countable number, but when the limit for infinit boxes is taken, the operator produces a continuum spectrum and the density matrix formalism crumbles completely. These problems forbid us to formulate Gibb's equilibrium states directly in thermodynamical limit.

The main contribution of Haag, Hugenholtz and Winnink, that will be discussed in Chapter \ref{chapKMS}, was to describe the precise condition of equilibrium of a system with infinite degrees of freedom, known as KMS condition, see reference \cite{haag67}. The KMS condition is a result of an independent work of Kubo, who introduced it, and Martin and Schwinger who defined thermodynamical Green functions. The KMS condition is central in the study of von Neumann algebras and thermodynamical equilibrium and will be central in this work.

Another central contribution in this direction was made by Tomita and Takesaki. Takesaki attended to a conference in operator algebras and applications where Haag, Hugenholtz and Winnink presented a seminar about the work cited above and Tomita presented a work answering an open question about the relation between an algebra and its commutant. Takesaki put the ideas together and created the Tomita-Takesaki modular theory. Of course it has physical significance since it is related with equilibrium states. Tomita-Takesaki Modular Theory proved to be a powerful tool.

Finally, for physical reasons, we expect some ``good behaviour'' from equilibrium states under perturbations. Araki has proved that for a bounded perturbation this is true, but not much has been done in the last 35 years, the only work was \cite{Derezinski03}. Our proposal is this thesis is to understand perturbations on KMS states in the light of noncommutative $L_p$ spaces.

\section{Thesis Organization}

This work is going to present a long revision because there is no text book with all necessary subjects and results to define modular theory, perturbation of KMS states, and noncommutative $L_p$ spaces accordingly (probably \cite{Takesaki2002} and \cite{Takesaki2003} are the closest to this aim). It is indispensable, for our goal, to define a coherent language and to understand all the connections between previous results. Moreover, while the evolution of important physical and mathematical concepts has been discussed above, more than a few words are usually required to make concepts and connections clear. Therefore, a more extensive analysis may be presented in some cases.

Below, a more detailed explanation of contents in each chapter:

Chapter \ref{chapCW}: this chapter is devoted to theory of $C^\ast$-algebras and $W^\ast$-algebras. After the first definitions we present the advances in the theory in order to prove the GNS-Representation Theorem and Functional Calculus, this course pass through states and weights, which are fundamental concepts by themselves.

Chapter \ref{chapKMS}: here we present the definition of dynamical systems, analytic elements and equivalences and important physical and mathematical proprieties of the so called KMS states. The majority of the presentation of this chapter can be found in \cite{Bratteli1} and \cite{haag67}.

Chapter \ref{chapTTMT}: we present the construction of both Modular and Relative Modular theory. Despite the first has many good presentations in articles and books, \eg \cite{Araki74}, \cite{Blackadar2006}, \cite{Bratteli1}, \cite{Dou72}, \cite{Takesaki2002}, the second one is difficult to find and it is presented with slightly different definitions, equivalent at least on the context. We also show useful Modular Operator proprieties involving analyticity, positive cones and dual cones. Since Relative Modular Operator is in the core of the definition of noncommutative $L_p$ spaces, we will present the construction using weights. We also present in this chapter the noncommutative Radon-Nikodym theorem, which is central to our results, in an direct organized way and in all details just with some minor modifications to help comprehension and to suit the language of this thesis, the original presentation can be found in \cite{Pedersen73} and \cite{sakai65}.

Chapter \ref{chapNCLp}: in this chapter we present the definition and theory of noncommutative $L_p$ spaces with results from an extensive list of references, with some different proofs due to the author.

Chapter \ref{chapExtensionPerturb}: this chapter contains the results of this project, which means, our extension of the theory of perturbations which includes unbounded perturbations. It starts discussing a property of bounded perturbations that does not hold for unbounded ones. A less restrictive property is proposed and examples of unbounded perturbation satisfying this property are given. Finally, we prove a inequality that allows us to state some conditions under which Dyson's series converges.
% flatex input end: [Introduction/introduction.tex]

% *************************** Main Matter *****************************
\renewcommand{\thechapter}{\arabic{chapter}}
\setcounter{chapter}{0}
%\include{Chapter1/chapter1}
% flatex input: [Chapter2/chapter2.tex]

%**************************** Second Chapter *****************************

\chapter{$C^\ast, W^\ast$-Algebras, States, Weights, Representations, and all that.}
\label{chapCW}

We start this chapter with well-known definitions and properties. Those definitions motivate the study of weights, which is much less common in the standard literature, as well as the GNS construction for weights. In fact, much of the technical steps in this chapter were written by the author and some other proofs were adapted or recreated in order to present a self-contained and linear development of the theory. 
%*********************** %Zero Section  ********************************
\section{Basic Definitions in $C^\ast, W^\ast$-Algebras }

\begin{definition}
A Banach $\ast$-algebra \index{algebra! Banach} \gls{balgebra} is a structure $(\mathfrak{B},\mathbb{K},\cdot,+,\circ,\ast,\|\cdot\|)$ consisting of a set $\mathfrak{B}$; a field $\mathbb{K}=\mathbb{R}, \mathbb{C}$; three binary operations: a scalar product $\cdot:\mathbb{K}\times \mathfrak{B} \to \mathfrak{B}$, a vector sum $+:\mathfrak{B}\times\mathfrak{B}\to\mathfrak{B}$ and an associative product $\circ:\mathfrak{B}\times \mathfrak{B} \to \mathfrak{B}$; an involution $\ast:\mathfrak{B} \to \mathfrak{B}$; and a norm $\|\cdot\|:\mathfrak{B} \to \mathbb{R}_+$, such that $(\mathfrak{B},\mathbb{K},\cdot,+,\|\cdot\|)$ is a Banach space over the field $\mathbb{K}$ and  $(\mathfrak{B},\mathbb{K},\cdot,+,\circ)$ is an associative algebra over the field $\mathbb{K}$ with the additional relation between the norm and the product
$$\|A\circ B\|\leq \|A\|\|B\| \qquad \forall A,B \in \mathfrak{B}.$$
\end{definition}

\begin{notation}
We used the notation $A\circ B$ to denote the multiplication but, as usual, we will omit this sign and write $AB$ instead of $A\circ B$.
\end{notation}

Note that we related the product to the norm, which ensures the continuity of the multiplication since, for a fixed $A, B\in \mathfrak{B}$ and $0<\epsilon<1$ given, take ${\delta=\frac{\epsilon}{2\left(\|B\|+1\right)\left(\|A\|+1\right)}}$, then $\forall A_1\times B_1 \in B_\delta(A)\times B_\delta(B)$ we get 
$$\begin{aligned}
\|AB-A_1B_1\|	&=\|AB-A_1B+A_1B-A_1B_1\|\\
				&\leq \|A-A_1\|\|B\|+\|A_1\|\|B-B_1\|\\
				&\leq \delta \|B\|+\|A_1\|\delta\\
				&\leq \delta\|B\|+\left(\|A\|+\delta\right)\delta\\
				&< \epsilon.
\end{aligned}$$

To this point no topological condition was required over the involution.

There are two interesting conditions available to impose over $\ast$:
\begin{flalign}
\label{Ccondition1}
\hspace{5.5 cm} &\|A^\ast A\|=\|A\|^2, &(B^\ast \text{ condition})\\
\label{Ccondition2}
\hspace{5.5 cm} &\|A^\ast A\|=\|A^\ast\|\|A\|. &(C^\ast \text{ condition})
\end{flalign} 

Of course a Banach algebra that satisfies the $B^\ast$ condition also meets the $C^\ast$ one, because due to the submultiplicative condition we have $\|A\|^2=\|A^\ast A\|\leq \|A^\ast\|\|A\| \Rightarrow \|A\|\leq\|A^\ast\|$, but with the substitution $A\mapsto A^\ast$ we obtain the opposite inequality, thus equality. The converse is true but non-trivial: a Banach $\ast$-algebra satisfying the $C^\ast$-condition also satisfies the $B^\ast$-condition. A stronger result which states that a complete normed $\ast$-algebra with continuous involution satisfies the $B^\ast$-condition that can be found in \cite{ArakiElliott73}. An alternative reference that includes an discussion on the topic is \cite{doran86}.
% See History of Banach Spaces and Linear Operatorsby A. Pietsch section 4.10.3

\begin{definition}
A (abstract) $C^\ast$-algebra \index{algebra! $C^\ast$} \gls{calgebra} is a Banach $\ast$-algebra that satisfies the $C^\ast$-condition. We say that $\calgebra$ is unital \index{algebra! unital} if it has an identity, \ie, and element $\mathbbm{1} \in \calgebra$ such that $\mathbbm{1}A=A\mathbbm{1}$ for every $A\in\calgebra$.
\end{definition} 

Since equations \eqref{Ccondition1} and \eqref{Ccondition2} are equivalent, we are going to use the first one, \ie, the $B^\ast$-condition.

\begin{proposition}[Polarization Identity]
Let $\calgebra$ be a $C^\ast$-algebra and $A, B \in \calgebra$. Then
\begin{equation}
\label{PolIden}
AB^\ast=\frac{1}{4}\sum_{n=0}^{3}\iu^n (A+\iu^n B)(A+\iu^n B)^\ast
\end{equation}
\end{proposition}

The previous proposition is a useful result, specially because we will see later it means that, in a unital $C^\ast$-algebra, every operator is a linear combination of four positive elements. Going a little further, we will see it is the case even when the $C^\ast$-algebra has no unit. 

\begin{notation}
	Throughout this work we will denote by $\hilbert$ a Hilbert space over $\mathbb{C}$ and by $\calgebra$ a von Neumann algebra acting on $\hilbert$.
\end{notation}

\begin{example}
	The first example of a $C^\ast$-algebra is $B(\hilbert)$, the set of all bounded operators on a Hilbert space $\hilbert$,  provided with the usual addition, scalar multiplication, adjoint operation, product, and norm.
	
	Of course, any norm closed $\ast$-subalgebra of $B(\hilbert)$ is another example, and, as we will see later, these are the only examples.
\end{example}

This motivates the following definition.

\begin{definition}
\label{DCCA}
A (concrete) $C^\ast$-algebra is a norm closed $\ast$-invariant sub-algebra of $B(\hilbert)$, where $\hilbert$ is a Hilbert space.
\end{definition}

One of the aims of this chapter is to prove that the abstract and the concrete definitions are equivalent.

\begin{definition}
Let $\hilbert$ be a Hilbert space, $B(\hilbert)$ the set of all bounded operators on $\hilbert$ provided with the usual addition, scalar multiplication, adjoint operation and product. The following families of seminorms define five vector space topologies on $B(\hilbert)$:
\begin{enumerate}[(i)]
\item $\rho(A)=\|A\|$; \hfill (norm or uniform topology)
\item $\rho_x(A)=\|Ax\|$, $x\in \hilbert$; \hfill \index{topology! strong operator} (strong operator topology)
\item $\rho_{x,y}(A)=|\ip{Ax}{y}|$, $x, y\in \hilbert$; \hfill \index{topology! weak operator} (weak operator topology) 
\item $\displaystyle  \rho_{(x_n)_{n\in\mathbb{N}}}(A)=\left(\sum_{n=1}^\infty\|Ax_n\|^2\right)^\frac{1}{2}$, $\displaystyle (x_n)_{n\in\mathbb{N}}\in \ell_2\left(\hilbert\right)$; \hfill \index{topology! ultra-strong operator} (ultra-strong operator topology)
\item $\displaystyle  \rho_{(x_n)_{n\in\mathbb{N}},(y_n)_{n\in\mathbb{N}}}(A)=\left(\sum_{n=1}^\infty|\ip{Ax_n}{y_n}|\right)^\frac{1}{2}$, $\displaystyle (x_n)_{n\in\mathbb{N}},(y_n)_{n\in\mathbb{N}}\in \ell_2\left(\hilbert\right)$.

\hfill \index{topology! ultra-weak operator} (ultra-weak operator topology)

\end{enumerate}
\begin{notation}
	We will use the acronyms SOT and WOT for strong operator topology and weak operator topology, respectively.
\end{notation}
\end{definition}

\begin{definition}
Let $\calgebra$ be a concrete $C^\ast$-algebra. We define the \index{commutant} commutant of $\calgebra$ to be the set
$$\calgebra^\prime=\left\{A^\prime \in B(\hilbert) \ \middle| \ AA^\prime=A^\prime A \quad \forall A\in \calgebra\right\}.$$ 
\end{definition}

We will skip the proof of von Neumann's Double Commutant Theorem, since it is a classical result which can be found in \cite{Bratteli1} and \cite{KR86}.
%******
\begin{theorem}[Double Commutant Theorem]\index{theorem! double commutant}]
\label{TDC}
Let $\algebra$ be a unital $\ast$-invariant subalgebra of $B(\hilbert)$. The following are equivalent:
\begin{enumerate}[(i)]
	\item $\algebra=\algebra^{\prime \prime}$;
	\item $\algebra$ is weak-operator closed;
	\item $\algebra$ is strong-operator closed.
	\item $\algebra$ is ultra-weak operator closed;
	\item $\algebra$ is ultra-strong operator closed.
\end{enumerate}
\end{theorem}

\begin{definition}
\label{DAvNA}
A (abstract) von Neumann algebra \index{algebra! von Neumann} \gls{nalgebra} is a $C^\ast$-algebra $\nalgebra$ which has a pre-dual, \ie, there exists a Banach space $\nalgebra_\ast$ such that $\left(\nalgebra_\ast\right)^\ast=\nalgebra$ as a Banach space.
\end{definition}

\begin{definition}
	\label{DCvNA}
A concrete von Neumann algebra is a weak-operator closed \mbox{$\ast$-invariant} subspace of $B(\hilbert)$, where $\hilbert$ is a Hilbert space.
\end{definition}

One of the central ingredients in von Neumann algebras, as noticed by Murray and von Neumann himself\footnote{see \cite{Glimm2006}.}, are the (orthogonal) projections. The following result is the first in this direction and it is of central importance.

We warn the reader that we will usually use the word projection meaning orthogonal projection.

\begin{theorem}
\label{projectionsLattice}\glsdisp{projections}{\hspace{0pt}}
%Takesaki vol I prop. 1.1
Let $\nalgebra$ be a von Neumann algebra. Then $$\nalgebra_p\doteq\left\{P\in \nalgebra \ \middle| \ P=P^\ast=P^2\right\}$$ (it means, the set of all orthogonal projections of $\nalgebra$) with its natural order given by $P\leq Q \Leftrightarrow QP=P$, is a complete lattice.
\end{theorem}
\begin{proof}
Without loss of generality (Theorem \ref{TGN}), consider $\nalgebra$ as a concrete von Neumann algebra, that is, a subalgebra of $B(\hilbert)$ for some Hilbert space $\hilbert$.
Let $\{P_i\}_{i\in I}$ be any family of projections in $\nalgebra$ and define $\displaystyle\bigwedge_{i\in I}P_i$ to be the orthogonal projection onto the closed subspace $\displaystyle \bigcap_{i\in I} P_i(\hilbert)$. Of course, $\displaystyle\bigwedge_{i\in I}P_i \in B(\hilbert)$. Furthermore, $\displaystyle\bigwedge_{i\in I}P_i \in B(\hilbert)$ is the greatest lower bound of the family $\{P_i\}_{i\in I}$, for it just noticing that $\displaystyle P_j\left(\hilbert\right) \subset \bigcap_{i\in I}P_i\left(\hilbert\right) \Rightarrow P_j\bigwedge_{i\in I}P_i=\bigwedge_{i\in I}P_i$ for any $j\in I$ and, for any other $Q\leq P_i \ \forall i\in I$ we must have $P_i Q=Q \ \forall i \in I \Rightarrow Q\left(\hilbert\right) \subset \displaystyle \bigcap_{i\in I} P_i(\hilbert)$ and thus $\displaystyle\bigwedge_{i\in I}P_i Q=Q$ which, by definition, means that $Q\leq\displaystyle\bigwedge_{i\in I}P_i$.

Notice that we just showed that any family of projections has a greatest lower bound in $B(\hilbert)$, but let $A\in \nalgebra^\prime$, by definition $AP_i=P_iA$, thus $A\left(P_i\hilbert\right)\subset P_i\left(\hilbert\right) \ \forall i \in I \Rightarrow A \displaystyle\bigwedge_{i\in I}P_i=\displaystyle\bigwedge_{i\in I}P_i A \Rightarrow \bigwedge_{i\in I}P_i \in \nalgebra^{\prime\prime}=\nalgebra$.

For the least upper bound just define $\displaystyle\bigvee_{i\in I}P_i=\mathbbm{1}-\bigwedge_{i\in I}\left(\mathbbm{1}-P_i\right) $.

\end{proof}

As we have mentioned, projections play a central role in understanding von Neumann algebras. We will see later that we can classify von Neumann algebras through them. Let us define a special class of von Neumann algebras with interesting properties related to states and cyclic vectors.

\begin{definition}
Let $\mathfrak{A}$ be a $\ast$-algebra. Two projections $P,Q\in \mathfrak{A}$ are pairwise orthogonal \index{orthogonal projection} if $PQ=0$.
\end{definition}

\begin{definition}
Let $\nalgebra$ be a von Neumann algebra. $\nalgebra$ is said to be $\sigma$-finite \index{algebra} if every set of non-zero pairwise orthogonal projections is at most countable.
\end{definition}

%*********************** %Fist Section  *********************************
\section{Positive Linear Functionals, States and Weights.}
To start this section we need the definition of a positive operator. We will devote Section \ref{SecSRoO} to this subject, for now, we will use the following definition: 

\begin{definition}
	\label{defpositive}\index{operator! positive}
Let $\calgebra$ be a unital $C^\ast$-algebra. $A\in\calgebra$ is said to be positive if $A=0$ or $A$ is self-adjoint and satisfies
$$\left\|\mathbbm{1}-\frac{A}{\|A\|}\right\|\leq1.$$
We denote by $\calgebra_+$ the set of positive operators in the algebra $\calgebra$.
\end{definition}
The definition above is not the most used in textbooks. In general, equivalent definitions that we will present soon in Theorem \ref{sqrt} or in Proposition \ref{EquiPositive} are more often used.
\begin{definition}[Positive Linear Functionals and States]
Let $\calgebra$ be a unital $C^\ast$-algebra
\begin{enumerate}[(i)]

\item  A linear functional $\omega$  on $\calgebra$ is said to be positive \index{positive functional} if
$$\omega(A)\geq 0 \quad \forall A\in \calgebra_+;$$

\item a positive linear functional \gls{state} on $\calgebra$ is said to be a state \index{state} if
$$\|\omega\|\doteq \sup_{\|A\|\leq 1} |\omega(A)|=1.$$
\end{enumerate}
\end{definition}

\begin{remark}
These definitions can be extended to non-unital $C^\ast$-algebras just adding a unity to it. If $\calgebra$ is unital we have $\omega(\mathbbm{1})=1$, because if $\omega$ is a state on $\calgebra$, $\|\omega\|=1 \Leftrightarrow \omega(\mathbbm{1}) =1$. Continuity is not required in the definition because it is a consequence of positiveness.

\end{remark}

\begin{proposition}[Cauchy-Schwarz inequality]
	\label{cauchyschwarz}
	Let $\calgebra$ be a $C^\ast$-algebra and $\omega$ a state. Then
	$$\begin{aligned}
	\omega(A^\ast B)	&=\overline{\omega(B^\ast A)};\\
	\left|\omega(A^\ast B)\right|^2 &\leq \omega(A^\ast A)\omega(B^\ast B).
	\end{aligned}$$
\end{proposition}
\begin{proof}
	It is a consequence of positivity that, for all $\lambda\in\mathbb{C}$,
	$$0\leq \omega\big((A+\lambda B)^\ast(A+\lambda B)\big)=\omega(A^\ast A)+\overline{\lambda}\omega(B^\ast A)+\lambda\omega(A^\ast B)+|\lambda|^2\omega(B^\ast B).$$
	
	In particular, for $\lambda=1,\iu$, we conclude that $\Im{\omega(A^\ast B)}=-\Im{\omega(B^\ast A)}$ and $\Re{\omega(A^\ast B)}=\Re{\omega(B^\ast A)}$, respectively. Hence $\omega(A^\ast B)=\overline{\omega(B^\ast A)}$ and
	$$0\leq \omega(A^\ast A)+\overline{\lambda\omega(A^\ast B)}+\lambda\omega(A^\ast B)+|\lambda|^2\omega(B^\ast B).$$
	
	If $\omega(B^\ast B)=0$ the inequality is trivial. For $\omega(B^\ast B)\neq0$, choose $\lambda=-\frac{\overline{\omega(A^\ast B)}}{\omega(B^\ast B)}$. It follows that
	$$0\leq \omega(A^\ast A)-\frac{|\omega(A^\ast B)|^2}{\omega(B^\ast B)}.$$
	
\end{proof}

\begin{proposition}
\label{PPC}
Let $\omega$ be a linear functional on a unital $C^\ast$-algebra $\calgebra$. The following are equivalent
\begin{enumerate}[(i)]
\item $\omega$ is positive;
\item $\omega$ is continuous and $\|\omega\|=\omega(\mathbbm{1})$.
\end{enumerate} 
\end{proposition}
\begin{proof}
$(i)\Rightarrow(ii)$ First, let us prove that
$$\left\{\omega(A) \in \mathbb{R} \ \middle| \ A\geq 0, \|A\|=1 \right\}$$ is bounded by some constant $M>0$. In fact, if this is not true, there exists a sequence of positive elements $A_n$ with $\|A_n\|=1$, such that $\omega(A_n)>n 2^n$. But then the operators $\displaystyle B_m=\sum_{n=1}^{m} 2^{-n} A_n$ are norm convergent to some positive element $B$ and it follows from positivity that $$\sum_{n=1}^m n<\omega(B_m)\leq \omega(B) \quad \forall m\in \mathbb{N}.$$
It is a contradiction. Hence, we can define
$$M=\sup\left\{\omega(A) \in \mathbb{R} \ \middle| \ A\geq 0, \|A\|=1 \right\}.$$

It follow from the polarization identity, equation \eqref{PolIden}, that, for each $\displaystyle A\in \calgebra$ with $\|A\|=1$, there exists $A_i\in \calgebra_+$ with $\|A_i\|\leq 1$ and $\displaystyle A=\sum_{j=0}^{3}\iu^jA_j$. Thus, $|\omega(A)|\leq 4M$, which means $\omega$ is continuous.

Furthermore, $\|\omega\|\geq \omega(\mathbbm{1})$. On the other hand, by the Cauchy-Schwarz inequality
$$|\omega(A)|^2= |\omega(A\mathbbm{1})|^2\leq |\omega(A)| |\omega(\mathbbm{1})| \Rightarrow |\omega(A)|\leq |\omega(\mathbbm{1})|.$$

$(ii)\Rightarrow(i)$ Dividing by its norm, we can suppose $\|\omega\|=1$, and by hypothesis, $\omega(\mathbbm{1})=1$.

Now, by Definition \ref{defpositive}, every norm-one positive element $A$ satisfies $\left\|\mathbbm{1}-A\right\|\leq 1$.
Then, $\omega(\mathbbm{1}-A)\leq 1$ and it follows that 
$$\omega(A)\geq 0 \quad \forall A\in\calgebra_+.$$

\end{proof}

We will use the end of this section to define the adjoint-functional.

\begin{definition}
	\label{defAdjoinfFunc}
	Let $\calgebra$ be a $C^\ast$-algebra and let $\omega$ be a linear functional on $\calgebra$. We define the adjoint $\omega^\ast$ of $\omega$ by
	$$\omega^\ast(A)=\overline{\omega(A^\ast)}, \quad A\in\calgebra.$$
\end{definition}

It will become clear, once we prove that all positive operators are of the form $A^\ast A$, that any positive linear functional is self-adjoint.

%*************************** Second Section  ******************************

\section{Weights}

Weights are, in some way, a natural generalization of positive linear functionals where the co-domain is the positive extended real line $\overline{\mathbb{R}}=[0,+\infty]$. \glsdisp{realextendedline}{\hspace{0pt}}\footnote{We are using the convention $0.\infty=0$.}

\begin{definition}\glsdisp{weight}{\hspace{0pt}}
A weight\index{weight} on a $C^*$ -algebra $\calgebra$ is a function $\phi:\calgebra_+ \to \overline{\mathbb{R}}_+$ such that
\begin{enumerate}[(i)]
\item $\phi(\lambda A)=\lambda \phi(A) \quad \forall A\in \calgebra_+$, $\forall \lambda \geq 0$;
\item $\phi(A+B)=\phi(A)+\phi(B)  \quad \forall A, B\in \calgebra_+.$
\end{enumerate}
\end{definition}

It is important to stress that we use the convention $\infty\cdot0=0$. 

\begin{definition}
\label{weightsets}
Let $\phi$ be a weight on a $C^\ast$-algebra $\calgebra$, we define:
\begin{enumerate}[(i)]
\item $ \displaystyle \mathfrak{N}_\phi=\mathfrak{D}^2_\phi=\{A \in \calgebra \ | \ \phi(A^\ast A)<\infty\}$; \glsdisp{Nphi}{\hspace{0pt}}
\item $ \displaystyle \mathfrak{F}_\phi=\{A \in \calgebra_+ \ | \ \phi(A)<\infty\}$; \glsdisp{Fphi}{\hspace{0pt}}
\item $ \displaystyle \mathfrak{M}_\phi=\mathfrak{D}^1_\phi=span\left[\mathfrak{F}_\phi \right]$; \glsdisp{Mphi}{\hspace{0pt}}
\item $ \displaystyle N_\phi= \{A \in \calgebra \ | \ \phi(A^\ast A)=0\}$. \glsdisp{N2phi}{\hspace{0pt}}
\end{enumerate}
\end{definition}

Some important definitions and results for a thorough understanding of the following, such as the Krein-Milman Theorem\footnote{see Theorem \ref{TKM}.}, are left to the Appendix \ref{AppKMT}.

\begin{proposition}
%Kadison 7.5.2 vol2
\label{simplepropweights}
The following properties hold:
\begin{enumerate}[(i)]
\item $ \displaystyle \mathfrak{F}_\phi \subset \mathfrak{N}_\phi\cap \mathfrak{N}_\phi^\ast $;
\item $ \displaystyle \mathfrak{F}_\phi$ is a face of $\calgebra_+$;
\item  $ \displaystyle \mathfrak{N}_\phi$ and $ \displaystyle N_\phi$ are left ideals of $\calgebra$;
\item $ \displaystyle \mathfrak{M}_\phi=\mathfrak{N}_\phi^\ast \mathfrak{N}_\phi=span\left[\{A^\ast A | A \in \mathfrak{N}_\phi\}\right]$;
\item $\displaystyle \mathfrak{M}_\phi$ is a $\ast$-subalgebra of $\calgebra$ and $\mathfrak{M}_{\phi+}=\mathfrak{M}_\phi\cap \calgebra_+=\mathfrak{F}_\phi$.
\end{enumerate}
\end{proposition}
\begin{proof}
$(i)$ It is obvious that $\mathfrak{F}_\phi \subset \mathfrak{N}_\phi$, since for each $A\in \mathfrak{F}_\phi$, $A^\ast A=AA\leq \|A\|A$ hence
$$0\leq \phi(\|A\|A-A^\ast A)=\|A\|\phi(A)-\phi(A^\ast A) \Rightarrow \phi(A^\ast A)\leq \|A\|\phi(A)<\infty.$$
Using now the positiveness, we also get $\mathfrak{F}_\phi \subset \mathfrak{N}_\phi^\ast$

$(ii)$ Trivial.

$(iii)$ Notice that, for all $A, B\in \calgebra$, $B^\ast A^\ast A B\leq \|A\|^2 B^\ast B$ because, if we call $D$ the unique positive square root of $\|A\|^2-A^\ast A$ which exists due to Theorem \ref{sqrt} and $\|A\|^2\mathbbm{1}\geq A^\ast A$, we have
$$\|A\|^2 B^\ast B-B^\ast A^\ast A B=B^\ast(\|A\|^2-A^\ast A)B=(DB)^\ast (DB)\geq0.$$
$$\Rightarrow \phi\big((AB)^\ast AB\big)=\phi(B^\ast A^\ast A B)\leq \|A\|^2 \phi(B^\ast B).$$
%Rieffel notes 208.pdf

Now, it is easy to see that when $A\in \calgebra$ and $B\in \mathfrak{N}_\phi$, we must have $AB\in \mathfrak{N}_\phi$, moreover, $\lambda \in \mathbb{C}, A\in \mathfrak{N}_\phi \Rightarrow \lambda A \in \mathfrak{N}_\phi$. Finally, 
$$(A+B)^\ast (A+B)+(A-B)^\ast (A-B)=2A^\ast A +2 B^\ast B.$$
Hence, if we take $A,B \in \mathfrak{N}_\phi$, we must have
$$\infty>\phi\left(2A^\ast A+2B^\ast B\right)=\phi\left((A+B)^\ast (A+B)+(A-B)^\ast (A-B)\right)\geq\phi\left((A+B)^\ast (A+B)\right),$$
from which it follows that $A+B \in \mathfrak{N}_\phi$.

The analogous properties for $N_\phi$ are trivial.

$(iv)$ The polarization identity states that, for $A,B \in \mathfrak{N}_\phi$,
$$B^\ast A=\frac{1}{4} \sum_{n=0}^{3}\iu^n(A+i^n B)^\ast(A+\iu^n B),$$
but the proof of $(iii)$ above says that $ A+i^n B \in \mathfrak{N}_\phi$, $0\leq n\leq 3$, and therefore the conclusion holds.

$(v)$ $\mathfrak{M}_\phi$ is closed with respect to $\ast$ due to Proposition \ref{cauchyschwarz} $(i)$. From the definition, every $A\in \mathfrak{M}_\phi$ can be written as a linear combination $A=A_1-A_2+\iu A_3-\iu A_4$ with $A_i\in \mathfrak{F}_\phi$. If $A\geq 0$, then $0\leq A=A_1-A_2\leq A_1 \Rightarrow A\in \mathfrak{F}_\phi$.

\end{proof}

\begin{remark}
	\label{extensionweight}
At this point, it is important to note that a weight $\phi$ admits a natural linear extension $\tilde{\phi}$ to $\mathfrak{M}_\phi$. It is a simple consequence of $(iv)$ in Proposition \ref{simplepropweights} and the uniqueness of polarization identity.

There will be no distinction between $\phi$ and $\tilde{\phi}$ in the following.
\end{remark}

Notice that these properties make the quotient $\mathfrak{N}_\phi/N_\phi$ a pre-Hilbert space, provided with $\ip{A}{B}_\phi=\phi(B^\ast A)$. We call the completion of this space $\hilbert_\phi$.
%**** linear na primeira ou na segunda?

\begin{definition}
A weight $\phi$ on a $C^*$-algebra $\calgebra$ is said to be:
\begin{enumerate}[(i)]
\item densely defined,\index{weight! densely defined}  if $\mathfrak{F}_\phi$ is dense in $\calgebra_+$;
\item faithful,\index{weight! faithful} if $\phi(A^\ast A)=0 \Rightarrow A=0$, \ $A\in \calgebra$;
\item normal,\index{weight! normal} if $\phi(\sup A_i)=\sup \phi(A_i)$ for all bounded increasing net $(A_i)_{i\in I} \in \calgebra_+$;
\item semifinite,\index{weight! semifinite} if $\mathfrak{M}_\phi$ is weakly dense in $\calgebra$;
\item tracial,\index{weight! tracial}\index{trace! see{ weight}} if $\phi(A^\ast A)=\phi(AA^\ast)$ for all $A\in \calgebra$. \glsdisp{trace}{\hspace{0pt}}
\end{enumerate}
\end{definition}

It is important to stress that the definition of a trace (\ie, a tracial weight) and a straightforward calculation using the polarization identity leads to the extension of a trace to $\mathfrak{M}_\phi$, as mentioned in Remark \ref{extensionweight}, satisfies
$\phi(AB)=\phi(BA)$ for all $A,B \in \mathfrak{M}_\phi$.

%%%************************************
\section{Characterization of Normal Functionals}

The aim of this section is to study the normality of functionals. Since this work will demand normality on the tracial weights, we are interested in obtaining some kind of continuity and a more workable form of these linear functionals.

For a more detailed presentation, the reader can look at \cite{Dixmier81}. We start presenting a useful result whose proof presented here can be found in \cite{murphy90}.

\begin{theorem}[Vigier]\index{theorem! Vigier}]
\label{vigier}
Let $\nalgebra\subset B(\hilbert)$ be a von Neumann algebra and let $(A_j)_{j\in J}\subset \nalgebra$ an increasing net of hermitian operators bounded above. Then,
$(A_j)_{j\in J}$ is convergent in the SOT. In addition, $\displaystyle A_j \xrightarrow{SOT} \sup_{j\in J}{A_j}\in \nalgebra$, where the suppremum is considered in the set of all self-adjoint elements of $\nalgebra$.  
\end{theorem}
\begin{proof}
Let $j_0 \in J$ and let $K=\{j\in J \ | \ j\geq j_0\}$. Then $(A_k-A_{j_0})_{k\in K}$ is a increasing net of positive operators bounded above and below. Let us denote this net by $(B_k)_{k\in K}$. Because of the properties of directed sets and limits, the desired convergence is equivalent to the convergence of the net $(B_k)_{k\in K}$.

Notice that, for every $x\in \hilbert$, $(\ip{B_k x}{x})_{k\in K} \subset \mathbb{R}_+$ defines an increasing net of real numbers such that
$$\ip{B_k x}{x}\leq\left\|B_k\right\| \|x\|\leq \sup_{k\in K}\|B_k\| \|x\|.$$
Hence, $(\ip{B_k x}{x})_{k\in K}$ is convergent. Using the polarization identity, we define the sesquilinear form
$$(x,y)=\sum_{l=0}^{3} i^l \sup_{k\in K} \ip{B_k (x+i^l y)}{x+i^l y}.$$

It is clear that $|(x,y)|\leq 4 \sup_{k\in K}{\|B_k\|}|\ip{x}{y}|$. It follows by the Riesz Theorem that there exists a positive operator $B\in \nalgebra$ such that $(x,y)=\ip{Bx}{y}$. Furthermore, it follows by the definition of the sesquilinear form that $B_j\leq B$.

Notice that 
$$\begin{aligned}
\|(B-B_j)x\|^2&= \ip{(B-B_j)x}{(B-B_j)x}\\
&\leq \|(B-B_j)x\| \ip{(B-B_j)x}{x}\\
&\leq \|B\|\ip{(B-B_j)x}{x}\\
&\xrightarrow{j} 0.
\end{aligned}$$
Thus, $B_j \xrightarrow{SOT}B$ and $B=sup_{j\in J} B_j$.

\end{proof}

Notice that normality is a kind of ``order continuity'', it will become clear in the next two lemmas. The next result shows some relation between the order, the SOT and the WOT.
 
%Dixmier, pag 47, corollary 5
\begin{lemma}
	Let $\nalgebra$ be a von Neumann algebra, $\mathfrak{I}\subset \nalgebra$ a two-sided ideal and $A\in\left(\overline{\mathfrak{I}}^{WOT}\right)_+$. Then there exists a increasing net $(A_i)_{i\in I} \subset\mathfrak{I}^+$ converging to $A$ in the SOT.
\end{lemma}
\begin{proof}
	First, to simplify the notation, lets write $\mathfrak{I}^+ \setminus\{0\}=\{A_j\}_{j\in J}$. Now, consider the set
	$$\mathcal{F}=\left\{K\subset J \ \middle| \ \sum_{K\in K_f}A_k \leq A, \ \forall K_f \subset K, K_f \mbox{ finite}\right\}$$
	provided with the partial order $K_1\preceq K_2 \Leftrightarrow K_1\subset K_2$.
	It is easy to see that it satisfies the requirements to conclude, by Zorn's lemma, the existence of a maximal element $I\in \mathcal{F}$.
	
	Consider $\displaystyle B=A-\sup\left\{\sum_{i\in I_f}A_i \ \middle| \ I_f \subset I, \ I_f \textrm{ finite}\right\}$ and let $(p_j)_{j\in J} \subset \mathfrak{I}$ be an increasing net of projections which converges to the identity of $\overline{\mathfrak{I}}^{WOT}$. By the ideal property, $B^\frac{1}{2}p_jB^\frac{1}{2} \in \mathfrak{I}^+$ for every $j \in J$ and $0\leq B^\frac{1}{2}p_jB^\frac{1}{2}\leq A$ since $0\leq B\leq A$, but it does not correspond to any index in $I$, in fact, if we had $B^\frac{1}{2}p_jB^\frac{1}{2}=A_r$, we would also have the inequality $$\sum_{i\in I_f}A_i \leq B+A_r\leq A \quad \forall I_f \subset K\cup\{r\}, I_f \mbox{ finite,}$$
	but the maximality of $I$ forbids this. The only remaining possibility is $$B^\frac{1}{2}p_jB^\frac{1}{2} =0 \quad \forall j \in J.$$
	Hence, $B=0$.
	
	Notice that, by Vigier's theorem, the net $\displaystyle \left(\sum_{i\in I_f} A_i\right)_{I_f\subset I}$, with the natural order on the finite subsets of $I$, converges to $A$ in the SOT.
	
\end{proof}

%Diexmier von Neumann algebras
\begin{lemma}
	\label{positiveneighbourhood}
	Let $\nalgebra$ be a von Neumann algebra, $A \in \nalgebra_+$ and $\phi,\psi$ two normal positive functionals such that $\phi(A)<\psi(A)$. Then, there exists $B\in \nalgebra_+\setminus\{0\}$ such that $$\phi(C)<\psi(C) \quad \forall C\in\nalgebra^+\setminus\{0\}  \ C\leq B.$$
\end{lemma}
\begin{proof}
	Let
	$$\mathcal{F}=\left\{P\in\nalgebra^+ \ \middle| \ P\leq A, \ \phi(P)\geq\psi(P)\right\}.$$
	
	Consider a chain $\{P_i\}_{i\in I}\subset \mathcal{F}$, we know $\displaystyle P=\sup_{i\in I}P_i\leq A$ and, by normality,
	$$\phi(P)=\sup_{i\in I} \phi(P_i)\geq \sup_{i\in I} \psi(P_i)=\psi(P).$$
	
	Since every chain has a maximal element, Zorn's lemma says $\mathcal{F}$ has a maximal element $Q$.
	
	Take $B=A-Q\in \nalgebra_+$. Of course, $B\leq A$, and, since $$\psi(B)=\psi(A)-\psi(Q)>\phi(A)-\phi(Q)=\phi(B)\geq0,$$
	$B\neq0$. Finally, for every $P\in\nalgebra^+$, $P\leq B$, we must have $\phi(P)<\psi(P)$, otherwise the positive element $A-P\geq Q$ would violate the maximality of $Q$.
	
\end{proof}

The next result is a classical result in Functional Analysis and can be found in \cite{con90}. It shows that, as a consequence of rigid structure of the field, strong and weak continuity are the same for linear functionals.

%Bratteli vol 1 pag 70 
%Conway IX, Proposition 5.1. pag. 281
\begin{lemma}
	\label{ultraweakcontinuous}
	Let $\nalgebra\subset B(\hilbert)$ be a von Neumann algebra and let $\omega$ be a linear functional on $\nalgebra$. Then $\omega$ is ultra-weakly (weakly) continuous if and only if it is ultra-strongly (strongly) continuous. Furthermore, for every ultra-strongly continuous functional there exists $(x_n)_{n\in \mathbb{N}}, (y_n)_{n\in \mathbb{N}} \in \hilbert$ satisfying $\displaystyle\sum_{n \in\mathbb{N}}\|x_n\|^2, \sum_{n \in\mathbb{N}}\|y_n\|^2 < \infty$ such that
	$$\omega(A)=\sum_{n \in\mathbb{N}}\ip{y_n}{Ax_n}, \quad A \in \nalgebra.$$
	
\end{lemma}
\begin{proof}
	It is obvious that ultra-weakly continuous functionals are ultra-strongly continuous.
	
	For the other implication, by taking a basic neighbourhood in the ultra-strong topology, there exists a sequence $(x_n)_{n\in\mathbb{N}} \in \hilbert$ with $\displaystyle\sum_{n\in \mathbb{N}}\|x_n\|^2 < \infty$ such that \sloppy ${\displaystyle |\omega(A)|^2<\sum_{n\in \mathbb{N}}\|Ax_n\|^2}$.
	
	We can now define the Hilbert space $\displaystyle\widetilde{\hilbert}=\bigoplus_{n\in \mathbb{N}}\hilbert$ and the norm continuous functional
	$$\tilde{\omega}\big((Ax_n)_{n\in\mathbb{N}}\big)=\omega(A)$$
	defined on the linear subspace $\left\{ (Ax_n)_{n\in\mathbb{N}} \in \widetilde{\hilbert} \ | \ A \in \nalgebra\right\}$. This functional can be extended to the whole space $\widetilde{\hilbert}$ (still denoted by $\tilde{\omega}$) using the Hahn-Banach Theorem\footnote{see Theorem \ref{normedHBT}.} and Riesz's Theorem ensures the existence of a vector $(y_n)_{n\in\mathbb{N}} \in \widetilde{\hilbert}$ (which means $\displaystyle \sum_{n\in \mathbb{N}}\|y_n\|^2 < \infty$) such that
	$$\tilde{\omega}\left((Ax_n)_{n\in\mathbb{N}}\right)=\omega(A)=\ip{(y_n)_{n\in\mathbb{N}}}{(Ax_n)_{n\in\mathbb{N}}}_{\widetilde{\hilbert}}=\sum_{n \in\mathbb{N}}\ip{y_n}{Ax_n}.$$
	
	This establishes not only the ultra-weakly continuity of $\omega$, but it gives us the mentioned general form of the ultra-weakly (and ultra-strongly, since they are the same) continuous functionals.
	
	The case when $\omega$ is strongly-continuous is basically the same.
	
\end{proof}

\begin{proposition}
	\label{WOT=SOT}
	Let $K\subset \nalgebra$ be a convex set. Then, the ultra-weak (weak) closure and the ultra-strong (strong) closure of $K$ coincide.
\end{proposition}
\begin{proof}
	Of course $\overline{K}^{SOT}\subset \overline{K}^{WOT}$, so the thesis is equivalent to show $\overline{K}^{WOT}\setminus \overline{K}^{SOT}=\emptyset$.
	
	Suppose there exists $x\in\overline{K}^{WOT}\setminus \overline{K}^{SOT}$. Then $\overline{K}^{SOT}-x$ is a closed convex subset that does not contain the null-vector. Let $U$ be a convex balanced neighbourhood of zero such that $(x+\overline{U}^{SOT})\cap \overline{K}^{SOT}=\emptyset$. Therefore $\overline{K}^{SOT}+\overline{U}^{SOT}-x$ is also a closed convex subset that does not contain the null-vector, by Corollary \ref{CTMD}, there exists a SOT-continuous functional $\phi$ such that $\phi(x)=0$ and $\phi(k)>0 \ \forall k\in \overline{K}^{SOT}+\overline{U}^{SOT}$. 
	
	Then, there exists $\alpha>0$ such that
	$\phi(x)=0$ and $\phi(k)>\alpha>0, \ \forall k\in \overline{K}^{SOT}$. This is not possible, as the previous lemma stated that such a $\phi$ is also WOT-continuous and $x\in \overline{K}^{WOT}$.
	
\end{proof}

\begin{notation}
	Let $\calgebra\subset B(\hilbert)$ be a $C^\ast$-algebra and $x\in \hilbert$. We denote by $\omega_x$ the continuous linear function on $\calgebra$ defined by
	$$\omega_x(A)=\ip{x}{Ax}, \quad A \in \calgebra.$$
\end{notation}

Notice that, if $\|x\|=1$, then $\omega_x$ is a state. It may not seem obvious if $\calgebra$ is not unital, but the existence of an approximate identity\footnote{see Definition \ref{defAppId}.} plays the same role in the proof. 

Finally, we will present the desired characterization of normal linear functionals.

%Dixmier pag 57, theorem 1
\begin{proposition}
	Let $\nalgebra\subset B(\hilbert)$ be a von Neumann algebra, $\phi$ a positive functional on $\nalgebra$. The following conditions are equivalent:
	\begin{enumerate}[(i)]
		\item $\phi$ is normal;
		\item $\phi$ is ultra-strongly continuous;
		\item there exists $(x_n)_{n\in\mathbb{N}}\subset \hilbert$, such that $\displaystyle\sum_{n=1}^{\infty}\|x_n\|^2 <\infty$ and $\displaystyle \phi=\sum_{n=1}^{\infty}\omega_{x_n}$.
	\end{enumerate}
\end{proposition}
\begin{proof}
	$(i)\Rightarrow (ii)$
	Let
	$$\mathcal{F}=\left\{P\in\nalgebra^+ \ \middle| \ P\leq\mathbbm{1} \mbox{ and } \nalgebra \ni A \mapsto \phi(AP) \mbox{ is ultra-weakly continuous}\right\}.$$
	
	Consider a chain $\{P_i\}_{i\in I}\subset \mathcal{F}$, we know $\displaystyle P=\sup_{i\in I}P_i \in \nalgebra^+$, $P_i \xrightarrow{SOT} P$ by the Vigier Theorem, so it also converges in the $WOT$, but on the unit ball the weak and ultra-weak topologies coincide. Furthermore, for $\|A\|\leq 1$,
	$$\left|\phi\left(A(P-P_i)\right)\right|^2\leq \phi(A^\ast A)\phi(P-P_i)\leq \phi(\mathbbm{1})\phi(P-P_i)$$
	
	Thus $A\mapsto \phi(AP)$ is the uniform limit of the ultra-weakly continuous positive functionals $\left(A \mapsto \phi(AP_i)\right)_{i\in I}$ on the unit ball, thus it is also ultra-weakly continuous on the unit ball. By linearity, it is ultra-weakly continuous on $\nalgebra$.
	
	Applying Zorn's lemma, there exists a maximal positive operator $Q\in \mathcal{F}$, $Q\leq 1$, such that $A\mapsto \phi(AQ)$ is ultra-weakly continuous.
	
	It remains to prove $Q=\mathbbm{1}$. In fact, if $\mathbbm{1}-Q>0$, we can find an ultra-weak positive functional $\psi$ such that $0\leq \phi(\mathbbm{1}-Q)<\psi(\mathbbm{1}-Q)$ by choosing a $z \in \hilbert$ such that $\phi(\mathbbm{1}-Q)<\ip{(\mathbbm{1}-Q)z}{z}$ and $\psi=\left(A\mapsto \ip{Az}{z}\right)$), for example.
	
	By Lemma \ref{positiveneighbourhood}, there exists $B\in\nalgebra^+\setminus\{0\}$ such that $B\leq \mathbbm{1}-Q$ and
	$$\phi(P)<\psi(P) \quad \forall C\leq B, \  C\in\nalgebra^+\setminus\{0\}.$$
	
	Hence, since for each $A\in\nalgebra\setminus\{0\}$ we have $BA^\ast A B \leq \|A\|^2B B \leq \|A\|^2 \|B\| B$, the Cauchy-Schwarz inequality gives us
	$$\left|\phi\left(\frac{AB}{\|A\|\|B\|^\frac{1}{2}}\right)\right|^2\leq\phi(1)\phi\left(\frac{BA^\ast A B}{\|A\|^2\|B\|}\right)<\psi\left(\frac{BA^\ast A B}{\|A\|^2\|B\|}\right)\Leftrightarrow \left|\phi\left(AB\right)\right)|^2<\psi\left(BA^\ast A B\right).$$
	
	Notice now that if $(A_i)_{i\in I}\subset\nalgebra$ is a net such that $A_i\to0$ ultra-strongly, then, for every $x,y\in \hilbert$, $\ip{B A_i^\ast A_i B x}{y}= \ip{A_i B x}{A_i B y}\leq\|A_i Bx\| \|A_i B y\| \to 0$. Thus, $\phi$ is ultra-strongly continuous, and by Lemma \ref{ultraweakcontinuous} $\phi$ is ultra-weakly continuous, since ultra-strongly and ultra-weakly functionals coincides on $\nalgebra$ \iffalse Dixmier pag 39, lemma 2 \fi. Hence, $A\mapsto\phi(AB)$ is ultra-weakly continuous.
	
	This contradicts the maximality of $Q$, since $A\mapsto \phi\left(A(Q+B)\right)$ is ultra-strongly continuous and $Q+B\leq Q+\mathbbm{1}-Q=\mathbbm{1}$. The conclusion follows.
	
	$(ii) \Rightarrow (iii)$ Due to Lemma \ref{ultraweakcontinuous}, is enough to prove that, for $y,z\in \hilbert$, there exists $x \in \hilbert$ such that $\nalgebra \ni A \mapsto \ip{y}{Az}=\omega_x$. In fact, we can see in the proof of Theorem \ref{cauchyschwarz} that positiveness implies that $\omega_{y,z}(T^\ast)=\overline{\omega_{y,z}(T)}$. Then, for every self-adjoint operator $A\in \nalgebra$, $$\ip{y}{Az}=\omega_{y,z}(A)=\omega_{y,z}(A^\ast)=\overline{\omega_{y,z}(A)}=\overline{\omega_{y,z}(A^\ast)}=\overline{\ip{y}{A^\ast z}}=\ip{A^\ast z}{y}=\ip{ z}{A y}$$
	$$\Rightarrow \ip{y+z}{A(y+z)}-\ip{y-z}{A(y-z)}=2(\ip{Ay}{z}+\ip{z}{Ay})=4\ip{y}{Az}.$$
	
	Since every operator is a linear combination os two self-adjoint operators, the previous equality holds for ever $A\in\nalgebra$.
	
	Set $\tilde{x}=y+z$. Hence $4\omega_{y,z}\leq \omega_{\tilde{x}}$. Now, Theorem \ref{commutantRN}, that will be proved later in a more general framework and that the reader should find no difficulty in adapting to this special case, guaranties the existence of a positive operator $H\in B(\hilbert)$ such that
	$4\omega_{x,y}(A)=\ip{H\tilde{x}}{A H\tilde{x}}$. The result follows setting $x=H\tilde{x}$.
	
	$(iii)\Rightarrow (i)$ it is obvious. 
	
\end{proof}

Before finishing this section, we want to stress that the previous theorem is strongly dependent on the Hahn-Banach Theorem and the Riesz Representation Theorem for Hilbert spaces. On its turn, the Hahn-Banach Theorem in the way it is applied, is a consequence of the convex vector topologies we have used.

%************************ Third Section  *******************************

\section{Representations and Spectral Analysis}

Spectral theory is well known for normal operators in $B(\hilbert)$ the space of bounded operators on a Hilbert space $\hilbert$. It allows us to construct a functional calculus for operators. It is also known that spectral theory and continuous functional calculus can be extended to $C^\ast$-algebras. An easy way to import all results of spectral theory of bounded operators are going to be shown next.

Our program in this section is to prove a representation theorem for $C^\ast$-algebras in $B(\hilbert)$ without using any Functional Calculus or Spectral Theory, then, ``import'' Functional Calculus from $B(\hilbert)$. In the von Neumann algebras case, using that the spectral projections are SOT-limits of operators in the algebra, we can also ``import'' the Spectral Theory.

Due to this choice to avoid Functional Calculus before proving a representation theorem, our presentation will be slightly different from the standard literature. There are several references on the subject, for example, \cite{Blackadar2006}, \cite{Bratteli1}, \cite{DS57}, \cite{KR83}, \cite{murphy90} and \cite{Takesaki2002}
 
\begin{definition}[Resolvent and Spectrum]
Let $\calgebra$ be a $C^\ast$-algebra, where we adjoin an identity if none is provided. Let $A\in \calgebra$.
\begin{enumerate}[(i)]
\item $\rho(A)\doteq\{\lambda \in \mathbb{C} \ |\ (\lambda\mathbbm{1}-A)\textrm{ has an inverse in } \calgebra\}$ is called the resolvent set \index{resolvent set} of $A$.

\item $\sigma(A)\doteq\mathbb{C}\setminus\rho(A)$ is called the spectrum\index{spectrum} of $A$.

\item $r(A)\doteq\sup\{|\lambda| \ |\ \lambda \in\sigma(A)\}$ is called the spectral radius \index{spectral radius} of $A$.
\end{enumerate}
\end{definition}

Note that this is the usual definition of the spectrum in $B(\hilbert)$ and that the spectral radius is a positive finite number due to the next lemma.

\begin{remark} \
We do not specify in which algebra the spectrum is taken. When necessary, we will write $\rho_\algebra(A)$ in order to fix the algebra $\algebra$, but it is an interesting fact (it will be shown later) that the spectrum does not depend on the $C^\ast$-algebra, \ie, the spectrum is the same in any unital $C^\ast$-algebra that has $A$ as an element.

\end{remark}

\begin{lemma}
\label{bspec}
Let $\calgebra$ be a unital $C^\ast$-algebra. If $A\in \calgebra$ and $|\lambda|>\|A\|$, then $\lambda \in \rho(A)$. In other words, $\lambda\in \sigma(A) \Rightarrow |\lambda| \leq \|A\|$.
\end{lemma}
\begin{proof}
Define $\displaystyle B_m=\sum_{n=0}^{m}{\lambda^{-(n+1)}A^n}$. This is an absolutely convergent series and, in particular, $\calgebra$ is a Banach space. Consequently, there exists $\displaystyle B=\lim_{m\rightarrow \infty}{B_m}$.

In order to show that $B$ is the inverse of $\lambda\mathbbm{1}-A$, note that
$$\begin{aligned}
(\lambda\mathbbm{1}-A)B_m	&=(\lambda\mathbbm{1}-A)\sum_{n=0}^{m}{\lambda^{-(n+1)}A^n} \\
							&=\sum_{n=0}^{m}{\lambda^{-n}A^n}-\sum_{n=1}^{m+1}{\lambda^{-n}A^n}\\
							&=\mathbbm{1}-\lambda^{-(m+1)}A^{m+1}.\\
\end{aligned}$$
Since $B_m\rightarrow B$, the left-hand side converges to $(\lambda\mathbbm{1}-A)B$ while by $\lambda^{-(m+1)}A^{m+1}\rightarrow 0$ the right-hand side goes to $\mathbbm{1}$.

The proof for $B(\lambda\mathbbm{1}-A)$ follows by the same argument.

\end{proof}

For normal operators it holds that $r(A)=\|A\|$. The proof of this fact is not difficult and can be found in \cite{Bratteli1}. We prefer not to present it in this work because we will use this result only in $B(\hilbert)$, where we are supposing it is known. The general result will be a consequence of the existence of a representation of the $C^\ast$-algebra. 

For the next theorem, we notice that, for a bounded $\ast$- homomorphism of Banach algebras $\Phi: \algebra_1 \to \algebra_2$, $\algebra_1/\ker{\Phi}$ is a Banach space since $\ker{\Phi}$ is closed and an $\ast$-algebra thanks to the ideal property of $\ker{\Phi}$, provided with the canonical operations bellow:
\begin{enumerate}[(i)]
	\item $[A]+[B]\doteq [A+B]$;
	\item $[A][B]\doteq[AB]$;
	\item $[A]^\ast\doteq[A^\ast]$;
	\item $\displaystyle \left\|[A]\right\|\doteq\inf_{\tilde{A}\in\ker{\Phi}}\left\|A+\tilde{A}\right\|$.
\end{enumerate}

\begin{theorem}[First Theorem of Isomorphism for Banach Algebras]\index{theorem!of isomorphism}
{\label{TI}}
Let $\algebra_1,\algebra_2$ be Banach $\ast$-algebras and $\Phi:\algebra_1 \to \algebra_2$ a bounded $\ast$-homomorphism, then there exists a unique $\ast$-isomorphism $\tilde{\Phi}: \algebra_1/\ker{\Phi} \to \Ran{\Phi}$ such that the following diagram commutes.

\begin{center}
\begin{minipage}{7 cm}
\begin{displaymath}
		\xymatrix{ \algebra_1 \ar[dr]_{\eta} \ar[rr]^\Phi & 			 											&	\Ran(\Phi) \subset \algebra_2						\\
																			 & \algebra_1/\ker{\Phi}	\ar[ur]_{\tilde{\Phi}} 									}
\end{displaymath}
\end{minipage}
\hfil\hspace{-4cm}
.
\end{center}

Furthermore, $\|\Phi\|=\|\tilde{\Phi}\|$.
\end{theorem}
\begin{proof}
Define $\tilde{\Phi}: \algebra_1/\ker{\Phi} \rightarrow \Ran(\Phi)$ such that $\tilde{\Phi}([A])=\tilde{\Phi}(\eta(A))=\Phi(A)$.

Of course we can check that this function is well defined, since we used a representing element of the equivalent class to define the function. If $[A_1]=[A_2] \in \algebra_1/\ker{\Phi}$ we have $A_1 = A_2 + K$ with $K \in \ker{\Phi}$ and then 
$$\tilde{\Phi}([A_1])=\Phi(A_1)=\Phi(A_1)+ \Phi(K)= \Phi(A_1 + K) = \Phi(A_2) = \tilde{\Phi}([A_2]).$$

$\tilde{\Phi}$ is again a $\ast$-homomorphism because
$$A_1, A_2 \in \algebra_1
\begin{cases} &\Rightarrow \tilde{\Phi}([A_1 + A_2])= \Phi(A_1+A_2)=\Phi(A_1)+\Phi(A_2) =\tilde{\Phi}([A_1])+\tilde{\Phi}([A_2]).\\
						&\Rightarrow \tilde{\Phi}([A_1 A_2])= \Phi(A_1A_2)=\Phi(A_1)\Phi(A_2)=\tilde{\Phi}([A_1])\tilde{\Phi}([A_2]).\\
						&\Rightarrow \tilde{\Phi}([A_1^\ast])= \Phi(A_1^\ast)=\Phi(A_1)^\ast =\tilde{\Phi}([A_1])^\ast.\\
\end{cases}$$

In order to show that $\tilde{\Phi}\in \mathcal{B}(\algebra_1/\ker{\Phi},\Ran(\Phi))$, it is enough to note that
$$\begin{aligned}
\|\tilde{\Phi}([A])\| 
&= \inf_{\tilde{A}\in \ker{\Phi}}{\|\Phi(A+\tilde{A})\|}\\
& \leq \inf_{\tilde{A}\in \ker{\Phi}}{\|\Phi\|\ \|A + \tilde{A}\|}\\
& = \|\Phi\| \inf_{\tilde{A}\in \ker{\Phi}}{\|A + \tilde{A}\|}\\
&=\|\Phi\| \ \|\ [A]\ \|.\\
\end{aligned}$$
hence $\|\tilde{\Phi}\| \leq \|\Phi\|$. In addition, 
$$\|\Phi(A)\| = \|\tilde{\Phi}\circ \eta (A)\| = \|\tilde{\Phi}([A])\| \leq \|\tilde{\Phi}\| \ \| \ [A] \ \| \leq \|\tilde{\Phi}\| \ \|A\|$$
and from the two inequalities it follows that $\|\Phi\| = \|\tilde{\Phi}\|$.

Finally, $\tilde{\Phi}$ is bijective because $\tilde{\Phi}([A])=\Phi(A)= 0 \Leftrightarrow A\in \ker{\Phi} \Leftrightarrow [A]=[0] $ and $y\in \Ran(\Phi) \Leftrightarrow y=\Phi(A)$ for some $A\in \algebra_1$. Thus, $\tilde{\Phi}([A])=\Phi(A)=y$.

We define $\tilde{\Phi}$ in such a way that $\Phi=\tilde{\Phi} \circ \eta$ and, in addition, if $\tilde{\Psi}$ is another continuous operator satisfying this identity we have $(\tilde{\Phi}-\tilde{\Psi})\eta(A) = 0$ for all $A\in \algebra_1$ and it follows that $\tilde{\Phi}=\tilde{\Psi}$.

\end{proof}

Although we are interested only in the case of Banach algebras, it is interesting to note that multiplication and involution play no role at the definition of the operator $\tilde{\Phi}$, that is, the very same definition works to prove this kind of theorem for groups, for example. 

\begin{definition}[Representation]
A representation \hspace{1pt} \index{representation} $\pi$ \hspace{1pt}of \hspace{1pt} a $C^\ast$-algebra $\calgebra$ is a \mbox{$\ast$-homomorphism} from $\calgebra$ into $B(\hilbert)$ for some Hilbert space $\hilbert$. The representation is said to be faithful \index{representation! faithfull} if $\pi$ is injective.
\end{definition}

\begin{theorem}
\label{nonexpensive}
Let $\calgebra$ be a $C^\ast$-algebra and let $\pi$ be a representation. Then $$\|\pi(A)\|\leq \|A\|.$$
\end{theorem}
\begin{proof}
From the Theorem \ref{TI}, there exists a $\ast$-isomorphism $\tilde{\pi}: \calgebra/\ker{\pi} \to \Ran{\pi}$.

Since $\tilde{\pi}$ is a $\ast$-isomorphism, $\tilde{\pi}([\mathbbm{1}])=\pi(\mathbbm{1})$ is in the image of $\pi$. Furthermore, $\lambda[\mathbbm{1}]-[A]$ has an inverse if and only if $\lambda\tilde{\pi}([\mathbbm{1}])-\tilde{\pi}([A])$ is invertible too. So, we have the identification $\sigma\left([A]\right)=\sigma\left(\tilde{\pi}([A])\right)$, in particular, from Lemma \ref{bspec} it follows that $r(A)\leq \|A\|$.

Now, we can use known results of (classical) spectral analysis for a normal element $\tilde{\pi}([A][A]^\ast)$, from which we conclude
$$\left\|\pi(A)\right\|^2=\left\|\tilde{\pi}([A])\right\|^2=\left\|\tilde{\pi}([A][A]^\ast)\right\|=r\left(\tilde{\pi}([A][A]^\ast)\right)\leq \|AA^*\|=\|A\|^2.$$

\end{proof}

A simple consequence of the preceding theorem is that a faithful representation is isometric.

%\begin{theorem}[Krein-Milman]
%\label{TKMx}
%\end{theorem}

\begin{definition}[Pure state]
A state in a $C^\ast$-algebra $\calgebra$ is said to be a pure state \index{state! pure} if it is an extremal point\footnote{See Definition \ref{ExtDef}.} of $\mathcal{S}_\calgebra\doteq\{\omega \in \calgebra^\ast \ | \ \|\omega \|=\omega(\mathbbm{1})=1\}$. When $\omega$ is not pure it is called a mixed\index{state! mixed} state.
\end{definition}

For now, the existence of pure states is not clear, but it will be shown soon that they exist in sufficient number such that its closed convex hull coincides with the weak$^\ast$-closure of the unit ball of $\calgebra^\prime$.

\begin{proposition}
\label{nullpurestates}
Let $A\in\calgebra$. If $\omega(A)=0$ for each pure state $\omega$, then $A=0$.
\end{proposition}
\begin{proof}
By Theorem \ref {TKM}, $B_{\calgebra^\prime}=\cchull{\mathcal{E}(\mathcal{S}_\calgebra)}$, where the closure is taken in the weak-$\ast$ topology. Then by continuity of the functionals in $\mathcal{S}_\calgebra$ we must have $\rho(A)=0$ for all $\rho \in B_{\calgebra^\prime}$ and using Corollary \ref{C3TMD}, this implies $A=0$.

\end{proof}

\begin{lemma}
\label{existencenorm}
Let $\calgebra$ be a $C^\ast$-algebra and $A \in \calgebra$. Then there exists a pure state $\omega$ such that $\omega(A^\ast A)=\|A\|^2$.
\end{lemma}
\begin{proof}
If $\calgebra$ has no unity, add it. Consider the subalgebra 
$$\mathfrak{A}_0=\{\alpha\mathbbm{1}+\beta A^\ast A \ | \ \alpha, \beta \in \mathbb{K}\}.$$

Define the linear functional $\tilde{\omega}: \mathfrak{A}_0 \to \mathbb{K}$ by $\tilde{\omega}(\alpha \mathbbm{1}+\beta A^\ast A)\doteq \alpha +\beta \|A\|^2$.

Note that
$$\begin{aligned}
\left|\tilde{\omega}(\alpha \mathbbm{1}+\beta A^\ast A)\right| &=\left|\alpha +\beta \|A\|^2\right| \\
&\leq \sup_{\lambda \in \sigma(A^\ast A)}{|\alpha + \beta\lambda|}\\
&=r(\alpha \mathbbm{1}+\beta A^\ast A)\\
&=\|\alpha \mathbbm{1}+\beta A^\ast A\|.
\end{aligned}$$
Hence, since $\tilde{\omega}(\mathbbm{1})=1$, $\tilde{\omega}$ is a state on $\mathfrak{A}_0$. Now, by the Hahn-Banach Theorem, it has a norm preserving extension to $\calgebra$.

This warrants that the closed and convex set 
$$\mathfrak{F}=\left\{\omega\in \mathcal{S}_\calgebra \ \middle| \ \omega(A^\ast A)=\|A\|^2\right\}\neq \emptyset.$$
Thus, it has a extreme point by the Krein-Milman\footnote{see Theorem \ref{TKM}.}.

Let $\omega$ be such an extreme point, and suppose $\omega=\lambda \omega_1+(1-\lambda)\omega_2$ for $\omega_1, \omega_2 \in \mathcal{S}_{\calgebra}$ and $0<\lambda<1$. If $\omega_1(A^\ast A)<\|A\|^2$ or $\omega_2(A^\ast A)<\|A\|^2$ we would have
$$\begin{aligned}
\|A\|^2
&=\omega(A^\ast A)\\
&=\lambda \omega_1(A^\ast A)+(1-\lambda)\omega_2(A^\ast A)\\
&< \lambda \|A\|^2+(1-\lambda)\|A\|^2\\
&=\|A\|^2\\
\end{aligned},$$
which is an absurd. Hence, $\omega_1(A^\ast A)=\omega_2(A^\ast A)=\|A\|^2$. It follows, by the definition, that $\omega_1,\omega_2 \in \mathfrak{F}$ and, by extremality of $\omega$, that
$\omega=\omega_1=\omega_2$.

Hence, $\omega$ is an extremal point in $\mathcal{S}_\calgebra$ as well and satisfies $\omega(A^\ast A)=\|A\|^2$. 

\end{proof}

Before presenting one of the most important theorems in operator algebras, we will introduce the notion of cyclic and separating vectors, which are important by themselves. These special vectors are both mathematically and physically significant: on one side, we will see soon that these vectors have the identity as support projections and we will use then to define Modular Theory in Chapter \ref{chapTTMT}; on the other, it is expected that the vacuum have such properties, \eg \, in Fock's space or Whitman's Fields formalism, moreover, Modular Theory has its own physical meaning. Some of the topics we just mentioned now will be rediscussed later.

\begin{definition}
	\label{def:cyclicandseparating}
	Let $\calgebra \subset B(\hilbert)$ be a concrete $C^\ast$-algebra, we say that a vector $\Omega$ is:
	\begin{enumerate}[(i)]
		\item cyclic for $\calgebra$ if $\overline{\{A\Omega \ | A\in \calgebra\} }=\hilbert$; \index{vector! cyclic}
		\item separating for $\calgebra$ if $A\Omega=0 \Rightarrow A=0$.\index{vector! separating}
	\end{enumerate}
\end{definition}

The GNS-construction, in honor of Gelfand, Naimark, and Segal who first obtained the result, is one of the most important results in Operator Theory. It proves the representation theorem we mentioned in the beginning of this section. Several references present this theorem, \eg, \cite{Blackadar2006}, \cite{Bratteli1}, \cite{doran86}, \cite{KR83}, and \cite{Takesaki2002}.

\begin{proposition}[GNS-Representation]
\label{GNS}
Let $\calgebra$ be a $C^\ast$-algebra and $\omega$ a state. Then, there exists a Hilbert space $\hilbert_\omega$, a cyclic and separating vector $\xi\in \hilbert_\omega$, and a representation $\pi_\omega$ of $\calgebra$ in $B(\hilbert_\omega)$ such that
$$\omega(A)=\ip{\xi}{\pi_\omega(A)\xi} \quad \forall A\in \nalgebra.$$
\end{proposition}
\begin{proof}
Define the closed left ideal (two-sided, due to $(i)$ in Proposition \ref{cauchyschwarz}, as it will be became clear below) $N_\omega=\left\{A \in \calgebra \ \middle| \ \omega(A^\ast A)=0 \right\}$ and $\hilbert_\omega=\calgebra/N_\omega$ provided with the inner product $\ip{[A]}{[B]}=\omega(A^\ast B)$.

First, this inner product is well defined because, if $N_1, N_2\in N_\omega$, then
\begin{equation}
\label{eq:calculationquotient}
\begin{aligned}
\ip{[A+N_1]}{[B+N_2]}
& = \omega\left((A+N_1)^\ast (B+N_2)\right)\\
&= \omega(A^\ast B)+\omega(N_1^\ast B)+\omega(A^\ast N_2)+\omega(N_1^\ast N_2)\\
&= \omega(A^\ast B)+\overline{\omega(B^\ast N_1)}+\omega(A^\ast N_2)+\omega(N_1^\ast N_2)\\
&=\omega(A^\ast B)\\
&=\ip{[A]}{[B]},\\
\end{aligned}
\end{equation}
where we used Proposition \ref{cauchyschwarz} $(i)$ and that $N_\omega$ is a left ideal.
 
 Positivity and sesquilinearity follow trivially from the positivity and linearity of $\omega$ and from the anti-linearity of $\ast$. It still remains to prove that $\ip{[A]}{[A]}=0 \Rightarrow [A]=0$, but this follows from the definition of quotient.
 
 Now, let us define the representation. Define the left representation by $$\pi_\omega(A)\left([B]\right)=[A B]=[A][B].$$
 
 Of course $\pi_\omega(A)$ is linear, $\pi_\omega(A B)=\pi_\omega(A)\pi_\omega(B)$ and  $\pi_\omega(A^\ast)=\pi_\omega(A)^\ast$, thus $\pi_\omega$ is a $\ast$-homomorphism. By definition of the quotient norm, $$\left\|\pi_\omega(A)\left([B]\right)\right\|=\left\|[A][B]\right\|\leq \left\|[A]\|\|[B]\right\|\leq \|A\| \|B\|$$ and that means $\pi_\omega(A) \in B(\hilbert_\omega)$.
 
 It remains just to prove the existence of a cyclic vector. Let $(E_\lambda)_{\lambda\in\Lambda}\in\calgebra$ an increasing approximate identity (see Theorem \ref{ExisAppId}), then the equivalent classes $([E_\lambda])_{\lambda\in\Lambda}$ form Cauchy sequence, thus convergent to some $\xi\in \hilbert_\omega$.
 
 It follows from the definition that this is a cyclic vector and
 $$\omega(A)=\sup_{\lambda \in \Lambda} \omega(E_\lambda A E_\lambda)=\sup_{\lambda \in \Lambda}\ip{[E_\lambda]}{\pi_\omega(A)[E_\lambda]}=\ip{\xi}{\pi_\omega(A)\xi}.$$
\end{proof}

\begin{remark}
We emphasize that $N_\omega$ is a two-sided ideal, but we only use it is a left ideal. This is because we could define a right representation in an analogous way. This fact is closely related to the modular operator.
\end{remark}

\begin{remark}
Note that the representation obtained in the previous result is not faithful. In fact, it is faithful if $N_\omega=\{0\}$, hence
the representation of Proposition \ref{GNS} is faithful on $\calgebra/N_\omega$.
\end{remark}

\begin{theorem}[Gelfand-Naimark]\index{theorem! Gelfand-Naimark}
\label{TGN}
Every $C^\ast$-algebra $\calgebra$ admits a faithful (isometric) representation.
\end{theorem}
\begin{proof}
Let $\mathcal{E}$ be the set of all pure states on $\calgebra$ and let $\hilbert_\omega$ and $\pi_\omega$ be the Hilbert space and the corresponding representation, respectively,  obtained in Theorem \ref{GNS}. Define
$$
\hilbert =\bigoplus_{\omega\in\mathcal{E}}{\hilbert_\omega}  \textrm{ and }
\pi=\bigoplus_{\omega\in\mathcal{E}}{\pi_\omega} \, ,
$$
where the direct sum is defined in Definition \ref{DS}.

By Proposition \ref{nullpurestates}, $\pi(A)=0 \Leftrightarrow A=0$. Thus $\pi$ is a $\ast$-isomorphism. Now, it follows from Theorem \ref{nonexpensive} that $\|\pi(A)\|\leq \|A\|$. On the other hand, Lemma \ref{existencenorm} warrants the existence of a pure state $\omega$ such that $\omega(A^\ast A)=\|A\|^2$. Hence $\|A\|^2\leq \|\pi(A)\|^2$ and the equality follows.

\end{proof}

This theorem is the first indicative of the necessity of weights in von Neumann algebras, because, although the direct sum works well for defining the new Hilbert space as well as the representation, such representation is not related to a state because the sum of the pure states may diverge. For weights, however, this is not a problem.

The GNS-Representation Theorem can be extended for weights, as it follows.

\begin{proposition}[GNS-Representation for Weights]
\label{GNSweight}
Let $\calgebra$ be a $C^\ast$-algebra and $\phi$ a faithful normal semifinite weight. Then, there exists a Hilbert space $\hilbert_\phi$, a cyclic and separating vector $\xi$, and a representation $\pi_\phi$ of $\calgebra$ in $B(\hilbert_\phi)$ such that
	$$\phi(A)=\ip{\xi}{\pi_\phi(A)\xi} \quad \forall A\in \mathfrak{M}_\phi.$$
\end{proposition}
\begin{proof}
The proof follows exactly the same steps of Proposition \ref{GNS}, just defining the Hilbert space $\hilbert_\phi$ as the completion of the pre-Hilbert space $\mathfrak{N}_\phi/N_\phi$, where $\mathfrak{N}_\phi$ and $N_\phi$ are as in Definition $\ref{weightsets}$,  and the representation is defined by
\vspace{6pt}

\hspace{0.08\textwidth}\begin{minipage}{0.4\textwidth}
$$\begin{aligned}
\pi_\phi: \ & \calgebra	&\to	\ & B\big(\hilbert_\phi\big)	\\
\quad	 & A			&\mapsto \ & \ \pi_\phi(A) \\
\end{aligned}$$
\end{minipage}, 
\begin{minipage}{0.4\textwidth}
$$\begin{aligned}
	\pi_\phi(A): \	& \hilbert_\phi  &\to \ & \ \hilbert_\phi \	   \\
	&B			\ &\mapsto \ & [AB]	\\
\end{aligned}$$
\end{minipage}\linebreak because Proposition \ref{simplepropweights} warrants the required properties. 

\end{proof}

%\textcolor{red}{Relations with traces and lower semi-continuous}

\begin{notation}
\label{GNSnotation}
Throughout this work, we will denote by $(\pi_\phi, \xi)$ any cyclic representation (in particular, the GNS-representaion related to the weight $\phi$) such that $\pi_\phi(A):\hilbert_\phi \to \hilbert_\phi$ with $\ip{\pi_\phi(A)\xi}{\pi_\phi(B)\xi}_\phi=\phi(A^\ast B)$ and $\pi_\phi(A)\pi_\phi(B)=\pi_\phi(A B)$.
\end{notation}
%**************Precisa que o peso seja fiel, o teorema de existência está depois

The proof of the next lemma is partially due to the author and it is used as a technical result to prove the Functional Calculus for polynomials.

\begin{lemma}
	\label{lemmax1}
	Let $\calgebra$ be a $C^\ast$-algebra, $A\in \calgebra$ a normal element, $K\subset \mathbb{R}$ a compact set with $-\|A\|,\|A\|\in K$, $|k|\leq\|A\| \ \forall k\in K$ and $p:K \to \mathbb{R}_+$ a polynomial. Then $\|p(A)\|\leq p(\|A\|)$. 
\end{lemma}

\begin{proof}
	Proceeding by induction on the degree of $p$. If $degree(p)=0$, there is nothing to prove. Now suppose the statement is true for any positive polynomial on $K$ of degree less than $n\in \mathbb{N}$ and take $p$ a positive polynomial with $degree(p)=n+1$ on $K$. Decompose 
	$$\begin{aligned}
	p(x)	&=xq(x)+r, \ degree(q)= n\\
	&=(xq(x)-\min_{|x|\leq \|A\|}{xq(x)})+\left(r+\min_{|x|\leq \|A\|}{xq(x)}\right), \ degree(q)= n.
	\end{aligned}$$
	
	It is important to notice that $\displaystyle xq(x)-\min_{|x|\leq \|A\|}{xq(x)}\geq0$ and $\displaystyle \min_{|x|\leq \|A\|}{p(x)}=r+\min_{|x|\leq \|A\|}{xq(x)}\geq0$, and just using the $C^\ast$-condition and the induction hypothesis we get $$\begin{aligned}
	\|p(A)\| &\leq\left\|Aq(A)-\min_{|x|\leq \|A\|}{xq(x)}\mathbbm{1}\right\|+\left\|\left(r+\min_{|x|\leq \|A\|}{xq(x)}\right)\mathbbm{1}\right\|\\
	&\leq\left(\|A\|q(\|A\|)-\min_{|x|\leq \|A\|}{xq(x)}\right)+\left(r+\min_{|x|\leq \|A\|}{xq(x)}\right)\\
	&=p(\|A\|).
	\end{aligned}$$
	
\end{proof}

\begin{theorem}[Continuous Functional Calculus]
\label{funccalculus}
Let $\calgebra$ be a $C^\ast$-algebra and $A\in \calgebra$ a self-adjoint operator. There exists a unique isometric $\ast$-isomorphism $$\Psi:\mathcal{C}\left(\sigma(A)\right)\to C^\ast\left(\{\mathbbm{1},A\}\right)$$
such that $\Psi(\mathbbm{1}_{\sigma(A)})=A$.
\end{theorem}
\begin{proof}

The case $A=0$ is trivial.

If $A\neq0$, first consider the case where $f$ is a polynomial $p$, $\displaystyle p(x)=\sum_{n=0}^{\infty}\alpha_n x^n$ where $\alpha_n=0$ apart from a finite index set. Define $\displaystyle\Psi(p)(A)=\sum_{n=0}^{\infty}\alpha_n A^n \in \calgebra$.

Now, let $f\in \mathcal{C}\left(\sigma(A)\right)$. If $\|A\|\notin \sigma(A)$, take $\tilde{f}:\mathcal{C}\left(\sigma(A)\cup\{\|A\|\}\right)$ the continuous extension of $f$ satisfying $f(\|A\|)=1$.

From Weierstrass's Approximation Theorem, for each $i\in \mathbb{N}$ there exists a polynomial $p_i$, defined by $\displaystyle p_n(x)=\sum_{n=0}^{\infty}\alpha^i_n x^n$ where $\alpha^i_n=0$ apart from finite number of $n$'s, such that

$$\left\|p_n-\left(\tilde{f}-3.2^{-n}\right)\right\|<2^{-n}.$$

This leads us to conclude that $(p_n)_{n\in\mathbb{N}}$ is a strictly increasing sequence, because, for each $t\in\sigma(A)\cup \{\|A\|\}$,
$$\begin{aligned}
\left|p_n(t)-\left(\tilde{f}(t)-3.2^{-n}\right)\right|<2^{-n} &\Rightarrow
-2^{-n}<p_n(t)-f(t)+3.2^{-n}<2^{-n} \\
& \Rightarrow f(t)-2^{-n+2}<p_n(t)<f(t)-2^{-n+1} \\
& \Rightarrow p_n<p_{n+1} \ \forall n \in \mathbb{N}.\\
\end{aligned}$$

Note that Lemma \ref{lemmax1} implies $$\|p_i(A)-p_j(A)\|=\left\|(p_i-p_j)(A)\right\|=\left\|(p_i-p_j)(\|A\|\mathbbm{1})\right\|=\left|(p_i-p_j)\left(\|A\|\right)\right|<2^{-n+1} $$

Hence, $\left(p_i(A)\right)_{i\in\mathbb{N}} \subset \calgebra$ is a Cauchy's sequence and must converge. The uniqueness of the limit in a Hausdorff space allows us to define
$$\Psi(f)=f(A)=\lim_{i\to \infty}p_i(A) \qquad \forall f\in\mathcal{C}(\sigma(A)) \textrm{ and } p_i\to f \textrm{ uniformly.}$$

All that remains to prove is uniqueness, but this is obvious: if $\Psi_1, \Psi_2$ are such a $\ast$-isomorphisms, they must satisfy
$$\begin{aligned}
\Psi_1\left(\mathbbm{1}_{\sigma(A)}\right)&=A=\Psi_2\left(\mathbbm{1}_{\sigma(A)}\right), \\
\Psi_1\left(1\right)&=\mathbbm{1} =\Psi_2\left(1\right), \\
\end{aligned}$$
but then they must coincide in all polynomials which constitute a dense subset, thus $\Psi_1=\Psi_2$.

\end{proof}

Only to clarify the previous theorem, note that we proved that a $C^\ast$-algebra is isomorphic to a closed sub-algebra of $B(\hilbert)$. We are supposing that the spectral theory for $B(\hilbert)$ is well known to avoid re-deducing it here. Then, since we know how to define $f(A)$ on the representation, where $A$ is a continuous function on the spectrum of $A$, and we also know that $f(A)$ is an element of the closure of the sub-algebra which is (norm) closed, we can use the isomorphism to return to the algebra.   

The previous theorem has an interesting consequence that shows how restrictive the $C^\ast$-condition is: if a $\ast$-algebra admits a norm that makes it a $C^\ast$-algebra, this norm is the only one with this property.   

\begin{corollary}
Let $\algebra$ be a $\ast$-algebra. There is at most one norm in $\algebra$ that makes it a $C^\ast$-algebra. In this case, this norm is given by
$\|A\|=\sqrt{r(A^\ast A)}$.
\end{corollary}

\begin{corollary}[Spectral Theorem for von Neumann Algebras]
	\label{ST}
	Let $\nalgebra$ be a von Neumann algebra and $A\in\nalgebra$ a self-adjoint operator, $\left\{E^A_\lambda\right\}_{\lambda\in\sigma(A)}$ its spectral resolution. Then, $\left\{E^A_\lambda\right\}_{\lambda\in\sigma(A)}\subset\nalgebra$ and
	$$A=\int_{\sigma(A)} \lambda dE^A_\lambda.$$
\end{corollary}
 
Again, the previous corollary in a consequence of the isomorphism between the von Neumann algebra and a norm-closed sub-space of $B(\hilbert)$, and the fact that the projections in the spectral resolution are obtained as a WOT-limit of the elements of the von Neumann algebra. In fact,
$$\chi_n(x)=\begin{cases}e^{-\frac{1}{nx^2}}, \quad &\textrm{if } x>0 \\ 0, \quad &\textrm{if } x\leq 0 \end{cases},$$
defines a strictly increasing sequence of infinitely differentiable functions in all the real line that converges punctually to
$$\chi(x)=\begin{cases}1, \quad &\textrm{if } x>0 \\ 0, \quad &\textrm{if } x\leq 0 \end{cases}.$$
Thus, for a self-adjoint element $A$ in the von Neumann algebra $\nalgebra$, we have that $(\chi_n(A-\lambda\mathbbm{1}))_{n\in\mathbb{N}}\subset\nalgebra$ and $\chi_n(A)\xrightarrow{SOT} \chi(A-\lambda\mathbbm{1})=E^A_\lambda=E^A_{(\lambda,\infty)}$.

\iffalse\begin{proposition}
Pure state- irreductile representation theorem 2.37 Reiffel Murphy 5.1.1, 5.1.5 e finalmente 5.1.6. Schur Lemma de http://math.northwestern.edu/~theojf/CstarAlgebras.pdf - Rieffel by Theo
\end{proposition}
\fi

We finish this section showing that the usual notion of positive operator coincides with the one given in Definition \ref{defpositive}.

\begin{proposition}
\label{positiveeq1}
$A\in B(\hilbert)$ is a positive operator if and only if $A$ is self-adjoint and $\sigma(A)\subset \mathbb{R}_+$.
\end{proposition}
\begin{proof}
$(\Rightarrow)$  $\displaystyle \left|\left|1-\frac{A}{\|A\|}\right|\right|\leq 1 \Rightarrow \sigma\left(1-\frac{A}{\|A\|}\right) \subset [-1,1]$. By Theorem \ref{ST}, it follows that $\sigma(A)\subset [0,2\|A\|]\cap [-\|A\|,\|A\|]$. Moreover, $A$ is self-adjoint by hypothesis.
	
$(\Leftarrow)$ On the other hand, if $A$ is self-adjoint and $\sigma(A) \subset [0,\|A\|]$. It follows, again by Theorem \ref{ST}, that 
$$\sigma\left(\mathbbm{1}-\frac{A}{\|A\|}\right)\subset [0,1] \Rightarrow \displaystyle \left|\left|1-\frac{A}{\|A\|}\right|\right|\leq 1.$$

\end{proof}

Another interesting consequence of the Spectral Theorem and the previous proposition is that, in a $C^\ast$-algebra $\calgebra$, every operator is a linear combination of four positive operators of $\calgebra$. In fact, $\xi_n(x)=x\chi_n(x)$ converges uniformly to the continuous function $\xi(x)=x\chi(x)$. Thus, for every self-adjoint operator $A\in\calgebra$, $\xi(A)$ and $(1-\xi)(A)$ are self-adjoint operators such that $\sigma(\xi(A)),\sigma((1-\xi)(A))\in\mathbb{R}_+$ and $A=\xi(A)-(1-\xi)(A)$. For the general case, notice that every operator $A\in \calgebra$ can be written as $\displaystyle A=\left(\frac{A+A^\ast}{2}\right) +\left(\iu \frac{A-A^\ast}{2\iu}\right)$, and both factors in parenthesis are self-adjoint. This is basically the proof of the next proposition.

\begin{proposition}
	\label{sum4positive}
	Let $\calgebra$ be a $C^\ast$-algebra and $A\in\calgebra$. Then, there exists four positive operators $A_j\in\calgebra_+$, $0\leq j\leq 3$, such that
	$$A=\sum_{j=0}^3 \iu^j A_j.$$
	Moreover, this decomposition is unique if we require $A_0A_2=0$ and $A_1A_3=0$.
\end{proposition}

\begin{notation}
\label{positivepart}
For a self-adjoint operator $A$ in a $C^\ast$-algebra, we denote $A_+ =\xi(A)$.
\end{notation}

% ********************************************************

\section{Projections Revisited}

Now that we have the Functional Calculus, which is a powerful tool, we can explore it to construct and understand the projections better. We have special interest in the von Neumann algebras case, because we already know that spectral projections can be obtained as a limit of a WOT-convergent sequence of the elements of the algebra and, since a von Neumann is WOT-closed, the spectral projections lie in the algebra. Another result that is indispensable here is Theorem \ref{projectionsLattice}, which also is a consequence of the WOT-closeness of von Neumann algebras.
 
\begin{definition}
Let $\nalgebra$ be a von Neumann algebra:
\begin{enumerate}[(i)]
%Takesaki, Operator Algebras II, Def. 1.18: s^\nalgebra(\phi)=e-f where \nalgebra e=\overline{\mathfrak{N}_\phi}^{WOT} and \nalgebra f=N_\phi; Pedersen-Takesaki, Radon-Nikodym, Proof of Lemma 5.1: s^\nalgebra(\phi)=1-P where P is the largest projection such that \phi(P)=0.
\item  we call the support of a weight \index{support! of a weight} $\phi$ on a von Neumann algebra $\nalgebra$ the smallest projection $P\in \nalgebra$ such that $\phi\left(\mathbbm{1}-P\right) =0$. Such a projection is denoted by $s^\nalgebra(\phi)$;

\item  we call the left support of $A\in \nalgebra$ \index{support! left} the smallest projection $P\in \nalgebra$ such that $PA=A$. We denote such a projection by $s^\nalgebra_L(A)$;

\item  we call the right support of $A\in \nalgebra$ \index{support! right} the smallest projection $P\in \nalgebra$ such that $AP=A$. We denote such a projection by $s^\nalgebra_R(A)$;

\item in the case where $\nalgebra$ is a concrete von Neumann algebra, we call the support of a vector \index{support! of a vector} $x\in \hilbert$ the smallest projection $P\in \nalgebra$ such that $Px=x$. Such a projection is denoted by $s^\nalgebra(x)$.

\end{enumerate}
\end{definition}

It is no coincidence that we use the word support in all items in the previous definition. They are all related by the next proposition.

\begin{proposition}
\label{supportrelations}
Let $\nalgebra$ be a von Neumann algebra and $\phi$ a semifinite normal weight, then the following relations holds:

\begin{enumerate}[(i)]
\item For each $x\in \hilbert_\phi$, $s^\nalgebra(\omega_x)=s^\nalgebra(x)$, where the vector state $\omega_x \in \mathcal{S}_\nalgebra$ is defined by $\omega_x(A)=\ip{x}{Ax}_\phi$;
\item For each $A\in \nalgebra$, $s^\nalgebra_L(A)=\pi_\phi^{-1}(P)$, where $P$ is the orthogonal projection on $\overline{\Ran(\pi_\phi(A))}$;
\item For each $A\in \nalgebra$, $s^\nalgebra_L(\mathbbm{1}-A)=\pi_\phi^{-1}(P)$ where $P$ is the orthogonal projection on $\Ker(\pi_\phi(A))$;
\item For each $A\in \nalgebra$, $s^\nalgebra_R(A)=\mathbbm{1}-\pi_\phi^{-1}(P)$ where $P$ is the orthogonal projection on $\Ker(\pi_\phi(A))$;
\item For each $A\in \mathcal{N}_\phi$, there exists $N\in N_\phi$ such that $s^\nalgebra(x_A)=\pi_\phi\left(s^\nalgebra(A+N)\right)$ where $x_A=A+N_\phi$;

\item $\displaystyle s^\nalgebra(\mathfrak{N}_\phi/N_\phi)=\bigvee_{x \in \mathfrak{N}_\phi/N_\phi}s^\nalgebra(x)=\pi_\phi\left(s^\nalgebra(\phi)\right)$.

\end{enumerate}
\end{proposition}

\begin{proof}
$(i)$ By definition, 
$$\begin{aligned}
\omega_x\left(\left(\mathbbm{1}-s^\nalgebra(x)\right) A\right) 
&=\ip{x}{\left(\mathbbm{1}-s^\nalgebra(x)\right)Ax}_\phi \\
&=\ip{\left(\mathbbm{1}-s^\nalgebra(x)\right)x}{Ax}_\phi\\
&=0.
\end{aligned}$$
Then, the minimality of $s^\nalgebra(\omega_x)$ implies that $s^\nalgebra(\omega_x)\leq s^\nalgebra(x)$.

On the other hand, $\omega_x\left(\left(\mathbbm{1}-s^\nalgebra(\omega_x\right) A\right)=0 \ \forall A \in \nalgebra$ implies
$$\begin{aligned}
\omega_x\left(\left(\mathbbm{1}-s^\nalgebra(\omega_x)\right)\mathbbm{1}\right) & =\ip{\left(\mathbbm{1}-s^\nalgebra(\omega_x)\right)x}{x}_\phi\\
&=\ip{ \left(\mathbbm{1}-s^\nalgebra(\omega_x)\right)x}{s^\nalgebra(\omega_x)x+\left(\mathbbm{1}-s^\nalgebra(\omega_x)\right)x}_\phi\\
&=\ip{ \left(\mathbbm{1}-s^\nalgebra(\omega_x)\right)x}{\left(\mathbbm{1}-s^\nalgebra(\omega_x)\right)x}_\phi\\
&=0.\\
\end{aligned}$$
 Hence, $\left(\mathbbm{1}-s^\nalgebra(\omega_x)\right)x=0 \Rightarrow s^\nalgebra(\omega_x)x=x \Rightarrow s^\nalgebra(x)\leq s^\nalgebra(\omega_x)$.

$(ii)$ Of course, $P$ is the minimal projection on $B(\hilbert_\phi)$ such that $P\pi_\phi(A)=\pi_\phi(A)$. Let $A^\prime \in \nalgebra^\prime$, since $\pi_\phi(A^\prime)\pi_\phi(A)=\pi_\phi(A)\pi_\phi(A^\prime)$ for any element $y \in \Ran(A)$ we have $\pi_\phi(A^\prime)y\in \Ran(A)$, thus $P$ also commutes with $A^\prime$ and, thanks to Theorem \ref{TDC} it is in $\phi_\phi(\nalgebra)$. Thus $\pi_\phi^{-1}(P) \in \nalgebra$ is the desired minimal projection.

$(iii)$ Let $P$ be the orthogonal projection onto $\Ker(\pi_\phi(A))$, for each $x\in \hilbert_\phi$ decompose $x=x_A+x_A^\perp$ where $x_A\in \Ker(\pi_\phi(A))$ and $x_A^\perp \in \Ker(\pi_\phi(A))^\perp$, then we can verify 
$$\begin{aligned}
P (\mathbbm{1}-A)x&=P(\mathbbm{1}-A)(x_A+x_A^\perp)\\
							&P(\mathbbm{1}-A)(x_A+x_A^\perp-x_A^\perp)\\
							&=Px_A\\
							&=x_A\\
							&=(\mathbbm{1}-A)x\\
\end{aligned}$$
to conclude that $s^\nalgebra_L(\mathbbm{1}-A) \leq P$.

This time, we will show that $P$ is in fact minimal.

Suppose it is not, then $s^\nalgebra_L(\mathbbm{1}-A) < P$ and there exists $x\in \Ker(A)\setminus s^\nalgebra_L(\mathbbm{1}-A)(\hilbert_\phi)$, which means $s^\nalgebra_L(\mathbbm{1}-A)x \neq x$, but then
$$\begin{aligned}
s^\nalgebra_L(\mathbbm{1}-A)(\mathbbm{1}-A)x=s^\nalgebra_L(\mathbbm{1}-A)x\neq x=(\mathbbm{1}-A)x \Rightarrow s^\nalgebra_L(\mathbbm{1}-A)(\mathbbm{1}-A)\neq (\mathbbm{1}-A), 
\end{aligned}$$
which contradicts the definition. Therefore, $P$ is the desired minimal projection.
 
$(iv)$ It follows from the same procedure used in $(ii)$ and $(iii)$ noticing that $\pi_\phi(A)\left(\mathbbm{1}-P\right)=0 \Rightarrow \mathbbm{1}-P$ is a projection on $\Ker(\pi_\phi(A))$ and the maximality implies that $\mathbbm{1}-\pi_\phi(A)$ is the orthogonal projection on $\Ker(\pi_\phi(A))$.

$(v)$ Note that $s^\nalgebra(x_A)x_A=x_A \Rightarrow s^\nalgebra(x_A)A+N_\phi=A+N_\phi \Rightarrow s^\nalgebra(x_A)A=A+N$ for some $N\in N_\phi$, and by the definition of projection it follows that 
$$\begin{aligned}
A+N &=\pi_\phi^{-1}\left(s^\nalgebra(x_A)\right)A=\pi_\phi^{-1}\left(s^\nalgebra(x_A)\right)^2A\\
	&=\pi_\phi^{-1}\left(s^\nalgebra(x_A)\right)(A+N),
\end{aligned}$$ from which we conclude that $\pi_\phi\left(s^\nalgebra(A+N)\right)\leq s^\nalgebra(x_A)$.

Moreover, $$\begin{aligned}
\pi_\phi\left(s^\nalgebra(A+N)\right)x_A	&=\pi_\phi\left(s^\nalgebra\left(A+N\right)\left(A+N_\phi\right)\right)\\
	&=\pi_\phi\left(s^\nalgebra(A+N)\left(A+N+N_\phi\right)\right)\\
	&=A+N+N_\phi\\
	&=x_A.\\
\end{aligned}$$
Thus, $s^\nalgebra(x_A)\leq \pi_\phi\left(s^\nalgebra(A+N)\right)$ and the equality follows.

$(vi)$ Note that, by definition, $\phi\left(s^\nalgebra(\phi)A\right)=\phi(A)$ for all $A\in \mathfrak{N}_\phi$. It follows that $\eta_\phi\left(s^\nalgebra(\phi)A\right)=\eta_\phi(A)$ for all $A \in \mathfrak{N}_\phi$. Thus
$$\pi_\phi\left(s^\nalgebra(\phi)\right)\eta_\phi(A)=\eta_\phi\left(s^\nalgebra(\phi)A\right)=\eta_\phi(A) \ \forall A\in \mathfrak{N}_\phi \Rightarrow \pi_\phi(s^\nalgebra(\phi))\geq \bigvee_{A\in \mathfrak{N}_\phi}s^\nalgebra(\eta_\phi(A)).$$

On the other hand, let $P \in \nalgebra$ be a projection such that $\displaystyle \pi_\phi(P)=\bigvee_{A\in \mathfrak{N}_\phi}s^\nalgebra(\eta_\phi(A))$. Then,
$$\begin{aligned}
\phi(B^\ast A)&=\ip{\eta_\phi(A)}{\eta_\phi(B)}\\
&=\ip{\pi_\phi(P)\eta_\phi(A)}{\eta_\phi(B)}\\
&=\ip{\eta_\phi\left(PA\right)}{\eta_\phi(B)}\\
&=\ip{\eta_\phi\left(PA\right)}{\eta_\phi(B)}\\
&=\phi\left(B^\ast PA\right) \ \forall B \in \mathfrak{N}, \ \forall A \in \mathfrak{N}.\\
\end{aligned}$$

It follows that $\displaystyle P\leq s^\nalgebra(\phi) \Rightarrow \bigvee_{x\in \mathfrak{N}_\phi/N_\phi}s^\nalgebra(x)\leq \pi_\phi\left(s^\nalgebra(\phi)\right)$ and we have the equality.

%****************checar isso, não é necessário nada para garantir que vale para todo elemento da álgebra? o N* e o N deve, ser densos? O peso contínuo?
\end{proof}

\begin{remark}
Note that in items $(ii)$, $(iii)$ and $(iv)$ we used the GNS-representation. This would not be necessary if we had considered a concrete von Neumann algebra. Nevertheless, item $(v)$ is strictly dependent on the GNS construction because of the specific choice of $A$ and $x_A$.
\end{remark}

The next result relates cyclicity and separability, which are connected with the existence of a Modular Operator, with support projection. The reader should be not surprised if those support projections appear again in relative modular theory or even in the Araki-Masuda Noncommutative $L_p$-Spaces.

\begin{proposition}
\label{supportcyclicseparating}
Let $\nalgebra$ be a concrete von Neumann algebra.
\begin{enumerate}[(i)]
\item $\Omega \in \hilbert$ is a separating vector if and only if $s^{\nalgebra}(\Omega)=\mathbbm{1}$;
\item $\Omega \in \hilbert$ is a cyclic vector if and only if $s^{\nalgebra^\prime}(\Omega)=\mathbbm{1}$.
\end{enumerate} 
\end{proposition}

\begin{proof}
$(i)$ It is quite obvious, if $\Omega$ is separating $\left(\mathbbm{1}-s^\nalgebra(\Omega)\right)\Omega=0 \Rightarrow \mathbbm{1}-s^\nalgebra(\Omega)=0$. On the other hand, suppose $A\in \nalgebra$ is such that $A\Omega = 0$. Then $(\mathbbm{1}-A)\Omega=\Omega$. Thus $\Omega \in \Ran(\mathbbm{1}-A)$ and $\mathbbm{1}=s^\nalgebra(\Omega)\leq s^\nalgebra_L(\mathbbm{1}-A)\leq \mathbbm{1} \Rightarrow s^\nalgebra_L(\mathbbm{1}-A)=\mathbbm{1}$. Using now item $(iii)$ in Proposition \ref{supportrelations}, $\Ker(A)=\hilbert \Rightarrow A=0$.

$(ii)$ If $\Omega$ is cyclic, $\left(\mathbbm{1}-s^{\nalgebra^\prime}(\Omega)\right)A\Omega=0$ for all $A\in \nalgebra$. But $\nalgebra\Omega$ is a dense subset of $\hilbert$. Therefore $\mathbbm{1}-s^{\nalgebra^\prime}(\Omega)=0$.

For the converse, we must look carefully at the orthogonal projection $P$ onto $\overline{\nalgebra\Omega}$. First, $AP=PAP$ for all $A\in \nalgebra$. In fact, fixing an element $x\in \hilbert$, $Px\in \overline{\nalgebra\Omega}$. It means that there exists a sequence $(A_n\Omega)_{n\in\mathbb{N}}\subset \nalgebra\Omega$ such that $A_n\Omega \to Px$, but then $(AA_n\Omega)_{n\in\mathbb{N}}\subset \nalgebra\Omega$ is such that $AA_n\Omega \to APx$ and $AA_n\Omega=P(AA_n\Omega)\to PAPx=APx$ as desired. Hence, $PA=(A^\ast P)^\ast=(PA^\ast P)^\ast=PAP=AP$ and it follows that $P\in \nalgebra^\prime$.

With this useful result everything becomes easy, since $\Omega \in \overline{\nalgebra \Omega}$, 
$$\begin{aligned}
P\Omega	=\Omega 	&\Rightarrow \mathbbm{1}=s^{\nalgebra^\prime}(\Omega)	\leq P \leq \mathbbm{1}\\
					& \Rightarrow P=\mathbbm{1}\\
					& \Rightarrow \overline{\nalgebra\Omega}=\hilbert.\\ 
\end{aligned}$$
\end{proof}

\begin{corollary}
Let $\nalgebra$ be a concrete von Neumann algebra and $\Omega \in \hilbert$. Then, the following statements are equivalent:
\begin{enumerate}[(i)]
	\label{cyclic_separating}
\item $\Omega$ is cyclic for $\nalgebra$;
\item $\Omega$ is separating for $\nalgebra^\prime$.
\end{enumerate} 
\end{corollary}
\begin{proof}
	$(\Rightarrow)$ Let $A^\prime\in \nalgebra^\prime$ such that $A^\prime\Omega=0$, then $AA^\prime\Omega=A^\prime A\Omega=0$ for all $A\in \nalgebra$, which means $A^\prime$ vanishes on the dense set $\{A\Omega \ | A\in \calgebra\}$, thus $A^\prime=0$.
	
	$(\Leftarrow)$ Let $P\in \nalgebra$ be the projection onto the closed subspace $\overline{\calgebra\Omega}$ by the previous proposition. Then $(\mathbbm{1}-P)\Omega=0 \Rightarrow P=\mathbbm{1} \Rightarrow \overline{\{A\Omega \ | A\in \calgebra\} }=\hilbert$.
\end{proof}

We have said that considering faithful normal semifinite states is not general, since it may not exist. The next theorem shows that considering weights with the same proprieties on von Neumann algebras is enough. The proof we present below was changed by the author and is based in the one found in \cite{Takesaki2003}.
%******### não lembro se a demonstração é minha
\begin{theorem}
\label{existenceweight}
%Takesaki vol 2 Theorem 2.7
Let $\nalgebra$ be a von Neumann algebra. There exists a faithful normal semifinite weight $\phi$ on $\nalgebra$.
\end{theorem}
\begin{proof}
Consider the set 
\begin{equation}
\label{equationx3}
\mathfrak{J}=\left\{J\subset \mathcal{S}_\nalgebra \ \middle| \ \forall \, \omega_1, \omega_2 \in J, \ \omega_1 \textrm{ is normal and } s^\nalgebra(\omega_1)s^\nalgebra(\omega_2)=0\right\}
\end{equation}
ordered by inclusion. It is not difficult  to see that the chains on this set have a maximal element given by the union of the chain elements. Hence, by Zorn's Lemma, it has a maximal element, say $I$.

We claim that $\displaystyle \sum_{\omega\in I}s^\nalgebra(\omega)=\bigwedge_{\omega\in I} s^\nalgebra(\omega)=\mathbbm{1}$. In fact, if it is not true, there exists an element $B \in \nalgebra$ such that $\displaystyle A=\left(\mathbbm{1}-\sum_{\omega\in I}s^\nalgebra(\omega)\right)B \neq 0$. Define 
$$\begin{aligned}
\omega_A: 	& \left(\sum_{\omega\in I}s^\nalgebra(\omega)\right)\nalgebra \oplus \mathbb{K}A	& \to & \ \mathbb{K} \\
			& \hspace{1.5cm} B+\alpha A												& \mapsto& \ \alpha .\\
\end{aligned}$$

Now, we use the Hahn-Banach Theorem to extend $\omega_A$ to the whole algebra (we will denote the extension by $\omega_A$ as well). This definition ensures that $s^\nalgebra(\omega_A)>0$ and $s^\nalgebra(\omega_A)s^\nalgebra(\omega)=0$ for all $\omega \in I$, thus $I\cup\{\omega_A\}>I$. This contradicts the maximality of $I$.

Consider now the weight 
\begin{equation}
\label{equationx4}
\displaystyle\phi(A)=\sum_{\omega\in I}\omega(A), \quad A\in \nalgebra_+.
\end{equation}
It is normal, because it is the limit of an increasing net of states and it is also semifinite because $\displaystyle s_J=\sum_{\omega\in J} s^\nalgebra(\omega)$, $J\subset I$ finite, form an increasing net of projections with the index set directed by inclusion, so it converges to its least upper bound, that is, to $\mathbbm{1}$, and it follows that
$$\displaystyle \bigcup_{\substack{J\subset I \\ J \textrm{ finite}}}\left(\sum_{\omega\in J}s^\nalgebra(\omega)\right)\nalgebra \subset \mathfrak{M}_\phi \Rightarrow \overline{\left(\bigcup_{\substack{J\subset I \\ J \textrm{ finite}}}\sum_{\omega\in J}s^\nalgebra(\omega)\right)\nalgebra}=\left(\sum_{\omega\in I}s^\nalgebra(\omega)\right)\nalgebra=\mathbbm{1}\nalgebra \subset \overline{\mathfrak{M}_\phi}.$$

Finally, $\phi$ is faithful because if $\phi(A)=0$ for $A\in \nalgebra_+$, then $\omega(A)=0$ for all $\omega\in I$, so $A \in \mathbbm{1}-s^\nalgebra(\omega)$, or equivalently, $A s^\nalgebra(\omega)=0$ for all $\omega\in I$. Hence, $$\displaystyle A=A\sum_{\omega\in I}s^\nalgebra(\omega)=\sum_{\omega\in I}A s^\nalgebra(\omega)=0.$$

\end{proof}

The following corollaries can be seen as consequences of Theorem \ref{TGN} or Theorem \ref{GNSweight} and Theorem \ref{existenceweight}. Before we start working with weights, let us analyse a particular case where states are sufficient for the theory.

\begin{proposition}
Let $\nalgebra$ be a concrete von Neumann algebra. The following are equivalent:
\begin{enumerate}[(i)]
\item $\nalgebra$ is $\sigma$-finite;
\item $\nalgebra$ admits a faithful normal state;
\item $\nalgebra$ is isometrically isomorphic to a von Neumann algebra which admits a cyclic and separating vector.
\end{enumerate} 
\end{proposition}

\begin{proof}
$(i)\Rightarrow (ii)$ it is a scholium of Theorem \ref{TGN}. From the hypothesis, the maximal element $I$ of $\mathcal{J}$ defined by equation \eqref{equationx3} must be countable, so just change the definition on equation (\ref{equationx4}) to
$$\displaystyle\phi(A)=\sum_{n \in \mathbb{N}}\frac{1}{2^n}\omega_n(A) \quad \forall A\in \nalgebra_+.$$

Since it converges in norm, $\frac{\phi}{\|\phi\|}$ is a state and the desired properties follows from the very same procedure used in Theorem \ref{existenceweight}.

$(ii)\Rightarrow (iii)$ let $\omega$ be the faithful normal state, by Theorem \ref{nonexpensive} $\|\pi_\omega(A)\|\leq \|A\|$. On the other hand, by definition presented in \ref{GNS}, $[B]=B+\{0\}\Rightarrow \|[B]\|=\|B\|$, therefore $\pi_\omega$ is isometric because $\|\pi_\omega(A)([\mathbbm{1}])\|=\|[A]\|=\|A\|$.

The normality of $\omega$ warrants that $\pi_\omega(\nalgebra)$ is also a von Neumann algebra.

Finally, the cyclic and separating vector is $[\mathbbm{1}]$, because $\pi_\omega(\nalgebra)[\mathbbm{1}]=[\nalgebra]=\hilbert_\omega$ and $\pi_\omega(A)[\mathbbm{1}]=0 \Leftrightarrow [A]=0 \Leftrightarrow A=0$.

$(iii)\Rightarrow (i)$ Let $\pi$ be the representation and $\Omega$ its cyclic and separating vector. Let $\{P_i\}_{i\in I}$ be a family of mutually orthogonal projection. Set $\displaystyle P=\bigvee_{i\in I}P_i=\sum_{i\in I}P_i$, then
\begin{equation}
\label{equationx5}
\begin{aligned}
\|\pi(P)\Omega\|^2	&=\ip{\pi(P)\Omega}{\pi(P)\Omega}\\
					&=\ip{\sum_{i\in I}P_i\Omega}{\sum_{i\in I}P_i\Omega}\\
					&=\sum_{i\in I}\ip{P_i\Omega}{P_i\Omega}\\
					&=\sum_{i\in I}\|\pi(P_i)\Omega\|.\\ 
\end{aligned}
\end{equation}
This means that the sum in equation \eqref{equationx5} has at most a countable number of non-null terms, which means $\pi(P_i)\Omega=0 \Rightarrow \pi(P_i)=0$ apart from a countable subset of $I$.

\end{proof}

%*************************** Fourth Section  *****************************

\section{Square Roots and Polar Decomposition of an Operator}
\label{SecSRoO}

Since we have already provided a $C^\ast$-algebra with a functional calculus, we could simply use it to get the square root of an element. Although this way might be easier, the next result will show the existence of square root in a different way and with some interesting consequences.  

The next result is a original proof of a well-known theorem which is based on some other proofs that can be found in classical books of operator theory.

\begin{theorem}
\label{sqrt}
Let $\calgebra$ be a $C^\ast$-algebra and let $A\in \calgebra$ be a positive operator. Then there exists a unique positive operator $B \in \calgebra$, such that $B^2=A$. Furthermore, $B$ is in the closure of the $\ast$-algebra generated by $A$ and $\mathbbm{1}$.
\end{theorem}

\begin{proof}
First we adjoin an identity, $\mathbbm{1}$, to the algebra.

The case $A=0$ is trivial. If $A \neq 0$, we can set $\tilde{A}=\frac{A}{\|A\|}$. Since $A$ is positive, $\|1-\tilde{A}\|\leq 1$.

Firstly, suppose that there exists $0<\alpha<1$ such that $\|1-\tilde{A}\|< \alpha$. Let $\overline{\llbracket A\rrbracket}$ be the closure of the $\ast$-algebra generated by $\tilde{A}$ and let
$$\displaystyle X=\left\{T \in \overline{\llbracket A\rrbracket} \ \middle| \ \|\mathbbm{1}-T\|\leq 1-\sqrt{1-\alpha}=\beta \right \} \neq \emptyset $$
be provided with the metric induced by the norm.

Note that the elements of $X$ are self-adjoint, thanks to Proposition \ref{positiveeq1} and because $A$ is also self-adjoint.

The function $f: X \rightarrow X$ given by $\displaystyle f(B)=B+\tfrac{1}{2}(\tilde{A}-B^2)$ is well defined, since
$$\|\mathbbm{1}-f(T)\|=\left|\left|\frac{1}{2}[(\mathbbm{1}-T)^2+(\mathbbm{1}-\tilde{A})]\right|\right|\leq \frac{1}{2}\left((1-\sqrt{1-\alpha})^2+ \alpha\right)\leq 1-\sqrt{1-\alpha}=\beta.$$

As the operators in $X$ commute with each other,
$$\begin{aligned}
\|f(B)-f(C)\|	& =\left|\left|B-C-\frac{1}{2}B^2+\frac{1}{2}C^2\right|\right|=\left|\left|\frac{1}{2}\left[(I-B)+(I-C)\right](B-C)\right|\right|\\
							& = \frac{1}{2}\left( \left|\left|I-B\right|\right|+\left|\left|I-C\right|\right| \right) \ \|B-C\| \leq \beta \|B-C\|.
\end{aligned}$$

It follows from Banach's Fixed Point Theorem that $f$ has exactly one fixed point, that is, there exists a unique $B \in X$ such that
$$B=B+\frac{1}{2}(\tilde{A}-B^2) \Rightarrow \tilde{A}=B^2 \Rightarrow A=\left(\sqrt{\|A\|}B\right)^2.$$ 

For the case $\|\mathbbm{1}-\tilde{A}\|=\|\mathbbm{1}-\frac{A}{\|A\|}\|=1$, consider a sequence of positive operators defined by $\displaystyle A_n=\tilde{A}+\frac{1}{2n}\mathbbm{1}$ and $\displaystyle \tilde{A}_n=\frac{A_n}{\|A_n\|}=\frac{A_n}{1+\frac{1}{2n}}$. Of course, $\tilde{A}_n \rightarrow \tilde{A}$ and thanks to Theorem \ref{ST},
$$\begin{aligned}
 \|\mathbbm{1}-\tilde{A}_n\|
 &=\left|\left|\left(1-\frac{1}{2n\|A_n\|}\right)\mathbbm{1}-\frac{\tilde{A}}{\|A_n\|}\right|\right|\\
 &\leq \max\left\{1-\frac{1}{2n\|A_n\|},1-\frac{1}{2n\|A_n\|}-\frac{1}{\|A_n\|}\right\}\\
 &=1-\frac{1}{2n\|A_n\|}.\\
\end{aligned}$$
Hence, for each $n\in\mathbb{N}$ there exists a unique $B_n$ with the desired properties such that $\tilde{A}_n=B_n^2$. Note that all the $B_n$ are in fact in the same algebra and commute with each other. Furthermore, note that $n>m \Rightarrow \tilde{A}_n\leq \tilde{A}_m \Rightarrow B_n \leq B_m$, since $A_n$ and $A_m$ commute, thus we have that, for every $n>m$,
\begin{equation}
\label{equationxx5}
\begin{aligned}
\|(B_m-B_n)(x)\|^2 	&=\ip{(B_m-B_n)x}{(B_m-B_n)x}\\
&=\ip{\tilde{A}_m x}{x}+\ip{\tilde{A}_n x}{x}-2\ip{B_m B_n x}{x} \\
&\leq \ip{\tilde{A}_m x}{x}+\ip{\tilde{A}_n x}{x}-2\ip{B_n^2 x}{x}\\
&=\ip{\tilde{A}_m x}{x}-\ip{\tilde{A}_n x}{x}\\
&= \ip{(\tilde{A}_m-\tilde{A}_n)(x)}{x}\\
&\leq\|\tilde{A}_m-\tilde{A}_n\| \ \|x\|^2 \\
&\leq\left[\left(\frac{1}{2m\|A_m\|}-\frac{1}{2n\|A_n\|}\right)+\left(\frac{1}{\|A_m\|}-\frac{1}{\|A_n\|}\right)\right] \|x\|^2 \\
&\begin{aligned}\Rightarrow \|B_m-B_n\| &\leq \|\tilde{A}_m-\tilde{A}_n\|^{\frac{1}{2}}\\
&\leq \sqrt{\left[\left(\frac{1}{2m\|A_m\|}-\frac{1}{2n\|A_n\|}\right)+\left(\frac{1}{\|A_m\|}-\frac{1}{\|A_n\|}\right)\right]}.
										\end{aligned}
\end{aligned}
\end{equation}
Hence, it is a Cauchy sequence and therefore $B_n \rightarrow B \in \overline{\llbracket\tilde{A}\rrbracket}$ and clearly $B^2=\tilde{A}$.

Moreover, equation \eqref{equationxx5} also ensures a continuity property in this particular case, with which we can prove uniqueness. In fact, if $C$ is a positive operator such that $C^2=\tilde{A}$ then, taking the limit, we have
$$\|B_m-C\| \leq \|\tilde{A}_m-\tilde{A}\|^{\frac{1}{2}} \Rightarrow \|B-C\|=0.$$

\iffalse  Prova da Literatura.
A última coisa que devemos fazer é garantir a unicidade, para isso, suponha que exista $C: \hilbert \rightarrow \hilbert$ positivo tal que $C^2=A$, mas então $C A= C^2 C=C^2 C= A C $ e portanto $C$ comuta com $B \in [\tilde{A}]=[A]$. Portanto, para $x\in \hilbert$ arbitrário e $y=(B-C)x$ teremos 
$$\displaystyle 0 \leq \ip{By}{y}+ \ip{Cy}{y}=\ip{(B+C)y}{y}= \ip{(B^2-C^2)x}{y}=0$$
$$\Rightarrow \ip{By}{y}=\ip{Cy}{y}=0.$$

Seja $D\in \mathcal{B}(\hilbert)$ positivo e tal que $D^2=C$, então $\|Dy\|^2=\ip{Dy}{Dy}=\ip{D^2y}{y}=0 \Rightarrow Cy=D^2y=0$, analogamente $By=0$. Finalmente, uma vez que $B$ e $C$ são auto-adjuntos
$$\|(B-C)x\|^2=\ip{(B-C)x}{(B-C)x}=\ip{(B-C)^2 x}{x}= \ip{(B-C)y}{x}=0 \ \forall x \in \hilbert \Rightarrow B=C$$
\fi
\end{proof}

\begin{definition}
Let $A:\hilbert\rightarrow \hilbert$ be a positive operator. We denote by $\sqrt{A}$ the unique positive operator such that $(\sqrt{A})^2=A$.
\end{definition}

Equation \eqref{equationxx5} in the proof of the theorem gives us:

\begin{scholium}
Let $(A_n)_{n\in\mathbb{N}}\subset \mathcal{B}(\hilbert)$ be a monotonic decreasing sequence of positive operators, then
$$\left\|\sqrt{A_n}-\sqrt{A_m}\right\| \leq \|A_n-A_m\|^{{\frac{1}{2}}}.$$
\end{scholium}

\begin{definition}
Let $A\in \mathcal{B}(\hilbert)$, we define $|A|\doteq\sqrt{A^\ast A}$.
\end{definition}

Note that the modulus function is always well defined since $A^\ast A$ is always a positive operator.

\begin{proposition}
	\label{EquiPositive}
$A:\hilbert \to \hilbert$ is a positive operator if and only if $A$ is self-adjoint and $\ip{Ax}{x} \geq 0$ for all $x \in \hilbert$.
\end{proposition}
\begin{proof}

If $A$ is positive, it is automatically self-adjoint by hypothesis, so it is enough to show the second condition. From Theorem \ref{sqrt}, we know that there exists a positive $B$ such that $A=B^2$. Thus
$$\ip{Ax}{x}=\ip{B B^* x}{x}=\|Bx\|^2 \geq0 \ \forall x \in \hilbert.$$

Reciprocally, it follows from the Cauchy-Schwarz inequality that $\ip{\frac{A}{\|A\|}x}{x} \leq \|x\|^2$ and, by hypothesis, $\ip{Ax}{x} \geq 0$ for all $x \in \hilbert$. Therefore, 
$$\begin{aligned}
&\left|\ip{\left(\mathbbm{1}-\frac{A}{\|A\|}\right)x}{x}\right|= \|x^2\|-\frac{1}{\|A\|}\ip{Ax}{x} \leq \|x\|^2\\
&\Rightarrow \left|\left|\mathbbm{1}-\frac{A}{\|A\|}\right|\right|=\sup_{\|x\|= 1}{\left|\ip{\left(\mathbbm{1}-\frac{A}{\|A\|}\right)x}{x}\right|}\leq 1.\end{aligned}$$

Finally, from Proposition \ref{positiveeq1}, we conclude that $A$ is positive.

\end{proof}

\begin{definition}
An operator $U \in \mathcal{B}(\hilbert)$ is said to be a partial isometry \index{partial isometry} if $\|Ux\|=\|x\|$ for all $x \in \ker{U}^\perp$. The subspace $\ker{U}^\perp$ is called the initial space of $U$ and $\ran{U}$ is called the final space of $U$.
\end{definition}

\begin{proposition}\index{operator! polar decomposition}
Let $A \in \mathcal{B}(\hilbert)$. Then, there exists a partial isometry $U$ with initial space $\Ker{A}^\perp$ and final space $\overline{\ran{A}}$ such that $A=U|A|$.
\end{proposition}
\begin{proof}
Notice that
\begin{equation}
\label{partialiso}
\|Ax\|^2=\ip{Ax}{Ax}=\ip{A^* Ax}{x}=\ip{|A|^2x}{x}=\ip{|A|x}{|A|x}=\|\, |A|x\|^2.
\end{equation}

Then, we can define define $\tilde{U}: \ran{|A|}\rightarrow \ran{A}$ by $U|A|x\doteq Ax$ for all $x \in \hilbert$. From the definition, it follows that $ \|\tilde{U}|A|x\|=\|\, |A|x\|$ and hence $\tilde{U}$ is a partial isometry. We can also extend $\tilde{U}$ to an operator $\bar{U}: \overline{\ran{|A|}} \to \overline{\ran{A}}$ using continuity. Finally, define $U \in \mathcal{B}(\hilbert)$ by
$$Ux=\begin{cases} \bar{U}x &\textrm{ if } x\in \overline{\ran{|A|}} \\ 0 &\textrm{ if } x \in \overline{\ran{|A|}}^\perp\end{cases}.$$
Thereafter $U$ is a partial isometry with initial space $\overline{\ran{|A|}}$ and final space $\overline{\ran{A}}$ such that $Ax=U|A|$.

All that remains to show is that $\overline{\ran{|A|}}=\ker{A}^\perp$ to conclude the assertion about the isometry. We begin by noticing that the equation \eqref{partialiso} says $\|Ax\|=\|\ |A|x\|$, hence $\ker{A}=\ker{|A|}$. Thus $x \in \ran{|A|}  \Rightarrow \exists y \in \hilbert$ such that $x=|A|y \Rightarrow \ip{x}{z}=\ip{|A|y}{z}=\ip{y}{|A|z}=0$  for all $z \in \ker{A}=\ker{|A|} \Rightarrow x\in \ker{A}^\perp$ and since $\ker{A}^\perp$ is closed $\overline{\ran{|A|}}\subset\ker{A}^\perp$.

On the other hand, let $x\in \ker{A}^\perp$ and $f\in \hilbert^*$ be a continuous functional that vanishes in $\ran{|A|}$. Thanks to Riesz Theorem, there exists $z \in \hilbert$ such that $f(y)=\ip{z}{y}$ for all $y \in \hilbert$. Note now that $\|Az\|^2=\| |A|z \|=\ip{Az}{Az}=\ip{z}{|A|^2 z}=f(|A|^2 z)=0 \Rightarrow z \in \ker{A}$. Finally, $f(x)=\ip{x}{z}=0$ for all $f\in \hilbert^*$ that vanishes in $\ran{|A|}$. Hence, as a consequence of the Hahn-Banach Theorem\footnote{see Theorem \ref{C2TMD}.}, $x \in \overline{\ran{|A|}}$.

\end{proof}

\section{Unbounded, Closed, and Affiliated Operators}

We finish this chapter presenting the definition of unbounded, closed, and affiliated operators. This theory is due to J. von Neumann and M. Stone, who developed it in their attempt to provide a mathematical formalization of quantum mechanics. These operators will play a very important role in this work, for obvious reasons.

\begin{definition}[Unbounded Operator]
Let $\mathcal{X}$ and $\mathcal{Y}$ be Banach spaces and \mbox{$\Dom{A}\subsetneq \mathcal{X}$} a subspace. An unbounded operator \index{operator! unbounded} is a linear map $A:\mathcal{X}\supset \Dom{A} \to \mathcal{Y}$ that cannot be continuously extended to the whole Banach space. 
\end{definition}

\begin{definition}
Let $\mathcal{X}$ and $\mathcal{Y}$ be Banach spaces and let $A:\mathcal{X}\supset \Dom{A}\to \mathcal{Y}$ be an operator. The graph of $A$ is defined as the set
$$\mathcal{G}(A)\doteq\left\{(x,Ax) \in \mathcal{X}\times \mathcal{Y} \ | \ x \in \Dom{A} \right\}.$$

We are considering $\mathcal{X}\times \mathcal{Y}$ provided with the product topology.
\end{definition}

\begin{definition}
Let $\mathcal{X}$ and $\mathcal{Y}$ be Banach spaces and let $\Dom{A}\subset \mathcal{X}$ be a subspace. An unbounded operator $A:\mathcal{X}\supset \Dom{A}\to \mathcal{Y}$ is said to be
\begin{enumerate}[(i)]
	\item densely defined, if $\overline{\Dom{A}}=\mathcal{X}$ \index{operator! densely defined};\index{operator! densely defined}
	\item closed \index{operator! closed} if its graph is closed, \ie, for any sequence $(x_n,Ax_n)_{n\in\mathbb{N}} \in \mathcal{G}(A)$ with $(x_n,Ax_n) \to (x,y)$ (which automatically is in $\overline{\mathcal{G}(A)}$) we have $(x,y)\in\mathcal{G}(A)$ which means $x\in \Dom{A}$ and $y=Ax$;\index{operator! closed}
	\item closable if $A$ admits a closed extension. We denote its closure by $\overline{A}$.
\end{enumerate}  
\end{definition}

From now on, we will use standard results about densely defined and closed operators without mentioning. More details about the topic can be found in \cite{kato95}, \cite{megginson98}, and \cite{reedv1}.

\begin{definition}
Let $\mathcal{X}$ and $\mathcal{Y}$ be Banach spaces, $A:\mathcal{X}\supset \Dom{A}\to \mathcal{Y}$ a closable operator. A core (or essential domain) of $A$ \index{operator! core} is a subset $D\subset \Dom{A}$ such that $\overline{\left.A\right|_D}=\overline{A}$. 
\end{definition}

%Eq. 3.13, Lemma 4 of von Neumann Algebras and Radon-Nikodym Theorem, H. Araki, pag. 319
\begin{lemma}
\label{lema3.13}
Let $A$ be a closable (or densely defined) positive operator and $D$ a core of $A$. Then $D$ is a core of $A^\gamma$, $0\leq \gamma\leq1$, and
$$\|A^\gamma x\|^2\leq \|x\|^2+\|Ax\|^2 \quad \forall x\in \Dom{A^\gamma}.$$
\end{lemma}
\begin{proof}
Let $A=\displaystyle\int{p dE_p}$ be the spectral decomposition of $A$ and $x\in \Dom{A^\gamma}\subset \Dom{A}$.
$E_{(0,1)}$ and $\mathbbm{1}-E_{(0,1)}$ are orthogonal \hspace{1pt} projections, furthermore \hspace{1pt} $E_{(0,1)} A^\gamma\leq \mathbbm{1}$ and  \mbox{$(\mathbbm{1}-E_{(0,1)})A^\gamma\leq(\mathbbm{1}-E_{(0,1)})A\leq A$} because of its spectral decomposition, hence, for every $x\in \Dom{A^\gamma}$,
\begin{equation}
\begin{aligned}
\label{eq3.13}
\|A^\gamma x\|^2	&=\|E_{(0,1)}A^\gamma x\|^2+\|(\mathbbm{1}-E_{(0,1)})A^\gamma x\|^2\\
					&\leq\|x\|^2+\|Ax\|^2.\\
\end{aligned}
\end{equation}

Now, let $x\in \Dom{\overline{A^\gamma}}$, also $x\in \Dom{\overline{A}}$. Then, there exists $(x_n)_{n\in\mathbb{N}}\subset D$, such that $x_n\rightarrow x$ and $\overline{A}x_n\rightarrow \overline{A}x$. By equation \eqref{eq3.13}, $\overline{A^\gamma} x_n \rightarrow \overline{A^\gamma} x$. Consequently, $D$ is a core for $\overline{A^\gamma}$, as well as for $A^\gamma$. 

\end{proof}

\begin{definition}[Adjoint]
	\index{operator! adjoint}
	Let $\hilbert_1, \hilbert_2$ be Hilbert spaces and $A:\hilbert_1\supset \Dom{A}\to \hilbert_2$ a densely defined unbounded operator, we set:
	\begin{enumerate}[(i)]
	\item $\displaystyle \Dom{A^\ast} \doteq \left\{ y \in \hilbert_2 \ \middle| \ \Dom{A}\ni x \xrightarrow{\omega_y} \ip{y}{Ax}_2 \textrm{ is a continuous linear functional}\right\}$;
	\item for $y\in\Dom{A^\ast}$, $A^\ast(y)=z$ where $z\in \hilbert_1$ is the unique vector such that \mbox{$\omega_y(x)=\ip{z}{x}_1$} for all $x\in\hilbert_1$.
	\end{enumerate}

One can  prove that the operator $A^\ast$ defined by the conditions above is linear. We call the operator $A^\ast$ the adjoint of $A$.
\end{definition}

An interesting result states that an operator is closable if and only if its adjoint is densely defined. Again, we stress that all standard results about the adjoint operator can be found again in \cite{kato95}, \cite{megginson98}, and \cite{reedv1}. 

\begin{definition}
	Let $\nalgebra\subset B(\hilbert)$ be a von Neumann algebra, we say that a linear operator $A:\Dom{A} \to \hilbert$ is affiliated to $\nalgebra$ if, for every unitary operator $U\in\nalgebra^\prime$, $UAU^\ast=A$. We denote that an operator is affiliated to $\nalgebra$  by $A\eta\nalgebra$ and the set of all affiliated operators \mbox{by $\nalgebra_\eta$}\glsdisp{affil}{\hspace{0pt}}.
\end{definition}

A very interesting way to look at this definition is through spectral projections. By this definition, the spectral projections of an affiliated operator belong to the von Neumann algebra. In fact, an equivalent condition to an operator $A$ being affiliated to a von Neumann algebra is that the partial isometry in the polar decomposition of $A$ and all the spectral projections of $|A|$ lie in the von Neumann algebra. This guarantees that an affiliated operator can be approached by algebra elements in the spectral sense.

\begin{lemma}
	Let $\nalgebra$ be a von Neumann algebra and $A\eta \nalgebra$.
	\begin{enumerate}[(i)]
		\item If $A$ is closable, then $\overline{A}\eta\nalgebra$;
		\item If $A$ is densely defined, then $A^\ast\eta \nalgebra$.
	\end{enumerate}
\end{lemma}
\begin{proof}
	$(i)$ The condition just warrants the existence of $\overline{A}$. Since it exists, $U\in\nalgebra^\prime$, $UAU^\ast=A \Rightarrow U\overline{A}U^\ast=\overline{A}$.
	
	$(ii)$ The condition plays the same role and, since $A^\ast$ exists, we have $U\in\nalgebra^\prime$, $UAU^\ast=A \Rightarrow UA^\ast U^\ast=A^\ast$.
	
\end{proof}
% flatex input end: [Chapter2/chapter2.tex]
%C algebras
% flatex input: [Chapter3/chapter3.tex]
%******************************** Fourth Chapter ************************

\chapter{Dynamical Systems, KMS States and their Physical Meaning}  %Title of the Fourth Chapter
\label{chapKMS}

KMS states are a generalization of quantum thermal equilibrium states given by the density matrix in finite-dimensional Hilbert spaces, defined, from the hamiltonian $H$, as follows
$$\rho_\beta=\frac{e^{- \beta H}}{Tr(e^{- \beta H})} \quad, \qquad \omega_{\rho_\beta}(A)=Tr(\rho_\beta A).$$

It is obvious in this situation that $\rho_\beta$ and $\omega_{\rho_\beta}$ admit an analytic extension for complex $\beta$.

Unfortunately, the density matrix usually does not survive the thermodynamic limit in an infinite-dimensional Hilbert space, but the KMS property does. Hence, Kubo, Martin and Schwinger, for which the initialism KMS stands for, proposed in \cite{kubo57} and \cite{martin59} that the analytic extension of the state, as it will be precisely defined soon, should define the thermal equilibrium states.

In this chapter we will present the general definition and main properties of dynamical systems in Operator Algebras context, an extensive discussion on mathematical aspects of KMS states, and finally a physical discussion on the topic.

%*************************** First Section  *******************************

\section{$C^\ast$-Dynamical Systems, $W^\ast$-Dynamical Systems, and Analytic Elements}
\label{sectionKMS} \index{dynamical system}

This section is devoted to present Dynamical Systems in the Operator Algebras context, as well as, the analytic elements, which play a crucial role for KMS state properties and definition.

\begin{definition}[$C^\ast$-Dynamical System]
\index{dynamical system! $C^\ast$}
A $C^\ast$-dynamical system $(\calgebra,G,\alpha)$ consists of a $C^\ast$-algebra $\calgebra$, a locally compact group $G$ and a strongly continuous homomorphism $\alpha$ of $G$ in $Aut(\calgebra)$.
\end{definition}

We clarify here that $Aut(\calgebra)$ is the set of all $\ast$-automorphisms of the $C^\ast$-algebra $\calgebra$.

\begin{definition}[$W^\ast$-Dynamical System]
\index{dynamical system! $W^\ast$}
A $W^\ast$-dynamical system $(\nalgebra,G,\alpha)$ consists of a von Neumann algebra $\nalgebra$, a locally compact group $G$ and a weakly continuous homomorphism $\alpha$ of $G$ in $Aut(\nalgebra)$.
\end{definition}

In particular, we denote by $(\calgebra,\alpha)$ the $C^\ast$-dynamical system with $\alpha$ a one-parameter group, $\mathbb{R} \ni t \mapsto \alpha_t \in Aut(\calgebra)$.

\begin{notation}
Let $X$ be a Banach space and $F\subset X^\ast$ a closed subspace, where $F$ is such that either $F=X^\ast$ or $F^\ast=X$. We denote by $\sigma(X,F)$ the locally convex topology of $X$ induced by functionals in $F$.
\end{notation}

%Note that either $\sigma(X,F)$ is the strong topology (when $F^\ast=X$) or $\sigma(X,F)$ is the weak topology (when $F=X^\ast$) due to the natural embedding of $X$ in $X^{\ast \ast}$.

\begin{definition}[Analytic Elements]
\label{Anal}
\index{analytic elements}
Let $\alpha$ be a one-parameter $\sigma(X,F)$-continuous group of isometries. An element $A \in X$ is analytic for $\alpha$ if there exists $\gamma\neq0$ and a function $f:\strip{\gamma}\to X$, where $\strip{\gamma}=\left\{z \in \mathbb{C} \ \middle | \ 0<\sgn{\gamma}\Im{z} < |\gamma|\right\}$, such that
\begin{enumerate}[(i)]
{
\item $f(t)=\alpha_t(A) \quad \forall t \in \mathbb{R}$;
\item $z \mapsto \varphi(f(z))$ is analytic in the strip $\strip{\gamma}$ for all $\varphi \in F$.
}\glsdisp{strip}{\hspace{0pt}}
\end{enumerate}
\end{definition}

\begin{notation}
	We will usually denote the function $f$ above as $f(z)=\alpha_z(A)$.
\end{notation}

Analytic elements will play an important role as these elements are usually easier to work with and the set of analytic elements is a dense subset, as we will see.

\begin{proposition}
\label{AnalDense}
Let $\alpha$ be a $\sigma(X,F)$-continuous group of isometries and denote \mbox{by $X_\alpha$} the set of entire elements of $X$ (analytic in the whole $\mathbb{C}$), then $$\overline{X_\alpha}^{\sigma(X,F)}= X .$$
\end{proposition}
\begin{proof}
Let $A\in X$ and define
\begin{equation}
\label{formulaAE}
A_n(z)=\sqrt{\frac{n}{\pi}}\int_\mathbb{R}{e^{-n(t-z)^2}\alpha_t(A) dt}.
\end{equation}
Note that $A_n$ is well defined since $e^{-n(t-z)^2}$ is an integrable function and intuitively $A_n=A_n(0)$ will approach $A$, because the coefficient function approaches Dirac's delta distribution.

First, note that, for $y \in \mathbb{R}$,
\begin{equation}
\label{1}
\begin{aligned}
A_n(y)	&=\sqrt{\frac{n}{\pi}}\int_\mathbb{R}{e^{-n(t-y)^2}\alpha_t(A) dt}\\
				&=\sqrt{\frac{n}{\pi}}\int_\mathbb{R}{e^{-n{t^\prime}^2}\alpha_{t^\prime+y}(A) dt} \\
				&=\alpha_y\left(A_n\right).	\\
\end{aligned}
\end{equation}

In order to show density, suppose $\varphi\in F$ and $\varepsilon>0$ are given. Then there exists a $\delta>0$ such that $\|\varphi(\alpha_t(A)-A)\|=\|\varphi(\alpha_t(A)-\alpha_0(A))\|\leq \frac{\varepsilon}{2}$ for every $t\in \mathbb{R}$ with $|t|<\delta$, because $\alpha$ is $\sigma(X,F)$-continuous. Now, choose $N \in \mathbb{N}$ such that for all $n>N$
$$\sqrt{\frac{n}{\pi}}\int_{\mathbb{R}\setminus{(-\delta,\delta)}}{e^{-nt^2}dt} < \frac{\varepsilon}{4 \|\varphi\| \|A\|} \ .$$

It follows that, for all $n>N$,
\begin{equation}
\label{2}
\begin{aligned}
|\varphi(A_n-A)|	& =	\left| \varphi \left( \sqrt{\frac{n}{\pi}}\int_\mathbb{R}{e^{-nt^2}\alpha_t(A)dt}-\sqrt{\frac{n}{\pi}}\int_\mathbb{R}{e^{-nt^2}A \ dt} \right) \right | \\
									& \leq \left| \sqrt{\frac{n}{\pi}}\int_\mathbb{R\setminus(-\delta,\delta)}{e^{-nt^2}\varphi\left(\alpha_t(A)-A\right)dt}\right| \\
									& \hspace{2.2cm} + \left|\sqrt{\frac{n}{\pi}}\int_\mathbb{(-\delta,\delta)}{e^{-nt^2}\varphi\left(\alpha_t(A)-A\right)dt}\right| \\
									& \leq  \sqrt{\frac{n}{\pi}}\int_{\mathbb{R}\setminus(-\delta,\delta)}{e^{-nt^2}\|\varphi\| \left( \|\alpha_t(A)\|+\|\alpha_0(A)\|\right)dt} \\
									&\hspace{1.95cm}+\sqrt{\frac{n}{\pi}}\int_\mathbb{(-\delta,\delta)}{e^{-nt^2}\|\varphi\left(\alpha_t(A)-A\right)\|dt} \\
									& < \varepsilon.
\end{aligned}
\end{equation}

This shows that $A_n \rightarrow A$ in the topology $\sigma(X,F)$. Hence, all that remains to show is that the $A_n(z)$'s are entire analytic. 

Again, suppose that $\varphi\in F$. Using the inequality $\left|\varphi\left(\tau_t(A)\right)\right|\leq \|\varphi\| \|A\|$, we conclude that
$$\begin{aligned}
\left|\frac{\varphi(A_n(z))-\varphi(A_n(z_0))}{z-z_0}-\sqrt{\frac{n}{\pi}}\int_\mathbb{R}{2n(t-z)e^{-n(t-z)^2}\varphi(\alpha_t(A))}dt\right|\\
&\hspace{-7.6cm}=\sqrt{\frac{n}{\pi}}\left|\int_\mathbb{R}{\left(\frac{e^{-n(t-z)^2}-e^{-n(t-z_0)^2}}{z-z_0}-2n(t-z)e^{-n(t-z)^2}\right)\varphi(\alpha_t(A))}dt\right|\\
&\hspace{-7.6cm}\leq\|\varphi\| \|A\|\sqrt{\frac{n}{\pi}}\int_\mathbb{R}{\left|\frac{e^{-n(t-z)^2}-e^{-n(t-z_0)^2}}{z-z_0}-2n(t-z)e^{-n(t-z)^2}\right|dt} 
.
\end{aligned}$$
Notice that the integral on the right-hand side goes to zero when $z\to z_0$ and the entire analyticity follows.

\end{proof}

We say that a vector-valued complex function $f$ is strong-analytic at an interior point $z_0$ of its domain if exists the limit $\displaystyle \lim_{z\to 0}{\frac{f(z_0+z)-f(z_0)}{z}}$, where the limit understood as a norm-limit in the vector space. Analogously, we say that a function is weak-analytic if, for every continuous linear functional $\phi$ on the vector space, we have $\phi\circ f$ is analytic. The next result will relate the two concepts. The relation with Definition \ref{Anal} is direct and clear.

\begin{proposition}
	Let $X$ be a Banach space, $D\subset \mathbb{C}$ an open subset, and let $f:D\to X$ a function. Then, $f$ is weak-analytic if and only if $f$ is strong analytic
\end{proposition}
\begin{proof}
	$(\Rightarrow)$ For every continuous linear functional $\phi\in X^\ast$. For every $z_0\in D$ there exists $r>0$ such that $D(z_0,r)\subset \overline{D(z_0,2r)}\subset D$. Now, for every element $z,w\in D(0,r)$, the Cauchy Integral formula for the circle $C=\{\zeta\in \mathbb{C} \ | \ |\zeta-z_0|=2r\}$ gives us that
	$$\begin{aligned}
	\left|\phi\left(\frac{f(z_0+z)-f(z_0)}{z}-\frac{f(z_0+w)-f(z_0)}{w}\right)\right|&=\frac{1}{2\pi}\left|\int_{C} {\frac{(z-w)\phi(f(\zeta))}{ (\zeta-z_0)(\zeta-z_0-z)(\zeta-z_0-w)}}d\zeta\right|\\
	&\leq \frac{|z-w|}{r^2} \sup_{\zeta\in C} |\phi(f(\zeta))|\\
	&\leq\frac{|z-w|}{r^2} \|\phi\|\sup_{\zeta\in C} \|f(\zeta)\|\\
	\end{aligned}$$ 
	Taking the supremum over all $\phi \in X^\ast$ such that $\|\phi\|=1$, we obtain, as a corollary of the Hahn-Banach Theorem, that
	\begin{equation}
	\label{calculationx42}
	\left\|\frac{f(z_0+z)-f(z_0)}{z}-\frac{f(z_0+w)-f(z_0)}{w}\right\|\leq\frac{|z-w|}{r^2}\sup_{\zeta\in C} \|f(\zeta)\|.
	\end{equation}
	
	Finally, applying the inequality above, we conclude that $x_n=\frac{f(z_0+z_n)-f(z_0)}{z_n}$, where $(z_n)_{n\in\mathbb{N}}\subset \mathbb{C}$ is such that $z_n\to 0$, defines a Cauchy sequence in the Banach space $X$. Hence, it is a convergent sequence for $x\in X$. It is easy to see that equation \eqref{calculationx42} guaranties that
	$$\lim_{z\to 0}\frac{f(z_0+z)-f(z_0)}{z}=x.$$
	
	$(\Leftarrow)$ It is trivial.
	
\end{proof}

\section{KMS States and Some of their Properties}

A general definition of KMS states can be found in any textbook of Operator Algebras such as \cite{Bratteli2}, \cite{KR86}, and \cite{Takesaki2002}. Here, we will follow mostly \cite{Bratteli2}, with just minor proofs and modifications are due to the author in order to suit better the ensuing application on this thesis.

Since we will need to use the Fourier Transform soon, let us first fix some notation.

\begin{notation}\glsdisp{fourierT}{\hspace{0pt}}
We denote by $\hat{f}$ the Fourier transform of $f$, a Schwartz function. The Fourier transform is a uniformly continuous (vector space) automorphism on Schwartz space, defined by
$$\hat{f}(k)=\frac{1}{\sqrt{2\pi}}\int_{\mathbb{R}}{f(x)e^{-\iu k x}dx}.$$
Note that we have also the inverse Fourier transform
$$f(x)=\frac{1}{\sqrt{2\pi}}\int_{\mathbb{R}}{\hat{f}(k)e^{\iu k x}dk}.$$

In particular, when $f$ is an infinitely differentiable function with compact support, the Fourier transform is well defined. We denote
 $$\mathcal{C}^\infty_0(\mathbb{R})=\bigl\{f:\mathbb{R}\to \mathbb{R} \mid \textrm{f is infinitely differentiable with compact support} \bigr\}.$$
\end{notation}

Before proceeding to a central theorem, it is necessary to analyse the consequences of a function $f$ being the Fourier transform of an infinitely differentiable function $\hat{f}$ with compact support.

\begin{lemma}
	\label{Schwartz_PaleyWiener}
If $f$ is the inverse Fourier transform of a function $\hat{f}\in \mathcal{C}^\infty_0(\mathbb{R})$ with $\supp(\hat{f})\subset[-M,M]$. Then $f$ is an entire function and, for each $n\in\mathbb{N}$, there exists some $K_n>0$ such that
$$|f(z)| \leq  K_n \frac{e^{ M |\Im z|}}{1+|z|^n} \qquad \forall z\in \mathbb{C}.$$
\end{lemma}
\begin{proof}
Since the $\hat{f}$ is infinitely differentiable, it is uniformly infinitely continuous differentiable in $[-M,M]$. Thus, we can explicitly write
$$\frac{df}{dz}\left(z\right)=\frac{\iu}{\sqrt{2\pi}}\int_{\mathbb{R}}{p\hat{f}(p)e^{\iu pz} dp}.$$
Therefore, $f$ is an entire analytic function.

Using integration by parts and noting that $\displaystyle \supp\left(\frac{d^n\hat{f}}{dp^n}\right)\subset \supp(\hat{f})$ we obtain
\begin{equation}
\label{estimativa}
\begin{aligned}
|f(z)|	&=\left| \frac{1}{\sqrt{2\pi}}\int_{\mathbb{R}}{\hat{f}(p)\frac{(-\iu)^n}{z^n}\frac{d^n}{dp^n}\left(e^{\iu pz}\right) dp} \right| \\
				&=\left| \frac{1}{\sqrt{2\pi}}\int_{\mathbb{R}}{\frac{d^n\hat{f}}{dp^n}(p)\frac{\iu^n}{z^n}e^{\iu pz} dp} \right|\\
				&\leq\frac{e^{ M |\Im z|}}{|z|^n}\left| \frac{1}{\sqrt{2\pi}}\int_{[-M,M]}{\frac{d^n\hat{f}}{dp^n}(p)e^{\iu p \Re z} dp} \right|\\
				&=K^1_n \frac{e^{ M |\Im z|}}{|z|^n} \\
\end{aligned}
\end{equation}

Finally, it follows from Weierstrass's theorem that within the compact set $\{z\in\mathbb{C} \mid |z|\leq 1\}$ there exist constants $K^2_n$ such that
$$|f(z)| \leq  K^2_n \frac{e^{ M |\Im z|}}{1+|z|^n}.$$
When $|z|>1$, we can use the previous inequality
$$|f(z)|\leq K^1_n \frac{e^{ M |\Im z|}}{|z|^n} \leq 2K^1_n \frac{e^{ M |\Im z|}}{2|z|^n} \leq 2K^1_n \frac{e^{ M |\Im z|}}{1+|z|^n}.$$
Hence, the desired inequality is valid for all $z\in \mathbb{C}$ if we take $K_n=\max{\{2K^1_n,K^2_n\}}$.

Since $p\hat{f}$ is infinitely continuous differentiable with $\supp(p\hat{f}) \subset \supp(\hat{f})$, we have
$$\left|\frac{df}{dz}\left(z\right)\right|\leq C_1 \frac{\left(e^{M\Im{z}}-e^{-M\Im{z}}\right)}{|\Im{z}|} \frac{2M}{|\Re{z}|}.$$
It follows from equation \eqref{estimativa} that, for fixed $y \in \mathbb{R}$, $\displaystyle\lim_{x\to \infty}{f(x+\iu y)}=0$ and
$$\begin{aligned}
|f(x+\iu y)|	&=\left| \int^{\infty}_{x}{\frac{df}{dz}\left(x^\prime+\iu y\right) dx^\prime} \right|.
\end{aligned}$$

\end{proof}

\begin{lemma}
\label{tauInv}
Let $(\calgebra, \tau)$ be a $C^\ast$-dynamical system and $\omega$ a state on $\calgebra$ such that, for any $A,B \in \calgebra$, there exists $\beta\neq0$ and a complex function $F_{A,B}$ which is analytic in \mbox{$\mathcal{D}_\beta=\left\{z\in \mathbb{C} \mid 0<  \sgn{\beta} \Im z<|\beta|\right\}$} and continuous on $\overline{\mathcal{D}_\beta}$ satisfying
\begin{equation}
\label{eqKMS3.0}
\begin{aligned}
F_{A,B}(t) &=& \omega(A\tau_t(B)) \ \forall t\in \mathbb{R}, \\
F_{A,B}(t+\iu \beta)& =& \omega(\tau_t(B)A) \ \forall t\in \mathbb{R}. \\
\end{aligned}
\end{equation}
Then, $\omega$ is $\tau$ invariant, \ie, $\omega(\tau_t(A))=\omega(A)$ for all $t \in \mathbb{R}$ and for all $A\in \calgebra$.
\end{lemma}
\begin{proof}
First suppose $\calgebra$ has an identity, then let $F_{\mathbbm{1},A}$ as in Theorem \ref{KMSequi} $(iii)$.
From equation \eqref{eqKMS3.0} we see that
$$\begin{aligned}
	F_{\mathbbm{1},A}(t)				&= \omega(\mathbbm{1}\tau_t(A)) \\
									&= \omega(\tau_t(A)\mathbbm{1}) \\
									&= F_{\mathbbm{1},A}(t+\iu \beta). \\
\end{aligned}$$

It follows from the Edge-of-the-Wedge Theorem that $F_{\mathbbm{1},A}$ can be extended to an entire analytic function with period $\iu \beta$. But, since $F_{\mathbbm{1},A}$ is bounded in the closed \mbox{strip $\overline{\strip{\beta}}$}, its extension is bounded as well. Using Liouville's theorem, it follows that the function $F_{\mathbbm{1},A}$ is constant.

If $\calgebra$ does not have an identity, use an approximation identity to conclude the invariance on a dense subset, and thus in the whole set using continuity.

\end{proof}

The conclusion of this lemma is a first expected property of a thermal equilibrium state, since the observables should have stopped evolving. For sure, this condition is not sufficient to ensure that a state is in thermal equilibrium. In fact, we also know that there exists another interesting class of states, called passive states, that are invariant by the dynamics but are not KMS in general.

The equations \eqref{eqKMS3.0} are known as the KMS condition. It first appeared in works by R. Kubo, P. Martin and J. Schwinger, \cite{kubo57} and \cite{martin59}, and it was studied by Haag, Hugenholtz and Winnink, \cite{haag67}, in the context of equilibrium states in the thermodynamic limit.

%Bratteli and Robinson, 
\begin{theorem}
\label{KMSequi}
Let $(\calgebra, \tau)$ be a $C^\ast$-dynamical system, $\beta \in \mathbb{R}$, and $\omega$ a state over $\calgebra$. The following statements are equivalent

\begin{enumerate}[(i)]

\item		There exists a dense subset $\tilde{\calgebra}$ of entire analytic elements such that
				\begin{equation}
				\label{eqKMS1}
				\omega(AB)=\omega\left(B\tau_{\iu \beta}(A)\right) \ \forall A,B \in \tilde{\calgebra};
				\end{equation} 
				
\item 	For any $A,B \in \calgebra$ there exists a complex function $F_{A,B}$ which is analytic in the strip $\strip{\beta}=\left\{z\in \mathbb{C} \mid 0<\sgn{\beta} \Im z<|\beta|\right\}$ and continuous and bounded on $\overline{\mathcal{D}_\beta}$ satisfying
				\begin{equation}
				\label{eqKMS2}
				\begin{aligned}
				F_{A,B}(t) &=& \omega(A\tau_t(B)) \ \forall t\in \mathbb{R}, \\
				F_{A,B}(t+\iu \beta)& =& \omega(\tau_t(B)A) \ \forall t\in \mathbb{R}; \\
				\end{aligned}
				\end{equation}
				
\item	For any $A,B \in \calgebra$ there exists a complex function $F_{A,B}$ which is analytic in $\mathcal{D}_\beta=\left\{z\in \mathbb{C} \mid 0<\sgn{\beta} \Im z<|\beta|\right\}$ and continuous on $\overline{\mathcal{D}_\beta}$ satisfying
				\begin{equation}
				\label{eqKMS3}
				\begin{aligned}
				F_{A,B}(t) &=& \omega(A\tau_t(B)) \ \forall t\in \mathbb{R}, \\
				F_{A,B}(t+\iu \beta)& =& \omega(\tau_t(B)A) \ \forall t\in \mathbb{R}; \\
				\end{aligned}
				\end{equation}
				
\item For any $A,B \in \calgebra$ and for any $f$ such that $\hat{f}\in \mathcal{C}^\infty_0(\mathbb{R})$, where $\hat{f}$ is the Fourier transform of $f$, the following relation is true

\begin{equation}
\label{eqKMS4}
\int_{\mathbb{R}}{f(t)\omega\left(A\tau_t (B)\right)dt}=\int_{\mathbb{R}}{f(t+\iu \beta)\omega\left(\tau_t (B)A\right)dt};
\end{equation}
\item For any $A\in \calgebra$ and for any $f$ such that $\hat{f}\in \mathcal{C}^\infty_0(\mathbb{R})$, where $\hat{f}$ is the Fourier transform of $f$, the measures $\mu_A(\hat{f})$ and $\nu_A(\hat{f})$, defined by
$$\begin{aligned}
\mu_A(\hat{f})&=&\int_{\mathbb{R}}{f(t)\omega(A^\ast\tau_t(A))dt},\\
\nu_A(\hat{f})&=&\int_{\mathbb{R}}{f(t)\omega(\tau_t(A)A^\ast)dt};
\end{aligned}$$
are equivalent positive Radon measures on $\mathbb{R}$ with Radon-Nikodym derivative
$$\frac{d\mu_A}{d\nu_A}(p)=e^{-\beta p};$$
			
\item Let $\delta$ be the infinitesimal generator of $\tau$. Then
\begin{equation}
\label{RAS}
\omega(A^\ast A) \log{\left(\frac{\omega(A^\ast A)}{\omega(A A^\ast)}\right)} \leq -\iu \beta \omega(A^\ast\delta(A)) \quad \forall A\in \Dom{\delta},
\end{equation}
where we are using
$$ x \log{\frac{x}{y}}=		\begin{cases} 
														x \log{\frac{x}{y}} & \textrm{ if } x >0, \ y>0\\
														0 									& \textrm{ if } x=0, \ y\geq 0\\
														+\infty 						& \textrm{ if } x >0, \ y=0\\
													\end{cases}$$

\end{enumerate}
\end{theorem}

\begin{proof}
$(i)\Rightarrow (ii)$ By definition, for all $A,B \in \tilde{\calgebra}$, $t\mapsto \omega(B\tau_t(A))$ admits an entire analytic extension (see Definition \ref{Anal}), so define
$$F_{A,B}(z)=\omega(B \tau_z(A)).$$
This function is entire analytic and satisfies \eqref{eqKMS2}.

Because $z\mapsto \tau_{z}(A)$ is analytic, we know it must be limited in the compact set $\{z\in \mathbb{C} | \ \Re{z}=0, \ 0 \leq \Im{z}\leq \beta\}$ and $F_{A,B}$ must be analytic in $\overline{\mathcal{D}_\beta}$ as well, since

$$|F_{A,B}(t+\iu y)|=\left|\omega\left(B\tau_t\left(\tau_{\iu y}(A)\right)\right)\right| \leq \|B\| \|\tau_{\iu y}(A)\| \leq \|B\| \ \sup_{0\leq y \leq \beta}{\|\tau_{\iu y}(A)\|}.$$

Now, for general $A,B \in \calgebra$, there exist sequences $(A_n)_{n\in\mathbb{N}},(B_n)_{n\in\mathbb{N}} \subset \calgebra$ such that $A_n\rightarrow A$ and $B_n \rightarrow B$, thanks to Proposition \ref{AnalDense}.

So $(F_n)_{n\in\mathbb{N}}$, $F_n=F_{A_n,B_n}$, is a sequence of entire analytic functions which are continuous and bounded in $\overline{\mathcal{D}_\beta}$, furthermore $F_n-F_m$ is also entire analytic, and thanks to the Maximum Modulus Principle it must assume its maximum at $\partial \overline{\mathcal{D}_\beta}$, so
$$\begin{aligned}
|F_n(z)-F_m(z)|
						& \leq \max{\left\{\sup_{t\in \mathbb{R}}{|(F_n-F_m)(t)|},\sup_{t\in \mathbb{R}}{|(F_n-F_m)(t+\iu \beta)}\right\}} \\
						& \hspace{-1.8cm} \leq \max{\left\{ \sup_{t\in \mathbb{R}}{|\omega(B_n\tau_{t}(A_n))-\omega(B_m\tau_{t}(A_m))|},\sup_{t\in \mathbb{R}}{|\omega(\tau_{t}(A_n)B_n)-\omega(\tau_{t}(A_m)B_m)|}\right\}} \\
						& \hspace{-1.8cm}\leq \sup_{t\in \mathbb{R}}{|\omega((B_n-B)\tau_{t}(A_n))+\omega(B\tau_{t}(A_n-A_m))+\omega((B-B_m)\tau_{t}(A_m))|} \\
						&\hspace{-1.2cm} +\sup_{t\in \mathbb{R}}{|\omega(\tau_{t}(A_n)(B_n-B))+\omega(\tau_{t}(A_n-A_m)B)+\omega(\tau_{t}(A_m)(B-B_m))|}\\
						&\hspace{-1.8cm}\leq 2\left(\|B_n-B\| \|A_n\|+\|B\| \|A_n-A_m\|+\|B-B_m\|\|A_m\|\right).
\end{aligned}$$

So, $(F_n)_{n\in\mathbb{N}}$ is a Cauchy sequence with the uniform norm, which means that its limit is a continuous bounded function on $\overline{\mathcal{D}_\beta}$ and analytic in $\mathcal{D}_\beta$. Hence, it is natural to define, for general $A,B \in \calgebra$,
$$F_{A,B}(z)=\lim_{n\rightarrow \infty}{F_{A_n,B_n}(z)},$$
where $(A_n)_{n\in\mathbb{N}},(B_n)_{n\in\mathbb{N}} \subset \tilde{\calgebra}$ are such that $A_n\rightarrow A$ and $B_n \rightarrow B$.

It follows from the uniform convergence that $F_{A,B}$ satisfies \eqref{eqKMS2}.

$(ii)\Rightarrow (iii)$ The two conditions are equal except for the requirement in $(ii)$ that $F_{A,B}$ is bounded in $\overline{\mathcal{D}_\beta}$, so $(iii)$ follows trivially.

$(iii)\Rightarrow (iv)$ Because $F_{A,B}(z) =\omega(A\tau_z(B))$ is analytic in $\mathcal{D}_\beta$ we can use Cauchy's Theorem in the region $\mathcal{D}^n_\beta=\left\{z\in \mathbb{C} \mid \sgn{\beta} \Im z<|\beta| \textrm{ and } |\Re z| \leq n \ \right\}$.

The importance of Lemma \ref{Schwartz_PaleyWiener} becomes clear: two of the boundary integrals must vanish when $n \rightarrow +\infty$, since $F_{A,B}$ is bounded in $\overline{\strip{\beta}}$ and
$$|f(z)| \leq  K_n \frac{e^{ M |\Im z|}}{1+|z|^n}.$$

Writing the remaining integrals we have the desired property:
\begin{equation}
\label{eq1}
\begin{aligned}
0=\int_{\partial \mathcal{D}_\beta}{f(z)\omega(A\tau_z B)dz}	& = \lim_{n\rightarrow \infty}{\int_{\partial \mathcal{D}^n_\beta}{f(z)\omega(A\tau_z B)}dz}\\
																														& = \lim_{n\rightarrow \infty}{\int_{\mathcal{D}^n_\beta}{f(z)F_{A,B}(z)}dz} \\
																														& = \int_{\mathbb{R}}{f(t)F_{A,B}(t)dz}- \int_{\mathbb{R}}{f(t+\iu \beta)F_{A,B}(t+\iu \beta)dz} \\
																														& = \int_{\mathbb{R}}{f(t)\omega(A\tau_t B)dz}- \int_{\mathbb{R}}{f(t+\iu \beta)\omega(\tau_t(B) A)}. \\
\end{aligned}
\end{equation}

$(iv)\Rightarrow (i)$ We are again going to construct a sequence of functions with compact support so that the inverse Fourier transforms approaches Dirac's delta distribution.
We know that 
$$h(x)\doteq \begin{cases} 0											 &\textrm{ if } x\leq 0\\
											e^{-\frac{1}{x^2}+1}   	 &\textrm{ if } x > 0\\ \end{cases}$$
is an infinitely continuous and differentiable function, so $f(x)=1-h(1-h(1-x))$ is a continuous and differentiable function which is decreasing, $f(0)=1$ and $f(1)=0$.

Finally, let 
$$\hat{f_n}(x)=\begin{cases} 	1							&\textrm{ if } |x|\leq n,\\
															f(|x|-n)			&\textrm{ if } n<|x|< n+1,\\
															0							&\textrm{ if } |x|\geq n+1.\\ \end{cases}$$

This definition leads us to

$$\lim_{n\rightarrow \infty}{\int_{\mathbb{R}}{f_n(x)g(x)dx}}=g(0).$$

So, changing variables in equation \eqref{eq1} and using the sequence $f_n\in \mathcal{C}_0^\infty(\mathbb{R})$ we obtain

$$\omega(A B)=\omega(\tau_{-\iu \beta}(B)A) \quad \forall B\in \calgebra_\tau, A \in \calgebra.$$

Equivalently,

$$\omega(B \tau_{\iu \beta}(A))=\omega(A B) \quad \forall B\in \calgebra_\tau, A \in \calgebra.$$

$(iv)\Leftrightarrow (v)$ First, let $(\hilbert_\omega,\pi_\omega, \Omega_\omega)$ be a representation of the algebra $\calgebra$, by Lemma \ref{tauInv} and Riesz Theorem, for each $A\in \calgebra$, there exists a unitary operator $U_\omega(t)\in B(\hilbert_\omega)$ such that $$\pi_\omega\left(\tau_t(A)\right)= U_\omega(t)\circ\pi_\omega(A)\circ U_\omega(t)^{-1}.$$

Let
\begin{equation}
\label{decespec}
U_\omega(t)=\int{e^{-\iu p t} dE_p}
\end{equation}
be the spectral decomposition of such a unitary operator.
\begin{equation}
\label{medida}
\begin{aligned}
\mu_A(\hat{f})	
		&=\int_{\mathbb{R}}{f(t)\omega(A^\ast\tau_t(A))dt}\\
		&= \frac{1}{\sqrt{2\pi}}\int_{\mathbb{R}}{\int_{\mathbb{R}}{e^{\iu k t}\hat{f}(k)dk}\ip{\pi_\omega(A)\Omega_\omega}{\pi_\omega(\tau_t(A))\Omega_\omega}dt} \\
		&= \frac{1}{\sqrt{2\pi}}\int_{\mathbb{R}}{\int_{\mathbb{R}}{e^{\iu k t}\hat{f}(k)dk}\ip{\pi_\omega(A)\Omega_\omega}{U_\omega(t)\pi_\omega(A)\Omega_\omega}dt} \\
		&= \frac{1}{\sqrt{2\pi}}\int_{\mathbb{R}}{\int_{\mathbb{R}}{e^{\iu k t}\hat{f}(k)dk}\ip{\pi_\omega(A)\Omega_\omega}{\int_{\mathbb{R}}{e^{-\iu p t} dE_p}\pi_\omega(A)\Omega_\omega}dt} \\
								&= \frac{1}{\sqrt{2\pi}}\int{\int_{\mathbb{R}}{\int_{\mathbb{R}}{e^{\iu (k-p) t}\hat{f}(k)\ip{\pi(A)\Omega_\omega}{dE_p \pi_\omega(A)\Omega_\omega}dk dt}}}\\
								&=\int{\hat{f}(k)\ip{\pi(A)\Omega_\omega}{ dE_k \pi_\omega(A)\Omega_\omega}}.
\end{aligned}
\end{equation}
So $\mu_A$ is a positive functional in $\mathcal{C}^{\infty}_0(\mathbb{R})$ and by the Riesz-Markov theorem there exists a unique Radon measure $d\mu_A$ such that
$$\mu_A(\hat{f})=\int_{\mathbb{R}}{\hat{f}d\mu_A}.$$

By the same steps we see that there exists a unique Radon measure $d\nu_A$ such that
$$\nu_A(\hat{f})=\int_{\mathbb{R}}{\hat{f}d\nu_A}.$$

Now, it follows trivially by the definition of Fourier transform that
$$\begin{aligned}
\widehat{e^{-\beta p}f}(t)		
								&=\frac{1}{\sqrt{2\pi}}\int_{\mathbb{R}}{e^{-\beta p}\hat{f}(p)e^{\iu p t}dt}\\
								&=\frac{1}{\sqrt{2\pi}}\int_{\mathbb{R}}{\hat{f}(p)e^{\iu(t+\iu \beta) p}dt}\\
								&=\hat{f}(t+\iu \beta).\\
\end{aligned}$$
Using now the hypothesis
$$\begin{aligned}
\int_{\mathbb{R}}{\hat{f}d\mu_A} = \mu_A(\hat{f})		&=\int_{\mathbb{R}}{f(t)\omega(A^\ast\tau_t(A))dt} \\
																										&=\int_{\mathbb{R}}{f(t+\iu \beta)\omega(\tau_t(A)A^\ast)dt}\\
																										&= \nu_A(e^{-\beta p}\hat{f}) = \int_{\mathbb{R}}{e^{-\beta p} \hat{f}d\nu_A}.\\
\end{aligned}$$
Note that from the last equality the converse follows.

$(v) \Leftrightarrow (vi)$ Let $U_\omega$ be as before and equation \eqref{decespec} its spectral decomposition. Stone's theorem ensures the existence of a self-adjoint operator $H_\omega$ such that
\begin{equation}
\label{generator}
U_\omega(t)=e^{\iu t H_\omega}.
\end{equation}
Differentiating the expressions \eqref{decespec} and \eqref{generator} and comparing the results at $t=0$, we conclude that
$$-\int{p dE_p}=H_\omega .$$
Note that the function $S(x,y)= x \log{\frac{x}{y}}$ is lower semicontinuous and homogeneous of degree 1. It is also jointly convex. In fact, there is nothing to prove, if any among the variables $x,y,u,v$ is zero and, for $x,y,u,v>0$, we can define
$$f(t)=S\left(t(x,y)+(1-t)(u,v)\right)-tS(x,y)-(1-t)S(u,v).$$
Its first and second derivatives are
$$\begin{aligned}
f'(t)		&=\frac{v x-u y}{t y+(1-t)v}+(x-u) \log \left(\frac{t x+(1-t)u}{t y+(1-t)v}\right)+u \log \left(\frac{u}{v}\right)-x \log \left(\frac{x}{y}\right) \\
f''(t)	&=\frac{(v x-u y)^2}{(t x+(1-t) u) (t y+(1-t) v)^2} .\\
\end{aligned}$$
Note that $f''(t)\geq 0$ and
$$\begin{aligned}
f'(0)		&=x-\frac{u y}{v}+x \log \left(\frac{u}{v}\right)-x \log \left(\frac{x}{y}\right)\\
				&= x-\frac{u y}{v}+x \log \left(\frac{uy}{vx}\right)\\
				&\leq x-\frac{u y}{v}+x \left(\frac{uy}{vx}-1\right) \\
				&= 0.\\
f'(1)	&=-u+\frac{v x}{y}+u \log \left(\frac{u}{v}\right)-u \log \left(\frac{x}{y}\right)\geq 0. \\
\end{aligned}$$
So, $f(0)=f(1)=0$, $f'(0)\leq0$, $f'(1)\geq0$ and $f'$ is strictly increasing, hence $f'$ has exactly one zero in $[0,1]$ (note that we can't have both $f'(0)=0$ and $f'(1)=0$) thus $f(t)\leq 0$. 

Using the lower semicontinuity, we conclude that the joint convexity remains with integrals, therefore

\begin{equation}
\label{RAS2}
\begin{aligned}
S\left(\omega(A^\ast A),\omega(AA^\ast)\right) &= S\left(\mu_A(1),\nu_A(1)\right)\\
																									&= S\left(\mu_A(1),\mu_A(e^{-\beta p})\right)\\
																									& \leq\mu_A\left(S(1,e^{-\beta p})\right) \\
																									& = \mu_A\left(\log( e^{-\beta p})\right)\\
																									&=\int_{\mathbb{R}}{\beta p\ip{\pi(A)}{dE_p \pi_\omega(A)}}\\
																									&=-\iu \beta \ip{\pi(A)\Omega_\omega}{\iu H_\omega \pi_\omega(A)\Omega_\omega} \\
																									&=-\iu \beta \omega(A^\ast \delta(A)),\\
\end{aligned}
\end{equation}
where we used that $H_\omega\Omega_\omega=0$, since we have already proved the $\tau$-invariance of $\omega$ in Proposition \ref{tauInv}.

For the converse, we have already seen in the last two steps of equation \eqref{RAS2} that
$$-\iu \beta \ip{\pi(A)\Omega_\omega}{\iu H_\omega \pi_\omega(A)\Omega_\omega}=-\iu \beta \omega(A^\ast \delta(A)).$$

But, since $H_\omega$ is self-adjoint, we have

$$\begin{aligned}
\omega(A^\ast \delta(A))								&=\ip{\pi(A)\Omega_\omega}{\iu H_\omega \pi_\omega(A)\Omega_\omega}\\
																				&= - \ip{\iu H_\omega\pi(A)\Omega_\omega}{ \pi_\omega(A)\Omega_\omega}\\
																				&= - \omega(\delta(A^\ast)A).\\
\end{aligned}$$

Thus, for a self-adjoint element $A$,
\begin{equation}
\begin{aligned}
\omega\left(\tau_t\left(A^2\right)\right)-\omega\left(A^2\right)	&= \int_0^t{\omega(\delta(\tau_{t^\prime}(A)^2))dt^\prime}\\
															&= \int_0^t{\omega\big(\delta(\tau_{t^\prime}(A))\tau_{t^\prime}(A)+\tau_{t^\prime}(A)\delta(\tau_{t^\prime}(A))\big)d t^\prime}\\
															&= 0.
\end{aligned}
\end{equation}

By proposition \ref{sqrt}, every positive element can be written as a square of a self-adjoint element. Using the polarization identity (equation \eqref{PolIden}), every element is a linear combination of four positive elements. Hence, we conclude that

$$\omega(\tau(A))=\omega(A)	 \quad \forall A\in \calgebra.$$

Now, if $\hat{f}\in \Dom{\delta}$, the element
$$A_f=\int_{\mathbb{R}}{f(t)\tau_t(A) dt}$$
is entire analytic for $\delta$ by the same arguments as used in Proposition \ref{Anal}.

In order to use expression \eqref{RAS}, let's calculate (using the spectral decomposition giving in equation \eqref{decespec} and equation \eqref{medida})
$$\begin{aligned}
-\iu \beta \omega(A_f^\ast \delta(A_f))		&= - \iu \beta \ip{\int_{\mathbb{R}}{f(t)U_\omega(t)(\pi_\omega(A))\Omega_\omega dt}}{ \iu H_\omega \int_{\mathbb{R}}{f(t)U_\omega(t)(\pi_\omega(A))\Omega_\omega dt}} \\
																						&\hspace{-19pt}= - \iu \beta \ip{\int_{\mathbb{R}}{\int{f(t)e^{-\iu p t}dE_p\pi_\omega(A)\Omega_\omega}dt}}{ \iu H_\omega \int_{\mathbb{R}}{\int{f(t)e^{-\iu p t}dE_p\pi_\omega(A)\Omega_\omega}dt}} \\
																						&\hspace{-19pt}= \ip{\int{\hat{f}(p) dE_p\pi_\omega(A)\Omega_\omega}}{\beta H_\omega \int{\hat{f}(p)dE_p\pi_\omega(A)\Omega_\omega}} \\
																						&\hspace{-19pt}= \ip{\pi_\omega(A)\Omega_\omega}{ \int{\int{ \int{-\beta r \overline{\hat{f}(q)}\hat{f}(p)dE_p dE_q dE_r\pi_\omega(A)\Omega_\omega}}}} \\
																						&\hspace{-19pt}= \ip{\pi_\omega(A)\Omega_\omega}{ \int{\int{ \int{-\beta p \overline{\hat{f}(p)}\hat{f}(p)dE_p \pi_\omega(A)\Omega_\omega}}}} \\
							&\hspace{-19pt}= \mu_A(k h), \\
\end{aligned}$$
where $h(p)=|\hat{f}(p)|^2$ and $k(p)=-\beta p$.

Similarly, one can calculate

$$ \begin{aligned}
-\iu \beta \omega(A_f \delta(A_f^\ast))		&= -\nu_A(k h), \\
\omega(A_f^\ast A_f)												&= \mu_A(h),\\
\omega(A_f A_f^\ast)												&= \nu_A(h).\\
\end{aligned} $$

Now, by hypothesis,

$$\begin{aligned}
\mu_A(k h)		& \geq S\bigl(\mu_A(h),\nu_A(h)\bigr),\\
-\nu_A(k h)	& \geq S\bigl(\nu_A(h),\mu_A(h)\bigr).\\
\end{aligned}$$

Since $h \in \mathcal{C}^\infty_0$, we can define
$$\begin{aligned}
\overline{p}(h)		&\doteq\sup_{x \in \supp(h)} x, \\
\underline{p}(h)		&\doteq\inf_{x \in \supp(h)} x, \\
\end{aligned}$$

and it follows that

$$\begin{aligned}
-\beta \underline{p}(h) \mu_A(h)		& \geq S(\mu_A(h),\nu_A(h))&=\mu_A(h) \log\left(\frac{\mu_A(h)}{\nu_A(h)}\right)\\
\beta \overline{p}(h) \nu_A(h)	& \geq S(\nu_A(h),\mu_A(h))		&=\nu_A(h) \log\left(\frac{\nu_A(h)}{\mu_A(h)}\right)\\
\end{aligned}$$

$$\Leftrightarrow e^{-\beta \underline{p}(h)}\nu_A(h)\geq \mu_A(h) \geq e^{-\beta \overline{p}(h)} \nu_A(h).$$

Now given $\varepsilon>0$, one takes a partition of unit $(h_n)_{n\in\mathbb{N}}$ (it could be finite if we take only a partition subordinated to an open cover of $\supp(h)$) such that
$$\left|e^{-\beta \underline{p}(h)}- e^{-\beta \overline{p}(h)}\right|< \varepsilon \nu_A(h).$$

From this property, and using Lebesgue's dominated convergence theorem, we conclude that
$$\left|\mu_A(h h_n)-\nu_A(k h h_n)\right|< \varepsilon \nu_A(h h_n) \Rightarrow \left|\mu_A(h)-\nu_A(k h)\right|< \varepsilon \nu_A(h) $$

So, $\mu_A(h)=\nu_A(k h)$.

\end{proof}
%###
It is worth to mention that, although it is a well-known fact, the proof above that $S(u,v)$ is jointly convex was written by the author.

\begin{definition}
\index{state! KMS}
A state is said to be a $(\tau,\beta)$-KMS state if it satisfies one of the conditions in Theorem \ref{KMSequi}. 
\end{definition}

We said earlier that one expected property to call a state an equilibrium state is $\tau$-invariance, but also some kind of concavity is expected from our knowledge of Thermodynamics. This concavity property is related to equivalence $(vi)$ in Theorem \ref{KMSequi}, but a discussion of this fact will be postponed.

%#### Essa demonstração é minha?
\begin{proposition}
Let $(\calgebra, \tau)$ be a $C^\ast$-dynamical system, $\beta \in \mathbb{R}$, and $\omega$ a state \mbox{over $\calgebra$}, then $\omega$ is a $(\tau,\beta)$-KMS state if and only if there exists a norm-dense $\ast$-subalgebra $\calgebra_\tau$ of $\tau$-analytic elements such that 
$$\omega(AB)=\omega\left({\tau_{-\iu\beta/2}(B)\tau_{\iu\beta/2}(A)}\right) \quad \forall A,B\in  \nalgebra_\tau.$$
\end{proposition}

\begin{proposition}
	\label{limitKMS}
	Suppose $\omega_n$, $n\in \mathbb{N}$, are uniformly bounded functionals on $\calgebra$ satisfying the  KMS condition for $(\tau^n,\beta)$ and
	\begin{enumerate}[(i)]
		
		\item For each $A\in \calgebra$ and each $t\in \mathbb{R}$
		$$\lim_{n\rightarrow \infty}{\tau^n_t(A)}=\tau_t(A),$$
		where $\{\tau_t\}_{t\in\mathbb{R}}$ is a group of isometries.
		
		\item For each $A\in \calgebra$, $\left(\omega_n(A)\right)_{n\in\mathbb{N}}$ is a convergent sequence and
		$$\omega(A)=\lim_{n\rightarrow \infty}{\omega_n(A)}$$ defines a state.
	\end{enumerate}
	Then $\omega$ is a $(\tau,\beta)$-KMS state.
\end{proposition}
\begin{proof}\footnote{From \cite{Haag64} with improvements by the author.} 
	Suppose $\varepsilon>0$ is given, in addition, fix $A,B\in \calgebra$ and $t\in\mathbb{R}$. Without loss of generality, suppose $\|\omega_n\|\leq 1 \ \ \forall n \in \mathbb{N}$. Take some $n_1 \in \mathbb{N}$ such that, for all $n>n_1$,
	$$\|\tau_t^{n}(A)-\tau_t(A)\| <\frac{\varepsilon}{6 \|B\|}.$$
	A basic consequence of this condition is that, for all $n,m>n_1$,
	$$\|\tau_t^{n}(A)-\tau^{m}_t(A)\| <\frac{\varepsilon}{3 \|B\|}.$$
	
	Consider now $m>n_1$ fixed. By condition $(ii)$ in the hypothesis, we can take $n_2\in \mathbb{N}$ such that for all $n>n_2$
	$$\|\omega_{n}(B \tau_t^{m}(A))-\omega(B \tau_t^{m}(A)) \| <\frac{\varepsilon}{3}$$
	
	Hence, for each $t\in \mathbb{R}$ and $n>\max\{{n_1,n_2}\}$,
	$$\begin{aligned}
	\|\omega_{n}\left(B \tau_t^{n}(A)\right)-\omega\left(B \tau_t(A)\right)\| &\leq \|\omega_{n}\left(B \tau_t^{n}(A)\right)-\omega_{n}\left(B \tau_t^{m}(A)\right) \| \\
	& \qquad+\|\omega_{n}\left(B \tau_t^{m}(A)\right)-\omega\left(B \tau_t^{m}(A)\right)\| \\
	& \qquad+\|\omega\left(B \tau_t^{m}(A)\right)-\omega\left(B \tau_t(A)\right)\| \\
	&<\|B\| \left\|\tau_t^{n}(A)- \tau_t^{m}(A)\right\| +  \frac{\varepsilon}{3} + \|B\| \left\|\tau_t^{n}(A)- \tau_t(A)\right\| \\
	&< \varepsilon.\\
	\end{aligned}$$
	
	Since the $\omega_n$'s are KMS states, for any $A,B \in \calgebra$ and for any $f$ such that $\hat{f}\in \mathcal{C}^\infty_0(\mathbb{R})$, equation \eqref{eqKMS4} must be satisfied. Using the Lemma $\ref{estimativa}$ with n=2, we know that
	$$|f(z)| \leq  K_2 \frac{e^{ M |\Im z|}}{1+|z|^2}.$$
	Hence, for every $\varepsilon>0$, there exists $R>0$ such that
	$$\begin{aligned}
	\left|\int_{|t|>R}{f(t)\omega(A\tau_t B)dt}\right|	&< \frac{\varepsilon}{4} ,\\
	\left|\int_{|t|>R}{f(t+\iu \beta)\omega(A\tau_t B)dt}\right|&< \frac{\varepsilon}{4} .\\
	\end{aligned}$$
	
	But the pointwise convergence becomes uniform in the compact set $[-R,R]$, so there exists $n_0\in\mathbb{N}$ such that, for all $n>n_0$,
	$$\left\|\omega_{n}\left(B \tau_t^{n}(A)\right)-\omega\left(B \tau_t(A)\right)\right\| < \frac{\varepsilon}{4R K_2 e^{M\beta}} \qquad \forall t \in [-M,M].$$
	Finally,
	$$\begin{aligned}
	&\left\|\int_{\mathbb{R}}{f(t)\omega(A\tau_t (B))dt}-\int_{\mathbb{R}}{f(t+\iu \beta)\omega(A\tau_t (B))dt}\right\| \\
	&\hspace{3.8cm} <\left\|\int_{|t|\leq R}{f(t)\omega(A\tau_t(B))dt}-\int_{|t|\leq R}{f(t+\iu \beta)\omega(A\tau_t (B))dt}\right\|+\frac{\varepsilon}{2}\\
	&\hspace{3.8cm} <\left\|\int_{|t|\leq R}{f(t)\omega_{n}(A\tau^{n}_t(B))dt} -\int_{|t|\leq R}{f(t+\iu \beta)\omega_{n}(A\tau^{n}_t (B))dt}\right\|+\varepsilon\\
	&\hspace{4cm} =\varepsilon, \\
	\end{aligned}$$
	and $\omega$ is a $(\tau,\beta)$-KMS state.
	
\end{proof}

We finish this section with a result of great importance. We said several times that KMS states ``survive the thermodynamical limit'', here we stated an equivalent result that shows it survive limits under certain hypothesis, but, in the next section, it will become clear that Proposition \ref{limitKMS} includes the reasonable physical situation.

\section{The Physical Meaning of KMS States}

Now, we are going to discuss some results first presented in \cite{haag67} that relate the KMS condition to equilibrium states.

A physically desirable property of an equilibrium system in a region with infinite-volume, $V$, is to be the limit of the restrictions of this system to the finite-volume regions, $V_f$, when $V_f\to V$.

Therefore we need to analyse some properties of equilibrium states for a finite-volume region. Consider a system with just one type of particles in a finite-volume region $V$, \ie, the space in focus is the Fock space $\displaystyle\mathfrak{F}(\hilbert)=\bigoplus_{n\geq0}{\hilbert^n}$ where $\hilbert^0=\mathbb{C}$. We will use $\hilbert=L_2\bigl(\mathbb{R}^3\bigr)$.

For physical reasons, we would like to define, for each $f\in \hilbert$, $\Psi^\ast(f)$ and $\Psi(f)$ (its conjugate) be the creation and annihilation operators of these particles defined by
\begin{equation}\begin{aligned}
%\label{creationaniquilation}
\Psi(f)\left(f_0\otimes\ldots \otimes f_n\right)&=\sqrt{n}(f,f_0)f_1\otimes\ldots \otimes f_n, \ n\geq 1 \textrm{ and } \Psi(f)(f_0)=0, \ f_0 \in \mathbb{C}, \\
\Psi^\ast(f)\left(f_0\otimes\ldots \otimes f_n\right)&=\sqrt{n+1}f\otimes f_0\otimes \ldots \otimes f_n. \\
\end{aligned}\end{equation}

The problem with this definition is that, in general, it does not define a bounded operator. To avoid this issue, since we don't want to extend our discussion going to the Weyl form of the CCR, we will restrict our attention to the CAR algebra. Nevertheless, one can, throughout the Weyl relations, extend the following discussion to the boson algebra. To a formal treatment on the subject we refer to \cite{Emch1972} and \cite{Bratteli2}. 

Hence, we remember that the Fock space can be decomposed as a direct sum of two closed subspaces $\mathfrak{F}_+(\hilbert), \mathfrak{F}_-(\hilbert) \subset \mathfrak{F}(\hilbert)$, the Bose-Fock space and the Fermi-Fock space, which are the spaces of symmetric and antisymmetric vectors, respectively. We can properly define the desired physical operator acting in $\mathfrak{F}_{-}(\hilbert)$ as

\begin{equation}\begin{aligned}
\label{creationaniquilation}
\Psi(f_{0})\left(f_1\otimes\ldots \otimes f_n\right)_A&=\begin{dcases}\frac{1}{n!}\sum_{\sigma\in S_{n}}\sgn{\sigma} \sqrt{n}(f_{n},f_{\sigma(1)})f_{\sigma(2)}\otimes\ldots \otimes f_{\sigma(n)}, \quad & n\geq 2\\ 0, &n=1\end{dcases} \\
\Psi^\ast(f_{0})\left(f_1\otimes\ldots \otimes f_n\right)_A&=\frac{1}{(n+1)!}\sum_{\sigma\in S_{n+1}}\sgn{\sigma}\sqrt{n+1}f_{\sigma(0)}\otimes f_{\sigma(1)}\otimes \ldots \otimes f_{\sigma(n)}, \\
\end{aligned}\end{equation}
where $S_{n}$ is the group of permutations of $n$ elements and $(f_1\otimes\ldots \otimes f_n)_A\in \mathfrak{F}_{-}(\hilbert)$.

From the definition, it follows that the $\Psi$'s satisfies the anticommutation relations
\begin{equation}\label{CR2}
\begin{aligned} \
[\Psi(f),\Psi^\ast(g)]_{+}	&=(g,f) &=\int_{\mathbb{R}^3}{\overline{f(x)}g(x)dx},\\
[\Psi(f),\Psi(g)]_{+}			&=0 ,\\
\end{aligned}
\end{equation}
where $[\cdot,\cdot]_+$ is the anticommutator defined by $[A,B]_+=AB+BA$.

Consider now the von Neumann algebra $\calgebra(V)$ generated by $$\bigl\{\Psi(f) \ \big| \ f\in L_2(V) \textrm{ and } \supp{(f)} \subset V  \bigr\}.$$ It follows from the construction that $\calgebra(V)$ is a weakly closed set of bounded operators on $\mathfrak{F}_-(\hilbert)$, so with respect to the operator norm, it is a $C^\ast$-algebra.

Consider now the set $$\calgebra = \overline{\bigcup_{V finite}{\calgebra(V)}}^{\|.\|},$$
which is also a $C^\ast$-algebra.

It is useful to denote $$\mathfrak{F}_-\left(\hilbert\right)^V=\calgebra(V) \mathfrak{F}_-(\hilbert).$$ Note that it is the Hilbert space appearing in the representation of the subspace $\calgebra(V)\subset \calgebra$.

Another important point is the existence of the zero-particle vector $\Omega_0$, that is, the state that satisfies
$$\Psi(f)\Omega_0=0 \quad \forall f \in L^2(\mathbb{R}^3).$$
Note that $\calgebra$, in fact $\calgebra \Omega=\mathfrak{F}_-(\hilbert)$ and $\calgebra(V) \Omega=\mathfrak{F}_-(\hilbert)^V$ are cyclic representations of these algebras.

$\calgebra$ is the $C^\ast$-algebra of physical interest, because it excludes the unwanted behaviour of global observables for infinite system since it contains only quasi-local quantities. For example, in the thermodynamic limit the total number of particles is not well defined, but, since we have finite density, we for every finite-volume region the particle number (operator) is well defined. The same argument would be true for energy. This restriction on the observables also reflects on the number of states. Indeed, the global energy density, for example, exists only as the limit of energy density in finite-volume regions.

Suppose the dynamics of the system is determined by a Hamiltonian $H$ as an operator on $\mathfrak{F}(\hilbert)$. As mentioned before, the number operator particle $N$ is also well defined in every finite-volume region $V$. As we are interested in changing the number of particles, it is useful to change to the Gibbs grand canonical ensemble. We set
$$H^\prime\doteq H-\mu N,$$
with $\mu$ a real number (it is the chemical potential).

Consider also the restricted Hamiltonian $H^\prime_{V}=H_V-\mu N_V$ to $V$, as well as the corresponding time evolution operator $U_V(t)=e^{-\iu t H^\prime_V}$.

The following three conditions are expected for a physically reasonable equilibrium system:

\begin{enumerate}[(i)]
	\item If $V=V_1\cup V_2$ and $V_1\cap V_2=\emptyset$, then
	$$H^\prime_V-H^\prime_{V_1}-H^\prime_{V_2}=H^\prime_S,$$
	where the surface term $H^\prime_S$. In other words, for every fixed time $t\in\mathbb{R}$ and for $(V_n)_{n\in\mathbb{N}}\subset\mathbb{R}^3$ with $V_n\subset V_{n+1}$ and  $\displaystyle \bigcup_{n \in \mathbb{N}}{V_n}=\mathbb{R}^3$, 
	$\displaystyle\left(U_{V_n}(t)AU_{V_n}^\ast(t)\right)_{n\in\mathbb{N}}$ must be a $\|\cdot\|$-Cauchy sequence.

	\item The interparticle forces must saturate, \ie, for every finite-volume region $V \subset \mathbb{R}^3$, for all $\beta >0$, and for a certain range of $\mu$-values, we have
	$$\Tr_{\mathfrak{F}(\hilbert)^V}\left(e^{-\beta H^\prime_V} \right) <\infty.$$
	When this is the case, we set
	$$\begin{aligned}\rho_V &\doteq\frac{e^{-\beta H^\prime_V}}{\Tr_{\mathfrak{F}(\hilbert)^V}\left(e^{-\beta H^\prime_V}\right)}\\
					\omega_V(A)&\doteq\Tr_{\mathfrak{F}(\hilbert)^V}\left(\rho_V A\right)\\
	\end{aligned}$$
	\item Let $(V_n)_{n\in\mathbb{N}}\subset \mathbb{R}^3$ such that $(V_n)\subset V_{n+1}$ and $\displaystyle \bigcup_{n \in \mathbb{N}}{V_n}=\mathbb{R}^3$, then, for each $A\in \calgebra$,
	$$\lim_{n\rightarrow \infty}{\omega_{V_n}(A)}=\omega(A).$$
\end{enumerate}
For simplicity, we will refer to these conditions as the $Haag-Hugenholtz-Winnink$ conditions or HHW-conditions. \index{HHW-conditions}

Notice that condition $(i)$ say that, for fixed $t\in\mathbb{R}$ and for $(V_n)_{n\in\mathbb{N}}\subset\mathbb{R}^3$ with $V_n\subset V_{n+1}$ and  $\displaystyle \bigcup_{n \in \mathbb{N}}{V_n}=\mathbb{R}^3$, we have
\begin{equation}
\label{cauchy1}
\lim_{n\rightarrow \infty}{U_{V_n}(t)A U_{V_n}^{-1}(t)}= U(t)A U^{-1}(t),
\end{equation}
but we are not (and should not be) demanding  any kind of uniform convergence in time parameter. A physical example that clarifies why we should not demand such a convergence is a system with one particle with velocity $v$ such that $v t> diam(V)$. In this situation, because the particle must reflect in the boundary of $V$, the dynamics in the whole space and in $V$ are very different, $\ie$, $U_{V_n}(t)A U_{V_n}^{-1}(t)$ and $U(t)A U^{-1}(t)$ have no reason to be close.

Now, let $z\in \mathbb{C}$ and $A \in \calgebra$, define
$$A_z^{V_n}\doteq e^{\iu H^\prime_{V_n} z}Ae^{-\iu H^\prime_{V_n} z}.$$

Although $A_z$ be can unbounded, $A_z^{V_n} e^{-\beta H^\prime_{V_n}}$ and $e^{- \beta H^\prime_{V_n}}A_z^{V_n}$ are bounded operators and of trace class for $0\leq\gamma\leq \beta$ and $-\beta \leq \gamma \leq 0$, respectively . In fact,
$$\begin{aligned}
A_z^{V_n} e^{- \beta H^\prime_{V_n}}	&=e^{\iu H^\prime_{V_n} z}Ae^{-\iu H^\prime_{V_n} z} e^{- \beta H^\prime_{V_n}}\\
							&=e^{- \gamma H^\prime_{V_n} t}\left(e^{\iu H^\prime_{V_n} t}Ae^{-\iu H^\prime_{V_n} t}\right) e^{- (\beta-\gamma) H^\prime_{V_n}},\\
\end{aligned}$$
where we see that all three terms are bounded for $0\leq \gamma\leq \beta$, as a consequence of the second HHW-condition, and, more than that, this operator must be trace-class because of the first and last terms. The argument for $ e^{- \beta H^\prime_{V_n}}A_z$ is identical.

For each $A, B \in \calgebra$, define the complex function
$$F_{A,B}(z)=\Tr(B A_z e^{-\beta H^\prime_{V_n}})=\omega_V\left(B A_z^{V_n}\right),$$
which is an analytic function in $\strip{\beta}$ and continuous in $\overline{\strip{\beta}}$, furthermore
$$F_{A,B}(t)=\Tr(B A_t e^{-\beta H^\prime_{V_n}})=\omega_V(B A_t^{V_n}).$$

Finally, we can compute
$$\begin{aligned}
F_{A,B}(t+\iu\beta)		
						&=\Tr\left(B A_{t+\iu \beta}^{V_n} e^{-\beta H^\prime_{V_n}}\right)\\
						&=\Tr\left(B e^{\iu H^\prime_{V_n} t}\left(e^{- \beta H^\prime_{V_n} t}Ae^{-\iu H^\prime_{V_n} t}\right) e^{- (\beta-\beta) H^\prime_{V_n}}\right)\\
						&=\Tr\left(B e^{\iu H^\prime_{V_n} t}\left(e^{- \beta H^\prime_{V_n} t}Ae^{-\iu H^\prime_{V_n} t}\right)\right)\\
						&=\Tr\left(A e^{-\iu H^\prime_{V_n} t}B e^{\iu H^\prime_{V_n} t}\left(e^{- \beta H^\prime_{V_n} t}\right)\right)\\
						&=\omega_{V_n}\left(A B_t^{V_n} \right).\\
\end{aligned}$$

Therefore, $\omega_{V_n}$ is a $(\tau\beta)$-KMS state for $\tau^{V_n}_t(A)=A_t^{V_n}$.

Notice that the third HHW-condition is equivalent to condition $(ii)$ in Proposition \ref{limitKMS} in that context.

\begin{corollary}
Suppose $(V_n)_{n\in\mathbb{N}}\subset\mathbb{R}^3$ with $V_n\subset V_{n+1}$ and  $\displaystyle \bigcup_{n \in \mathbb{N}}{V_n}=\mathbb{R}^3$.
Then, if each KMS state $\omega_{V_n}$ satisfies the first and third HHW-conditions, the state $\displaystyle\omega(A)=\lim_{n\rightarrow \infty}{\omega_{V_n}(A)}$ is also a KMS state.
\end{corollary}
\begin{proof}
It is simply the statement of Theorem \ref{limitKMS} with $\tau^V_t(A)=U_V(t)A U_V(t)^{-1}$ and $\tau_t(A)=U(t)A U(t)^{-1}$, since, as we have said, it follows from the first HHW-condition that, for fixed $t\in\mathbb{R}$,

$$\lim_{n\rightarrow \infty}{\tau_t^{V_n}(A)}=\lim_{n\rightarrow \infty}{U_{V_n}(t)A U_{V_n}^{-1}}= U(t)A U^{-1}(t)=\tau_t(A).$$

\end{proof}
 
%Now, by condition $(ii)$ the Gibbs states can be expressed as
%$$\omega_V(A)=Tr(\rho A)$$
%where $\rho$ is a positive operator (matrix), so, it has a unique positive square root.
%Denote
%$$\xi=\rho^{\frac{1}{2}}$$
%follow from the definition that $\\xi_0$ is a Hilbert-Schmidt operator and
%$$\rho^{-1}=\frac{e^{\beta H^\prime}}{Tr\left(e^{-\beta H^\prime}\right)}$$
%Is a unbounded positive operator by the spectral theorem with $\overline{\Dom{\rho^{-1}} }=\calgebra$
%and must be so to $\xi_0$, thus the image of $\xi_0$ is a dense subspace of $\hilbert_{\mathfrak{F}}$ or, equivalently,
%$$\begin{aligned}
%A\xi_0		&=0 \Rightarrow A=0\\
%\xi_0 A		&=0 \Rightarrow A=0\\
%\end{aligned}$$
%Coming back to the state definition
%$$0=\omega_V(A^\ast A)=Tr\left(\xi_0 A^\ast A \xi_0\right) \Rightarrow \xi_0 A=0 \Rightarrow A=0$$

% coming back to the previous discussion about the cyclic representation, let  

%********************************** Second Section  **************************************

% flatex input end: [Chapter3/chapter3.tex]
% KMS
% flatex input: [Chapter4/chapter4.tex]
\chapter{The Tomita-Takesaki Modular Theory}
\label{chapTTMT}
\setlength\epigraphwidth{10cm}
\setlength\epigraphrule{0pt}
\epigraph{[The Tomita-Takesaki] \textit{theorem is a beautiful example of ``prestabilized harmony'' between physics and mathematics.}}{\textit{Rudolf Haag, Local Quantum Physics}}
%******************* First Section ***********************

\section{Modular Operator and Modular Conjugation}
\label{secModular}

It was Tomita's ideas presented in a conference in 1967 that started this theory, but undoubtedly it was Takesaki's work, \cite{takesaki70.2}, that put all those ideas in a solid basis. That is why we call it nowadays the Tomita-Takesaki Modular Theory.

Modular Theory was responsible for significant advance in Operator Algebras and its applications to Quantum Field Theory. For example, it plays a central role in Connes's classification of type III algebras. In the realm of physics, the first to notice the relation between the Tomita-Takesaki Modular Theory and physics, more specific, with equilibrium states, were Haag, Hugenholtz, and Winnink, as presented in the previous chapter.

This section is devoted to present the definitions and main properties of the modular operator and the modular conjugation, as well as, to present characterizations of these operators, which give us important properties of the positive cone.

This topic is a standard subject and can be found in classical books \eg \, \cite{Bratteli1}, \cite{KR83} and \cite{Takesaki2002} or even in \cite{Araki74}.

%****************Add some comments about the relation between the modular theory and KMS states.

Here cyclic and separating vectors \footnote{see Definition \ref{def:cyclicandseparating}} will be crucial. Note that, although the definition of a cyclic and of a separating vector makes explicit reference to concrete $C^\ast$-algebras, it is not difficult to rewrite this definition using (cyclic) representations. It is also fundamental the relation between cyclic and separating stated in Corollary \ref{cyclic_separating}.

\begin{proposition}
Let $\nalgebra$ be a von Neumann algebra. $\Omega$ is cyclic for $\nalgebra \Leftrightarrow \Omega$ is separating for $\nalgebra^\prime$.
\end{proposition}

Let us now define two operators, which will give rise to the operators that give name to this section. For the cyclic and separating vector $\Omega$, define the anti-linear operators:
\begin{center}
\begin{minipage}[c]{5cm}
$$\begin{aligned}
S_0 :	&\{A\Omega \in \hilbert \ | A\in \calgebra\} 		&\to		& \hspace{0.4cm} \hilbert\\
		&\hspace{1.2cm} A\Omega 								&\mapsto	&\hspace{0.2cm} A^\ast\Omega
\end{aligned}$$
\end{minipage}\hspace{0.75cm},\hspace{0.75cm}
\begin{minipage}[c]{5cm}
$$\begin{aligned}
F_0 :	& \{A^\prime\Omega \in \hilbert \ | A^\prime\in \calgebra^\prime\} 		&\to		& \hspace{0.4cm} \hilbert \\
		& \hspace{1.2cm} A^\prime\Omega 								&\mapsto	&\hspace{0.2cm}A^{\prime\ast}\Omega
\end{aligned}$$
\end{minipage}
\end{center}

Note that the domains of the operators are dense subspaces.

\begin{lemma}
\label{closureS}
The operators $S_0$ and $F_0$ are closable operators. Moreover, $S_0^\ast=\overline{F_0}$ and $F_0^\ast=\overline{S_0}$.
\end{lemma}
\begin{proof}
By definition,
$$\begin{aligned}
\ip{A\Omega}{F_0A^\prime\Omega} &= \ip{A\Omega}{A^{\prime\ast}\Omega}\\
								&= \ip{A^\prime\Omega}{A^{\ast}\Omega} \\
								&= \ip{A^\prime\Omega}{S_0 A\Omega}\\
								&= \overline{ \ip{S_0 A\Omega}{A^\prime\Omega}}\\
								&=\ip{A \Omega}{S_0^\ast A^\prime\Omega},
\end{aligned}$$
since the vector on the right-hand side forms a dense subset we must have $S_0^\ast A^\prime\Omega=F_0A^\prime\Omega$, thus $S_0^\ast$ extends $F_0$ and since an operator is closable if and only if its adjoint is densely defined \iffalse **************** fazer! Theorem \ref{CDD} \fi. The analogous calculation shows that $F_0^\ast$ extends $S_0$. Furthermore, $F_0\subset S_0^\ast \Rightarrow \overline{S_0}=S_0^{\ast\ast}\subset F_0^\ast $, but since $S_0\subset F_0^\ast$ the equality follows.
\end{proof}

\begin{notation}
We will denote $S=\overline{S_0}$ and $F=\overline{F_0}$.
\end{notation}
An important point to stress now is that we omitted the dependence on $\Omega$ to keep the notation clean, but we will mention it in the following.

Moreover, even though $S$ is not a bijection, it is injective and we will write $S^{-1}$ (which is equal to $S$) to denote its inverse over its range. The same holds for $\Delta$, which will be defined soon.

\begin{definition} \glsdisp{modDelta}{\hspace{0pt}} \glsdisp{modJ}{\hspace{0pt}} \index{operator! modular}
We denote by $J_\Omega$ and $\Delta_\Omega$ the unique anti-linear partial isometry and positive operator, respectively, in the polar decomposition of $S$, \ie,   $S=J_\Omega\Delta_\Omega^{\frac{1}{2}}$. $J_\Omega$ is called the modular conjugation and $\Delta_\Omega$ is called the modular operator. 
\end{definition}

Note that the existence and uniqueness of these operators are stated in the Polar Decomposition Theorem.
%*********************** Fazer teorema da decomposição polar

\begin{proposition}
\label{MOProp}
The modular conjugation and the modular operator satisfy the following relations:

\centering{\begin{minipage}{0.7\textwidth}
\begin{multicols}{2}
\begin{enumerate}[(i)]
\item $\Delta_\Omega = FS,$
\item $\Delta_\Omega^{-1} =SF,$
\item $J_\Omega^\ast J_\Omega=\mathbbm{1},$
\item $\Delta_\Omega^{-\frac{1}{2}}=J_\Omega \Delta_\Omega^{\frac{1}{2}}J_\Omega^\ast,$
\item $J_\Omega=J_\Omega^\ast,$
\item $F=J_\Omega\Delta_\Omega^{-\frac{1}{2}}.$
\end{enumerate}
\end{multicols}
\end{minipage}}
\end{proposition}
\begin{proof}
$(i):$ Of course, $\Delta_\Omega=S^\ast S$ by the polar decomposition and the property follows just by using Lemma \ref{closureS}.

$(ii):$ Note that $S_0=S_0^{-1}$ and $F_0=F_0^{-1}$, thus $S=S^{-1}$ and $F=F^{-1}$ and using $(i)$ $\Delta_\Omega^{-1}=(FS)^{-1}=S^{-1}F^{-1}=SF$.

$(iii):$ Since $\Delta_\Omega$ has a dense range in $\hilbert$, $J_\Omega^\ast J_\Omega$ must be a densely defined projection, thus $J_\Omega^\ast J_\Omega=\mathbbm{1}$.

$(iv):$ Using $(iii)$ we must have $J_\Omega\Delta_{\Omega}^{\frac{1}{2}}=S=S^{-1}=\Delta_{\Omega}^{-\frac{1}{2}}J_\Omega^{-1}= \Delta_{\Omega}^{-\frac{1}{2}}J_\Omega^\ast \Rightarrow \Delta_{\Omega}^{-\frac{1}{2}}=J_\Omega\Delta_{\Omega}^{\frac{1}{2}}J_\Omega^\ast$.

$(v):$ Since $S^\ast S$ and $S S^\ast$ are unitary equivalent though $J_\Omega$, \ie,  $S^\ast S=J S S^\ast J^\ast$, and using $(iii)$ and $(iv)$ it follows
$J_\Omega\Delta_{\Omega}^{\frac{1}{2}}=J_\Omega\Delta_{\Omega}^{\frac{1}{2}}J_\Omega^\ast J_\Omega=\Delta_{\Omega}^{-\frac{1}{2}}J_\Omega$.

Finally, by $(iii)$, $S=S^{-1}=J_\Omega^\ast\Delta_{\Omega}^{\frac{1}{2}} \Rightarrow J_\Omega^\ast=J_\Omega$.

$(vi):$ Just use $F=S^\ast$, $(v)$ and $(iv)$.

\end{proof}

The epigraph of this chapter was stated by Haag in \cite{Haag64} and it is about the theorem below, one of the most important results in Tomita-Takesaki Modular Theory which is extremely significant to both Physics and Mathematics.

For the proof of the theorem we refer to \cite{Takesaki2002} and \cite{Bratteli1}.

\begin{theorem}[Tomita-Takesaki]\index{theorem! Tomita-Takesaki}
\label{TTT}
Let $\nalgebra$ be a von Neumann algebra with cyclic and separating vector $\Omega$, then
$$J_\Omega \nalgebra J_\Omega=\nalgebra^\prime, \qquad \Delta_\Omega^{\iu t} \nalgebra \Delta_\Omega^{-\iu t}=\nalgebra \quad \forall t \in \mathbb{R}.$$
\end{theorem}

Just to mention an immediate consequence of the theorem above, notice that such a result implies that, for fixed $t\in\mathbb{R}$, $A\xmapsto{\sigma^\Omega_t}\Delta_\Omega^{\iu t}A\Delta_\Omega^{-\iu t}$ defines an automorphism of the algebra. Hence, $\left\{\sigma^\Omega_t \right\}_{t\in\mathbb{R}}$ is a one-parameter group of isometries.

\begin{definition}[Modular Automorphism Group]
	\label{def:modularautomosphism}
Let $\nalgebra$ be a von Neumann algebra with cyclic and separating vector $\Omega$ and let $\Delta_\Omega$ be the associated modular operator. For each $t\in\mathbb{R}$, define the automorphism $\sigma^\Omega:\nalgebra \to \nalgebra$ by $\sigma^\Omega_t(A)=\Delta_\Omega^{\iu t}A\Delta_\Omega^{-\iu t}$. We call the modular automorphism group\index{modular! automorphism group} the one-parameter group $\left\{\sigma^\Omega_t \right\}_{t\in\mathbb{R}}$.
\end{definition}

The modular automorphism group is very important in Physics because it is related with KMS states and, more than that, with the time evolution of the system, as we will mention in the suitable moment.

\begin{notation}
	We will denote the modular automorphism group with respect to a cyclic and separating vector $\Omega$ by $\left\{\sigma^\Omega_t \right\}_{t\in\mathbb{R}}$. In addition, given a faithful normal semifinite weight $\phi$ on a von Neumann algebra, we will denote by $\left\{\sigma^\phi_t \right\}_{t\in\mathbb{R}}$ the modular automorphism group with respect to the cyclic and separating vector obtained in the GNS-construction.
\end{notation}

The last result of this section presents some useful properties and a characterization of the modular conjugation. The proof of the theorem below can be found in \cite{Araki74}, as well as, an extensive treatment of the subject.

\begin{theorem}
Let $\nalgebra$ be a von Neumann algebra and $\Omega\in \hilbert$ a cyclic and separating vector. Then, an operator $J$ is the modular conjugation with respect to $\Omega$ if and only if
\begin{enumerate}[(i)]
\item $\ip{Jx}{Jy}=\ip{y}{x} \ \forall x,y \in \hilbert$,
\item $J^2=\mathbbm{1}$,
\item $J\nalgebra J=\nalgebra^\prime$,
\item $J\Omega=\Omega$,
\item $\ip{\Omega}{AJAJ\Omega}\geq 0 \ \forall A\in \nalgebra$, and the equality holds, if and only if, $A=0$.
\end{enumerate} 
\end{theorem}
\begin{proof}
$(\Rightarrow)$ $(i)-(iv)$ were already proved in Theorem \ref{MOProp}.
For $(v)$ note that,
$$\begin{aligned}
\ip{\Omega}{AJAJ\Omega} &=\ip{\Omega}{AJA\Omega}\\
						&=\ip{\Omega}{AJ\Delta^{-\frac{1}{2}}\Delta^{\frac{1}{2}}A\Omega}\\
						&=\ip{\Omega}{A\Delta^{\frac{1}{2}}J\Delta^{\frac{1}{2}}A\Omega}\\
						&=\ip{\Omega}{A\Delta^{\frac{1}{2}}A^\ast\Omega}\\
						&=\ip{A^\ast\Omega}{\Delta^{\frac{1}{2}}A^\ast\Omega}\\
						&\geq 0,
\end{aligned}$$
since $\Delta^{\frac{1}{2}}$ is a positive operator.

Now, let $A\in\nalgebra$ and suppose $\ip{\Omega}{AJAJ\Omega}=0$, then $\Delta^{\frac{1}{4}}A^\ast\Omega=0\Rightarrow A^\ast\Omega=0\Rightarrow A=0$.

$(\Leftarrow)$ From $(i)$ and $(ii)$, $J$ is anti-unitary. Let us define an operator $T$ in the dense domain $\nalgebra\Omega$ by
$$TA\Omega=JA^\ast \Omega \quad \forall A\in \nalgebra.$$

From this definition, $T=JS_0$, which means that $T$ is also closable ($J$ is an isometry), $\overline{T}=JS$ and $\overline{T}^\ast=S^\ast J$, but using $(iii)$ and $(iv)$, $JA\Omega=JAJ\Omega \in \nalgebra^\prime\Omega\in \Dom{S^\ast}$. Hence
 $$\begin{aligned}
T^\ast A\Omega 	&= S^\ast JA\Omega \\
				&= S^\ast JAJ\Omega \\
				&= J A^\ast J\Omega \\
				&= J A^\ast \Omega \\
				&= T A\Omega. \\
 \end{aligned}$$
Thus $T$ and $T^\ast$ coincides on a dense subset and, therefore, we must have $\overline{T}=T^\ast$.

Now, for all $A \in \nalgebra$,
$$\begin{aligned}
\ip{A\Omega}{TA\Omega}	&= \ip{A\Omega}{JA^\ast\Omega}\\
						&= \ip{A\Omega}{JA^\ast J \Omega}\\
						&\geq 0.\\
\end{aligned}$$
This mean, by definition, means that $\overline{T}$ is a positive operator and, furthermore, $$\overline{T}=JJ_\Omega\Delta_\Omega^{\frac{1}{2}} = u\Delta_\Omega^{\frac{1}{2}}, $$
 where $u=JJ_\Omega$ is a unitary operator, since it is a product of two anti-unitary operators. But then, $T^\ast \overline{T} = \Delta_\Omega^{\frac{1}{2}} u^\ast u \Delta_\Omega^{\frac{1}{2}}=\Delta_\Omega$ and, by Theorem \ref{sqrt}, we conclude $\overline{T} =\Delta_\Omega^\frac{1}{2}$. Thus $u=\mathbbm{1}\Rightarrow J=J_\Omega$.
 
\end{proof}

The last comment we would like to add at this section is that, thanks to the Reeh-Schlieder Theorem, it is possible to obtain a modular operator for the algebra of local observables using the vacuum state, see \cite{Araki99} for more details. It has still been a very important and interesting result that corroborates the believe that von Neumann algebras are the right framework to describe quantum systems. Another interesting question is what is the interpretation of this local modular operator? The answer for this question is unknown in general, but in some spacial cases, \eg \ the one described by Bisognano and Wichmann in \cite{bisognano75} and \cite{bisognano76}, it has a reasonable physical interpretation.

\section{The Cones $V^\alpha_\Omega$}
\label{SecCones}

The general definition of a cone\index{cone} in a vector space is a subset closed by the operation of multiplying by a positive scalar. Analogously, a convex cone is a subset $C$ of the vector space such that, if $x,y\in C$ and if $\alpha,\beta\geq0$, $\alpha x+\beta y$ lies in $C$. Here, we are interest in specific cones, thus, presenting any general theory is beyond the scope of this thesis. 

The cones discussed here plays an interesting role in Araki's perturbation theory, but also in the Araki-Masuda Noncommutative $L_p$-Spaces and multiple-time KMS condition. Such a connection is nowhere surprising, since both are dependent on the possibility of analytically extending the modular operator.

In this section we will follow strictly the results presented in \cite{Araki74}, but other appearances and applications of those cones can be seen in \cite{Araki73}, \cite{Araki74}, \cite{Araki82}, and also in \cite{takesaki70.2} in special case and with another name.

\begin{definition}\glsdisp{cones}{\hspace{0pt}}
\label{DefCone} \index{cone}
Let $\Omega \in \hilbert$ be a cyclic and separating vector and $\alpha \in \mathbb{R}$. We define the cones $V^\alpha_\Omega$ of a von Neumann algebra as follows:
$$V^\alpha_\Omega=\overline{\{\Delta^\alpha_\Omega A\Omega \ | \ A\in\nalgebra , \ A \geq 0\}}^{w}$$
\end{definition}

\begin{definition}
%reference Positive Linear Maps of Operator Algebras By Erling Størmer
\label{DefDualCone} \index{cone! dual}
Let $C$ be a closed and convex cone in $\hilbert$. We define the dual cone as the set
$$C^{\circ}=\{\xi \in \hilbert \ | \ \ip{\xi}{\eta}\geq0, \ \forall \eta \in C\}.$$
\end{definition}

The just defined cones and its duals have very interesting properties in connections with Modular Theory, as stated in the next proposition.

\vspace{1 cm}

\begin{theorem}
\label{theorem3in2} %Some properties of modular conjugation
$V^\alpha_\Omega$ has the following properties:
\begin{enumerate}[(i)]
\item $V_{\Omega}^\alpha$ is a pointed weakly closed convex cone invariant under $\Delta_{\Omega}^{\iu\alpha}$;

\item $V_{\Omega}^\alpha\subset\Dom{\Delta_{\Omega}^{\frac{1}{2}-2\alpha}}$ and
\begin{equation}
\label{equationcone1}
J_\Omega \Phi=\Delta_{\Omega}^{\frac{1}{2}-2\alpha} \Phi \, , \quad \Phi\in V_{\Omega}^\alpha;
\end{equation}

\item $\Delta_{\Omega}^{\alpha}V_{\Omega}^0$ is a dense subset of $V_{\Omega}^\alpha$;

\item $J_\Omega V_\Omega^\alpha= V_\Omega^{\frac{1}{2}-\alpha}$;

\item The dual of $V_{\Omega}^\alpha$ is $V_{\Omega}^{\frac{1}{2}-\alpha}$;

\item $V_{\Omega}^\alpha=\Delta_{\Omega}^{\alpha-\frac{1}{4}}\left\{V_\Omega^{\frac{1}{4}}\cap\Dom{\Delta_{\Omega}^{\alpha-\frac{1}{4}}}\right\}$;

\item If $A \in \nalgebra$ and $A\Omega \in V_{\Omega}^\alpha$, then $\Delta_\Omega^{\iu z}A\Delta_\Omega^{-\iu z}$ is bounded by $\|A\|$ for $z\in \strip{2\alpha}$\footnote{see Definition \ref{Anal} or the List of Symbols.} and satisfies
$$\begin{aligned}
\overline{\Delta_\Omega^{-2\alpha}A\Delta_\Omega^{2\alpha}}	&=A^\ast\\
\overline{\Delta_\Omega^{-\alpha}A\Delta_\Omega^{\alpha}}	&\geq 0\\
\end{aligned}$$ 
where the bar indicates the closure;

\item If $\Phi \in V_{\Omega}^{\alpha}$, $\alpha\leq \frac{1}{4}$ and there exists $M>0$ such that $\omega_\Phi \leq M \omega_\Omega$, then there exists $A\in \nalgebra$ such that $\Phi=A\Omega$ and $\|A\|\leq M^{\frac{1}{2}}$;

\item If $A\Omega \in V_{\Omega}^{\alpha}$, $A\in \nalgebra$, then $\left(\|A\|-A\right)\Omega \in V_{\Omega}^{\alpha}$.
\end{enumerate}
\end{theorem}

\begin{proof}
$(i)$ By definition $V^\alpha_\Omega$ is weakly closed, and it is convex because the convex combinations of positive elements is positive. It is invariant under $\Delta_{\Omega}^{\iu t}$ because $\Delta_{\Omega}^{\iu t}A\Delta_{\Omega}^{-\iu t}\in \nalgebra$ thanks to Theorem \ref{TTT} and
$$\Delta_{\Omega}^{\iu t}\left(\Delta_{\Omega}^{\alpha}A\Omega\right) = \Delta_{\Omega}^{\alpha}\left(\Delta_{\Omega}^{\iu t}A\Omega\right) =\Delta_{\Omega}^{\alpha}\left(\Delta_{\Omega}^{\iu t}A\Delta_{\Omega}^{-\iu t}\Omega\right) \in V^\alpha_\Omega.$$

$(ii)$ Let $A\in \nalgebra$, $A\geq 0$, then
\begin{equation}
\label{equationcone2}
\begin{aligned}
J_\Omega(\Delta^\alpha_\Omega A\Omega)
	&=J_\Omega\left(\Delta^{\alpha-\frac{1}{2}}_\Omega \Delta^{\frac{1}{2}}_\Omega A\Omega\right)\\
	&=\Delta^{\frac{1}{2}-\alpha}_\Omega J_\Omega\Delta^{\frac{1}{2}}_\Omega A\Omega\\
	&=\Delta^{\frac{1}{2}-\alpha}_\Omega A^\ast\Omega\\
	&=\Delta^{\frac{1}{2}-\alpha}_\Omega A\Omega\\
	&=\Delta^{\frac{1}{2}-2\alpha}_\Omega\left(\Delta^{\alpha}_\Omega A\Omega\right).\\
\end{aligned}
\end{equation}
which proves equation \eqref{equationcone1}. This also proves the statement, because $V^\alpha_\Omega$ is convex and weakly closed, thus strongly closed, hence equation \eqref{equationcone1} must hold in the whole domain since the operators are closed.

$(iii)$ $V^0_\Omega$ is convex and weakly closed, hence strongly closed. Let $\Phi \in V^0_\Omega$, then there exists $(A_n)_{n\in\mathbb{N}} \subset \nalgebra$, a sequence of positive operators, such that $A_n\Omega \rightarrow \Phi$. Using now that $\Delta_\Omega^\frac{1}{2}$ is a positive operator, we have from equations \eqref{1} and \eqref{eq3.13},  that
$$\begin{aligned}
\left\|\Delta^\alpha_\Omega(A_n\Omega-\Phi)\right\|^2
	&=\left\|\Delta^\alpha_\Omega(A_n\Omega-\Phi)\right\|^2 \\
	&\leq\left\|\Delta^{\frac{1}{2}}_\Omega(A_n\Omega-\Phi)\right\|^2+\left\|A_n\Omega-\Phi\right\|^2 \\
	&=\left\|J_\Omega(A_n\Omega-\Phi)\right\|^2+\left\|A_n\Omega-\Phi\right\|^2 \rightarrow 0.\\	
\end{aligned}$$

Hence $\Delta^\alpha_\Omega V^0_\Omega \subset V^\alpha_\Omega$, but $\Delta^\alpha_\Omega V^0_\Omega$ is closed and contains the set under the closure sign in Definition \ref{DefCone}, from where we have the equality. 

$(iv)$  $J_\Omega V_\Omega^\alpha \subset V_\Omega^{\frac{1}{2}-\alpha}$ was proved in equation \eqref{equationcone2} simply by noticing that $J_\Omega$ is continuous and $V_\Omega^{\frac{1}{2}-\alpha}$ is closed. As for the opposite inclusion, just apply $J_\Omega$ to both sides of equation and use $J_\Omega^2=\mathbbm{1}$. 

$(v)$Let $A\in\nalgebra$, $A\geq0$, $\xi, \Omega\in V^{\frac{1}{2}-\alpha}_\Omega$ and $(B_n)_{n\in\mathbb{N}} \subset \nalgebra$, $B_n\geq0$ such that $B_n\Omega \rightarrow \xi$. By definition, $\Delta_\Omega^\alpha A\Omega\in V^0_\Omega$ and
$$\begin{aligned}
\ip{\Delta_\Omega^\alpha A\Omega}{\Delta_{\Omega}^{\frac{1}{2}-\alpha}B_n\Omega}	&=\ip{J_\Omega\Delta_{\Omega}^{\frac{1}{2}}B_n\Omega}{J_\Omega A\Omega}\\
&= \ip{B_n\Omega}{J_\Omega AJ_\Omega\Omega}\\
&= \ip{\Omega}{B_n J_\Omega AJ_\Omega\Omega}\\
& \geq 0,\\
\end{aligned}$$
since $J_\Omega AJ_\Omega \in \nalgebra^\prime$, thanks to Theorem \ref{TTT}. Hence, $\ip{\Delta_\Omega^\alpha A\Omega}{\xi}\geq0$ and we conclude, since $V^{\frac{1}{2}-\alpha}_\Omega$ is closed, that $V_\Omega^\alpha\subset V^{\frac{1}{2}-\alpha}_\Omega$.

Now, let $B\in\nalgebra$, $B\geq0$, $\eta\in V^{\alpha}_\Omega$ and $(A_n)_{n\in\mathbb{N}} \subset \nalgebra$, $A_n\geq0$ such that $\Delta_\Omega^\alpha A_n\Omega \rightarrow \eta$. The exact same calculation gives us $\ip{\eta}{\Delta_\Omega^{\frac{1}{2}-\alpha} B\Omega}\geq0$ and hence $V_\Omega^\alpha\supset V^{\frac{1}{2}-\alpha}_\Omega$.

$(vi)$ First, suppose $\alpha\leq \frac{1}{4}$.

Take $\xi \in V^\alpha_\Omega$. Then, there exists a sequence $(A_n)_{n\in\mathbb{N}} \subset \nalgebra_{+}$ such that $A_n\Omega\rightarrow\xi$. From Lemma \ref{eq3.13} it follows that
$$\begin{aligned}
\left\|\Delta_\Omega^{\frac{1}{4}}\left(A_n-A_m\right)\Omega\right\|^2
	&=\left\|\Delta_\Omega^{\frac{1}{4}-\alpha}\left(\Delta_\Omega^{\alpha}\left(A_n-A_m\right)\Omega\right)\right\|^2 \\
	&\leq\left\|\Delta_\Omega^{\alpha}\left(A_n-A_m\right)\Omega\right\|^2 + \left\|\Delta_\Omega^{\frac{1}{2}-\alpha}\left(A_n-A_m\right)\Omega\right\|^2\\
	&=\left\|\Delta_\Omega^{\alpha}\left(A_n-A_m\right)\Omega\right\|^2 + \left\|J_\Omega\left(A_n-A_m\right)\Omega\right\|^2\\
	& \rightarrow 0.\\
\end{aligned}$$
Thus, $\left(\Delta_\Omega^{\frac{1}{4}}A_n\Omega\right)_{n\in\mathbb{N}}$ is a Cauchy sequence. Consequently, $\left(\Delta_\Omega^{\frac{1}{4}}A_n\Omega\right)\rightarrow \eta \in V_\Omega^{\frac{1}{4}}$. Since $\Delta^{\frac{1}{4}-\alpha}_\Omega$ is a closed operator and 
$$\ip{\Delta^\alpha_\Omega A_n\Omega}{ \Delta^{\frac{1}{4}}_\Omega A_n\Omega } =\ip{\Delta^\alpha_\Omega A_n\Omega}{\Delta_\Omega^{\frac{1}{4}-\alpha}\left(\Delta_\Omega^{\alpha}A_n\Omega\right) }\rightarrow\ip{\xi}{\eta}=\ip{\xi}{ \Delta^{\frac{1}{4}-\alpha}\xi}$$ 

Hence, $\Delta^{\frac{1}{4}-\alpha}\xi\in V^\frac{1}{4}_\Omega\cap\Dom{\Delta^{\alpha-\frac{1}{4}}}$ or, equivalently, $V^\alpha_\Omega \subset \Delta^{\frac{1}{4}-\alpha}\left(V^\frac{1}{4}_\Omega\cap\Dom{\Delta^{\alpha-\frac{1}{4}}}\right)$.

On the other hand, let $\xi\in V^\frac{1}{4}_\Omega\cap\Dom{\Delta^{\alpha-\frac{1}{4}}}$ and $\eta \in V^0_\Omega$. Then
$$\ip{\Delta^{\frac{1}{2}-\alpha}_\Omega\eta}{\Delta^{\alpha-\frac{1}{4}}_\Omega\xi}=\ip{\Delta^{\frac{1}{4}}_\Omega\eta}{\xi}\geq0$$ by $(iii)$ and $(v)$. Moreover, since the vectors on the right-hand side form a dense set, 
$$\ip{\zeta}{\Delta^{\alpha-\frac{1}{4}}_\Omega}{\xi}\geq0 \qquad \forall \zeta\in V^{\frac{1}{2}-\alpha}_\Omega.$$
Hence, $$\Delta^{\alpha-\frac{1}{4}}_\Omega\xi \in \left(V^{\frac{1}{2}-\alpha}_\Omega\right)^\circ=V^\alpha_\Omega.$$
In other words, $V^\alpha_\Omega \supset \Delta^{\frac{1}{4}-\alpha}\left(V^\frac{1}{4}_\Omega\cap\Dom{\Delta^{\alpha-\frac{1}{4}}}\right)$.

For $\alpha>\frac{1}{4}$ simply use the dual relation in $(v)$.

$(vii)$ Since $A\Omega \in V_{\Omega}^\alpha$ and $A\Omega \in \Dom{\Delta_\Omega^{\frac{1}{2}-2\alpha}}$, one has
$$J_\Omega A\Omega=\Delta_{\Omega}^{\frac{1}{2}}A^\ast \Omega= \Delta_{\Omega}^{\frac{1}{2}-2\alpha}A \Omega$$ from $(ii)$. Thus $A\Omega \in \Dom{\Delta_{\Omega}^{-2\alpha}}$ and $\Delta^{-2\alpha}A\Omega=A^\ast \Omega$.
The boundedness follows by Lemma \ref{lemma6} and the positivity follow from $(ii)$, $(iii)$ and $(v)$. If $A_1=\overline{\Delta_\Omega^{-2\alpha}A\Delta_\Omega^{2\alpha}}$, then $A\Delta_\Omega^\alpha \Phi= \Delta_\Omega^\alpha A_1\Phi$ for $\Phi$ in a dense subset. Hence, $\Delta_\Omega^\alpha A^\ast\Omega=A_1^\ast\Omega$, which implies $J_\Omega\Delta_\Omega^{\frac{1}{2}-\alpha} A\Omega=J_\Omega\Delta^{\frac{1}{2}} A_1\Omega$ and we conclude that $A\Omega=\Delta_\Omega^\alpha A_1 \Omega$. Note now that $A_1\geq 0$, hence $A\Omega \in V_\Omega^\alpha$.

$(viii)$ Since $\Phi\in V^\alpha_\Omega$, where $\alpha\leq\frac{1}{4}$, and $M>0$ is such that $\omega_\Phi\leq M \omega_\Omega$, it follows from Proposition \ref{PPC} that $(A,B)\mapsto \omega_\Phi(A^\ast B)=\ip{A\Phi}{B\Phi}$ induces a bounded sesquilinear form. By the Riesz Representation Theorem, there must exist a bounded operator $T^\prime \in B(\hilbert)$ such that 
\begin{equation}
\label{eqriesz}
\ip{A\Phi}{B\Phi}=\ip{T^\prime A\Omega}{B\Omega}.
\end{equation}

Moreover, since $0\leq \omega_\Phi(A^\ast A) =\ip{T^\prime A\Omega}{A\Omega}\leq M\omega_\Omega(A^\ast A)$ and the vectors of the form $A\Omega$ are dense, $0\leq T^\prime \leq M \mathbbm{1}$ and, in addition, $\ip{B\Omega}{T^\prime C A \Omega}=\omega_\Phi(B^\ast AC)=\omega_\Phi((C^\ast B)^\ast A)=\ip{B\Omega}{C T^\prime A \Omega}$. Consequently, $T^\prime\in \nalgebra^\prime$. Rewriting equation \eqref{eqriesz} in terms of $Q^\prime = \sqrt{T^\prime}$, one concludes that $\omega_\Phi =\omega_{Q^\prime \Omega}$ with $Q^\prime \in \nalgebra^\prime$ and $\|Q^\prime\|\leq M^\frac{1}{2}$. This implies that there should exist a partial isometry $u\in \nalgebra^\prime$ satisfying $\Omega=u Q^\prime \Phi$.

Using $(ii)$,
$$\Delta^{\frac{1}{2}-2\alpha}_\Omega uQ^\prime \Omega=\Delta^{\frac{1}{2}-2\alpha}_\Omega\Phi=J_\Omega\Phi=j_\Omega(u Q^\prime)\Omega$$
thus, by $(iv)$, $J_\Omega \in V^{\frac{1}{2}-\alpha}_\Omega$, by $(vii)$, $Q_1=j_\Omega(uQ^\prime) \in \nalgebra$.

Using Lemma \ref{lemma6}, there exists a family of operators $\tau_\Omega(z)$ analytic on $\strip{\frac{1}{2}-\alpha}$ and continuous on the border. Take
$$Q=\tau_\Omega\left(\frac{\iu}{2}-2\iu \alpha\right)Q_1 \in \nalgebra.$$
Then, from condition $(iii)$ in Lemma \ref{lemma6}, $\Phi=Q\Omega$ and from condition $(iv)$ (and the fact that $J_\Omega$ is an isometry), $\|Q\|\leq \|Q_1\|\leq M^\frac{1}{2}$.

$(ix)$ Let $A\in V^\alpha_\Omega$. Using $(vii)$, $\Delta_\Omega^{-\alpha}A\Delta_\Omega^{\alpha}$ is positive and bounded by $\|A\|$.Hence,
$$\Delta_\Omega^{-\alpha}\left(\|A\|\mathbbm{1}-A\right)\Delta_\Omega^{\alpha}=\|A\|\mathbbm{1}-\Delta_\Omega^{-\alpha}A\Delta_\Omega^{\alpha}$$
is also bounded by $\|A\|$, positive and affiliated with $\nalgebra$. By $(vii)$ again, we have the desired result.
 
\end{proof}

It is immediate from item $(v)$ in \ref{theorem3in2} that $V_\Omega^\frac{1}{4}$ is the unique self-dual cone. This property has very interesting consequences, and we will use this cone several times, which justifies a new notation.

\begin{notation}
We set $V_\Omega\doteq V_\Omega^\frac{1}{4}$ and $j_\Omega(Q)=J_\Omega Q J_\Omega, where $Q$\in\nalgebra$.
\end{notation}

The selfdual cone $V_\Omega$ has special properties. In particular, it has a strong connection with the modular conjugation, namely, if $\Phi \in V_\Omega$, then, $J_\Phi=J_\Omega$. The proof of this fact can be found in \cite{Araki74}.

\begin{theorem}
The cone $V_\Omega$, where $\Omega$ is a cyclic and separating vector for the von Neumann algebra $\nalgebra$, has the following properties:
\begin{enumerate}[(i)]
\item $V_\Omega$ is a pointed closed selfdual (weakly) closed convex cone;
\item $V_\Omega=\overline{\{Qj_\Omega(Q)\Omega | Q\in \nalgebra\}}^{\|\cdot\|}$;
\item  one has
$$\begin{aligned}
&\Delta_\Omega^{\iu t} V_\Omega = V_\Omega & \quad & \forall t \in \mathbb{R};\\
&J_\Omega x = x & \quad & \forall x \in V_\Omega;\\
&Q j_\Omega(Q) V_\Omega \subset V_\Omega & \quad & \forall Q\in \nalgebra;\\
&\ip{x}{Qj_\Omega(Q)y}\geq 0 & \quad & \forall x,y \in V_\Omega \textrm{ and } \forall Q \in \nalgebra;\\
\end{aligned}$$
\item If $\Phi \in V_\Omega$ is a cyclic or separating vector for $\nalgebra$, then $\Phi$ is cyclic and separating for $\nalgebra$ and $V_\Omega=V_\Phi$;
\item If $\Phi$ is cyclic and separating for $\nalgebra$, then $\Phi \in V\Omega$ if and only if $J_\Phi=J_\Omega$ and, for all positive $Z$ in the center of the algebra,
$$\ip{\Phi}{Z\Phi}\geq 0 $$
\item Any $\Phi\in \hilbert$ has a unique decomposition
$$\Phi = (\Phi_1-\Phi_2)+\iu(\Phi_3-\Phi_4)$$
such that $\Phi_i \in V_\Omega$, $1\leq i\leq 4$,  and $\Phi_1\perp\Phi_2$, $\Phi_3 \perp\Phi_4$.
\end{enumerate}
\end{theorem}
\begin{proof}
$(i)$ It is immediate from $(i)$ and $(v)$ in Theorem \ref{theorem3in2}.

$(ii)$ Let $A_n=A_n(0) \in X_\tau$ be a sequence of analytic elements defined by equation \eqref{formulaAE} with respect to the isometry defined by $\tau(A)=\Delta_\Omega^{\iu t} A\Delta_\Omega^{-\iu t}$.

For these elements we have

$$\begin{aligned}
A_n j_\Omega(A_n)\Omega	
	&=A_n J_\Omega A_n J_\Omega \Omega \\
	&=A_n \Delta^\frac{1}{2}_\Omega A_n^\ast \Omega \\
	&=\left(\int_\mathbb{R}{e^{-nt^2}\Delta_\Omega^{\iu t}A \Delta_\Omega^{-\iu t} dt}\right)\Delta_\Omega^{\frac{1}{2}}\left( \int_\mathbb{R}{e^{-nt^2}\Delta_\Omega^{\iu t}A \Delta_\Omega^{-\iu t} dt}\right)\Omega \\
	&=\left(\int_\mathbb{R}{e^{-nt^2}\Delta_\Omega^{\iu t}A \Delta_\Omega^{-\iu t+\frac{1}{4}} dt}\right) \left(\int_\mathbb{R}{e^{-nt^2}\Delta_\Omega^{\iu t-\frac{1}{4}}A \Delta_\Omega^{-\iu t} dt}\right)\Omega \\
	&=\Delta_\Omega^\frac{1}{4}\left(\int_\mathbb{R}{e^{-nt^2}\Delta_\Omega^{\iu( t+\frac{\iu}{4})} A \Delta_\Omega^{-\iu(t+\frac{\iu}{4})} dt}\right)\left(\int_\mathbb{R}{e^{-nt^2}\Delta_\Omega^{\iu( t+\frac{\iu}{4})}A \Delta_\Omega^{-\iu(t+\frac{\iu}{4})} dt}\right)\Delta_\Omega^{-\frac{1}{4}}\Omega \\
	&=\Delta_\Omega^\frac{1}{4} \tau_{\frac{\iu}{4}}(A_n)\tau_{\frac{\iu}{4}}(A_n)^\ast \Omega \in V_\Omega.\\
\end{aligned}$$

But, as we already  know, $A_n j_\Omega(A_n)\rightarrow Aj_\Omega(A)$. since $V_\Omega$ is closed, we have
$$\overline{\{Qj_\Omega(Q)\Omega | Q\in \nalgebra\}}\subset V_\Omega.$$
%******************* fazer o restante dos itens

We will skip the proof of the remaining items, but the reader can find them in \cite{Araki74}.
\end{proof}

%******************************************************************

%*********************** Second Section  *******************************
\section{Analyticity of Modular Automorphisms}
\label{AnalyticalModAut}

It must be clear by now how important analyticity is in this work. The reader can find this section a little technical (if he does not think this is the case of almost all this thesis) and he is right to think so. This section is basically made of lemmas about the possibility to analytically extend the action of the modular automorphism group. Nevertheless, it is indispensable to perturbation theory. The main references here are \cite{Araki74} and \cite{Araki73.2}.

Now, the comment about the relation between cones and analyticity that precedes Section \ref{SecCones} may start making sense.

Let $\Psi$ be a cyclic and separating vector for the von Neumann algebra $\nalgebra$.

\begin{notation}\glsdisp{modaut}{\hspace{0pt}}
	We denote by $\nalgebra_\Psi$ the set of all operators such that there exists a family of bounded linear operators $\tau_\Psi(z)Q$ which are entire analytic on $z \in \mathbb{C}$ and satisfy
	$$\tau_\Psi(t)Q=\Delta^{\iu t}_\Psi Q \Delta_\Psi^{-\iu t} \ \qquad \forall t \in \mathbb{R}.$$
	We also denote
	$$\begin{aligned}
	\nalgebra_{\Psi 1}&= \nalgebra_\Psi\cap\nalgebra, & D_{\Psi1}=\nalgebra_{\Psi1}\Psi \\
	\nalgebra_{\Psi 2}&= \nalgebra_\Psi\cap\nalgebra^\prime, & D_{\Psi2}=\nalgebra_{\Psi2}\Psi \\
	\end{aligned}$$
\end{notation}

%\Lemma 5 in [2] - Some properties of modular conjugation...
\begin{lemma}
	\label{lemma5}
	Let $Q \in \nalgebra$. Then, $Q\Psi \in \Dom{\Delta^{\frac{1}{2}+\alpha}_\Psi}$ if and only if $Q^\ast\Psi\in\Dom{\Delta^{-\alpha}_\Psi}$.
	
	Furthermore, if $Q\Psi \in \Dom{\Delta^{\frac{1}{2}+\alpha}_\Psi}$ for $\alpha>0$, there exists a family of closable operators $\hat{\tau}_\Psi(z)Q$, $z\in\strip{\alpha}$ with a common domain $D_{\Psi 2}$ such that
	\begin{enumerate}[(i)]
		\item $\hat{\tau}_\Psi(z)Q$ is affiliated with $\nalgebra$;
		\item for each $x \in D_{\Psi 2}$, $\hat{\tau}_\Psi(z)Qx$ is analytic for z in the strip $\strip{\alpha}$ and continuous on the boundary;
		\item for each $x \in D_{\Psi 2}$, $\hat{\tau}_\Psi(z)Qx=\Delta_\Psi^{\iu z}Q\Delta_\Psi^{-\iu z}x$;
		\item for each $x \in D_{\Psi 2}$, $\left(\hat{\tau}_\Psi(z)Q\right)^\ast x=\Delta_\Psi^{\iu \bar{z}}Q^\ast\Delta_\Psi^{-\iu \bar{z}}x$.
	\end{enumerate}
\end{lemma}
\begin{proof}
	The first part follows trivially from the properties of the modular operator and the modular conjugation.
	
	Assuming that $Q\Psi \in \Dom{\Delta^{\frac{1}{2}+\alpha}_\Psi}$ for $\alpha>0$, we can define $A_z$ by
	\begin{equation}
	\label{defAz}
	A_z Q^\prime \Psi = Q^\prime \Delta_\Psi^{\iu z}Q\Psi \quad \forall Q^\prime\in \nalgebra_{\Psi2}, \ z\in \strip{\alpha},
	\end{equation}
	because $\Psi$ is separating for $\nalgebra^\prime\supset \nalgebra_{\Psi2}$. This defines a linear operator.
	
	Furthermore, since $\Delta^{\frac{1}{2}}_\Psi Q \Psi=J_\Psi Q^\ast \Psi $, we have
	$$\begin{aligned}
	\ip{Q^\prime_1\Psi}{A_z Q^\prime_2 \Psi}
	&=\ip{Q^{\prime \ast}_2 Q^\prime_1 \Psi}{ \Delta^{\iu z}_\Psi Q\Psi}\\
	&=\ip{\Delta^{-\iu \bar{z}} Q^{\prime \ast}_2 Q^\prime_1\Psi}{Q\Psi} \\
	&=\ip{\Delta^{\frac{1}{2}}_\Psi\tau\left(-\bar{z}+\frac{\iu}{2}\right)\left(Q^{\prime \ast}_2 Q^\prime_1\right)\Psi}{J_\Psi\Delta^{\frac{1}{2}}_\Psi Q^\ast\Psi} \\
	&=\ip{J_\Psi\left(\tau\left(-\bar{z}+\frac{\iu}{2}\right)\left(Q^{\prime \ast}_2 Q^\prime_1\right)\right)^\ast\Psi}{J_\Psi\Delta^{\frac{1}{2}}_\Psi Q^\ast\Psi} \\
	&=\ip{J_\Psi\tau\left(-z-\frac{\iu}{2}\right)\left(Q^{\prime \ast}_1 Q^\prime_2\right)\Psi}{J_\Psi\Delta^{\frac{1}{2}}_\Psi Q^\ast\Psi} \\
	&=\ip{\Delta^{\frac{1}{2}}_\Psi Q^\ast\Psi}{\tau\left(-z-\frac{\iu}{2}\right)\left(Q^{\prime \ast}_1 Q^\prime_2\right)\Psi} \\
	&=\ip{\Delta^{\frac{1}{2}}_\Psi Q^\ast\Psi}{\Delta^{-\iu z+\frac{1}{2}}_\Psi Q^{\prime \ast}_1 Q^\prime_2\Psi} \\
	&=\ip{Q_1^\prime\Delta_\Psi\Delta^{\iu \bar{z}}_\Psi Q^\ast\Psi}{ Q^\prime_2\Psi}, \\
	\end{aligned}$$
	where $Q^\ast\Psi\in\Dom{\Delta^{-\alpha}_\Psi}$ as proved previously, hence in $\Dom{\Delta^{\iu \bar{z}}_\Psi}$. It follows that the dense set $D_{\Psi2}\subset{\Dom{A^\ast_z}}$ and this implies that $A_z$ is closable. Consider the family of closable operators $A_z=\hat{\tau}(z)Q$, let us prove the properties above.
	
	$(i)$ Let $Q_1^\prime, Q_2^\prime \in \nalgebra_{\Psi2}$. Using the definition given in equation \eqref{defAz}, we have
	$$Q_1^\prime A_z Q_2^\prime \Psi=Q_1^\prime Q_2^\prime \Delta^{\iu z}_\Psi Q \Psi=A_zQ_1^\prime Q_2^\prime \Psi.$$
	This means that $\hat{\tau}(z)Q$ commutes with $Q_1^\prime\in \nalgebra_{\Psi2}$ which implies that it is affiliated with $(\nalgebra_{\Psi2})^\prime=\nalgebra$.
	
	$(ii)$ Let $Q^\prime\in \nalgebra_{\Psi2}$, and $z\in \strip{\alpha}$, then $$\left(\hat{\tau}(z)Q\right)Q^\prime \Psi=A_z Q^\prime \Psi=Q^\prime \Delta_\Psi^{\iu z}Q\Psi$$
	and continuity and analyticity follows.
	
	$(iii)$ Using the previous calculation and $(i)$,
	$$\begin{aligned}
	\left(\hat{\tau}(z)Q\right)Q^\prime \Psi&=Q^\prime \Delta_\Psi^{\iu z}Q\Psi\\
	&=\Delta^{\iu z}_\Psi\left(\tau_\Psi(-z)Q^\prime \right)Q\Psi\\
	&=\Delta^{\iu z}_\Psi Q\left(\tau_\Psi(-z)Q^\prime \right)\Psi\\
	&=\Delta_\Psi^{\iu z}Q\Delta_\Psi^{-\iu z}Q^\prime \Psi.
	\end{aligned}$$
	
	$(iv)$ Using $(iii)$, we have
	$$\begin{aligned}
	\ip{Q^\prime_1\Psi}{ \left(\hat{\tau}(z)Q\right) Q^\prime_2 \Psi}
	&=\ip{Q^\prime_1\Psi}{\Delta^{\iu z}_\Psi Q \Delta^{-\iu z}_\Psi Q^\prime_2 \Psi} \\
	&=\ip{\Delta^{\iu \bar{z}}_\Psi Q^\ast \Delta^{-\iu \bar{z}}_\Psi Q^\prime_1\Psi}{ Q^\prime_2 \Psi} \\
	&=\ip{\left(\hat{\tau}(z)Q^\ast\right) Q^\prime_1\Psi}{Q^\prime_2 \Psi}.
	\end{aligned}$$
\end{proof}

\begin{corollary}
	\label{clemma5}
	Let $Q \in \nalgebra$. Then $Q\Psi \in \Dom{\Delta^{-\alpha}_\Psi}$ if and only if $Q^\ast\Psi\in\Dom{\Delta^{\frac{1}{2}+\alpha}_\Psi}$.
	
	Furthermore, if $Q\Psi \in \Dom{\Delta^{-\alpha}_\Psi}$ for $\alpha>0$, then there exists a family of closable operators $\hat{\tau}_\Psi(z)Q$, $z\in\strip{\alpha}$, with a common domain $D_{\Psi 2}$ such that
	\begin{enumerate}[(i)]
		\item $\hat{\tau}_\Psi(z)Q$ is affiliated with $\nalgebra$;
		\item for each $x \in D$, $\hat{\tau}_\Psi(z)Qx$ is analytic for z in the strip $\strip{\alpha}$ and continuous on the boundary;
		\item for each $x \in D$, $\hat{\tau}_\Psi(z)Qx=\Delta_\Psi^{\iu z}Q\Delta_\Psi^{-\iu z}x$;
		\item for each $x \in D$, $\left(\hat{\tau}_\Psi(z)Q\right)^\ast x=\Delta_\Psi^{\iu \bar{z}}Q^\ast\Delta_\Psi^{-\iu \bar{z}}x$.
	\end{enumerate}
\end{corollary}

%Lemma 6 in [2]

\begin{lemma}
	\label{lemma6} %Some properties of modular conjugation...
	Let $Q \in \nalgebra$. Suppose there exists $Q_1\in \nalgebra$ and $\alpha\neq0$ such that \mbox{$\Delta_\Psi^\alpha Q \Psi=Q_1\Psi$}. Then, there exists a family of operators $\tau_\Psi(z)$, $z \in \overline{\strip{\alpha}}$, which is analytic on $\strip{\alpha}$ and continuous on the boundary such that
	\begin{enumerate}[(i)]
		\item $\tau_\Psi(z)Q x = \Delta_\Psi^{\iu z} Q \Delta_\Psi^{-\iu z} x \ \ \forall x \in \Dom{ \Delta_\Psi^{-\iu z}}$.
		\item $\left(\tau_\Psi(z)\right)^\ast Q x = \Delta_\Psi^{\iu \overline{z}} Q^\ast \Delta_\Psi^{-\iu \overline{z}} x \ \ \forall x \in \Dom{ \Delta_\Psi^{\iu z}}$.
		\item $\left\|\tau_\Psi(z)Q\right\| \leq \max\{\|Q\|,\|Q_1\|\} $.
		\item $\tau_\Psi(0)Q=Q$ and $\tau_\Psi(-\iu \alpha)Q=Q_1$. 
	\end{enumerate}
\end{lemma}

\begin{proof}
	First, for any $Q_1\in \nalgebra$, $Q_1\Psi \in \Dom{\Delta^{\frac{1}{2}}_\Psi}$, hence $Q\Psi\in\Dom{\Delta^{\frac{1}{2}+\alpha}_\Psi}$. It follows from Lemma \ref{lemma5} that there exists a family of closable operators $\tau_\Psi(z)Q$, which are analytic on $\strip{\alpha}$ and continuous on its closure.
	
	Define
	$$f(z)\doteq \ip{x}{\hat{\tau}_\Psi(z) Q y} \qquad x,y \in D_{\Psi2}.$$
	Then, if we take $y=Q_2^\prime \Psi$, it follows from the Maximum Modulus Principle, \iffalse ************ Fazer! \fi and the fact that $\hat{\tau}_\Psi Q$ is analytic on the strip and continuous on its closure, that
	$$\begin{aligned}
	|f(z)|	&=| \ip{x}{Q_2^\prime\Delta^{\iu z}_\Psi Q\Psi}|\\
	&=| \ip{x}{Q^{\prime}_2\Delta^{\iu z}_\Psi Q\Psi}|\\
	&\leq \|x\| \|Q^{\prime}_2\Delta^{\iu z}_\Psi Q\Psi\|\\
	&\leq \|x\| \max\{\|Q^{\prime}_2Q\Psi\|,\|Q^{\prime}_2\Delta^{\alpha}_\Psi Q\Psi\|\}\\
	&\leq \|x\| \max\{\|Q^{\prime}_2Q\Psi\|,\|Q^{\prime}_2 Q_1\Psi\|\}\\
	&\leq \|x\| \|y\|\max\{\|Q\|,\|Q_1\|\}\\
	\end{aligned},$$
	and $(iii)$ holds.
	
	Now, the uniform boundedness and the analyticity (continuity) in the dense set $D_{\Psi2}$ implies that the closure of $\hat{\tau}_\Psi(z)Q$ is also analytic on the strip $\strip{\alpha}$ and continuous on its closure.
	Furthermore, the properties in the lemma are consequences of the properties in Lemma \ref{lemma5}.
	
\end{proof}

%Lemma 7 in H. Araki, Some properties of modular conjugation operators of a von Neumann albegra and a non-commutative Radon-Nikodym Theorem

\begin{lemma}
	\label{lemma7}
	Let $\Psi$ be a cyclic and separating vector for $\nalgebra$ and $S\in\nalgebra$ with a bounded inverse $S^{-1} \in \nalgebra$ such that $S\Psi \in V_{\Psi}$. If $Q_1, Q \in \nalgebra$ are such that $\Delta^{\frac{1}{2}}_\Psi Q \Psi=Q_1 \Psi$, then
	$\displaystyle\Delta^{\frac{1}{2}}_{S\Psi}Q(S\Psi)=Q_2(S\Psi),$	where $Q_2=SQ_1S^{-1}.$
\end{lemma}
\begin{proof}
	Since $S\Psi \in V_{\Psi}$, $J_{S\Psi}=J_\Psi$.
	$$\begin{aligned}
	\Delta^{\frac{1}{2}}_{S\Psi}QS\Psi 
	&=J_{S\Psi}Q^\ast S\Psi &=& \ j_{S\Psi}\left(Q^\ast\right)S\Psi \\
	&=S j_{S\Psi}\left(Q^\ast\right)\Psi &=& \ S j_{\Psi}\left(Q^\ast\right)\Psi \\
	&=S J_{\Psi}Q^\ast\Psi &=& \ S \Delta^{\frac{1}{2}}_{\Psi}Q\Psi\\
	&=S \Delta^{\frac{1}{2}}_{\Psi}S^{-1}(SQ)\Psi&=& \ Q_2(S\Psi). \\
	\end{aligned}$$
\end{proof}

\section{Bounded Perturbations on KMS states}
\label{BoundedPert}

This section is devoted to the development of the theory of perturbation of KMS states. This theory was initiated by H. Araki in \cite{Araki73} and \cite{Araki74} for bounded perturbations and developed for a class of unbounded perturbations by J. Derezi{\'n}ski, C.-A. Pillet, and V. Jak{\v{s}}i{\'c} in \cite{Derezinski03}. Since our aim in this thesis is to extend this theory, what will be done in Chapter \ref{chapExtensionPerturb}, it is indispensable to present it here with some details. Moreover, several techniques and results we developed in Chapter \ref{chapExtensionPerturb} are based on the one developed by Araki.

Quoting the definition given by Fujiwara in \cite{Fujiwara52}, ``[...] the ordered exponential operators [...] will be called briefly ``expansional'' operators [...]''.

\begin{definition}
	\label{DysonSeries}\index{expansional}
	Let $\nalgebra$ be a von Neumann algebra, $t\mapsto A(t)\in\nalgebra$ a strong-continuous function such that $\displaystyle \sup_{0\leq t\leq T}\|A(t)\|=r_A(T)<\infty$ for all $T\in\mathbb{R}_+$. For each $t\in\mathbb{R}_+$ define
	$$\begin{aligned}
	Exp_r\left(\int_0^t;A(s) ds\right)	&=\sum_{n=0}^{\infty}{\int_{0}^{t} dt_1\ldots \int_{0}^{t_{n-1}}dt_{n}A(t_n)\ldots A(t_1)}; \\
	Exp_l\left(\int_0^t;A(s) ds\right)	&=\sum_{n=0}^{\infty}{\int_{0}^{t} dt_1\ldots \int_{0}^{t_{n-1}}dt_{n} A(t_1)\ldots A(t_n)}; \\
	\end{aligned}$$
where the term for $n=0$ is the identity.
\end{definition}

Note that these operators are well defined since $\|A(t_i)\|\leq r_A(t)$ for every $1\leq i\leq n$ and ${\int_{0}^{t} dt_1\ldots \int_{0}^{t_{n-1}}dt_{n}=\frac{t^n}{n!}}$. Thus the series converge absolutely.

It is important to mention that these operators are basically the Dyson series. In fact, if given $(t_1,\ldots, t_n) \in \mathbb{R}^n$, we set a permutation $\sigma:\{1,\ldots,n\}\to\{1,\ldots,n\}$ such that $t_{\sigma(n)}\leq t_{\sigma(n-1)}\leq \ldots t_{\sigma(1)}$, and we define the operators $T,\widetilde{T}:\nalgebra\to\nalgebra$ by
$$\begin{aligned}
T\left(A(t_1)\ldots A(t_n)\right)&=A(t_{\sigma(1)})\ldots A(t_{\sigma(n)});\\
\widetilde{T}\left(A(t_1)\ldots A(t_n)\right)&=A(t_{\sigma(n)})\ldots A(t_{\sigma(1)});
\end{aligned}$$
then, 
	$$\begin{aligned}
Exp_r\left(\int_0^t;A(s) ds\right)	&=\sum_{n=0}^{\infty}{\int_{0}^{t} dt_1\ldots \int_{0}^{t}dt_{n}\frac{T\left(A(t_n)\ldots A(t_1)\right)}{n!}}; \\
Exp_l\left(\int_0^t;A(s) ds\right)	&=\sum_{n=0}^{\infty}{\int_{0}^{t} dt_1\ldots \int_{0}^{t}dt_{n} \frac{\widetilde{T}\left(A(t_n)\ldots A(t_1)\right)}{n!}}. \\
\end{aligned}$$

Let us examine some properties of these operators that can be found in \cite{Araki73.2}.

\begin{proposition}
	\label{exponentials}
	Let $\nalgebra$ be a von Neumann algebra, $t\mapsto A(t)\in\nalgebra$ a strong-continuous function such that $\displaystyle \sup_{0\leq t\leq T}\|A(t)\|=r_A(T)<\infty$ for all $T\in\mathbb{R}_+$. Then, the following properties hold
	\begin{enumerate}[(i)]
		\item $\displaystyle \frac{d}{dt}{Exp_r\left(\int_0^t;A(s) ds\right)}=Exp_r\left(\int_0^t;A(s) ds\right)A(t)$;\\
		$\displaystyle \frac{d}{dt}{Exp_l\left(\int_0^t;A(s) ds\right)}=A(t)Exp_l\left(\int_0^t;A(s) ds\right)$;
		\item $\displaystyle Exp_l\left(\int_0^t;-A(s) ds\right) Exp_r\left(\int_0^t;A(s) ds\right)=\mathbbm{1}$;
		\item $\displaystyle Exp_r\left(\int_0^t;A(s) ds\right) Exp_l\left(\int_0^t;-A(s) ds\right)=\mathbbm{1}$;
		\item $\displaystyle Exp_r\left(\int_0^t;A(s) ds\right) Exp_r\left(\int_0^{t^\prime};A(s+t) ds\right)=Exp_r\left(\int_0^{t+t^\prime};A(s) ds\right)$;
		
		$\displaystyle Exp_l\left(\int_0^{t^\prime};A(s+t) ds\right) Exp_l\left(\int_0^t;A(s) ds\right)=Exp_l\left(\int_0^{t+t^\prime};A(s) ds\right)$.
	\end{enumerate}
\end{proposition} 
\begin{proof}
	$(i)$ This follows directly from the definition.
	
	$(ii)$ Let us examine the product of partial sums. It is sufficient to calculate the product with the same number of terms. Because absolute convergence, we are able to change the sum's order, thus
	
	\begin{equation}
	\label{eq:calculationDyson}
	\begin{aligned}
		&\sum_{n=0}^{N}{(-1)^n\int_{0}^{t} dt_1\ldots \int_{0}^{t_{n-1}}dt_{n}A(t_1)\ldots A(t_n)} \sum_{k=0}^{N}{\int_{0}^{t} dt_1^\prime\ldots \int_{0}^{t_{k-1}^\prime}dt_{k}^\prime A(t_k^\prime)\ldots A(t_1^\prime)} \\
		& =
		\sum_{n=0}^{N}\sum_{k=0}^{N}{{(-1)^n\int_{0}^{t} dt_1\ldots \int_{0}^{t_{n-1}}dt_{n}A(t_1)\ldots A(t_n)}}\times\\
		&\hspace{5.65cm}\times{\int_{0}^{t} dt_{n+1}^\prime\ldots \int_{0}^{t_{n+k-1}^\prime}dt_{n+k}^\prime A(t_{n+k}^\prime)\ldots A(t_{k+1}^\prime)} \\ 
		&= \mathbbm{1}+
		\sum_{i=1}^{2N}\sum_{j=0}^{i}{{(-1)^{i-j}\int_{0}^{t} dt_1\ldots \int_{0}^{t_{i-j-1}}dt_{i-j}A(t_{1})\ldots A(t_{i-j})}}\times \\
		&\hspace{6.4cm}\times{\int_{0}^{t} dt_{i-j+1}\ldots \int_{0}^{t_{i-1}}dt_{i}A(t_{i})\ldots A(t_{i-j+1})} \\ 
		&= \mathbbm{1}+
		\sum_{i=1}^{2N}\sum_{j=0}^{i}{{(-1)^{i-j}\int_{0}^{t} dt_1\ldots \int_{0}^{t_{i-j-1}}dt_{i-j}A(t_{1})\ldots A(t_{i-j})}}\times \\
		& \hspace{4.1cm}\times{\left(\int_{0}^{t_{i-j}}+\int_{t_{i-j}}^{t}\right) dt_{i-j+1}\ldots \int_{0}^{t_{i-1}}dt_{i}A(t_{i})\ldots A(t_{i-j+1})} \\
		& =\mathbbm{1}+
		\sum_{i=1}^{2N}\sum_{j=0}^{i}{(-1)^{i-j}\int_{0}^{t} dt_1\ldots \int_{0}^{t_{i-1}}dt_{i}\left[A(t_{1})\ldots A(t_{i-j})\right]\left[ A(t_{i})\ldots A(t_{i-j+1})\right]} \\
		&\hspace{0.75cm}+ \sum_{i=1}^{2N}\sum_{j=0}^{i}{(-1)^{i-j}\int_{0}^{t} dt_1\ldots \int_{0}^{t_{i-1}}dt_{i}A(t_{1})\ldots\left[A(t_{i-j+1}) A(t_{i})\right]\ldots A(t_{i-j})} \\
		& =\mathbbm{1}+
		\sum_{i=1}^{2N}\sum_{j=0}^{i}{(-1)^{i-j}\int_{0}^{t} dt_1\ldots \int_{0}^{t_{i-1}}dt_{i}\left[A(t_{1})\ldots A(t_{i-j})\right] \left[A(t_{i})\ldots A(t_{i-j+1})\right]} \\
		&\hspace{0.75cm}+\sum_{i=1}^{2N}{(-1)^{i}\int_{0}^{t} dt_1\ldots \int_{0}^{t_{i-1}}dt_{i}A(t_{1})\ldots A(t_{i})} \\
		&\hspace{0.75cm}- \sum_{i=1}^{2N}\sum_{j=0}^{i+1}{(-1)^{i-j}\int_{0}^{t} dt_1\ldots \int_{0}^{t_{i-1}}dt_{i}\left[A(t_{1})\ldots A(t_{i-j})\right] \left[A(t_{i})\ldots A(t_{i-j+1})\right]} \\
		&=\mathbbm{1} + \sum_{i=1}^{2N}{(-1)^{i}\int_{0}^{t} dt_1\ldots \int_{0}^{t_{i-1}}dt_{i}A(t_{1})\ldots A(t_{i})}\\
		&\phantom{=\mathbbm{1}}-\sum_{i=1}^{2N}{(-1)^{i+1}\int_{0}^{t} dt_1\ldots \int_{0}^{t_{i-1}}dt_{i}A(t_{1})\ldots A(t_{i})} \\
		&= \mathbbm{1}.\\
		\end{aligned}\end{equation}
	Note that we have indicated with brackets special terms which are just $\mathbbm{1}$ either if $j=0$ or $j=i$, respectively.
	
	$(iii)$ This follow from the same argument of $(ii)$ or using (i).
	
	$(iv)$ We prefer to analysis the series in the last items, but the differential equation approach is very often preferable. For this item, notice that, thanks to $(i)$, $Exp_r\left(\int_0^{t};A(s) ds\right)$ is the unique solutions of \mbox{$f^\prime(x)=f(x)A(x)$} with the initial condition $f(0)=\mathbbm{1}$ at $x=t$. Applying the result for $A=A(t)$, we have that $Exp_r\left(\int_0^{t^\prime};A(s+t) ds\right)$ is the unique solution of $f^\prime(x)=f(x)A(x+t)$ with $f(0)=\mathbbm{1}$ for $x=t^\prime$. Thus, the right-hand side is the unique solution of $f^\prime(x)=f(x)A(x+t)$ with $f(0)=Exp_r\left(\int_0^{t};A(s) ds\right)$ at $x=t^\prime$. Hence, it is the unique solution of $f^\prime(x)=f(x)A(x)$ at $x=t+t^\prime$, which must be $Exp_r\left(\int_0^{t+t^\prime};A(s) ds\right)$.
	
	The other equality follows using an equivalent argument.
	
\end{proof}

We define the operation below just to prove the following important identities:

\begin{definition}
	Let $\nalgebra$ be a von Neumann algebra, $t\mapsto A(t)\in\nalgebra$ and $t\mapsto B(t)\in\nalgebra$ strong-continuous functions such that $$ \sup_{0\leq t\leq T}\|A(t)\|=r_A(T)<\infty \quad \textrm{ and } \quad \sup_{0\leq t\leq T}\|B(t)\|=r_B(T)<\infty, \qquad T\in\mathbb{R}_+.$$
	Define, for each $t\in\mathbb{R}$,
	$$(B\star A) (t)=Exp_r\left(\int_0^t;B(s) ds\right)A(t)Exp_l\left(\int_0^t;-B(s) ds\right).$$
\end{definition}

\begin{proposition}
		Let $\nalgebra$ be a von Neumann algebra, $t\mapsto A(t)\in\nalgebra$ and $t\mapsto B(t)\in\nalgebra$ strong-continuous functions such that $$ \sup_{0\leq t\leq T}\|A(t)\|=r_A(T)<\infty \quad \textrm{ and } \quad \sup_{0\leq t\leq T}\|B(t)\|=r_B(T)<\infty, \qquad T\in\mathbb{R}_+.$$
		Then,
	$$Exp_r\left(\int_0^t;(B\star A)(s) ds\right)Exp_r\left(\int_0^t;B(s) ds\right)=Exp_r\left(\int_0^t;A(s)+B(s) ds\right)$$
\end{proposition}
\begin{proof}
	Define $$f(t)= Exp_r\left(\int_0^t;(B\star A)(s) ds\right)Exp_r\left(\int_0^t;B(s) ds\right).$$ Notice that, by Proposition \ref{exponentials},
	$$\begin{aligned}
	f^\prime(t)&=Exp_r\left(\int_0^t;(B\star A)(s) ds\right)(B\star A)(t)Exp_r\left(\int_0^t;B(s) ds\right)+f(t)B(t)\\
	&=Exp_r\left(\int_0^t;(B\star A)(t) ds\right)Exp_r\left(\int_0^t;B(s) ds\right)A(t)+f(t)B(t)\\
	&=f(t)(A(t)+B(t)).
\end{aligned}$$
Hence, the thesis follows by the unicity of the solution of the differential equation above with the initial condition $f(0)=\mathbbm{1}$.

\end{proof}

Notice that, in particular, the case $B(t)=B$ gives a simple and expected result, namely,
$Exp_r\left(\int_0^t;B(s) ds\right)=Exp_l\left(\int_0^t;B(s) ds\right)=e^{tB}$. Using this equality in previous proposition, we obtain
$$Exp_r\left(\int_0^t;e^{tB}A(s)e^{-tB} ds\right)e^{tB}=Exp_r\left(\int_0^t;A(s)+B ds\right)$$
and setting $A(t)=A$ we have
\begin{equation}
\label{Duhamel}
Exp_r\left(\int_0^t;e^{tB}Ae^{-tB} ds\right)e^{tB}=e^{t(A+B)}.
\end{equation}

A central case in perturbations is the case $A(t)=\tau_t^\Psi(A)$ where $\left\{\sigma^\Psi_t\right\}_{t\in\mathbb{R}}$ is the modular automorphism group. Now, the self-dual cone start playing a special part in the theory. As we have already said, and will repeat it here because of importance of this statement, we have that $J_\Phi=J_\Psi$ when $\Phi\in V_\Psi$.

\begin{lemma}
	\label{lemma5.2}
	Let $\nalgebra\subset B(\hilbert)$ be a von Neumann algebra, $\Psi\in\hilbert$ a cyclic and separating vector and $\left\{\sigma_t^\Psi\right\}_{t\in\mathbb{R}}$ the modular automorphism group. Let $Q\in \calgebra$ be such that $Q\Psi \in V_\Psi$, then $\sigma_t^\Psi(Q)$ has an analytic continuation $\sigma_z^{\Psi}(Q)\in \calgebra$, which is continuous on the closure, for $$z\in \overline{\strip{\gamma}}=\{z \in \mathbb{C} | \ - \gamma \leq \Im z \leq 0\}$$ and $\|\sigma^\Psi_{-\iu \gamma}(Q)-1\|\leq L$ where $\gamma \in \left(0,\frac{1}{8}\right)$. Let $F(t)=2\pi sech(2\pi t)^2$ and
	$$Q_F=\int_{\mathbb{R}}{F(t) \sigma_t^\Psi(Q) dt},$$
	then $\sigma_t^\Psi(Q_F)$ has an analytic continuation $\sigma_z^{\Psi}(Q_F)\in \calgebra$ for \\ $z\in \{w \in \mathbb{C} \ | \ - \gamma-\frac{1}{4} \leq \Im w \leq \gamma+\frac{3}{4}\}$ and such that
	$$h_1\defeq \sigma_{\frac{\iu}{4}}^\Psi(Q_F)-2=\sigma_{\frac{\iu}{4}}^\Psi((Q-\mathbbm{1})_F)\in \calgebra,$$
	$h_1^\ast=h_1$ and $\|h_1\|\leq 2L$.
	Furthermore,
	$$\begin{aligned}
	\Psi_1& \defeq\Psi(h_1)\defeq Exp_r\left(\int_{0}^{\frac{1}{2}};\sigma_{-it}^{\Psi}(h_1)dt\right)\Psi \\
	Q_1 & \defeq Q Exp_l\left(\int_{0}^{\frac{1}{2}};-\sigma_{-it}^{\Psi}(h_1)dt\right) \\
	\end{aligned}$$
	satisfies $Q_1\Psi_1=Q\Psi$, $\sigma^{\Psi1}_z(Q_1) \in \calgebra$ for $z\in\overline{\strip{\gamma_1}}$, $\gamma_1\in(0, \gamma)$, and
	$$\left\|\sigma^{\Psi1}_{-\gamma_1 \iu}(Q_1) -1\right\|\leq \left(L^2+(1+L)L^\prime\right)e^{\frac{L}{2}}$$
	where $L^\prime=\frac{1}{2} \left(\pi L \log\left(2(\gamma-\gamma_1)\right)\right)^2e^{-\pi L\log\left(2(\gamma-\gamma_1)\right)}$.
\end{lemma}
\begin{proof}
	Since $Q\Psi \in V_{\Psi}$, Theorem \ref{theorem3in2} $(vii)$ implies that $\sigma^{\Psi}_t(Q)$ admits an analytic continuation,  $\sigma^{\Psi}_z(Q)$ for $z \in \strip{\frac{1}{2}}$ which is continuous on the closure, such that $\left(\sigma^{\Psi}_t(Q)\right)^\ast=\sigma^{\Psi}_{t+\frac{\iu}{2}}(Q)$.
	
	Now, by hypothesis $\sigma^{\Psi}_z(Q)$ also has an analytic continuation, continuous on the boundary, for $z\in\overline{\strip{\gamma}}$. It follow from the edge-of-the-wedge theorem that $\sigma^{\Psi}_z(Q)$ is analytic for $-\gamma < \Im z < \frac{1}{2}+\gamma$ and continuous for $-\gamma \leq \Im z \leq \frac{1}{2}+\gamma$ and such that $\left(\sigma^{\Psi}_z(Q)\right)^\ast=\sigma^{\Psi}_{\bar{z}+\frac{\iu}{2}}(Q)$. Furthermore, using Maximum Modulus Principle, we have, for $\Im z\in[-\gamma,\frac{1}{2}+\gamma]$,
	\begin{equation}
	\label{ineq1}
	\begin{aligned}
	\|\sigma^{\Psi}_z(Q-\mathbbm{1}) \|	&\leq\max\left\{\left\|\sigma^{\Psi}_{-\iu \gamma}(Q-\mathbbm{1})\right\|,\left\|\sigma^{\Psi}_{\frac{\iu}{2}+\iu \gamma}(Q-\mathbbm{1})\right\|\right\}\\
	&=\max\left\{\left\|\sigma^{\Psi}_{-\iu \gamma}(Q-\mathbbm{1})\right\|,\left\|\left(\sigma^{\Psi}_{-\iu \gamma}(Q-\mathbbm{1})\right)^\ast\right\|\right\}\\
	&\leq L. \\
	\end{aligned}
	\end{equation}
	
	Now, $Q_F\Psi$ is in $\Dom{\Delta_\Psi^{\iu z}}$ for $\Im z\in\left[-\frac{1}{4},\frac{1}{4}\right]$. In fact, the Fourier transform of $F(t)$ is
	$$\hat{F}(k)=\frac{1}{\sqrt{2\pi}}\int_{\mathbb{R}}{2\pi sech(2\pi t)^2 e^{\iu kt} dt}=\begin{cases}\frac{k}{\sqrt{2\pi}\left(e^\frac{k}{4}-e^{-\frac{k}{4}}\right)} & \textrm{ if } k\neq 0\\\frac{2}{\sqrt{2\pi}}& \textrm{ if } k=0\end{cases},$$
	thus $e^{t k}\hat{F}(k)$ is bounded for $|\Re{t}|<\frac{1}{4}$.
	Hence, $\hat{F}(\log(\Delta_\Psi))Q\Psi\in{\Dom{\Delta_\Psi^{\iu z}}}$ for $|\Im z|< \frac{1}{4}$. Moreover, by the Fourier transform definition
	\begin{equation}
	\label{xxx}
	\hat{F}(\log(\Delta_\Psi))Q\Psi=\frac{1}{\sqrt{2\pi}}\int_{\mathbb{R}}F(t)\Delta_\Psi^{\iu t}Q\Psi dt=\frac{1}{\sqrt{2\pi}}\int_{\mathbb{R}}F(t)\sigma^\Psi_t(Q)\Psi dt=\frac{1}{\sqrt{2\pi}}Q_F \Psi.
	\end{equation}	
	Then we have an analytic continuation, continuous on the boundary, for $\Im{z_1} \in [-\gamma, \frac{1}{2}+\gamma]$ and $-\frac{1}{4}\leq \Im{z_2} \leq \frac{1}{4}$. Consequently, for  $z=z_1+z_2, \Im{z} \in [-\gamma-\frac{1}{4},\frac{3}{4}+\gamma]$,
	$$\sigma^{\Psi}_z(Q_F)=\int_{\mathbb{R}}{F(t-z_1) \sigma_t^\Psi\left(\sigma_{z_2}^\Psi(Q)\right) dt}.$$
	
	It follows now from $Q\Psi \in V_\psi$ and Theorem \ref{theorem3in2} $(vii)$ that $\sigma^\Psi_{\frac{\iu}{4}}(Q_F) \geq 0$ and so $h_1=h_1^*$.
	
	By definition, 
	$$\mathbbm{1}_F=\int_{\mathbb{R}}{F(t)\sigma^\Psi_t(\mathbbm{1})}=\int_{\mathbb{R}}{F(t)} \mathbbm{1}=2 \mathbbm{1},$$ it follows that
	$\sigma_{\frac{\iu}{4}}^\Psi(Q_F)-2=\sigma_{\frac{\iu}{4}}^\Psi((Q-\mathbbm{1})_F)$. Using again the Maximum Modulus Principle in the strip $\strip{\frac{1}{2}}$ and the inequality $\eqref{ineq1}$,
	$$\|h_1\|=\|\sigma^\Psi_{\frac{\iu}{4}}(Q_F)-2\| \leq \left\|\sigma_{\frac{\iu}{4}}^\Psi((Q-1)_F)\right\|=\left\|\int_{\mathbb{R}}F(t)\sigma^\Psi_{\frac{i}{4}+t}(Q-\mathbbm{1})dt\right\|\leq 2 L.$$
	
	It is a consequence of Proposition \ref{exponentials} that $Q_1\Psi_1=Q\Psi$.
	
	Let $$Q^\prime_1=Exp_l\left(\int_0^{\frac{1}{2}};-\sigma^\Psi_{-\iu s}(h_1)ds\right)-1 +\int_0^{\frac{1}{2}}{\sigma^\Psi_{-\iu s}(h_1)ds},$$
	
	Noticing that
	$$\int_{-\frac{1}{4}}^{\frac{1}{4}}{\hat{F}(k) e^{kt}dt}=\frac{1}{\sqrt{2\pi}}$$
	
	we can use equation \eqref{xxx} to get
	$$\left(\int_0^{\frac{1}{2}}{\sigma^\Psi_{-\iu s}(h_1)ds}\right)\Psi=\left(\int_{-\frac{1}{4}}^{\frac{1}{4}}{\sigma^\Psi_{-\iu s}((Q-\mathbbm{1})_F)ds}\right)\Psi=\int_{-\frac{1}{4}}^{\frac{1}{4}}{\Delta_\Psi^{s}(Q-\mathbbm{1})_F\Psi ds}=(Q-\mathbbm{1})\Psi.$$
	
	Because $\Psi$ is separating, we conclude that
	%	$$\int_0^{\frac{1}{2}}{\sigma^\Psi_{-\iu s}(h_1)ds}=Q-1$$
	$Q_1=\mathbbm{1}-(Q-\mathbbm{1})^2+Q Q_1^\prime$ and then
	$$\|\sigma^\Psi_z(Q_1)-\mathbbm{1}\|\leq \|\sigma^\Psi_z(Q)-\mathbbm{1}\|+\left(\mathbbm{1}+\|\sigma^\Psi_z(Q)-\mathbbm{1}\|\right)\|\sigma^\Psi_z(Q^\prime_\mathbbm{1})\|,$$
	and
	$$\begin{aligned}
	\left\|\sigma^\Psi_z(Q^\prime_1)\right\| 
	&\leq \sum_{n=2}^{\infty}{\frac{1}{n!}\int_0^{\frac{1}{2}}\left\|\sigma^\Psi_{z-\iu s}(h_1)\right\|^n}\\
	&=e^{\|\sigma^\Psi_{z-\iu s}(h_1)\|}-1-\left\|\sigma^\Psi_{z-\iu s}(h_1)\right\|\\
	&\leq \left\|\sigma^\Psi_{z-\iu s}(h_1)\right\|^2e^{\frac{\left\|\sigma^\Psi_{z-\iu s}(h_1)\right\|}{2}}.
	\end{aligned}$$
	
	Finally, we can choose $a(s)$ such that $\Im z- s+ a(s)+\frac{1}{4}\in[-\gamma,\frac{1}{2}]$ and $|a(s)|<\frac{1}{4}$,
	$$\begin{aligned}
	\left\|\sigma^\Psi_{z-\iu s}(h_1)\right\| 
	&\leq \int_{\mathbb{R}}{|F(t+\iu a(s))|\left\|\sigma^{\Psi}_{z-\iu s+\iu a(s)+\frac{\iu}{4}}(Q-\mathbbm{1})dt\right\|} \\
	&\leq L \int_{\mathbb{R}}{|F(t+\iu a(s))|dt}\\
	&=8\sqrt{2\pi}L \frac{a(s)}{\sin(4\pi a(s))}.
	\end{aligned}$$
	
	Suppose now $\Im z \in [-\gamma_1,0] $. We can set
	$$a(s)=\frac{s-\frac{1}{4}}{1-2(\gamma-\gamma_1)}.$$
	
	As a consequence of this choice, $|a(s)|<\frac{1}{4}$ and $\left|\frac{a(s)(\frac{1}{4}-|a(s)|)}{\sin(4\pi a(s))}\right|\leq 2^{-5}$. Hence,
	$$\left\|\sigma^\Psi_{z-\iu s}(h_1)\right\|\leq \sqrt{2\pi}L|\log2(\gamma-\gamma_1) |.$$
\end{proof}

\begin{lemma}
	\label{lemma5.3}
	Let $\nalgebra\subset B(\hilbert)$ be a von Neumann algebra, $\Psi\in\hilbert$ a cyclic and separating vector and $\left\{\sigma^\Omega_t \right\}_{t\in\mathbb{R}}$ the modular automorphism group with respect to $\Psi$. Let $Q\in \nalgebra$ be such that $Q\Psi\in V_{\Psi}$, $\gamma \in [0,\frac{1}{8}]$, $\sigma_z^\Psi Q\in \nalgebra$ for $z\in \overline{\strip{-\gamma}}$ and $\|\sigma^\Psi_{-\iu \gamma}Q-\mathbbm{1}\|\leq L_0\leq (4\pi \log\gamma)^{-2}$. Then, there exists $h\in\calgebra$, $h=h^\ast$, such that
	$$Q\Psi=\Psi(h)\defeq Exp_r\left(\int_0^{\frac{1}{2}}{;\sigma^\Psi_{\iu s}(h)ds}\right)\Psi.$$
\end{lemma}
\begin{proof}
	Define $\gamma_n=2^{-n}\gamma$. By \ref{lemma5} there exists the extension required in Lemma \ref{lemma5.2} and, by Lemma \ref{lemma5.2}, there exists vectors $\Psi_n$, $Q_n \in \nalgebra$, $h_n\in\nalgebra$ with $h_n=h_n^\ast$ and $L_n>0$ such that $\Psi_n=\Psi_{n-1}(h_n)$ with $\Psi_0=\Psi$ and $\|h_n\|\leq 2 L_{n-1}$. Furthermore $\sigma^{\Psi_n}_t(Q_n)$ can be extended analytically for $\Im z \in (-\gamma_n,0)$ and continuously for $\Im{z} \in [-\gamma_n,0]$ satisfying 
	$\|\sigma^{\Psi_n}_{-\iu\gamma_n}(Q_n)-\mathbbm{1}\|\leq L_n$, $Q_n\Psi_n=Q\Psi$,
	$$\begin{aligned}
	L_n&=(L^2_{n-1}+(1+L_n-1)L^\prime_{n-1})e^{\frac{L_{n-1}}{2}} \quad \textrm{ and } \\ L^\prime_{n-1}&=\frac{1}{2}\left(\pi L_n-1\log2(\gamma-\gamma_n)\right)^2 e^{-\pi L_{n-1}\log2(\gamma-\gamma_n)}.
	\end{aligned}$$
	One can prove that $L_n\leq 2^{-n}L_0$ and since $\|h_n\|\leq 2L_{n-1}$, $\displaystyle\sum_{n=0}^{\infty}{h_n}\in \nalgebra$ and $\displaystyle\lim_{n\rightarrow \infty}{\Psi_n}=\Psi(h)$ where $\Psi_n=\Psi(h1+\ldots+h_n)$. %fazer os resultados necessários para concluir isso (4.1 e 4.5 relative hamiltonians)

	By \, Lemma \ref{lemma5.2}, \, we \, have \, an \, analytic \, continuation \, $\sigma^{\Psi_n}_z(Q_n)$ \, for \mbox{$\Im{z} \in [-\gamma_n-\frac{1}{4},\frac{3}{4}+\gamma_n]$} and, by Lemma \ref{lemma6}, $\|Q-\mathbbm{1}\|\leq\|\sigma^{\Psi_n}_{-\iu \gamma_n}(Q_n)-\mathbbm{1}\|\leq L_n \rightarrow 0$.
	Hence $Q\Psi=\lim_{n\rightarrow \infty}{ Q_n\Psi_n}=\Psi(h)$.
	
\end{proof}

Now the theory starts taking form. The previous result shows that the vectors in the self-dual cone under some other assumptions are achievable by a perturbation. The the next proposition, which are the main result on Araki's perturbation theory, has basically the same interpretation.

\begin{proposition}
	Let $\nalgebra\subset B(\hilbert)$ be a von Neumann algebra, $\Psi\in\hilbert$ a cyclic and separating vector and $\left\{\sigma^\Omega_t \right\}_{t\in\mathbb{R}}$ the modular automorphism group with respect to $\Psi$. Let $Q\in \nalgebra$ be such that $J_\Psi Q\Psi=Q\Psi$ and $\sigma_t^\Psi Q$ has an analytic continuation for $z\in \overline{\strip{-\frac{1}{2}}}$. Then there exists $h\in\nalgebra$, $h=h^\ast$, such that $$e^Q \Psi=\Psi(h).$$
\end{proposition}
\begin{proof}
	Since $Q^\ast \Psi=J_\Psi\Delta^{\frac{1}{2}}_\Psi Q \Psi=\Delta^{-\frac{1}{2}}_\Psi J_\Psi Q \Psi=\Delta^{-\frac{1}{2}}_\Psi Q \Psi$, $Q\Psi \in \Dom{\Delta_\Psi^{-\frac{1}{2}}}$ and, by Lemma \ref{lemma6}, $\sigma_t^\Psi (Q)$ admits an analytic continuation to $z\in \strip{\frac{1}{2}}$, furthermore it satisfies $\sigma_{\frac{\iu}{2}}^\Psi (Q)=Q^\ast$ thus $\sigma_{\bar{z}}^\Psi (Q)^\ast=\sigma_{z+\frac{\iu}{2}}^\Psi(Q)$, in particular $\sigma_{\frac{\iu}{4}}^\Psi (Q)^\ast=\sigma_{\frac{\iu}{4}}^\Psi (Q)$.
	
	Let $\Phi=e^{tQ}\Psi$, by Theorem \ref{theorem3in2} $(vii)$, $\sigma_{\frac{\iu}{4}}^\Psi (e^{tQ})=e^{t\sigma_{\frac{\iu}{4}}^\Psi (Q)}\geq0$ implies $\Phi_t\in V_\Psi$, $t\in \mathbb{R}$.
	
	It follows from Lemma \ref{lemma7} that
	$$\Delta_{\Phi_t}^{\frac{1}{2}}Q\Phi_t=e^{tQ}\sigma_{-\frac{\iu}{2}}^\Psi (Q)e^{-tQ}\Phi_t.$$
	
	Again, by Lemma \ref{lemma6}, $\sigma_{z}^{\Phi_t} (Q)$ has an analytic continuation for $z\in\overline{\strip{\frac{1}{2}}}$ such that
	$$\left\|\sigma_{z}^{\Phi_t} (Q)\right\|\leq a=\max\left\{\|Q\|, e^{2t\|Q\|}\left\|\sigma_{-\frac{\iu}{2}}^{\Psi_t} (Q)\right\|\right\}.$$
	
	Now, choose $N\in \mathbb{N}$ such that $e^{\frac{a}{N}}-1\leq (4\log\gamma)^{-2}$ for a fixed $\gamma \in [0,\frac{1}{8}]$. Then,
	$$\left\|\sigma_{-\iu \gamma}^{\Phi_t} (e^{\frac{Q}{N}})\right\|=\left\|e^{\frac{\sigma_{-\iu \gamma}^{\Phi_t}(Q)}{N}}-1\right\|\leq e^{ \frac{\left\|\sigma_{-\iu \gamma}^{\Phi_t}(Q)\right\|}{N}} -1\leq (4\log\gamma)^{-2}.$$
	
	Now Lemma \ref{lemma5.3} guarantees the existence of $h_n\in \nalgebra$ for $n\leq N$ such that $\Phi_{\frac{n}{N}}=\Phi_{\frac{n-1}{N}}(h_n)$ with $h_n=h_n^\ast$. Then
	$$e^Q\Psi=\Phi_{\frac{N}{N}}=\Psi\left(\sum_{n=1}^{N}{h_n}\right).$$
\end{proof}

An interesting question one might ask now is ``now we know what kind of vector we can achieve by a perturbation, but is this set big (in some sense) in the Hilbert space?''. The answer to this question is presented below.

%Remark in Relative Hamiltonians for Faithful Normal States, H. Araki, pag. 199
\begin{corollary}
	\label{PerturbationsareDense}
	Let $\nalgebra\subset B(\hilbert)$ be a von Neumann algebra, $\Psi\in\hilbert$ cyclic and separating and $\left\{\sigma^\Omega_t \right\}_{t\in\mathbb{R}}$ the modular automorphism group with respect to $\Psi$. The set 
	$$\left\{e^Q \Psi \ \middle| \ \begin{minipage}{10.5cm} $Q\in \nalgebra$, $J_\Psi Q\Psi=Q\Psi$ and \\ $\sigma_z(Q)$ has an analytic continuation for $ -\frac{1}{2}\leq\Im{z}\leq 0$\end{minipage}\right\}$$ is dense in $V_\Psi$.
\end{corollary}
\begin{proof}
	Note first that the vectors of the form $\Delta^{\frac{1}{4}}_\Psi A \Psi$ with $A\geq0$ are dense in $V_\Psi$, by definition. Furthermore, by equation \eqref{eq3.13}, for any self-adjoint operator $B\in \nalgebra$ such that $A-B\in\Dom{\Delta_\Psi^ \frac{1}{4}}$, we have
	$$\left\|\Delta_\Psi^{\frac{1}{4}}(A-B)\Psi\right\|^2 \leq 2 \|(A-B)\Psi\|^2.$$
	
	Hence, it is enough to prove that, for any $A\in\nalgebra_+$, there is a sequence of vectors in the set defined above that converges to $\Delta^{\frac{1}{4}}_\Psi A \Psi$.
	
	Let $A\in \nalgebra_+$ and let $A=\int{\lambda dE_\lambda}$ be its spectral decomposition. We set
	$$\begin{aligned}
	A_L&=A\left(E_L-E_{\frac{1}{L}}\right)+\frac{1}{L}E_{\frac{1}{L}}+(1-E_L)\\
	A_{L,\beta}&=\sqrt{\frac{\pi}{n}}\int_{\mathbb{R}}{e^{-\frac{t^2}{\beta}\sigma^\Psi_{t}}(\log A_L)dt}.
	\end{aligned}$$
	We have already proved in Proposition \ref{AnalDense} that
	$$\lim_{\beta\rightarrow0}{e^{A_{L,\beta}}\Psi}=A_L\Psi.$$
	Hence
	$$\lim_{L\rightarrow\infty}{\lim_{\beta\rightarrow0}{e^{A_{L,\beta}}\Psi}}=A\Psi.$$
	Therefore $Q_{L,\beta}=\sigma^{\Psi}_{-\frac{\iu}{4}}(A_{L,\beta})$ satisfies
	$$\lim_{L\rightarrow\infty}{\lim_{\beta\rightarrow0}{e^{Q_{L,\beta}}\Psi}}=\Delta^{\frac{1}{4}}_\Psi A\Psi.$$
	In addition $J_\Psi Q_{L,\beta}\Psi = Q_{L,\beta} \Psi$ since $\sigma^{\Psi}_{-\frac{\iu}{4}}(Q_{L,\beta})=A_{L,\beta}=A_{L,\beta}^\ast$.
	
\end{proof}

The results presented in this section lead to the important result by Araki that we are goint to state below.

\begin{theorem}[Araki's Perturbation Theorem] \index{theorem! Araki's Perturbation}
Let $(\nalgebra,\tau)$ be a $W^\ast$-dynamical system, $\omega$ a $(\tau,\beta)$-KMS state, $\Omega_\omega$ its vector representation throughout the \mbox{GNS-representation}, $H_\omega\in\pi_\omega(\nalgebra)$ the hamiltonian of $\tau$ and $Q=Q^\ast\in \pi_\omega(\nalgebra)$ a perturbation. The perturbed dynamics is defined by
$$\tau_t^Q(A)=\pi_\omega^{-1}\left(e^{\iu t(H_\omega+Q)}\pi_\omega(A)e^{-\iu t(H_\omega+Q)}\right).$$
Then, $\Omega_\omega\in \Dom{e^{-\frac{\beta}{2}(H_\omega+Q)}}$ and $\omega^Q(A)=\ip{\Psi^Q}{\pi_\omega(A)\Psi^Q}$ is as $(\tau^Q,\beta)$-KMS state where $\Psi^Q=\frac{e^{-\frac{\beta}{2}(H_\omega+Q)}\Omega_\omega}{\|e^{-\frac{\beta}{2}(H_\omega+Q)}\Omega_\omega\|}$.
\end{theorem}

Just to clarify, the hamiltonian mentioned in the previous theorem is the generator of the unitary time-evolution operator of the $W^\ast$-dynamics on the represented algebra. The proof of the last theorem can be found in \cite{Araki73}, but all the ingredients of the proof are already here.

\section{Radon-Nikodym Derivative}

Takesaki says in \cite{Takesaki2003} that ``the theories of weights, traces and states are often referred as non commutative integration. If the von Neumann algebra in question is abelian, then our theory is precisely the theory of measures and integration''. Hence, it is inevitable to ask about the existence of analogous for useful results such as the Radon-Nikodym Theorem.

The first version of a noncommutative analogous for the Radon-Nikodym Theorem appears in \cite{Dye52} and \cite{Segal53} and the generalization we are about to present is from \cite{Pedersen73}.

It is worth to mention some interesting interpretations here. The Radon-Nikodym derivative can give us an operator in the algebra that connects two different weights, what is exactly what we are looking for in perturbation theory. From the physical point of view, the Radon-Nikodym derivative can be used to obtain a relative hamiltonian or a relative entropy (see \cite{Araki73} and \cite{Araki76}). On the other hand, the Connes cocycle and the modular automorphism group have strong connections with the Radon-Nikodym derivative, when it exists.

% \textcolor{red}{###*****************I'm not sure if it is the Pedersen-Takesaki version, because I just wrote this proof and I use the article just to see how to avoid the inequality hypothesis. In addition, Pedersen and Takesaki don't say nothing about weights}

Before we present Radon-Nikodym-type theorems, we need Kadison's characterization of the extremal points of the unity ball in a $C^\ast$-algebra (see \cite{Kadison51}) and the polar decomposition of linear functionals.

\begin{theorem}
	Let $\calgebra$ be a $C^\ast$-algebra and $B_1=\{A \in calgebra \ | \ \|A\|\leq1\}$ its unitary ball. Then
	$\displaystyle \mathcal{E}(B_1)=\big\{U \in calgebra \ \big| \ U \textrm{ is a partial isometry and } \left(\mathbbm{1}-U U^\ast\right)\calgebra\left(\mathbbm{1}-U^\ast U\right)=\{0\} \big\}$\footnote{The definition of the set $\mathcal{E}(B_1)$ and the Krein-Milman Theorem can be found in Appendix \ref{AppKMT}.}.
\end{theorem}
\begin{proof}
	First of all, notice that $B_1$ is a non-empty WOT-compact set. Then, by the Krein-Milman Theorem, $\mathcal{E}(B_1)\neq \emptyset$.
	
	Let $U \in \mathcal{E}(B_1)$, then $U^\ast U$ is a self-adjoint operator and $\|U^\ast U\|\leq 1$, this implies that $\sigma(U^\ast U)\subset[0,1]$. Suppose $\lambda \in \sigma(U^\ast U)\setminus\{0,1\}$, then there exists $\epsilon>0$ such that \mbox{$(t-\epsilon,t+\epsilon)\subset[0,1]$} and we can find a positive infinitely differentiable function ${f: [0,1]\to [0,1]}$ such that $f(x)=0$ if $x \notin (t-\epsilon,t+\epsilon)$ and $0<f(\lambda)<\min\{\sqrt{\lambda},\sqrt{1-\lambda}\}$.
	
	As a result of this choice $f(U^\ast U)$ commutes with $U^\ast U$ and 
	$$\begin{aligned}
	\sigma\left(U^\ast U(\mathbbm{1}\pm f(U^\ast U))^2 \right)\subset [0,1]&\Rightarrow \|(\mathbbm{1}\pm f(U^\ast U))U^\ast U(\mathbbm{1}\pm f(U^\ast U))\|\leq 1\\
	&\Rightarrow \|U(\mathbbm{1}\pm f(U^\ast U))\|\leq 1\\
	&\Rightarrow U(\mathbbm{1}\pm f(U^\ast U))\in B_1,\\
	\end{aligned}$$
	but $U=\frac{1}{2}U(\mathbbm{1}+ f(U^\ast U))+\frac{1}{2}U(\mathbbm{1} -f(U^\ast U))$, which contradicts the extremality of $U$. Hence $\sigma(U^\ast U)\subset\{0,1\}$ and it follows that $U^\ast U$ is a projection, which means that $U$ is a partial isometry.
	
	Also, denote $P=U^\ast U$, $Q=UU^\ast$ and let $A\in \left(\mathbbm{1}-Q\right)\calgebra\left(\mathbbm{1}-P\right)$, $\|A\|\leq 1$, and $x \in \hilbert$,
	
	$$\begin{aligned}
	\|(U\pm A)x\|^2&=\|U(Px)\pm A\left((\mathbbm{1}-P)x\right)\|^2\\
	&=\|QU(Px)\pm (\mathbbm{1}-Q)A\left((\mathbbm{1}-P)x\right)\|^2\\
	&=\|V(Px)\|^2+\|A\left((\mathbbm{1}-P)x\right)\|^2\\
	&\leq \|x\|.
	\end{aligned}$$
	Hence $U\pm A \in B_1$, and by the extremality of $U$, $$U=\frac{1}{2}(U+A)+\frac{1}{2}(U-A)\Rightarrow A=0.$$
	
	On the other hand, suppose $U$ is a partial isometry such that $\left(\mathbbm{1}-Q\right)\calgebra\left(\mathbbm{1}-P\right)=\{0\}$. If $U=\frac{1}{2}(A+B)$ for some $A,B\in B_1$, then, for every $x \in \Ran(P)$ with $\|x\|=1$ we have $1=\ip{Ux}{Ux}=\frac{1}{2}\ip{Ax}{Ux}+\frac{1}{2}\ip{Bx}{Ux}$, but since both $\ip{Ax}{Ux}$ and $\ip{Bx}{Ux}$ are elements of the closed unit disk in $\mathbb{C}$ and $1$ is extremal in the disc, we have
	$$\ip{Ax}{Ux}=\ip{Bx}{Ux}=\ip{Ux}{Ux}=1\Rightarrow Ax=Bx=Ux, \forall x \in Ran(P).$$	
	
	We already have $AP=BP=U$. Lets now show that $QA(1-P)=QB(1-P)=0$. Suppose it is not the case, then we can take $z\in Ran(QA(\mathbbm{1}-P))\setminus\{0\}$ with $|z|$=1, which means there exists $x\in \hilbert$ such that $QA(\mathbbm{1}-P)x=z$. At the same time, since $z\in Ran(Q)$ and $Q$ is the final projection of $U$, there exists $y\in \hilbert$ such that $z=Uy=APy$. Notice that $1=\|z\|=\|QA(\mathbbm{1}-P)x\|\leq \|(\mathbbm{1}-P)x\|$ and $1=\|z\|=\|APy|=\|UPy\|=\|Py\|$.
	
	Take $\theta\in\left[0,\frac{\pi}{2}\right]$ such that $\cot(\theta)=\|(\mathbbm{1}-P)x\|$ and $w=\cos(\theta)Py+\frac{\sin(\theta)}{\|(\mathbbm{1}-P)x\|}(\mathbbm{1}-P)x$. Then,
	$$\|QAw\|=\left|\cos(\theta)+\tan(\theta)\sin(\theta)\right|\|z\|=\left(1+\frac{1}{\|(\mathbbm{1}-P)x\|^2}\right)^\frac{1}{2}>1, $$
	but this contradicts $\|QA\|\leq 1$. Then, $QA(\mathbbm{1}-P)=0$, thus, by hypothesis, we must have ${(\mathbbm{1}-Q)A(\mathbbm{1}-P)=0}$. Hence, $A(\mathbbm{1}-P)=0$ and, by the same argument, $B(\mathbbm{1}-P)=0$. It follows that $A=B=U$, thus $U$ is extremal.
	
\end{proof}

\begin{theorem}[Polar Decomposition of Functionals] \index{theorem! polar decomposition of functionals}
	Let $\nalgebra\subset B(\hilbert)$ be a von Neumann algebra and $\phi$ a WOT-continuous bounded functional on $\nalgebra$. Then, there exists a positive normal bounded functional $\psi$ on $\nalgebra$ and a partial isometry $U\in \nalgebra$, extreme in the unit ball, such that $\phi(A)=\psi(AU)$ and $\psi(A)=\phi(A U^\ast)$. 
\end{theorem}

\begin{proof}
	The case $\phi$=0 is trivial. Suppose $\phi\neq 0$ and let 
	$$\mathcal{F}=\{A\in\nalgebra \ | \ \|A\|\leq 1 \textrm{ and } \phi(A)=\|\phi\|\},$$
	since the unit ball of $\nalgebra$, $B_1$, is WOT-compact, there exists $\tilde{V}\in B_1$ such that $|\phi(\tilde{V})|=\|\phi\|$, hence $\phi(V)=\|\phi\|$ for $V=e^{-\iu Arg(\phi(\tilde{V}))}\tilde{V}$. Hence $\mathcal{F}\neq \emptyset$.
	
	Now, $\mathcal{F}$ is a compact face in $B_1$, thus $\mathcal{E}(\mathcal{F})\subset \mathcal{E}(B_1)$. Let $W \in \mathcal{E}(\mathcal{F})$, then, $W$ is a partial isometry satisfying $(\mathbbm{1}-W W^\ast)\nalgebra (\mathbbm{1}-W^\ast W)=\{0\}$.
	
	Define $\psi$ by $\psi(A)=\phi(AW)$ for every $A\in \nalgebra$. Then, $$|\psi(A)|=\phi(AW)\leq \|\phi\|\|A\|\|W\|=\|\phi\|\|A\|\Rightarrow \|\psi\|\leq \|\phi\|,$$
	thus $\psi(\mathbbm{1})=\phi(W)=\|\phi\|=\|\psi\|$, it follows by Proposition \ref{PPC} that $\psi$ is positive.
	
	Let $s^\nalgebra(\psi)$ be the support of $\psi$, then $s^\nalgebra(\psi)\leq W W^\ast$, since $$\psi(W W^\ast)=\phi(WW^\ast W)=\phi(W)=\psi(\mathbbm{1}) \Rightarrow \mathbbm{1}-WW^\ast\in N_{\psi}$$ and $N_{\psi}$ is a left ideal.
	
	Let $U=W^\ast s^\nalgebra(\psi)$, we have $U^\ast U=s^\nalgebra(\psi)W W^\ast s^\nalgebra(\psi)=s^\nalgebra(\psi)$, then
	$$\phi(AU^\ast)=\phi\left(A s^\nalgebra(\psi) W\right)=\psi\left(As^\nalgebra(\psi)\right)=\psi(A).$$
	
	Suppose now that there exists $A\in\nalgebra$, $\|A\|=1$, such that 
	$\phi\left(A(\mathbbm{1}-U U^\ast )\right)>0 $. Then, for every $t \in \left[0,\frac{\pi}{2}\right]$ such that $\cot\left(\frac{t}{2}\right)>\frac{\|\phi\|}{\phi\left(A(\mathbbm{1}-U U^\ast )\right)}$, we have
	$$\begin{aligned}
	\phi\big(\cos(t) U^\ast +\sin(t) A(\mathbbm{1}-U U^\ast )\big)
	&=\cos(t)\phi(U^\ast)+\sin(t)\phi\big(A(\mathbbm{1}-U U^\ast )\big)\\
	&=\cos(t)\phi(W)+\sin(t)\phi\left(A(\mathbbm{1}-U U^\ast )\right)\\ 
	&=\cos(t)\|\phi\|+\sin(t)\phi\left(A(\mathbbm{1}-U U^\ast )\right)\\ &>\|\phi\|.\end{aligned}$$
	
	If follows that $\phi\left(A(\mathbbm{1}-U U^\ast )\right)=0\Rightarrow \phi(A)=\phi(AUU^\ast)=\psi(AU)$.
	
\end{proof}

\begin{notation}
	The $\psi$ obtained in the previous theorem is called the modulus of $\phi$ and usually denoted by $|\phi|$. The polar decomposition of $\phi$ can be written as $\phi=\hat{U}|\phi|$, where $U$ is a partial isometry.
\end{notation}

It may be worth to mention that the modulus of a normal function is unique and, if we require $U^\ast U =s^\nalgebra(\phi)$, so is $U$.

\begin{proposition}
	\label{selfadjointfunctional}
	Let $\nalgebra$ be a von Neumann algebra and let $\phi$ be a positive linear functional on $\nalgebra$. For $H\in\nalgebra$, define $\hat{H}$ by $\left(\hat{H}\right)\phi(A)=\phi(AH)$. Then, if $\hat{H}\phi$ is self-adjoint\footnote{See Definition \ref{defAdjoinfFunc}.}, $$\left|\left(\hat{H}\phi(A)\right)\right|=|\phi(AH)|\leq \|H\||\phi(A)|, \ \forall A\in\nalgebra.$$
\end{proposition}
\begin{proof}
	By self-adjointness,
	$$\phi(AH)=\hat{H}\phi(A)=(\hat{H}\phi)^\ast(A)=\overline{\hat{H}\phi(A^\ast)}= \overline{\phi(A^\ast H)}=\phi(H^\ast A).$$
	
	Hence $\phi\left(A H^{2^{n+1}}\right)=\phi\left(\left(H^{2^n}\right)^\ast AH^{2^n}\right)$. Then, for every $A\geq 0$,
	$$\begin{aligned}
	|\phi(AH)|
	&=\left|\phi(A^\frac{1}{2}A^\frac{1}{2}H)\right|\\
	&\leq \left|\phi(A)\right|^\frac{1}{2}\left|\phi(H^\ast A H)\right|^\frac{1}{2}\\
	&= \left|\phi(A)\right|^\frac{1}{2}\left|\phi(A H^2)\right|^\frac{1}{2}\\
	&\leq \left|\phi(A)\right|^{\sum_{i=1}^{n} 2^{-i}}\left|\phi(A H^{2^n})\right|^\frac{1}{2^n}\\
	&=\left|\phi(A)\right|^{1-2^{-n}}\left|\phi(A H^{2^n})\right|^\frac{1}{2^n}\\
	&\leq\left|\phi(A)\right|^{1-2^{-n}}\left(\|\phi\|\|A\|\|H\|^{2^n}\right)^\frac{1}{2^n}\\
	&\xrightarrow{n\to\infty} \|H\|\left|\phi(A)\right|.\\
	\end{aligned}$$
	
\end{proof}

Next, we will present what can be seen as a prototype version of the desired theorem. The problem with the next is basically that the Radon-Nikodym derivative obtained does not lie in the von Neumann algebra, but in its commutant.

\begin{proposition}
	\label{commutantRN}
	Let $\nalgebra$ be a von Neumann algebra, $\phi, \psi$ be normal semi-finite weights on $\nalgebra$ such that $\psi\leq \phi$. Then, there exists an operator $H^\prime \in \pi_\phi(\nalgebra)^\prime$, $0\leq H^\prime\leq \mathbbm{1}$, such that
	$$\psi(A)=\ip{H^\prime \pi_\phi(A^\ast) \Phi}{\Phi}_\phi, \quad \forall A\in\nalgebra.$$
\end{proposition}
\begin{proof}
	We will use the GNS-representation for the weight $\phi$. Notice that, by hypothesis, $\mathfrak{N}_\phi\subset \mathfrak{N}_\psi$ and $N_\phi\subset N_\psi$.
	
	Define the sesquilinear form $\tilde{\psi}: \mathfrak{N}_\phi/N_\phi\times \mathfrak{N}_\phi/N_\phi \to \mathbb{C}$ given by $\tilde{\psi}([A],[B])=\psi(B A^\ast)$ which is well defined by the same calculation presented in equation \eqref{eq:calculationquotient}.
	
	By Cauchy-Schwarz's inequality and the hypothesis, $$|\tilde{\psi}([A],[B])|^2=|\psi(BA^\ast )|^2\leq \psi(A^\ast A) \psi(B^\ast B)\leq \phi(A^\ast A) \phi(B^\ast B)=\|[A]\|_\phi^2 \|[B]\|_\phi^2.$$
	
	Hence $\tilde{\psi}$ admits a unique extension to a sesquilinear form on $\hilbert_\phi$, also denoted \mbox{by $\tilde{\psi}$}. By the Riesz Theorem, there exists a unique operator $H^\prime\in B(\hilbert_\phi)$ such that
	\begin{equation}
	\label{eq:riesz1}
	\tilde{\psi}(x,y)=\ip{H^\prime x}{y}_{\phi} \quad \forall x,y \in \hilbert_\phi.
	\end{equation}
	It also follows by the same theorem that $\|H^\prime \|\leq 1$ and, by positiveness of $\tilde{\psi}$, it follows that $H^\prime \geq 0$.
	
	In addition, for every $A,B\in\mathfrak{N}_\phi$ and $C\in\nalgebra$ we have
	$$\begin{aligned}
	\ip{(H^\prime \pi_\phi(C)-\pi_\phi(C)H^\prime )[A]}{[B]}_\phi&=\ip{H^\prime \pi_\phi(C)[A]}{[B]}_\phi-\ip{H^\prime [A]}{\pi_\phi(C)^\ast [B]}_\phi\\
	&=\ip{H^\prime [CA]}{[B]}_\phi-\ip{H^\prime [A]}{[C^\ast B]}_\phi\\
	&=\tilde{\psi}\left((CA)^\ast B\right)-\tilde{\psi}\left(A^\ast C^\ast B\right)\\
	&=0.
	\end{aligned}$$
	Thus, $(H^\prime \pi_\phi(C)-\pi_\phi(C)H^\prime )[A]=0 \quad \forall A\in \mathfrak{N}_\phi \Rightarrow H^\prime \pi_\phi(C)=\pi_\phi(C)H^\prime  \Rightarrow H^\prime \in \nalgebra^\prime$.
	
	Finally, equation \eqref{eq:riesz1}, in the special case $(E_\alpha)_\alpha \subset \mathfrak{N}_\phi$, is an approximation identity and $A\in \nalgebra_+$, can be rewritten as
	
	$$\psi(E_\alpha A E_\alpha)=\tilde{\psi}(\pi_\phi(A^\ast) [E_\alpha],[E_\alpha])=\ip{H^\prime \pi_\phi(A^\ast)[E_\alpha]}{[E_\alpha]}_{\phi}$$
	and the thesis follows by normality and the polarization identity.
	
\end{proof}
\begin{remark}
	It is common in the literature to express the previous proposition as
	$$\exists H^\prime \in \nalgebra^\prime \ \mid \ \psi(A)=\phi(H^\prime A), \quad \forall A\in \nalgebra_+.$$
	It is important to stress that this is an abuse of notation, because there is no reason for $H^\prime A \in \nalgebra$, since $H^\prime \in \nalgebra^\prime$. Nevertheless, the expression makes sense if it is seen as an extension of $\phi$ to $B(\hilbert)$
\end{remark}

The next theorem is a noncommutative version of the Radon-Nikodym Theorem. Notice that it is in the hypothesis that one weight must dominate the other one, which is exactly the absolutely continuous hypothesis in the commutative case. 

\begin{theorem}[Sakai-Radon-Nikodym] \index{theorem! Sakai-Radon-Nikodym}
	\label{TSRN}
	Let $\nalgebra$ be a von Neumann algebra, $\phi, \psi$ be normal functionals on $\nalgebra$ such that $\psi\leq \phi$. Then, there exists an operator $H\in \nalgebra$, $0\leq H\leq \mathbbm{1}$, such that
	$$\psi(A)=\phi(HAH), \quad \forall A\in \nalgebra_+.$$
\end{theorem}
\begin{proof}
	
	By Theorem \ref{commutantRN}, there exists $H^\prime\in\pi_\phi(\nalgebra)^\prime$, $0\leq H^\prime\leq \mathbbm{1}$, such that
	$$\psi(A)=\ip{\pi_{\phi}(A) H^\prime\Phi}{H^\prime\Phi}_\phi, \ \forall A\in\nalgebra_+.$$
	Consider the WOT-continuous functional $\phi^\prime: \pi_\phi(\nalgebra)^\prime \to \mathbb{C}$ given by $\phi^\prime(A^\prime)=\ip{A^\prime H^\prime \Phi}{\Phi}_\phi$.
	
	Let $|\phi^\prime|=\widehat{U^{\prime\ast}}\phi^\prime$ be the polar decomposition of $\phi^\prime$, then $$|\phi^\prime|=\widehat{U^{\prime\ast} H^\prime}\omega,$$  where ${\omega(A^\prime)=\ip{A^\prime\Phi}{\Phi}_\phi}$ and, by Proposition \ref{selfadjointfunctional}, $$|\phi^\prime|(A)=\widehat{U^{\prime\ast} H^\prime}\omega(A)\leq \|U^\ast H^\prime\| \omega(A)\leq \|U^{\prime\ast}\|\|H^\prime\|\omega(A)\leq \omega(A).$$
	
	Using now the Proposition \ref{commutantRN} for $|\phi^\prime|$ and $\omega$, there exists $H\in \nalgebra^{\prime\prime}=\nalgebra$, $0\leq H\leq \mathbbm{1}$ such that
	$$|\phi^\prime|(A^\prime)=\ip{A^\prime H\Phi}{\Phi}_\phi.$$
	
	Now, we have all the elements we will need to conclude the proof. It just remains to do some calculations.
	
	Notice that, for every $A^\prime \in \pi_\phi(\nalgebra)^\prime, $ $$\begin{aligned}
	\ip{H^\prime \Phi}{A^{\prime}\Phi}_\phi&=\ip{A^{\prime\ast} H^\prime \Phi}{\Phi}_\phi\\
	&=\phi^\prime(A^{\prime\ast})\\
	&=\widehat{U^\prime}|\phi^\prime|(A^{\prime\ast})\\
	&=\widehat{U^\prime}\ip{A^{\prime\ast} H \Phi}{\Phi}_\phi\\
	&=\ip{A^{\prime\ast} H U^\prime\Phi}{\Phi}_\phi\\
	&=\ip{U^\prime H \Phi}{A^{\prime}\Phi}_\phi \ ,\\
	\end{aligned}$$	from which, since $\pi_\phi(\nalgebra)^\prime\Phi$ is dense in $\hilbert$, we conclude $H^\prime\Phi=U^\prime H\Phi$.
	
	Moreover, for every $A^\prime \in\pi_\phi(\nalgebra)^\prime$,
	$$\begin{aligned}
	\ip{ H\Phi}{A^{\prime}\Phi}_\phi&=\ip{A^{\prime\ast} H\Phi}{\Phi}_\phi\\
	&=|\phi^\prime|(A^{\prime\ast})\\
	&=\widehat{U^{\prime\ast}}\phi^\prime(A^{\prime\ast})\\
	&=\ip{A^{\prime\ast} H^\prime U^{\prime\ast}\Phi}{\Phi}_\phi\\
	&=\ip{ U^{\prime\ast}H^\prime\Phi}{A^{\prime}\Phi}_\phi\\
	\end{aligned}$$
	and, again by the density argument, we conclude that $H\Phi=U^{\prime\ast}H^\prime\Phi$.
	
	Finally, 
	$$\begin{aligned}
	\psi(A)&=\ip{A H^\prime \Phi}{H^\prime\Phi}_\phi\\
	&=\ip{A H^\prime \Phi}{U^\prime H\Phi}_\phi\\
	&=\ip{HA U^{\prime\ast}H^\prime \Phi}{\Phi}_\phi\\
	&=\ip{HA H \Phi}{\Phi}_\phi\\
	&=\phi(HAH) \quad \forall A\in\nalgebra_+.\\
	\end{aligned}$$

\end{proof}
\iffalse
\begin{remark}
	The proof can be easily generalized to $C^\ast$-algebras just avoiding the use of projection by using continuous functions with support in a small neighbourhood of the desired point of the spectrum.
	
	The proof above was written by the author because it is very direct and, although it is some way new, the methods used are not surprising at all once previous proofs of the already know Radon-Nikodyn theorem (see \cite{Sakai65,KR86}) uses the polar decomposition of functionals that, in its turn, uses the Krein-Milman theorem in a very similar way.
\end{remark}
\fi

Using the GNS-representation, we can generalise Theorem \ref{TSRN} for weights. Although it is just a straightforward result, the following proof is due to the author.

\begin{theorem}[Sakai-Radon-Nikodym for Weights]\index{theorem! Sakai-Radon-Nikodym}
	\label{TSRNweights}
	Let $\nalgebra$ be a von Neumann algebra, $\phi, \psi$ be normal semi-finite weights on $\nalgebra$ such that $\psi\leq \phi$. Then, there exists an operator $H\in \nalgebra$, $0\leq H\leq \mathbbm{1}$, such that
	$$\psi(A)=\phi(HAH) \quad \forall A\in \nalgebra_+.$$
\end{theorem}
\begin{proof}
	As in Proposition \ref{commutantRN}, there exists a unique sesquilinear form $\tilde{\psi}:\hilbert_\phi \times \hilbert_\phi\to \mathbb{C}$ such that $\tilde{\psi}([A],[B])=\psi(A^\ast B)$ for all $A,B\in\mathfrak{N}_\phi$.
	
	Then, by Theorem \ref{TSRN}, there exists a unique $H\in\pi_\phi\left(\nalgebra\right)_+$, $0\leq H\leq \mathbbm{1}$, such that $\phi(A^\ast A)=\tilde{\psi}([A],[A])=\ip{A^\ast A H\Phi}{H\Phi}=\phi(\pi^{-1}_\phi(H)A^\ast A \pi^{-1}_\phi(H)), \ A\in\mathfrak{N}_\phi$, and the statement of the theorem follows by semi-finiteness and normality.
	
\end{proof}

Now, we will start studying the modular condition because we would like to have not only a noncommutative version of the Radon-Nikodym Theorem, but also an unbounded noncommutative version and it was noticed by Pedersen and Takesaki that the modular automorphism group is the key to do this. 

\begin{definition}
	\label{ModularCondition}
	\index{modular! condition}
	Let $\calgebra$ be a $C^\ast$-algebra and let $\{\sigma_t\}_{t\in\mathbb{R}}$ be a one-parameter group of automorphisms of $\calgebra$. A lower semi-continuous weight $\phi$ is said to satisfy the modular condition for $\{\sigma_t\}_{t\in\mathbb{R}}$ if
	\begin{enumerate}[(i)]
		\item $\phi=\phi\circ\sigma_t$ for every $t\in\mathbb{R}$;
		\item for every $A,B\in \mathfrak{N}_\phi\cap \mathfrak{N}_\phi^\ast$, there exists a complex function $F_{A,B}$ which is analytic on the strip $\strip{1}=\left\{z\in \mathbb{C} \mid 0< \Im z<1\right\}$ and continuous and bounded on its closure satisfying
		\begin{equation}
		\label{eq:ModCond}
		\begin{aligned}
		F_{A,B}(t) &= \phi(A\sigma_t(B)) \quad \forall t\in \mathbb{R}, \\
		F_{A,B}(t+\iu) &= \phi(\sigma_t(B)A) \quad \forall t\in \mathbb{R}. \\
		\end{aligned}
		\end{equation}
	\end{enumerate}
\end{definition}

It is evident that equation \eqref{eq:ModCond} and equation \eqref{eqKMS2} are basically the same for $\beta=-1$. Together which Lemma \ref{tauInv} it becomes evident that the modular condition is the KMS condition for $\beta=-1$ or, equivalently, that $\phi$ is a $\left(\{\sigma_t\}_{t\in\mathbb{R}},\beta\right)$-KMS state for $\sigma_t=\tau^\phi_{-\beta t}$. Of course we will suppose known all results we have already proved for KMS states.

The inconvenience in the sign of $\beta$ is just a consequence of the difference between mathematicians and physicists convention on the sign of the modular automorphism group (for mathematicians) and time evolution operator (for physicists). 

\begin{proposition}
	\label{twosidedmodule}
	Let $\phi$ be a normal semifinite weight on a von Neumann \mbox{algebra $\nalgebra$}, $\{\tau^\phi_t\}_{t\in \mathbb{R}}$ its modular automorphism group and $\nanalytic$ the set of all entire $\{\tau^\phi_t\}_{t\in \mathbb{R}}$-analytic elements of $\nalgebra$. Then
	\begin{enumerate}[(i)]
		\item $\mathfrak{N}_\phi\cap \mathfrak{N}_\phi^\ast$ is a two-sided module over $\nanalytic$; 
		\item $\mathfrak{M}_\phi$ is a two-sided module over $\nanalytic$.
	\end{enumerate}
\end{proposition}
\begin{proof}
	In this proof we will refer to the items of Proposition \ref{simplepropweights}.
	
	$(i)$ The proof that $\mathfrak{N}_\phi\cap \mathfrak{N}_\phi^\ast$ is a $\ast$-subalgebra is basically already in \mbox{Proposition \ref{simplepropweights}}. So, it remains to show the module property.
	
	Let $N\in \mathfrak{N}_\phi\cap \mathfrak{N}_\phi^\ast$ and $A\in\nanalytic$, of course $AN\in\mathfrak{N}_\phi$, since $\mathfrak{N}_\phi$ is a left ideal as stated in item $(iii)$. It remains to show that $AN\in\mathfrak{N}_\phi^\ast$.
	
	It follows from the modular condition that there exists a complex function $F_{ANN^\ast,A^\ast}$ which is analytic on the strip $\strip{1}=\left\{z\in \mathbb{C} \mid 0< \Im z<1\right\}$ and continuous and bounded on its closure satisfying 
	$$\begin{alignedat}{2}
	F_{ANN^\ast,A^\ast}(t) &= \phi\left(ANN^\ast\tau^\phi_t(A^\ast)\right) \quad \forall t\in \mathbb{R}, \\
	F_{ANN^\ast,A^\ast}(t+\iu)&= \phi\left(\tau^\phi_t(A^\ast)ANN^\ast\right) \quad \forall t\in \mathbb{R}. \\
	\end{alignedat}$$
	Since $\tau^\phi_t(A^\ast)=\tau^\phi_t(A)^\ast$, $A^\ast\in\nanalytic$ and
	$$\begin{aligned}
	\phi\left(ANN^\ast\tau^\phi_t(A^\ast)\right)&=F_{ANN^\ast,A^\ast}(t)\\
	&=\phi\left(\tau^\phi_{-\iu}(A^\ast)ANN^\ast\right)\\
	&\leq \phi\left(N^\ast A^\ast\tau^\phi_{-\iu}(A^\ast)^\ast\tau^\phi_{-\iu}(A^\ast)A N\right)^\frac{1}{2}\phi(N N^\ast)^\frac{1}{2}\\
	&\leq\|A\|\left\|\tau^\phi_{-\iu}(A^\ast)\right\|\phi(N^\ast N)^\frac{1}{2}\phi(N N^\ast)^\frac{1}{2}\\
	&<\infty.
	\end{aligned}$$
	
	$NA\in \mathfrak{N}_\phi\cap \mathfrak{N}_\phi^\ast$ follows by the same argument, just noticing that $NA=(A^\ast N^\ast)^\ast$.
	
	$(ii)$ Again, it is already done in item $(v)$ that $\mathfrak{M}_\phi$ is a $\ast$-subalgebra.
	
	By definition $\mathfrak{M}_\phi=span\left[\mathfrak{F}_\phi\right]$ and by item $(i)$ $\mathfrak{F}_\phi\subset \mathfrak{N}_\phi\cap \mathfrak{N}_\phi^\ast$. The conclusion is now obvious.
	
\end{proof}

We will skip the proof of the next theorem and refer to \autocite[Chapter VIII \S 1, Theorem 1.2]{Takesaki2003}.

\begin{theorem}
	\label{UniqueModCond}
	Let $\nalgebra$ be a von Neumann algebra and $\phi$ a faithful normal semifinite weight on $\nalgebra$. Then, there exists a unique one parameter group of automorphisms $\{\tau_t\}_{t\in\mathbb{R}}$ satisfying the modular condition.
\end{theorem}
\iffalse \begin{proof}
	\textcolor{red}{******************DO IT!!!!!!!!!!! It is a shame you haven't done this.}
\end{proof} \fi

\begin{corollary}
	\label{modularautomosphism}
	Let $\nalgebra_1, \nalgebra_2$ be von Neumann algebras, $\phi$ be a normal semifinite weight on $\nalgebra_1$, and $\pi:\nalgebra_1\to\nalgebra_2$ an isomorphism. Then, the modular automorphism group of $\phi\circ\pi$ is $\{\pi^{-1}\circ\tau_t^\phi\circ\pi\}_{t\in\mathbb{R}}$.
\end{corollary}
\begin{proof}
	Let us prove that $\phi\circ\pi$ satisfies the two conditions of Definition \ref{ModularCondition} for the one-parameter group $\{\pi^{-1}\circ\tau_t^\phi\circ\pi\}_{t\in\mathbb{R}}$.
	
	$(i)$ $(\phi\circ\pi)\circ(\pi^{-1}\circ\tau_t\circ\pi)=\phi\circ\tau_t\circ\pi=\phi\circ\pi$, where we have used that $\phi\circ\tau_t=\phi$.
	
	$(ii)$ For every $A,B\in\mathfrak{N}_{\phi}\cap \mathfrak{N}_{\phi}^\ast$ there exists a complex function $F_{A,B}$ which is analytic on the strip $\strip{1}=\left\{z\in \mathbb{C} \mid 0< \Im z<1\right\}$ and continuous and bounded on its closure, satisfying
	$$\begin{aligned}
	F_{A,B}(t) &= \phi(A\sigma_t(B)) =&&\hspace{-0.2cm} \phi\circ\pi\left(\pi^{-1}(A)\pi^{-1}\circ\sigma_t\circ\pi\left(\pi^{-1}(B)\right)\right) \\
	F_{A,B}(t+\iu) &= \phi(\sigma_t(A)B) =&&\hspace{-0.2cm} \phi\circ\pi\left(\pi^{-1}\circ\sigma_t\circ\pi(B)\pi^{-1}(A)\right) \\
	\end{aligned} \qquad \forall t\in \mathbb{R}.$$
	
	Notice that $\mathfrak{N}_{\phi\circ\pi}\cap \mathfrak{N}_{\phi\circ\pi}^\ast=\pi^{-1}\left(\mathfrak{N}_\phi\right)\cap\pi^{-1}\left(\mathfrak{N}_\phi^\ast\right)$. The thesis follows choosing $F_{\pi^{-1}(A),\pi^{-1}(B)}=F_{A,B}$.
	
\end{proof}

\begin{definition}[Centralizer of a Weight] \index{weight! centralizer}
	Let $\nalgebra$ be a von Neumann algebra, $\phi$ a faithful normal semifinite weight on $\nalgebra$, and $\tau^\phi=\{\tau^\phi_t\}_{t\in\mathbb{R}}$ the modular automorphism group of $\phi$. We define \index{weight! centralizer} the centralizer of $\phi$ as the set
	$$\mathfrak{M}_{\tau^\phi}=\left\{A\in \nalgebra \ | \ \tau^\phi_t(A)=A \ \forall t\in\mathbb{R}\right\}.$$
\end{definition}

Notice that it follows from linearity and normality that $\mathfrak{M}_{\tau^\phi}$ is a von Neumann subalgebra of $\nalgebra$.

\begin{theorem}
	\label{centralizercommute}
	Let $\nalgebra$ be a von Neumann algebra, $\phi$ a faithful normal semifinite weight. Then, $A\in\mathfrak{M}_{\tau^\phi}$ if and only if
	\begin{enumerate}[(i)]
		\item $A\mathfrak{M}_\phi \subset \mathfrak{M}_\phi$ and $\mathfrak{M}_\phi A\subset \mathfrak{M}_\phi$;
		\item $\phi(AB)=\phi(BA)$ \ $\forall B\in\mathfrak{M}_\phi$.
	\end{enumerate}
\end{theorem}
\begin{proof}
	Let's denote $\tau^\phi=\{\tau^\phi_t\}_{t\in\mathbb{R}}$ the modular automorphism group of $\phi$.
	
	$(\Rightarrow)$ Since $A\in\mathfrak{M}_{\tau^\phi}$, $A$ is an entire analytic element. Then, condition $(i)$ follows from $(ii)$ in Proposition \ref{twosidedmodule}.
	
	Since $B\in\mathfrak{M}_\phi$, $B= C^\ast D$ with $C,D \in \mathfrak{N}_\phi$. By Proposition \ref{twosidedmodule} and the modular condition, there exists an analytic functions on the strip $\strip{1}$, continuous and bounded on its closure, such that
	$$\begin{alignedat}{4}
	F_{C^\ast,DA}(t) &=\phi\left(C^\ast\tau^\phi_t(D A)\right) &=\phi\left(C^\ast\tau^\phi_{t}(D)A\right) & \\
	F_{C^\ast,DA}(t+\iu) &= \phi\left(\tau^\phi_t(DA)C^\ast\right) &=\phi\left(\tau^\phi_{t}(D)AC^\ast\right) &=\phi\left(\tau^\phi_{t}(D)AC^\ast\right) \\
	\end{alignedat} \qquad \forall t\in, \mathbb{R}$$
	and
	$$\begin{alignedat}{3}
	F_{\tau^\phi_{t}(D),A\tau^\phi_{-t}(C^\ast)}(t)
	&=\phi\left(\tau^\phi_{t}(D)\tau^\phi_{t}\left(A\tau^\phi_{-t}(C^\ast)\right)\right)
	&=\phi\left((\tau^\phi_{t}(D)AC^\ast\right) \\
	F_{\tau^\phi_{t}(D),A\tau^\phi_{-t}(C^\ast)}(t+\iu) 
	&=\phi\left(\tau^\phi_t\left(A\tau^\phi_{-t}(C^\ast)\right) \tau^\phi_{t}(D)\right)
	&=\phi\left(AC^\ast \tau^\phi_{t}(D)\right) \\
	\end{alignedat} 	\qquad \forall t\in \mathbb{R}.$$

	Now, by the previous equation and by $A\in\mathfrak{M}_\phi$ we have $F_{C^\ast,DA}(t+\iu)=F_{\tau^\phi_{t}(D),AC^\ast}(t)$ and $F_{C^\ast,DA}(t)=F_{\tau^\phi_{t}(D),AC^\ast}(t+\iu)$, so we can define the bounded function $G:\mathbb{C} \to \mathbb{C}$ by
	$$G(z)=\begin{cases}F_{C^\ast,DA}(z-2n\iu) & \textrm{if } 2n\leq Im(z)\leq 2n+1, \ n\in\mathbb{Z}\\ F_{\tau^\phi_{t}(D),AC^\ast}(z-(2n+1)\iu) & \textrm{if } 2n+1\leq Im(z)\leq 2n+2, n\in\mathbb{Z}.\\\end{cases}$$
	
	It follows from the edge-of-the-wedge theorem that $G$ is entire analytic and it is also bounded. Thus it is constant, by Liouville's theorem. Hence,
	$$\phi(B A)=F_{C^\ast,DA}(0) =G(0)=G(2\iu)= F_{\tau^\phi_{t}(D),AC^\ast}(\iu)=\phi(AB) \quad \forall B\in\mathfrak{M}_{\tau^\phi}.$$
	
	$(\Leftarrow)$ Notice that the assumptions $(i)$ warrants that, for any $B\in\mathfrak{N}_\phi$, we can again define an analytic function on $\strip{1}$ which is continuous and bounded on $\overline{\strip{1}}$, namely
	$F(t)=\phi\left(\tau^\phi_t(A)B\right)=\phi\circ\tau^\phi_t\left(A\tau^\phi_{-t}(B)\right)$.
	
	Form assumption $(ii)$, this function is periodic with period $\iu$. Then, using a definition similar to the one we used for $G$ above, $F$ can be analytically extended to the whole complex plane by the edge-of-the-wedge theorem. Since the extension is bounded,  by Liouville's theorem, $F$ is constant.
	
	Hence, $\phi\left(\tau^\phi_t(A)B\right)=F(t)=F(0)=\phi(AB) \quad \forall B\in\mathfrak{N}_\phi\Rightarrow \tau^\phi_t(A)=A$.
	
\end{proof}

\begin{lemma}
	\label{centralizerorder}
	Let $\nalgebra$ be a von Neumann algebra, $\phi$ a faithful normal semifinite weight on $\nalgebra$,
	$$\begin{aligned}
	\mathfrak{M}_{\tau^\phi}\ni H \mapsto & \ \phi_H&: \ &\nalgebra^+ &\to &\hspace{1.1cm}\overline{\mathbb{R}}\\
	& \ & \ &A&\mapsto & \hspace{0.2cm} \phi\left(H^\frac{1}{2}AH^\frac{1}{2}\right)\\
	\end{aligned}$$
	is an order preserving map on the weights on $\nalgebra$.
\end{lemma}
\begin{proof}
	Normality is evident and semifiniteness follows from Proposition \ref{twosidedmodule}. For the order, notice that $H, K\in \mathfrak{M}_\phi^+$ implies $H^\frac{1}{2}, K^\frac{1}{2}, (H+K)^\frac{1}{2}\in \mathfrak{M}_\phi^+$, thus $$\phi\left(H^\frac{1}{2}A H^\frac{1}{2}\right)=\phi(HA) \textrm{ and } \phi\left(K^\frac{1}{2}A K^\frac{1}{2}\right)=\phi(KA).$$
	Hence, if $K\leq H$, 
	$$\begin{aligned}
	\phi(H^\frac{1}{2}A H^\frac{1}{2}) &=\phi(HA)\\	&=\phi\left(((H-K)+K)A\right)\\
	&=\phi\left(((H-K)A\right)+\phi(KA)\\
	&=\phi\left((H-K)^\frac{1}{2}A(H-K)^\frac{1}{2}\right)+\phi\left(K^\frac{1}{2}AK^\frac{1}{2}\right)\\
	&\geq \phi\left(K^\frac{1}{2}AK^\frac{1}{2}\right).\\
	\end{aligned}$$
\end{proof}

The next lemma ensures that we can extend our notion of Radon-Nikodym derivative for unbounded operators. The proof is partially due to the author. 
%####**************The independence of the choice is the net is one of my proofs.
\begin{lemma}
\label{unboundedderivativeweight}
	Let $\nalgebra$ be a von Neumann algebra, $\phi$ a faithful normal semifinite weight on $\nalgebra$ and $H\eta \mathfrak{M}_{\tau^\phi}^+$.  If $\left(H_i\right)_{i\in I} \in  \mathfrak{M}_{\tau^\phi}^+$ is an increasing net such that $H_i \to H$, then
	\begin{equation}
	\label{eq:definitionx1}
	\phi_H(A)=\lim_{H_i \to H}\phi\left(H_i A\right)=\lim_{H_i \to H}\phi\left(H_i^\frac{1}{2}AH_i^\frac{1}{2}\right)=\sup_{i\in I}{\phi\left(H_i^\frac{1}{2}AH_i^\frac{1}{2}\right)}, \quad A\in\nalgebra
	\end{equation}
	defines a normal semifinite weight $\phi_H$ on $\nalgebra$, which is independent of the choice of the net $(H_i)_{i\in I}$ with $H_i \to H$.	
	
	 In addition, $\phi_H$ is faithful if and only if $H$ is non-singular and, if $(H_i)_{i\in I}$ is an increasing  net of positive operators affiliated with $\mathfrak{M}_{\tau^\phi}^+$ such that $H_i\to H$, then 
	 $$\phi_H=\sup_{i\in I}\phi_{H_i}.$$
	
\end{lemma}
\begin{proof}
	By Lemma \ref{centralizerorder}, $\phi_H(A)$ is well defined since it is the limit of a positive increasing net of real numbers. Furthermore, it is easy to see (by the same Lemma and normality of $\phi$) that it is a normal weight. It remains to prove that it is semifinite.
	
	Let $\left\{E^{H}_\lambda\right\}_{\lambda\in\mathbb{R}_+}$ be the spectral resolution of $H$. By $(ii)$ in Proposition \ref{twosidedmodule}, $\displaystyle \bigcup_{n\in\mathbb{N}}E_{[0,n)}\mathfrak{M}_\phi E_{[0,n)}\in\mathfrak{M}_\phi$ and it is WOT-dense in $\nalgebra$.
	
	To prove the independence of the net, let $H_n=E^H_{[0,n)}H E^H_{[0,n)}$ be fixed and let $(K_j)_{j\in J}\in \mathfrak{M}_{\tau^\phi}$ be another increasing net such that $K_j \to H$. Denote by $\phi_H$ the normal semifinite weight defined by equation \eqref{eq:definitionx1} for the sequence $(H_n)_{n\in\mathbb{N}}$.
	
	We know that $K_{j,n}=E^{H}_{[0,n)}K_jE^{H}_{[0,n)}$ is an increasing net with $\displaystyle \sup_{j\in J}K_{j,n}=H_n$ and $\displaystyle \sup_{n\in \mathbb{N}}K_{j,n}=K_j$. Let's use the GNS-representation throughout $\phi$. Notice that, $\pi_\phi(K_{j,n})\xrightarrow[SOT]{j} \pi_\phi(H_n)$ and $\pi_\phi(K_{j,n})\xrightarrow[SOT]{n} \pi_\phi(K_j)$ thanks to Vigier's theorem\footnote{see Theorem \ref{vigier}}. Then
	$$\begin{aligned}
	&\sup_{j\in J}\phi(K_{n,j}A)=\sup_{j\in J}\ip{\pi_{\phi}(K_{j,n})\Phi}{\pi_{\phi}(A)\Phi}_\phi=\ip{\pi_{\phi}(H_n)\Phi}{\pi_{\phi}(A)\Phi}_\phi=\phi(H_nA),\\
	&\sup_{n\in \mathbb{N}}\phi(K_{n,j}A)=\sup_{n\in \mathbb{N}}\ip{\pi_{\phi}(K_{j,n})\Phi}{\pi_{\phi}(A)\Phi}_\phi=\ip{\pi_{\phi}(K_j)\Phi}{\pi_{\phi}(A)\Phi}_\phi=\phi(K_jA).\\
	\end{aligned}$$
	Hence,	
	$$\begin{aligned}
	\label{eq:calculationX17}
	\phi_H(A)&=\sup_{n\in \mathbb{N}}{\phi\left(H_n^\frac{1}{2}AH_n^\frac{1}{2}\right)}\\
	&=\sup_{n\in \mathbb{N}}\sup_{j\in J}{\phi\left(K_{n,j}^\frac{1}{2}AK_{n,j}^\frac{1}{2}\right)}\\
	&\leq \sup_{j\in J}{\phi\left(K_j^\frac{1}{2}AK_j^\frac{1}{2}\right)}\\
	&=\sup_{j\in J}\sup_{n\in \mathbb{N}}{\phi\left(K_{n,j}^\frac{1}{2}AK_{n,j}^\frac{1}{2}\right)}\\
	&=\phi_H(A).\\
	\end{aligned}$$
	
	For the last statement, let $(H_i)_{i\in I}$ an increasing net of operators affiliated with $\mathfrak{M}_{\tau^\phi}$ such that $H_i\to H$. Then we can define $H_{i,n}=E^{H_i}_{[0,n)} H_i E^{H_i}_{[0,n)}$ and using what we get in equation \eqref{eq:calculationX17} we obtain
	$$\begin{aligned}
	\phi_{H}=\sup_{i\in I}\phi_{H_{i,n}}=\sup_{i\in I}\phi_{H_i}.
	\end{aligned}$$
\end{proof}

\begin{notation}
Henceforth, when $H$ is a positive unbounded operator affiliated with $\mathfrak{M}_{\tau^\phi}^+$ we will consider the weight $\phi_H$ and usually write $\phi(HA)$ instead of $\phi_H(A)$.

\end{notation}

%\textcolor{red}{****************** It seems $\psi$ don't have to be faithful in the proof below and in the next one. Check it!!!}
 
\begin{theorem}[Pedersen-Takesaki-Radon-Nikodym] \index{theorem! Pedersen-Takesaki-Radon-Nikodym}
	\label{TPTRN}
	Let $\phi$ and $\psi$ be two normal semifinite weights on a von Neumann algebra $\nalgebra$. Suppose in addition that $\phi$ is faithful and $\psi$ is invariant under the modular automorphism group of $\phi$, $\{\tau^\phi_t\}_{t\in \mathbb{R}}$, and $\psi\leq \phi$. Then, there exists a unique operator $H\in \mathfrak{M}_{\tau^\phi}$ with $0\leq H\leq \mathbbm{1}$ and such that $\psi(A)=\phi(HA)$ for all $A\in \nalgebra_+$.
\end{theorem}
\begin{proof}
	
	By Proposition \ref{commutantRN}, there exists $H^\prime \in \pi_\phi(\nalgebra)^\prime$, $0\leq H^\prime \leq \mathbbm{1}$, such that $$ \psi(A)=\ip{H^\prime A^\ast \Phi}{\Phi}_\phi \quad \forall A\in\nalgebra_+.$$
	
	Notice that the invariance of $\psi$ implies, for $A\in\nalgebra_\phi$,
	$$\psi\left(\tau^\phi_{t}(A)\right)=\ip{H^\prime\tau^\phi_{t}(A)\Phi}{\Phi}_\phi=\ip{\tau^\phi_{t}(H^\prime)A\Phi}{\Phi}_\phi=\ip{H^\prime A\Phi}{\Phi}_\phi=\psi(A),$$	thus $H^\prime$ is invariant under $\{\tau^\phi_t\}_{t\in \mathbb{R}}$.
	
	By Tomita's theorem, we have that $H=J_\phi H^\prime J_\phi \in \nalgebra$ and it is also $\{\tau^\phi_t\}_{t\in \mathbb{R}}$-invariant. Furthermore, for every $A\in \mathfrak{N}_\phi\cap \mathfrak{N}_\phi^\ast$,
	
	\begin{equation}
	\label{eq:calculationRN}
	\begin{aligned}
	\phi\left(H^\frac{1}{2}\tau^\phi_{z-\frac{\iu}{2}}(A)H^\frac{1}{2}\right)
	&=\ip{\tau^\phi_{z-\frac{\iu}{2}}(A)^\ast J_\phi H^{\prime\frac{1}{2}} J_\phi\Phi}{J_\phi H^{\prime \frac{1}{2}} J_\phi\Phi}_\phi\\
	&=\ip{\tau^\phi_{z-\frac{\iu}{2}}(A)^\ast \Delta_\phi^{-\frac{1}{2}}J_\phi\Delta_\phi^{-\frac{1}{2}} H^{\prime \frac{1}{2}}\Phi}{\Delta_\phi^{-\frac{1}{2}}J_\phi \Delta_\phi^{-\frac{1}{2}} H^{\prime\frac{1}{2}} \Phi}_\phi\\
	&=\ip{\tau^\phi_{z-\frac{\iu}{2}}(A^\ast) \Delta_\phi^{-\frac{1}{2}} H^{\prime \frac{1}{2}} \Phi}{\Delta_\phi^{-\frac{1}{2}}H^{\prime \frac{1}{2}} \Phi}_\phi\\
	&=\ip{\Delta_\phi^{-\frac{1}{2}}\tau^\phi_{z-\frac{\iu}{2}}(A^\ast) \Delta_\phi^{-\frac{1}{2}} H^{\prime\frac{1}{2}} \Phi}{H^{\prime\frac{1}{2}} \Phi}_\phi\\	
	&=\ip{\tau^\phi_{z}(A^\ast) H^{\prime\frac{1}{2}} \Phi}{H^{\prime\frac{1}{2}} \Phi}_\phi\\	
	&=\ip{\tau^\phi_{z}(A^\ast) H^{\prime} \Phi}{ \Phi}_\phi\\	
	&=\psi\left(\tau^\phi_{z}(A)\right).	
	\end{aligned}
	\end{equation}
	
	Using the analyticity of the left-hand side of equation \eqref{eq:calculationRN} (see Corollary \ref{clemma5}) and the constancy of the right-hand side on the line $\Im{z}=0$, there remains no possibility but the constancy of the analytic extension for the strip $\strip{\frac{1}{2}}$.
	%#### Minha prova?
	Finally, basically undoing the steps in equation \eqref{eq:calculationRN}, we get, for every ${A\in \mathfrak{N}_\phi\cap \mathfrak{N}_\phi^\ast}$,
	$$\begin{aligned}
	\phi\left(H^\frac{1}{2} A H^\frac{1}{2}\right)&=\ip{A^\ast H^\prime\Phi}{\Phi}_\phi\\
	&=\ip{A^\ast J_\phi \Delta_\phi^{-\frac{1}{2}} H^\prime \Delta_\phi^{-\frac{1}{2}} J_\phi\Phi}{\Phi}_\phi\\
	&=\ip{A^\ast H\Phi}{\Phi}_\phi\\
	&=\phi(HA).
	\end{aligned}$$
	
	Hitherto we have proved that $\psi(A)=\phi(HA)$ for all $A \in \mathfrak{N}_\phi\cap \mathfrak{N}_\phi^\ast$, but this result can be extended for every $A\in\nalgebra_+$ using semifiniteness and normality.  

\end{proof}

This theorem reveals some interesting property about the nature of the modular automorphism group, stated below. Curiously, this corollary will be used to improve the result of which it comes.

\begin{corollary}
	\label{invariancecommute}
	Let $\phi$ and $\psi$ be two faithful normal semifinite weights on a von Neumann algebra $\nalgebra$, and $\{\tau^\phi_t\}_{t\in\mathbb{R}}$ and $\{\tau^\psi_t\}_{t\in\mathbb{R}}$ be their automorphism groups, respectively. The following are equivalent:
	\begin{enumerate} [(i)]
		\item $\psi=\psi\circ\tau^\phi_t$;
		\item $\{\tau^\phi_t\}_{t\in \mathbb{R}}$ and $\{\tau^\psi_t\}_{t\in\mathbb{R}}$ commute;
		\item $\phi=\phi\circ\tau^\psi_t$.
	\end{enumerate}  
\end{corollary}
\begin{proof}
	$(i)\Rightarrow (ii)$ By Corollary \ref{modularautomosphism}, $\displaystyle \tau^{\psi}_t=\tau^{\psi\circ\tau^\phi_t}_t={\tau^{\phi}}^{-1}_t\circ\tau^{\psi}_t\circ\tau^{\phi}_t\Rightarrow \tau^{\phi}_t\circ\tau^{\psi}_t=\tau^{\psi}_t\circ\tau^{\phi}_t$.
	
	$(ii)\Rightarrow (iii)$ If the automorphism groups commute, by Corollary \ref{modularautomosphism}, $\phi\circ\tau^\psi_t$ has $\{\tau^\phi_t\}_{t\in\mathbb{R}}$ as its automorphism group. In addition, it is obvious that $\phi\circ\tau^\psi_t$ is normal and semifinite.
	
	Then, also $\omega=\phi\circ\tau^\psi_t+\phi$ is normal and has $\{\tau^\phi_t\}_{t\in\mathbb{R}}$ as its automorphism group. Let's prove it is semifinite.
	
	Since $\phi$ and $\phi\circ\tau^\psi_t$ are invariant under the action of $\{\tau^\phi_t\}_{t\in \mathbb{R}}$, so are $\mathfrak{N}_\phi$, $\mathfrak{N}_{\phi}^\ast$, $\mathfrak{N}_{\phi\circ\tau^\psi_t}$, and $\mathfrak{N}_{\phi\circ\tau^\psi_t}^\ast$. Thus, by item $(iv)$ in Proposition \ref{simplepropweights}, also $\mathfrak{M}_\psi$ and $\mathfrak{M}_{\phi\circ\tau^\psi_t}$ are invariant under the action of the automorphism group.
	
	We can use the very same proof of Proposition \ref{AnalDense} to obtain, for any $A\in \mathfrak{M}_\phi$, a sequence $(A_n)_{n\in\mathbb{N}} \subset \mathfrak{M}_\phi$ of analytic elements for $\{\tau^\phi_t\}_{t\in \mathbb{R}}$ such that $A_n\xrightarrow{WOT} A$ throughout equation \eqref{formulaAE}. The same holds for any $A\in \mathfrak{M}_{\phi\circ\tau^\psi_t}$.
	
	Hence, the sets $\mathcal{M}_\phi\cap\mathcal{M}_\mathcal{A}$ and $\mathcal{M}_{\phi\circ\tau^\psi_t}\cap\mathcal{M}_\mathcal{A}$ are WOT-dense in $\mathcal{M}_\psi$, which itself is WOT-dense in $\nalgebra$, because the multiplication is separately WOT-continuous.
	
	Finally, by Proposition \ref{twosidedmodule}, $\left(\mathcal{M}_\phi\cap\mathcal{M}_\mathcal{A}\right)\left(\mathcal{M}_{\phi\circ\tau^\psi_t}\cap\mathcal{M}_\mathcal{A}\right)\subset \mathcal{M}_\phi\cap\mathcal{M}_{\phi\circ\tau^\psi_t}$, and the set on the right-hand side is WOT-dense.
	
	Since $\omega$ is a faithful normal semifinite weight on $\nalgebra$, such that $\phi=\phi\circ\tau^\omega_t$ and $\phi\leq\omega$, by Theorem \ref{TPTRN}, there exists a unique operator $K\in \nalgebra$, invariant for $\{\tau^\omega_t=\tau^\phi_t\}_{t\in \mathbb{R}}$, such that $$
	\phi(A)=\omega(KA)=\phi\circ\tau^\psi_t(KA)+\phi(KA) \quad \forall A\in \nalgebra_+.$$
	
	Since both $\phi$ and $\phi\circ\tau^\psi_t$ are faithful, $\{0,1\}\notin \sigma(K)$ and we can define the positive operator $H=\frac{K}{\mathbbm{1}-K}$, which is affiliated to $\mathfrak{M}_{\tau\phi}$. Let $\left\{E^H_\lambda\right\}_{\lambda\in\mathbb{R}}$ be the spectral projections of $H$, we now that $\displaystyle\bigcup_{n\in\mathbb{N}}E^H_{(0,n)}\nalgebra_+ E^H_{(0,n)}$ is dense in $\nalgebra_+$ and
	\begin{equation}
	\label{eq:calculationX7}
	\phi(A)=\phi\circ\tau^\psi_t(HA) \qquad  \forall A\in \bigcup_{n\in\mathbb{N}}E^H_{(0,n)}\nalgebra_+ E^H_{(0,n)}.
	\end{equation}
	
	By Lemma 1 in \cite{takesaki70} or Theorem 2.11 in \cite{Takesaki2003}, we know that the automorphism group for a weigh as in equation \eqref{eq:calculationX7} is given by $$\begin{aligned}
	\tau^\phi_t(A)=\tau^\omega_t\left(H^{\iu t}AH^{-\iu t}\right)&=\tau^\phi_t(H^{\iu t}AH^{-\iu t})\\
	&\Rightarrow A=H^{\iu t}AH^{-\iu t} \quad \forall A\in \mathfrak{N}_\phi, \\
	&\Rightarrow H\eta\mathcal{Z}(\nalgebra).
	\end{aligned}$$
	
	If $H\neq \mathbbm{1}$, there exists a projection $P\in\mathcal{F}_\phi \cap \mathfrak{M}_{\tau^\phi}$ such that $P\leq E^H_{(a,b)}$ for some $a,b\in\mathbb{R}$ with $1\neq(a,b)$. Then, either $HP<P$ or $HP>P$, but this leads to a contradiction since
	$$
	\phi(P)\neq\phi(HP)=\phi\circ\tau^\psi_t(P)=\phi\left(\tau^\psi_t(P)\right)=\phi(P).$$
	
	The conclusion is that $H=\mathbbm{1}$ and then $\phi(A)=\phi\circ\tau^\psi_t(HA)=\phi\circ\tau^\psi_t(A)$.
	
	$(iii)\Rightarrow (i)$ is obvious just applying $(i)\Rightarrow(iii)$ for $\phi$ instead of $\psi$. 
\end{proof}

In order to generalise the previous theorem, we will need the following lemma.

\begin{lemma}
	\label{sumsemifinite}
	Let $\phi$ and $\psi$ be two faithful normal semifinite weights on a von Neumann algebra $\nalgebra$. Suppose $\psi$ is invariant under the modular automorphism group $\{\tau^\phi_t\}_{t\in \mathbb{R}}$ of $\phi$. Then $\phi+\psi$ is semifinite.
\end{lemma}
\begin{proof}
	
	Since $\psi$ is invariant under the action of $\{\tau^\phi_t\}_{t\in \mathbb{R}}$ and $\{\tau^\psi_t\}_{t\in \mathbb{R}}$, so is $\phi$, by Corollary \ref{invariancecommute}. Hence, as before, $\mathfrak{M}_\phi$ and $\mathfrak{M}_\psi$ are invariant under the action of those automorphism groups.
	
	Since the automorphism groups commute, by Corollary \ref{invariancecommute}, we can use the very same proof of Proposition \ref{AnalDense} to obtain, for any $A\in \mathfrak{M}_\phi$, a sequence $(A_n)_{n\in\mathbb{N}} \subset \mathfrak{M}_\phi$ of analytic elements for both $\{\tau^\phi_t\}_{t\in \mathbb{R}}$ and $\{\tau^\psi_t\}_{t\in \mathbb{R}}$ such that $A_n\xrightarrow{WOT} A$ throughout the following expression
	$$A_n=\frac{n}{\pi}\int_\mathbb{R}\int_\mathbb{R}{e^{-n(t^2+s^2)}\tau^\psi_t\circ\tau^\phi_s(A) dt ds}.$$
	
	Hence, $\mathcal{M}_\phi\cap\mathcal{M}_\mathcal{A}$ is a WOT-dense in $\mathcal{M}_\phi$, which in turn is WOT-dense in $\nalgebra$. By the very same argument, $\mathcal{M}_\psi\cap\mathcal{M}_\mathcal{A}$ is a WOT-dense set in $\nalgebra$.
	
	Finally, $\left(\mathcal{M}_\phi\cap\mathcal{M}_\mathcal{A}\right)\left(\mathcal{M}_\psi\cap\mathcal{M}_\mathcal{A}\right)\subset \mathcal{M}_\phi\cap\mathcal{M}_\psi$ and the right-hand side is WOT-dense in $\nalgebra$, because the multiplication is separately WOT-continuous.
	
\end{proof}

Finally, the theorem that gives name to this section in all  generality we will need.

\begin{theorem}[Pedersen-Takesaki-Radon-Nikodym]\index{theorem! Pedersen-Takesaki-Radon-Nikodym}
	\label{TPTRN2}
	Let $\phi$ and $\psi$ be two normal semifinite weights on a von Neumann algebra $\nalgebra$. Suppose, in addition, that $\phi$ is faithful and $\psi$ is invariant under the modular automorphism group $\{\tau^\phi_t\}_{t\in \mathbb{R}}$ of $\phi$. Then, there exists a unique positive operator $H\eta\mathfrak{M}_{\tau^\phi}$ such that $\psi(A)=\phi(HA)$, for all $A\in \nalgebra_+$.
\end{theorem}
\begin{proof}
	Let $s^\nalgebra(\psi)$ be the support projection for $\psi$. Notice that $\phi$ and $\psi$ are faithful normal semifinite weights for the von Neumann algebra $s^\nalgebra(\psi) \nalgebra s^\nalgebra(\psi)$ and $s^\nalgebra(\psi)$ is $\{\tau^\phi_t\}_{t\in \mathbb{R}}$-invariant. By Lemma \ref{sumsemifinite}, $\phi+\psi$ is a faithful normal semifinite weight on $s^\nalgebra(\psi) \nalgebra s^\nalgebra(\psi)$. Since $\phi\leq \phi+\psi$ and $\phi$ is $\left\{\tau^{\phi+\psi}_t\right\}$-invariant as a consequence of Corollary \ref{invariancecommute}, Theorem \ref{TPTRN} states that there exists a positive $\left\{\tau^{\phi+\psi}_t\right\}$-invariant operator $K\in s^\nalgebra(\psi) \nalgebra s^\nalgebra(\psi)$ with $0\leq K\leq \mathbbm{1}$ such that
	$$\phi\left(s^\nalgebra(\psi)As^\nalgebra(\psi)\right)=\phi\left(Ks^\nalgebra(\psi)As^\nalgebra(\psi)\right)+\psi\left(Ks^\nalgebra(\psi)As^\nalgebra(\psi)\right) \quad \forall A\in\mathfrak{N}_\phi.$$
	
	In addition, since $\phi$ is faithful on $s^\nalgebra(\psi) \nalgebra s^\nalgebra(\psi)$, $0\notin\sigma(K)$ and thus we can define $H=\frac{\mathbbm{1}-K}{K}$. Let $\left\{E^H_\lambda\right\}_{\lambda\in\mathbb{R}_+}$ be the spectral resolution of $H$ and define $H_n=HE_{(0,n)}=E_{(0,n)} H E_{(0,n)}$, then
	
	\begin{equation}
	\label{eq:calculationX8}
	\begin{aligned}
	\psi(E_{(0,n)}AE_{(0,n)})&=\psi(E_{(0,n)}A)\\
	&=\psi\left(s^\nalgebra(\psi)E_{(0,n)}As^\nalgebra(\psi)\right)\\
	&=\phi\left(H_n s^\nalgebra(\psi)E_{(0,n)}As^\nalgebra(\psi)\right)\\
	&=\phi\left(H_n E_{(0,n)} A\right)\\
	&=\phi\left(H_n^\frac{1}{2} AH_nH_n^\frac{1}{2}\right) \quad \forall A\in\mathfrak{N}_\phi.
	\end{aligned}
	\end{equation}
	
	Using normality, the desired invariance follows from equation \eqref{eq:calculationX8}.
	
\end{proof}

\section{Modular Theory for Weights}

In the spirit of Section \ref{secModular}, we will start by defining a closable operator, but now for a weight. Here, we are following reference \cite{Takesaki2003}.

Let $\phi$ be a semifinite weight, $\mathcal{N}_\phi$, $N_\phi$, $\hilbert_\phi$ as in \ref{weightsets} and \ref{GNSnotation} and let $\eta_\phi:\mathfrak{N}_\phi \to \mathfrak{N}_\phi/N_\phi$ be the quotient operator. Set $\mathfrak{D}_\phi=\eta_\phi\left(\mathfrak{N}_\phi\cap \mathfrak{N}_\phi^\ast \right)$ and define an the anti-linear operator by
\begin{equation}
\label{defS}
\begin{aligned}
 S^0_{\Phi}: & \mathfrak{D}_\phi &\to		& \ \hilbert_\phi\\
			 & \eta_\phi(A) 	&\mapsto 	&\ \eta_\phi(A^\ast).\\
\end{aligned}
\end{equation}

Note that this operator is defined on a dense subset of $\hilbert_\phi$. In fact, note that if we take $\eta_\phi(A) \in \mathfrak{D}^\perp \cap \eta_\phi(\mathfrak{N}_\phi)$, we get, for every $B,C \in \mathfrak{N}_\phi,$
$$\ip{\pi_\phi(B)\eta_\phi(A)}{\eta_\phi(C)}_\phi=\ip{\eta_\phi(A)}{\pi_\phi(B^\ast)\eta_\phi(C)}_\phi=\ip{\eta_\phi(A)}{\eta_\phi(B^\ast C)}_\phi=0.$$
Since $\mathfrak{N}_\phi/N_\phi$ is dense in $\hilbert_\phi$, we have no option unless $\pi_\phi(B)\eta_\phi(A)=0$ for all $B \in \mathfrak{N}_\phi$, and thus $\mathfrak{D}^\perp_\phi=\{0\}$.

\begin{lemma}
%Theorem 3.12 Boey Thesis
\label{lemma1}
Let $\nalgebra$ be a von Neumann algebra, $\phi$ a weight and
\begin{equation}
\label{eq6}
\Phi_\phi=\bigl\{\omega \in \nalgebra^+_\ast \ | \ \exists \varepsilon>0 \textrm{ such that } (1+\varepsilon)\omega\leq \phi\bigr\}.
\end{equation}
For each $\omega \in \Phi_\phi$ there exists a positive operator $H_\omega \in \pi_\phi(\nalgebra)^\prime$ and a vector $\zeta_\omega \in \hilbert_\phi$ such that $H_\omega^\frac{1}{2}\eta_\phi(A)=\pi_\phi(A)\zeta_\omega$ and
$$\omega(A)=\ip{\pi_\phi(A)\zeta_\omega}{\zeta_\omega}_\phi.$$
\end{lemma}
\begin{proof}

Note that, for $A,B \in \mathfrak{N}_\phi$, 
$$\begin{aligned}
\left|\omega(B^\ast A)\right|	&\leq \omega(A^\ast A)^\frac{1}{2}\omega(B^\ast B)^\frac{1}{2} \\
								&\leq (1+\varepsilon)^{-1}\phi(A^\ast A)^\frac{1}{2}\phi(B^\ast B)^\frac{1}{2}\\
								&= (1+\varepsilon)^{-1}\|A\|_\phi \|B\|_\phi.
\end{aligned}$$

This means that the sesquilinear form $\ip{\eta_\phi(A)}{\eta_\phi(B)}_\omega=\widetilde{\omega}(\eta_\phi(B)^\ast \eta_\phi(A))$ is bounded in $\hilbert_\phi$, and by Riesz's theorem, there exists a bounded operator $H_\omega \in B(\hilbert_\phi)$ such that $$\widetilde{\omega}(\eta_\phi(B)^\ast \eta_\phi(A))=\ip{H_\omega\eta_\phi(A)}{\eta_\phi(B)}_\phi.$$

Note that, for any $A,B$ as above and $C \in \nalgebra$,
$$\begin{aligned}
\ip{H_\omega\pi_\phi(C)\eta_\phi(A)}{\eta_\phi(B)}_\omega 
				&=\omega\left(\eta_\phi(B)^\ast \eta_\phi(C A)\right)\\
				&=\omega\left(\eta_\phi(C^\ast B)^\ast\eta_\phi(A)\right)\\
				&=\ip{H_\omega\eta_\phi(A)}{\eta_\phi(C^\ast B)}\\
				&=\ip{\eta_\phi(C)H_\omega\eta_\phi(A)}{\eta_\phi(B)}.\\
\end{aligned}$$
Hence $H_\omega \in \pi_\phi(\nalgebra)^\prime$.

Now, we use the cyclic vector $\xi_\omega$ obtained by the GNS-construction though $\omega$.
$$\begin{aligned}
\|\pi_\omega(A)\xi_\omega\|^2_\omega=\omega(A^\ast A)\leq (1+\varepsilon)^{-1}\phi(A^\ast A)=(1+\varepsilon)^{-1}\|\eta_\phi(A)\|_\phi^2.
\end{aligned}$$
From this we conclude that $\pi_\omega(A)\xi_\omega \overset{\Psi}{\mapsto} \eta_\phi(A)$ is a bounded operator, initially defined in the dense subspace which can be extended to $\hilbert_\omega$ through continuity, resulting in a bijective operator.

But 
$$\begin{aligned}
\ip{\Psi \pi_\phi(A)}{\Psi\pi_\phi(B)} 
		&= \ip{\Psi \pi_\phi(A)}{\pi_\omega(B)\xi_\omega}\\
		&= \ip{ \pi_\omega(A)\xi_\omega}{\pi_\omega(B)\xi_\omega}\\
		&=\omega(B^\ast A)\\
		&= \ip{ H_\omega \eta_\phi(A)}{\eta_\phi(B)}.\\
\end{aligned}$$
It follows that $\Psi^\ast \Psi = H_\omega$. By the polar decomposition theorem, there exists a partial isometry $u$ such that of $\Psi=u H_\omega^\frac{1}{2}$.

On the other hand, $\Psi\pi_\phi(A)\eta_\phi(B)=\Psi \eta_\phi(AB)=\pi_\omega(AB)\xi_\omega=\pi_\omega(A)\Psi \eta_\phi(B)$ and it follows that $\Psi\pi_\phi(A)=\pi_\omega(A)\Psi$.

In terms of the polar decomposition,,

$$uH_\omega^\frac{1}{2} \pi_\phi(A)=\Psi\pi_\phi(A)=\pi_\omega(A)\Psi=\pi_\omega(A)uH_\omega^\frac{1}{2}$$

But the partial isometry is chosen such that $u[H_\omega^\frac{1}{2}\hilbert_\phi]^\perp=0$, thus so does $u\pi_\phi(A)$. Since we proved the operators are equal in $[H_\omega^\frac{1}{2}\hilbert_\phi]$, they are equal everywhere.

Define $\zeta_\omega=u\xi_\omega\in \hilbert_\phi$. This definition gives us

$$\begin{aligned}
\pi_\phi(A)\zeta_\omega	&=\pi_\phi u^\ast \xi_\omega \\
						&=(u_\omega\pi_\phi(A^\ast))^\ast\xi_\omega\\
						&=(\pi_\omega(A^\ast)u)^\ast\xi_\omega\\
						&=u^\ast\pi_\omega(A)\xi_\omega\\
						&=u^\ast \Psi\eta_\phi(A)\\
						&=H_\omega^{\frac{1}{2}}\eta_\phi(A).\\
\end{aligned}$$

Finally, take $A=B^\ast C$ where $B,C \in \mathfrak{N}_\phi$

$$\begin{aligned}
\ip{\pi_\phi(A)\zeta_\omega}{\zeta_\omega} 
	&=\ip{\pi_\phi(C)\eta_\omega}{\pi_\phi(B)\eta_\omega}\\
	&=\ip{H_\omega^\frac{1}{2} \eta_\phi(C)}{H_\omega^\frac{1}{2} \eta_\phi(B)}\\
	&=\omega(B^\ast C) \\
	&=\omega(A).
\end{aligned}$$

\end{proof}

\begin{lemma}
%Boey Theorem 3.11 and Takesaki II Theorem 2.6
Let $\nalgebra$ be a von Neumann algebra, $\phi$ a weight, $\Phi_\phi$ as in equation \eqref{eq6}. Let also $H_\omega \in \pi_\phi(\nalgebra)^\prime$ and $\zeta \in \hilbert_\phi$ be  a positive operator and a vector, respectively, as in Lemma \ref{lemma1}.

Then, the set
$$\bigcup_{\omega_1,\omega_2\in \Phi_\phi}\left\{H_{\omega_1}^\frac{1}{2}\pi_\phi(\nalgebra)^\prime\zeta_{\omega_2}\right\}$$
is dense in $\hilbert_\phi$.
\end{lemma}

\begin{proof}
Note that, for $A\in \mathfrak{D}_\phi$,
\begin{equation}
\label{eq5}
\begin{aligned}
\|\eta_\phi(A)\|^2 	&= \sup_{\omega\in \Phi_\phi}\left\{\omega(A^\ast A)\right\}\\
					&= \sup_{\omega\in \Phi_\phi}\left\{\|A\zeta_\omega\|^2\right\}\\
					&= \sup_{\omega\in \Phi_\phi}\left\{\|H_\omega^\frac{1}{2}\eta_\phi(A)\|^2\right\}\\
					&=\sup_{\omega\in \Phi_\phi}\left\{\ip{H_\omega\eta_\phi(A)}{\eta_\phi(A)}\right\}\\
\end{aligned}
\end{equation}
but the net $(H_\omega)_{\omega\in\Phi_\phi}$ is an increasing net, and, by equation \eqref{eq5}, we must have $\displaystyle\bigcup_{\omega_1\in \Phi_\phi}\{H_{\omega_1}^\frac{1}{2}\pi_\phi(\nalgebra)^\prime\zeta_{\omega_2}\}^\perp=\{0\}$. It follows that this set is dense in $\pi_\phi(\nalgebra)^\prime \zeta_{\omega_2}$.

Let $P$ be the projection on $\overline{span\{\pi_\phi(\nalgebra)\zeta_\omega | \omega \in \Phi_\phi\}}\subset \hilbert_\phi$. Since $\pi(\nalgebra)$ is also a von Neumann algebra, it must contain $P$. Let $S\in \nalgebra$ such that $\pi_\phi(S)=\mathbbm{1}-P$, then
$$\begin{aligned}
\phi(S)	&=\sup_{\omega\in \Phi_\phi}\{\omega(S)\}\\
		&=\sup_{\omega\in \Phi_\phi}\{\ip{\pi_\phi(S)\zeta_\omega}{\zeta_\omega}\}\\
		&=\sup_{\omega\in \Phi_\phi}\{\ip{(1-P)\zeta_\omega}{\zeta_\omega}\}\\
		&=0.\\
\end{aligned}$$

This means that $S\in N_\phi \Rightarrow \pi_\phi(S)=0\Rightarrow P=\mathbbm{1}$.
\end{proof}

\begin{lemma}
The operator $S_0$ is closable.
\end{lemma}
\begin{proof}

It is enough to show that the adjoint of $S_0$ is densely defined.

By Lemma \ref{lemma1}, for each $\omega \in \Phi_\phi$ there exists a positive operator $H_\omega \in \pi_\phi(\nalgebra)^\prime$ and a vector $\zeta_\omega \in \hilbert_\phi$ such that $H_\omega^\frac{1}{2}\eta_\phi(A)=\pi_\phi(A)\zeta_\omega$ and
$\omega(A)=\ip{\pi_\phi(A)\zeta_\omega}{\zeta_\omega}$.

Take $\omega_1, \omega_2 \in \Phi_\phi$, $B\in \pi_\phi(\nalgebra)^\prime$ and $A\in \mathfrak{D}_\phi$,
$$\begin{aligned}
\ip{S_0\eta_\phi(A)}{H_{\omega_1}^\frac{1}{2}B\zeta_{\omega_2}}
	&=\ip{\eta_\phi(A^\ast)}{H_{\omega_1}^\frac{1}{2}B\zeta_{\omega_2}}\\
	&= \ip{\pi_\phi(A^\ast) \zeta_{\omega_1}}{B\zeta_{\omega_2}}\\
	&= \ip{\zeta_{\omega_1}}{\pi_\phi(A)B\zeta_{\omega_2}}\\
	&= \ip{\zeta_{\omega_1}}{B\pi_\phi(A)\zeta_{\omega_2}}\\
	&= \ip{\zeta_{\omega_1}}{BH_{\omega_2}^\frac{1}{2}\eta_\phi(A)}\\
	&= \ip{H_{\omega_2}^\frac{1}{2}B^\ast\zeta_{\omega_1}}{\eta_\phi(A)}.\\
\end{aligned}$$
So, $H_{\omega_2}^\frac{1}{2}B^\ast\zeta_{\omega_1}$ must be in the domain of $S_0^\ast$ and, moreover, $S_0^\ast H_{\omega_1}^\frac{1}{2}B\zeta_{\omega_2}=H_{\omega_2}^\frac{1}{2}B^\ast\zeta_{\omega_1}$.

\end{proof}

\begin{definition}
We denote by $J_\phi$ and $\Delta_\phi$ be the unique anti-linear partial isometry and the positive operator, respectively, in the polar decomposition of $\overline{S_0}$, \ie,    $S=J_\phi\Delta_\phi^{\frac{1}{2}}$. $J_\phi$ is called the modular conjugation operator and $\Delta_\phi$ is called the modular operator. 
\end{definition}

\section{Relative Modular Theory}

Although it might seem repetitive to define the modular theory for a weight and later a relative modular theory, it is not the case, since we can use modular theory to construct the relative modular theory in a much more elegant way using what we call a balanced weight. Such construction is also interesting because it gives us a powerful tool to deal with these operators.

The construction we are presenting here is based in \cite{Izumi1999} and \cite{Takesaki2003}.

\begin{notation}
We denote the von Neumann algebra of $2\times2$ matrices of elements in the von Neumann algebra $\nalgebra$ by
$$M_{2\times2}=\left\{\begin{pmatrix} A_{11} & A_{12} \\ A_{21} & A_{22}\end{pmatrix} \ \middle| \ A_{ij}\in \nalgebra, \ i=1,2, \ j=1,2\right\}.$$
\end{notation}

\begin{definition}[Balanced Weight]
\label{balancedweight}
Let $\nalgebra$ be a von Neumann algebra and $\phi, \psi$ weights. We define the \index{weight! balanced} balanced weight
$$\begin{aligned}
\theta_{\phi,\psi}:	& \quad M_{2\times2}(\nalgebra)	&\to	& \ \overline{\mathbb{R}}\\
					& \begin{pmatrix} A_{11} & A_{12} \\ A_{21} & A_{22}\end{pmatrix} &\mapsto	& \ \phi(A_{11})+\psi(A_{22}).
\end{aligned}$$
\end{definition}

\begin{notation}
To simplify the notation, we will not use $\theta_{\phi,\psi}$, instead, we will only write $\theta$.

In addition, it is much easier to write $\displaystyle \sum_{i,j=1}^{2}A_{ij}e_{ij}$ instead of $\begin{pmatrix} A_{11} & A_{12} \\ A_{21} & A_{22}\end{pmatrix}$, where
$$e_{11}=\begin{pmatrix} 1 & 0 \\ 0 & 0\end{pmatrix}, \quad
e_{12}=\begin{pmatrix} 0 & 1 \\ 0 & 0\end{pmatrix}, \quad
e_{21}=\begin{pmatrix} 0 & 0 \\ 1 & 0\end{pmatrix}, \quad
e_{22}=\begin{pmatrix} 0 & 0 \\ 0 & 1\end{pmatrix}, \quad$$
\end{notation}

As expected, some of the good properties we have studied for weights are inherited by the balanced weight, as we will see in the next lemma.

\begin{lemma}
Let $\nalgebra$ be a von Neumann algebra and $\phi, \psi$ weights.
\label{lemmabalanced}
\begin{enumerate}[(i)]
\item $\theta_{\phi,\psi}$ is a normal semifinite weight if and only if $\phi, \psi$ have these properties;
\item $\theta_{\phi,\psi}$ is faithful if and only if at least one of $\phi$ and $\psi$ is faithful;
\item $\mathfrak{N}_{\theta}=\{A \in M_2(\nalgebra) \ | A_{11}, A_{21} \in \mathfrak{N}_\phi \textrm{ and } A_{12}, A_{22} \in \mathfrak{N}_\psi \}$.
\end{enumerate}
\end{lemma}

\begin{proof}
$(i)$ First $$\theta(A)=\sup_{\substack{\omega_\phi \in \Phi_\phi \\ \omega_\psi \in \Phi_\psi}}{\{\omega_\phi(A_{11})+\omega_\psi(A_{22})\}}.$$
Hence $\theta$ is normal.

Note that, if we take projections $p_\phi\in \mathfrak{M}_\phi$ and $p_\psi\in \mathfrak{M}_\psi$, $\theta_{\phi,\psi}(p_\phi e_{11}+p_\psi e_{22})=\phi(p_\phi)+\psi(p_\psi)<\infty$. Thus, $p_\phi e_{11}+p_\psi e_{22} \in \mathfrak{M}_\theta$. Since $\mathfrak{M}_\phi,\mathfrak{M}_\psi$ are dense in the algebra, $1$ must be in the closure of both and, particularly, it must be in the closure of projections, hence $1_{M_2(\nalgebra)}=e_{11}+e_{22} \in \mathfrak{M}_\theta$ and it is dense.

$(ii)$ As a consequence of positivity, $\theta_{\phi,\psi}(A^\ast A)=0$ if and only if $\phi(A)=\psi(A)=0$.

$(iii)$ It is a direct consequences of the definition of $A^\ast A$.

\end{proof}

Let us see what happens with the modular operators for the weight $\theta_{\phi,\psi}$. First, the domain of $S_0$ will be, using $(iii)$ in Lemma \ref{lemmabalanced},
$$\mathfrak{N}_\theta\cap\mathfrak{N}_\theta^\ast=\{A\in M_2(\nalgebra) \ | 
\ A_{11}\in \mathfrak{N}_\phi\cap\mathfrak{N}_\phi^\ast,
A_{12}\in \mathfrak{N}_\phi^\ast\cap\mathfrak{N}_\psi,
A_{21}\in \mathfrak{N}_\phi\cap\mathfrak{N}_\psi^\ast,
A_{22}\in \mathfrak{N}_\psi\cap\mathfrak{N}_\psi^\ast \}.$$

This suggests that the following closed subspaces of $\hilbert_\theta$ will play an interesting role:

$$\begin{aligned}
&\hilbert_{\phi\phi}=\overline{span\left\{\eta_\theta\left(\mathfrak{N}_\phi\cap\mathfrak{N}_\phi^\ast \right)e_{11}\right\}},\\
&\hilbert_{\phi\psi}=\overline{span\left\{\eta_\theta\left(\mathfrak{N}_\phi\cap\mathfrak{N}_\psi^\ast\right)e_{21}\right\}},\\
&\hilbert_{\psi\phi}=\overline{span\left\{\eta_\theta\left(\mathfrak{N}_\psi^\ast\cap\mathfrak{N}_\phi\right)e_{12}\right\}},\\
&\hilbert_{\psi\psi}=\overline{span\left\{\eta_\theta\left(\mathfrak{N}_\psi\cap\mathfrak{N}_\psi^\ast\right)e_{22}\right\}}.\\
\end{aligned}$$

Note, that we immediately have that
$$\Dom{S} \subset\hilbert_{\phi\phi}\oplus\hilbert_{\phi\psi}\oplus\hilbert_{\psi\phi}\oplus\hilbert_{\psi\psi}\subset\hilbert_\theta.$$

Now, take $Ae_{21}\in \hilbert_{\phi\psi}\cap\Dom{S}e_{21}$. Then, there exists a sequence $$\left(\eta_\theta(X_n) e_{21}\right)_{n\in\mathbb{N}} \subset \eta_\theta\left(\mathfrak{N}_\phi\cap\mathfrak{N}_\phi^\ast\right)e_{21}$$ converging to $Ae_{21}$. Consequently,
$$S(Ae_{21})=\lim_{n\to\infty}S(\eta_\theta(X)e_{21})=\lim_{n\to\infty}\eta_\theta(X^\ast)e_{12}\in \hilbert_{\psi\phi}\cap\Dom{S}.$$
Proceeding in this way, we get
$$\begin{aligned}
S(\hilbert_{\phi\phi})\subset\hilbert_{\phi\phi}, &\quad
S(\hilbert_{\phi\psi})\subset\hilbert_{\psi\phi},\\
S(\hilbert_{\psi\phi})\subset\hilbert_{\phi\psi}, &\quad
S(\hilbert_{\psi\psi})\subset\hilbert_{\psi\psi}.\\
\end{aligned}$$

Thus the operator $S$ can be represented as a matrix acting on a subspace of $\hilbert_{\phi\phi}\oplus\hilbert_{\phi\psi}\oplus\hilbert_{\psi\phi}\oplus\hilbert_{\psi\psi}$, namely
$$\begin{pmatrix}
S_{\phi\phi} 	& 0 	& 0 			& 0\\
0				& 0		& S_{\psi\phi}	& 0\\
0				& S_{\phi\psi}	& 0		& 0 \\
0				& 0		& 0		& S_{\psi\psi}
\end{pmatrix}.$$
Moreover, note that $S^2=1_{\Dom{S}}$, thus it is bijective on its range.

The way we define $S_{\phi\phi}, S_{\phi\psi}, S_{\psi\phi}$ and $S_{\psi\psi}$, they are operators acting in $\hilbert_{\phi\phi}, \hilbert_{\phi\psi},\hilbert_{\psi\phi}$ and $\hilbert_{\psi\psi}$, respectively. One can easily see that these spaces can be identified isometrically with subspaces of $\hilbert_\phi$ or $\hilbert_\psi$ through the extension of the following:
$$\begin{aligned}
U_{\phi,1}\eta_\theta(A e_{11})=\eta_\phi(A) \quad \forall A\in \mathfrak{N}_\phi\cap\mathfrak{N}_\phi^\ast,\\
U_{\phi,2}\eta_\theta(A e_{21})=\eta_\phi(A) \quad \forall A\in \mathfrak{N}_\phi\cap\mathfrak{N}_\psi^\ast,\\
U_{\psi,1}\eta_\theta(A e_{12})=\eta_\psi(A) \quad \forall A\in \mathfrak{N}_\phi^\ast\cap\mathfrak{N}_\psi,\\
U_{\psi,2}\eta_\theta(A e_{22})=\eta_\psi(A) \quad \forall A\in \mathfrak{N}_\psi\cap\mathfrak{N}_\psi^\ast.\\\end{aligned}$$

Using this identification, we can bring the components of $S$ back as operators between Hilbert spaces as
$$\begin{aligned}
\widetilde{S}_{\phi\phi}=U_{\phi,1}S_{\phi\phi}U_{\phi,1}^\ast, &\quad
\widetilde{S}_{\psi\phi}=U_{\psi,2}S_{\psi\phi}U_{\phi,2}^\ast,\\
\widetilde{S}_{\phi\psi}=U_{\phi,2}S_{\phi\psi}U_{\phi,2}^\ast, &\quad
\widetilde{S}_{\psi\psi}=U_{\psi,2}S_{\psi\psi}U_{\psi,2}^\ast.\\
\end{aligned}$$

We will not make distinction in the notation of the operators just defined.

Let us look to the polar the composition of $S_\theta$:
\begin{align}\label{RelModOp}
\Delta_\theta=S_\theta^\ast S_\theta	&=
\begin{pmatrix}
S_{\phi\phi}^\ast S_{\phi\phi}	& 0	& 0	& 0\\
0	& S_{\psi\phi}^\ast S_{\phi\psi}& 0 & 0 \\
0	& 0	&S_{\phi\psi}^\ast S_{\psi\phi}	& 0\\
0& 0 & 0 & S_{\psi\psi}^\ast S_{\psi\psi}\\
\end{pmatrix},\\
\label{RelModConj}
J_\theta &=
\begin{pmatrix}
J_{\phi\phi} & 0 & 0 & 0\\
0 & 0 & J_{\psi\phi}  & 0\\
0 & J_{\phi\psi} & 0 & 0\\
0 & 0 & 0 & J_{\psi\psi}\\
\end{pmatrix}.
\end{align}

It is immediate from the definition that $\widetilde{S}_{\phi\phi}$ and $\widetilde{S}_{\psi\psi}$ are the usual operators defined by equation $\ref{defS}$ for the weights $\phi$ and $\psi$.

The two new closed operators $S_{\phi\psi}$ and $S_{\psi\phi}$ will be used to defined the relative modular operators.

\begin{definition}\index{operator! relative modular}
We denote by $J_{\phi\psi}$ and $\Delta_{\phi\psi}$ the unique anti-linear partial isometry and positive operator, respectively, in the polar decomposition of $S_{\phi\psi}$, \ie,   \mbox{$S_{\phi\psi}=J_{\phi\psi}\Delta_{\phi\psi}^{\frac{1}{2}}$}. $J_{\phi\psi}$ is called the relative modular conjugation operator and $\Delta_{\phi\psi}$ is called the relative modular operator.
\end{definition}

Of course, it is possible to do the same definition for $S_{\psi\phi}$, but it would be redundant. One has:

%*****\textcolor{red}{Atention - complete the proposition}
\begin{proposition}
The relative modular conjugation and the modular operator satisfy the following relations:
\begin{enumerate}[(i)]
\item $\Delta_{\phi\psi}=S_{\psi\phi}^\ast S_{\phi\psi}$;
\item $J_{\phi\psi}^\ast=J_{\psi\phi}$;
\item $S_{\phi\psi}=J_{\phi\psi}\Delta_{\phi\psi}^{\frac{1}{2}}=\Delta_{\psi\phi}^{-\frac{1}{2}}J_{\phi\psi}$;
\item $S_{\phi\psi}^\ast=\Delta_{\phi\psi}^{\frac{1}{2}}J_{\psi\phi}=J_{\psi\phi}\Delta_{\psi\phi}^{-\frac{1}{2}}$.
\end{enumerate}
\end{proposition}
\begin{proof}
	All the items follows by definition and Proposition \ref{MOProp} applied to the balanced weight used in the definition of these operators. In special, equations \eqref{RelModOp} and \eqref{RelModOp} are quite useful.
	
\end{proof}

The first big difference between the usual modular operator and the relative one can be seen in the kernel of these operators. The relative modular operator has no trivial kernel, as one could expect.

\begin{proposition}
\label{deltaKernel}
Let $\nalgebra$ be a von Neumann algebra and $\phi, \psi$ states $$\ker{\Delta_{\phi\psi}}=\left(1-s^{\nalgebra^\prime}(\psi)s^{\nalgebra}(\phi)\right)\hilbert_{\phi\psi}.$$
\end{proposition}
\begin{proof}
We will use the balanced weight $\theta=\theta_{\phi,\psi}$ and denote by $\widetilde{\nalgebra}=M_{2\times2}(\nalgebra)$.

Let us look to the kernel of $\Delta_\theta$.
$$\hilbert_\theta=s^{\widetilde{\nalgebra}}(\theta)s^{\widetilde{\nalgebra}^\prime}(\theta)\hilbert_\theta\oplus \left(\mathbbm{1}-s^{\widetilde{\nalgebra}}(\theta)\right)s^{\widetilde{\nalgebra}^\prime}(\theta)\hilbert_\theta \oplus
\left(\mathbbm{1}-s^{\widetilde{\nalgebra}}(\theta)\right)\hilbert_\theta.$$
But $$\begin{aligned}
S_\theta\left(\left(\mathbbm{1}-s^{\widetilde{\nalgebra}}(\theta)\right)\eta_\theta(A)\right) 
&=\eta_\theta\left(A^\ast\left(\mathbbm{1}-s^{\widetilde{\nalgebra}}(\theta)\right)\right) \\
&=\pi_\theta(A)^\ast\eta_\theta\left(\mathbbm{1}-s^{\widetilde{\nalgebra}}(\theta)\right)\\
&=0. \\
\end{aligned}$$
Doing the analogous calculation, we find
$$S_\theta\left(\left(\mathbbm{1}-s^{\widetilde{\nalgebra}^\prime}(\theta)\right)\eta_\theta(A)\right) =0.$$
 
Hence the modular operator can be decomposed as $$\Delta_\theta=\widetilde{\Delta}_\theta\oplus0\oplus0.$$

It follows that $\ker{\Delta_\theta}\subset\left(\mathbbm{1}-s^{\widetilde{\nalgebra}}(\theta)s^{\widetilde{\nalgebra}^\prime}(\theta)\right)\hilbert_\theta$. Now, it remains to be shown that $\ker{\widetilde{\Delta}_\theta}=\{0\}$.

But this is a consequence of Propositions \ref{supportrelations} and \ref{supportcyclicseparating} that $\eta_\theta(1)$ is a cyclic and separating vector for the algebra $s^{\widetilde{\nalgebra}}(\theta)s^{\widetilde{\nalgebra}^\prime}(\theta)\pi_\theta(\nalgebra)s^{\widetilde{\nalgebra}^\prime}(\theta)s^{\widetilde{\nalgebra}}(\theta)$ acting on the Hilbert space $s^{\widetilde{\nalgebra}}(\theta)s^{\widetilde{\nalgebra}^\prime}(\theta)\hilbert_\theta$. Hence, the operator $\widetilde{\Delta}_\theta$ is the usual modular operator related to the cyclic and separating vector $\eta(1)$ that is injective.

Since the relation is true for the balanced weight $\Delta_\theta$, it is easy to obtain the general relation though the matrix expression.

\end{proof}

Note that, in general, $\Delta_{\phi\psi}$ has no inverse, since its kernel it not trivial. But $\widetilde{\Delta}_{\phi\psi}$, which appears in the last proposition is injective, hence, it has a inverse on its image.

\begin{notation}
We denote by $\Delta_{\phi\psi}^{-1}:\Ran{\Delta_{\phi\psi}}\oplus\left(s^{\widetilde{\nalgebra}}(\phi)s^{\widetilde{\nalgebra}^\prime}(\psi)\hilbert_{\phi\psi}\right)$ the operator defined by $$\Delta_{\phi\psi}^{-1}=\widetilde{\Delta}_{\phi\psi}^{-1}\oplus 0.$$
\end{notation}

Just stressing what we said in the beginning of this section, the great advantage of using the balanced weight is to have all the necessary ingredients for the construction as immediate consequences of the usual modular theory. One can check that this construction leads to the same more workable definition throughout the GNS-vector $\Phi$ and $\Psi$ for the weights $\phi$ and $\psi$, respectively, namely,
$$S_{\phi,\psi}\left(A\Phi+\left(\mathbbm{1}-s^{\nalgebra^\prime}(\psi)\right)x\right)=s^{\nalgebra^\prime}(\psi)A^\ast \Psi \, , \qquad x\in \eta_\theta\left(\mathfrak{N}_\psi\cap\mathfrak{N}_\phi^\ast\right).$$

This operator is well defined, since the decomposition on the left-hand side is in orthogonal subspaces, as a consequence of Proposition \ref{supportcyclicseparating}. It is also interesting to mention that, by virtue of Proposition \ref{supportcyclicseparating}, this definition is equivalent to the one in which we take the quotient by $N_\phi$, what makes the GNS-vector $\Psi$ cyclic for $\nalgebra/N_\psi$. Moreover, for $S_{\phi,\psi}^\ast$, $\Phi$ is cyclic for $\nalgebra^\prime/N_\phi$.

In the above definition Proposition \ref{deltaKernel} is trivial. This definition is used, for example, in \cite{Araki82} by Araki and Masuda in order to construct the noncommutative $L_p$-space based on the Hilbert Space.

Again, notice that the Relative Modular Theory is closely related with Relative Hamiltonians and Relative Entropy, as seen in \cite{Araki76} and \cite{Araki73}, Connes cocycle, as one can see in \cite{Takesaki2003}.
% flatex input end: [Chapter4/chapter4.tex]
%Tomita Takesaki
% flatex input: [Chapter5/chapter5.tex]
%******************************** Fifth Chapter *************************

\chapter{Noncommutative $L_p$-Spaces }  %Title of the Fifth Chapter
\label{chapNCLp}

Noncommutative $L_p$-spaces are analogous to the Banach spaces of the $p$-integrable functions with respect to a measure. The study of these spaces goes back to the works of Segal, \cite{Segal53}, and Dixmier, \cite{Dixmier53}, which depend on the existence of a normal faithful semifinite trace. It was just 25 years later that Haagerup in \cite{haagerup79} proposed a generalization of the Segal-Dixmier $L_p$-spaces which included the type III von Neumann algebras. As a consequence of a Hilsum's paper, \cite{hilsum81}, which answers to a question on spacial derivatives raised by Connes, Araki and Masuda could propose a definition, equivalent to that one proposed by Haagerup, of noncommutative $L_p$-space based just in the Hilbert space of a concrete von Neumann algebra. 

We are going to present the three constructions of noncommutative $L_p$-spaces mentioned in the above paragraph. The Segal-Dixmier one is presented in much more detail, since we will use this construction in our results in the last chapter. The Haagerup and Araki-Masuda constructions will be just presented without proofs, first because there is a immediate connection with our results in the last chapter of the thesis, because of the appearance of a trace in both, second because very interesting discussions can be made based on that.

Our interest in these spaces for a class of perturbations can be justified on two well-known facts for classical $L_p$-spaces that still hold in the noncommutative case: they admit unbounded functions (operators) and have a useful duality property.

%**************************** %First Section  ******************************

\section{Measurability with Respect to a Trace}

\ifpdf
\graphicspath{{Chapter1/Figs/Raster/}{Chapter1/Figs/PDF/}{Chapter1/Figs/}}
\else
\graphicspath{{Chapter1/Figs/Vector/}{Chapter1/Figs/}}
\fi

In the first section of this chapter, we present the useful theory of noncommutative measure whith respect to a normal faithful semifinite trace on a von Neumann algebra, which is the basis for the Segal-Dixmier noncommutative $L_p$-spaces (and for noncommutative geometry, but it is another topic). We use \cite{terp81} very often in this section.

%********************************** %First Section  **************************************

Henceforth, we will use $\tau$ and call it simply a trace in the text of next chapters, meaning a normal faithful semifinite trace. It is important to note that supposing the existence of such a trace restricts our options of algebras to the semifinite ones.

\begin{definition}
Let $\nalgebra$ be a von Neumann algebra, $\tau$ a normal faithful semifinite trace, and $\varepsilon, \delta>0$. Define $$D(\varepsilon,\delta)=\left\{A\in \nalgebra_\eta \ \middle| \ \exists p\in \nalgebra_p \textrm{ such that } p\hilbert\subset\Dom{A}, \ \|Ap\|\leq\varepsilon \textrm{ and } \tau(\mathbbm{1}-p)\leq\delta\right\}.$$\footnote{see Definition}
\end{definition}

Some important properties we are about to present rely on the equivalence of projections. Two projections $p, \, q \in \nalgebra_p$ are said equivalent, in the sense of Murray and von Neumann, if there exists a partial isometry $u\in \nalgebra$ such that $u^\ast u=p$ and $uu^\ast=q$. This also defines a partial order in $\nalgebra_p$, namely $p\preceq q$ if, and only if $p\sim q^\prime \leq q$, \ie, $p$ is equivalent to a subprojection of $q$. One of these important equivalences is $s_R^\nalgebra(A)\sim s_L^\nalgebra(A)$, that is, the right and left support of an operator are equivalent projections. This standard result is a consequence of the polar decomposition, since for the partial isometry, $u$, in the polar decomposition of $A$ holds $uu^\ast= s_L^\nalgebra(A)$ and $u^\ast u= s_R^\nalgebra(A)$. We can use it to deduce another important equivalence, as follows
$$s^\nalgebra_L(q(\mathbbm{1}-p))=proj[\Ran\left(q(\mathbbm{1}-p)\right)]=q-(p\wedge q).$$
$$proj[\Ker\left(q(\mathbbm{1}-p)\right)]=p+(\mathbbm{1}-p)\wedge (\mathbbm{1}-q)\Rightarrow s^\nalgebra_R(q(\mathbbm{1}-p))=\mathbbm{1}-p-(\mathbbm{1}-p\vee q)= (p\vee q)-p.$$

Then $(p\vee q)-p \sim q-(p\wedge q)$. It also has an interesting consequence: if two projections $p$ and $q$ are such that $(p\wedge q)=0$ (which means the intersection of the ranges is the null space), then $p\preceq \mathbbm{1}-q$.

\begin{proposition}
\label{Dsumandproduct}
Let $\varepsilon_1,\varepsilon_2, \delta_1, \delta_2>0$. Then
	\begin{enumerate}[(i)]
		\item $D(\varepsilon_1,\delta_1)+D(\varepsilon_2,\delta_2)\subset D(\varepsilon_1+\varepsilon_2,\delta_1+\delta_2)$;
		\item $D(\varepsilon_1,\delta_1)D(\varepsilon_2,\delta_2)\subset D(\varepsilon_1\varepsilon_2,\delta_1+\delta_2)$.
	\end{enumerate}
\end{proposition}
\begin{proof}
	$(i)$ Let $A_i\in D(\varepsilon_i,\delta_i)$, $i=1,2$. By definition, there exist projections $p_i$ such that $p_i\hilbert\subset\Dom{A_i}$,$\|Ap_i\|\leq\varepsilon_i$ and $\tau(\mathbbm{1}-p_i)\leq\delta_i$. Taking $q=p_1\wedge p_2$, we get 
	$$q\hilbert \subset\Dom{A_1}\cap\Dom{A_2}=\Dom{A_1+A_2},$$
	$$\|(A_1+A_2)q\|\leq \|A_1 q\|+\|A_2 q\|\leq \|A_1 p_1\|+\|A_1 p_2\|\leq\varepsilon_1+\varepsilon_2 \, \textrm{ and }$$
	$$\begin{aligned}
		\tau(\mathbbm{1}-q) &= \tau(\mathbbm{1}-p_1\wedge p_2)\\
	&=\tau\bigl((\mathbbm{1}-p_1)\vee(\mathbbm{1} -p_2)\bigr)\\
	&\leq \tau(\mathbbm{1}-p_1)+\tau(\mathbbm{1}-p_2)\leq \delta_1+\delta_2.
	\end{aligned}$$
	
	$(ii)$ Let $A_i$, $p_i$, $i=1,2$, as before. Note that $(\mathbbm{1}-p_1)A_2p_2\in \nalgebra$, then we can define $r=[\Ker{\left((\mathbbm{1}-p_1)A_2 p_2\right)}]$.
	
	It follow from this definition that, for every $x \in r\hilbert$, $A_2p_2 x\in p_1\hilbert\subset\Dom{A_1}$. This means $r\hilbert\subset\Dom{A_1A_2p_2}$, which would imply $(p_2\wedge r)\hilbert\subset \Dom{A_1A_2}$ and 
	$$\begin{aligned}
	\|A_1A_2(p_2\wedge r)\|	&=\|A_1p_1A_2p_2(p_2\wedge r)\|\\
	&\leq \|A_1p_1\|\|A_2p_2\|\\
	&=\varepsilon_1\varepsilon_2.
	\end{aligned}$$
	
	Now, $\mathbbm{1}-r=s^\nalgebra_R\left((\mathbbm{1}-p_1)A_2p_2\right) \sim s^\nalgebra_L\left((\mathbbm{1}-p_1)A_2p_2\right)\leq \mathbbm{1}-p_1$, and thus
	$$\begin{aligned}
	\tau(\mathbbm{1}-p_2\wedge r)&=\tau\bigl((\mathbbm{1}-p_2)\vee(\mathbbm{1}-r)\bigr)\\
	&\leq \tau(\mathbbm{1}-p_2)+\tau(\mathbbm{1}-p_1)\\
	&\leq\delta_1+\delta_2.
	\end{aligned}$$
	
\end{proof}
This proposition shows a connection between these sets and the operation in the algebra. In fact, the properties above will be used soon to prove that the family $\left\{\nalgebra_\tau \cap D(\varepsilon,\delta)\right\}_{\varepsilon>0, \delta>0}$ are a basis of neighbourhoods of zero for a vector topology.

\begin{proposition}
	\label{proposition d<->tau}
	Let $A\eta\nalgebra$ be a closed densely defined operator, $\varepsilon, \delta>0$ and \\ $E_{(\varepsilon,\infty)}=s^\nalgebra_L\left((|A|-\varepsilon)_+\right)$ \footnote{see Notation \ref{positivepart}.} the spectral projection of $|A|$ on the interval ${(\varepsilon,\infty)}$. Then,
	$$A\in D(\varepsilon,\delta) \Leftrightarrow \tau\left(E_{(\varepsilon,\infty)}\right)\leq\delta.$$
\end{proposition}
\begin{proof}
	$(\Leftarrow)$ Just take $p=\mathbbm{1}-E_{(\varepsilon,\infty)}$, then $p\hilbert \subset\Dom{|A|}$, $\|Ap\|=\|A(\mathbbm{1}-E_{(\varepsilon,\infty)})\|\leq \varepsilon$ and, by hypothesis, $\tau(\mathbbm{1}-p)\leq\delta$.
	
	$(\Rightarrow)$ Suppose $A\in D(\varepsilon,\delta)$. There exists $p\in\nalgebra_p$ such that $p\hilbert\subset\Dom{A}$, $\|Ap\|\leq \varepsilon$ and $\tau(\mathbbm{1}-p)\leq\delta$. Let us take $\left\{E_{(\lambda,\infty)}\right\}_{\lambda \in \mathbb{R}_+}$ the spectral projections of $|A|$. Notice that $\||A|E_{(\varepsilon,\infty)}x\|>\varepsilon\|E_{(\varepsilon,\infty)}x\| \ \forall x \in \hilbert$, while $\||A|p x\|\leq \varepsilon \|px\|$. It follows that $\left(E_{(\varepsilon,\infty)} \wedge p\right)=0$ or, in other words, $E_{(\varepsilon,\infty)}\preceq \mathbbm{1}-p$.
	
	Finally, the previous inequality leads to $\tau(E_{(\varepsilon,\infty)})\leq\tau(\mathbbm{1}-p)\leq \delta$.
	
\end{proof}

The next result is another one about the good-behaviour of the algebra operations, namely, some kind of continuity of the involution. 

\begin{proposition}
	\label{AinD=>A*inD}
	Let $A\eta \nalgebra$ be a closed densely defined operator.
	$$A\in D(\varepsilon,\delta)\Leftrightarrow A^\ast \in D(\varepsilon,\delta).$$
\end{proposition}
\begin{proof}
	Let $A=u|A|$  be the polar decomposition of $A$, so $|A^\ast|^2=A A^\ast=u|A|\left(u|A|\right)^\ast=u|A|^2u^\ast$.
	
	Using the previous identity, where $E_{(\varepsilon,\infty)}^{|A|}$ and $E_{(\varepsilon,\infty)}^{|A^\ast|}$  are the spectral projections of $|A|$ and $|A^\ast|$ respectively, we have $$\begin{aligned}
	\tau\left(E_{(\varepsilon,\infty)}^{|A^\ast|}\right)&=\tau\left(u E_{(\varepsilon,\infty)}^{|A|} u^\ast\right)\\
	&=\tau\left(E_{(\varepsilon,\infty)}^{|A|}\right)\\
	&\leq\delta.
	\end{aligned}$$ 
\end{proof}

\begin{definition}
	Let $\nalgebra$ be a von Neumann algebra and $\tau$ a normal faithful semifinite trace. A subspace $V\subset\hilbert$ is said to be $\tau$-dense if, for all $\delta>0$, there exists $p\in\nalgebra_p$ such that $p\hilbert\subset V$ and $\tau(\mathbbm{1}-p)<\delta$.	
\end{definition}

\begin{proposition}
	\label{pn->1} \index{$\tau$-! dense}
	Let $\nalgebra$ be a von Neumann algebra and $\tau$ a normal faithful semifinite trace. A subspace $V\subset\hilbert$ is $\tau$-dense if, and only if, there exists an increasing sequence of projections $(p_n)_{n\in\mathbb{N}}\subset\nalgebra_p$ such that $p_n\to \mathbbm{1}$ and $\tau(\mathbbm{1}-p_n)\to 0$ and $\displaystyle \bigcup_{n\in\mathbb{N}} p_n\hilbert\subset V$.
\end{proposition}

\begin{proof}
	Of course, the existence of such a sequence of projections implies the $\tau$-density of $V$.
	
	On the other hand, if $V$ is $\tau$-dense, there exists, for each $n\in\mathbb{N}$, a projection $q_n$ such that $q_n\hilbert\subset V$ and ${\tau(\mathbbm{1}-q_n)<2^{-n}}$.
	
	Now, in order to obtain an increasing sequence, define
	$\displaystyle p_n=\bigwedge_{k\geq n}q_k$.
	
	To show that $p_n\to \mathbbm{1}$, define $p=\displaystyle \bigvee_{n\in\mathbb{N}}q_n$. Then, for all $n\in \mathbb{N}$,
	\begin{equation}\begin{aligned}
	\label{increasingtau}
	\tau(\mathbbm{1}-p)&\leq \tau(\mathbbm{1}-p_n)\\ &=\tau\left(\mathbbm{1}-\bigwedge_{k\geq n}q_k\right)\\
	&=\tau\left(\bigvee_{k\geq n}(\mathbbm{1}-q_k)\right)\\
	&\leq \sum_{k\geq n}\tau(\mathbbm{1}-q_k)\\
	&\leq \sum_{k\geq n}2^{-k} \\
	&=2^{1-n},
	\end{aligned}\end{equation}
	but this is only possible if $\tau(\mathbbm{1}-p)=0$ and since $\tau$ is faithful, $p=\mathbbm{1}$. It follows by Vigier's Theorem\footnote{see Theorem \ref{vigier}.} that $p_n\to\mathbbm{1}$ in the SOT.
	
\end{proof}

\begin{corollary}
	\label{taudense=>dense}
	If $V\subset \hilbert$ is a $\tau$-dense subspace, $V$ is dense in $\hilbert$.
\end{corollary}

\begin{corollary}
	\label{intersec2taudense}
	Let $V_1,V_2 \subset\hilbert$ be $\tau$-dense subspaces. Then $V_1\cap V_2$ is $\tau$-dense. 	
\end{corollary}
\begin{proof}
	First, although we are going to present a proof based on the previous proposition, it is quite interesting to notice that this result has already been proved. In fact, the proof can be found in the proof of Proposition \ref{Dsumandproduct}.
	
	Let $(p^i_n)_{n\in\mathbb{N}}\in \nalgebra_p$, $i=1,2$, be such that $\tau(\mathbbm{1}-p^i_n)\to0$ and $\displaystyle \bigcup_{n\in\mathbb{N}} p^i_n\hilbert\subset V_i$. Define $q_n=p^1_n\wedge p^2_n$. Then
	$$\begin{aligned}
	\tau(\mathbbm{1}-q_n) &= \tau(\mathbbm{1}-p^1_n\wedge p^2_n)\\
	&=\tau\left((\mathbbm{1}-p^1_n)\vee(\mathbbm{1} -p^2_n)\right)\\
	&\leq \tau(\mathbbm{1}-p^1_n)+\tau(\mathbbm{1}-p^2_n) \to 0.\\
	\end{aligned}$$
	
	Furthermore, $q_n=p^1_n\wedge p^2_n$ is the projection in the intersection of the image of \mbox{$p^1_n$ and $p^2_n$}, which is a subset of $V_1\cap V_2$ for every $n\in\mathbb{N}$.
	
\end{proof}

The two last corollaries show, respectively, that $\tau$-denseness is stronger than the usual density in the Hilbert space, and that $\tau$-dense subspaces are stable under finite intersections. This last result is very important when we are looking for a candidate of algebra that includes unbounded operators. Indeed, we are interested in densely defined operators, but the sum and product of such two operators is only defined in the intersection of their domains. Hence, Corollary \ref{intersec2taudense} warranties that the sum and product of two operators with $\tau$-dense domains is again an operator with $\tau$-dense domain, thus, densely defined thanks to Corollary \ref{taudense=>dense}. A question that is much more delicate is about the convergence of a series of operators with $\tau$-dense domains, since the domain of such series would be the intersection of all the infinity domains.

\begin{proposition}
	Let $\nalgebra\subset B(\hilbert)$ be a von Neumann algebra and let \mbox{$A_1, A_2\in \nalgebra_\eta$} be\, closed \, (densely \, defined) \, operators \, such \, that\, there\, exists\, a \,$\tau$-dense\, subspace \mbox{$V\subset \Dom{A_1}\cap\Dom{A_2}$} where $\left. A_1\right|_V=\left. A_2\right|_V$. Then $A_1=A_2$.
	
\end{proposition}

\begin{proof}
	
	Note that the balanced weight $\theta=\theta_{\tau,\tau}$\footnote{see Definition \ref{balancedweight}.} is a normal faithful semifinite trace in $M_{2\times2}(\nalgebra)$.
	
	Let $p_i$ be the projection on the graph of $A_i$, $i=1,2$. Notice that $$M_{2\times2}(\nalgebra)^\prime=\left\{\begin{pmatrix} A^\prime & 0 \\ 0 & A^\prime\end{pmatrix}\middle| A^\prime \in \nalgebra^\prime\right\} \textrm{ and }$$
	$$
	\begin{pmatrix} A^\prime & 0 \\ 0 & A^\prime\end{pmatrix}\begin{pmatrix} x  \\  A_i x\end{pmatrix}=
	\begin{pmatrix} A^\prime x  \\  A^\prime A_i x\end{pmatrix}=
	\begin{pmatrix} A^\prime x  \\  A_i A^\prime x\end{pmatrix}\in \Gamma(A_i) \qquad \forall x\in \Dom{A^\prime}
	$$
	hence, $p_i \in M_{2\times2}(\nalgebra)^{\prime\prime}=M_{2\times2}(\nalgebra)$.

	By hypothesis, there exists $p\in \nalgebra_p$ such that $p\hilbert\subset V$ and $\tau(\mathbbm{1}-p)<\delta$.
	
	Take $r=s^\nalgebra_L(p_1-p_2)=[\Ran(p_1-p_2)]$ and notice $r\wedge p=0$. Thus $r\leq \mathbbm{1}-(p\oplus p)$ and it follows that $\theta(r\oplus r)\leq\theta\left((\mathbbm{1}-p)\oplus(\mathbbm{1}-p)\right)=\tau(\mathbbm{1}-p)+\tau(\mathbbm{1}-p)<2\delta$. Since $\delta>0$ is arbitrary, $\theta(r\oplus r)=0\Rightarrow p_1=p_2$.
	
\end{proof}

\begin{definition}
	\index{$\tau$-! measurable}
	Let $\nalgebra$ be a von Neumann algebra and $\tau$ be a normal faithful semifinite trace. A closed (densely defined) operator $A\in \nalgebra_\eta$ is said $\tau$-measurable if $\Dom{A}$ is $\tau$-dense. We denote by $\nalgebra_\tau$ the set of all $\tau$-measurable operators.
\end{definition}

Notice that by the previous proposition, if $A$ is a $\tau$-measurable operator and $B$ extends $A$, we must have $A=B$. This, in turn, implies that a $\tau$-measurable symmetric operator is self-adjoint.

\begin{definition}
	\label{defpremeasurable}\index{$\tau$-! premeasurable}
	Let $\nalgebra$ be a von Neumann algebra and $\tau$ be a normal faithful semifinite trace. An operator $A\eta\nalgebra$ is said $\tau$-premeasurable if, $\forall \delta>0$, there exists $p\in\nalgebra_p$ such that $p\hilbert \subset \Dom{A}$, $\|Ap\|<\infty$ and $\tau(1-p)\leq \delta$.  
\end{definition}

An equivalent way to define a $\tau$-premeasurable operator relies on $D(\varepsilon,\delta)$: $A$ is $\tau$-premeasurable if, and only if, $\forall \delta>0$, there exists $\varepsilon>0$ such that $A\in D(\varepsilon,\delta)$.

Another interesting thing to notice is that a $\tau$-premeasurable operator is densely defined since $\Dom{A}$ must be $\tau$-dense.

\begin{proposition}
	\label{equivalencetaumeasurability}
	Let $\nalgebra$ be a von Neumann algebra, $\tau$ a normal faithful semifinite trace, $A \eta \nalgebra$ a closed densely defined operator, and $\left\{E_{(\lambda,\infty)}\right\}_{\lambda\in\mathbb{R}_+}$ the spectral decomposition of $|A|$. The following are equivalent:
	\begin{enumerate}[(i)]
		\item $A$ is $\tau$-measurable;
		\item $|A|$ is $\tau$-measurable;
		\item $\forall \delta>0 \ \exists \varepsilon>0$ such that $A\in D(\varepsilon,\delta)$;
		\item $\forall \delta>0 \ \exists \varepsilon>0$ such that $\tau\left(E_{(\varepsilon,\infty)}\right)<\delta$;
		\item $\displaystyle \lim_{\lambda\to\infty}\tau\left(E_{(\lambda,\infty)}\right)=0$;
		\item $\exists \lambda_0>0$ such that $\tau\left(E_{(\lambda_0,\infty)}\right)<\infty$.
	\end{enumerate}
\end{proposition}
\begin{proof}
	
	$(i)\Leftrightarrow(iii)$ Simply rewrite the definition, as mentioned before;
	
	$(i)\Leftrightarrow(ii)$ Just notice that $A\in D(\varepsilon,\delta)\Leftrightarrow |A|\in D(\varepsilon,\delta)$, which follow from Proposition \ref{proposition d<->tau};
	
	$(iii)\Leftrightarrow(iv)$ It is Proposition \ref{proposition d<->tau};
	
	$(iv)\Leftrightarrow(v)$ $\left(E_{(\lambda,\infty)}\right)_{\lambda\in\mathbb{R}}$ is a decreasing net of projections. Let $\delta>0$, there exists $\varepsilon>0$ such that $\tau\left(E_{(\varepsilon,\infty)}\right)<\delta$. Hence, for every $\lambda>\varepsilon$, $\tau\left(E_{(\varepsilon,\infty)}\right)<\tau\left(E_{(\lambda,\infty)}\right)<\delta$ and the other implication is analogous;
	
	$(v)\Rightarrow(vi)$ is obvious;
	
	$(vi)\Rightarrow (v)$ Let $\lambda_0>0$ be such that $\tau\left(E_{(\lambda_0,\infty)}\right)<\infty$. Define the increasing upper bounded net $\left(E_{(\lambda_0,\infty)}-E_{(\lambda,\infty)}\right)_{\lambda>\lambda_0}$, notice $\left(\tau\left(E_{(\lambda_0,\infty)}-E_{(\lambda,\infty)}\right)\right)_{\lambda>\lambda_0}\subset \mathbb{R}_+$ is also an increasing net and $\left(\tau\left(E_{(\lambda,\infty)}\right)\right)_{\lambda>\lambda_0}\subset \mathbb{R}_+$ is a decreasing net, hence, both nets have limits. By normality of $\tau$
	
	$$\begin{aligned}
	\lim_{\lambda\to \infty}\tau\left(E_{(\lambda_0,\infty)}-E_{(\lambda,\infty)}\right)&=\sup_{\lambda>\lambda_0}\tau\left(E_{(\lambda_0,\infty)}-E_{(\lambda,\infty)}\right)\\
	&=\tau\left(E_{(\lambda_0,\infty)}-\bigwedge_{\lambda >\lambda_0}E_{(\lambda,\infty)}\right)\\
	&=\tau\left(E_{(\lambda_0,\infty)}\right).\\
	\end{aligned}$$
	
\end{proof}

The last theorem is in the very core of the theory on Noncommutative Integrations. It was Segal in \cite{Segal53} who fist noticed some of these ideas. 

\section{The Segal-Dixmier Noncommutative $L_p$-Spaces}

Hitherto, we have presented the theory of noncommutative measure which enables us to start presenting now the first approach to noncommutative spaces. Our aim here is to prove useful results, in general, based on the standard results in Functional Analysis. Minkowski's and H\"older's inequalities, in particular, play a central role, as well as, the well-known duality between spaces with H\"older conjugated\footnote{$p,\,q \in \mathbb{R}$ are said H\"older conjugated if $\frac{1}{p}+\frac{1}{q}=1$.} indices. We highlight that generalization of the H\"older inequality will be in the core of our results on perturbation of KMS states.
 
\begin{proposition}
	$\nalgebra_\tau$ provided with the usual scalar operations and involution, and the following vector operations is a $\ast$-algebra:
	\begin{enumerate}[(i)]
		\item $A\bm{+}B=\overline{A+B}$;
		\item $A\bm{\times}B=\overline{AB}$.
	\end{enumerate}
\end{proposition}

\begin{proof}
	First of all, so that the previous definition can make sense, we must guarantee that for every $A,B\in \nalgebra_\tau$, $A+B$ and $AB$ are closable. In fact, for $A,B\in \nalgebra_\tau$ and for every $\delta>0$ the previous proposition guarantees that there exist $\varepsilon_A, \varepsilon_B>0$ such that
	$$\begin{aligned}
	A\in D\left(\varepsilon_A,\frac{\delta}{2}\right) &\Rightarrow A^\ast\in D\left(\varepsilon_A,\frac{\delta}{2}\right),\\
	B\in D\left(\varepsilon_B,\frac{\delta}{2}\right) &\Rightarrow B^\ast\in D\left(\varepsilon_B,\frac{\delta}{2}\right),
	\end{aligned}$$
	where the implication is a consequence of Proposition \ref{AinD=>A*inD}. Also, it means that $A^\ast \in \nalgebra_\tau$. Then, by Proposition \ref{Dsumandproduct},
	$$\begin{aligned}
	A^\ast+B^\ast &\in D\left(\varepsilon_A+\varepsilon_B,\delta\right),\\
	A^\ast B^\ast &\in D\left(\varepsilon_A\varepsilon_B,\delta\right).\\
	\end{aligned}$$
	
	These inclusions, as commented after Definition \ref{defpremeasurable}, means that $A^\ast+B^\ast$ and $B^\ast A^\ast$ are $\tau$-premeasurables, which implies they are densely defined. Hence $A+B\subset (A^\ast+B^\ast)^\ast$ and $AB\subset (B^\ast A^\ast)^\ast$ admit closed extensions, so they are closable.
	
	Now, $\overline{A+B}$ and $\overline{AB}$ are closed densely defined operators, for which condition $(iii)$ in the Proposition \ref{equivalencetaumeasurability} holds thanks to Proposition \ref{Dsumandproduct}, hence $A\bm{+}B, A\bm{\times}B \in \nalgebra_\tau$. 
	
	Just remains to prove the identities for $\bm{+}, \bm{\times}$ and their relations with $*$. Let $A,B,C \in \nalgebra_\tau$, then all operators in the following equalities are closed. As they coincide on a $\tau$-dense subspace, the equalities holds:
	$$\begin{aligned}
	(A\bm{+}B)\bm{+}C&=A\bm{+}(B\bm{+}C),	&\qquad (A\bm{\times}B)\bm{\times}C&=A\bm{\times}(B\bm{\times}C),\\
	(A\bm{+}B)\bm{\times}C&=A\bm{\times}C\bm{+}B\bm{\times}C, &\qquad C\bm{\times}(A\bm{+}B)&=C\bm{\times}A\bm{+}C\bm{\times}B,\\
	(A\bm{+}B)^\ast &= A^\ast\bm{+}B^\ast, &\qquad (A\bm{\times}B)^\ast&=B^\ast\bm{\times}A^\ast.\\
	\end{aligned}$$
	
\end{proof}

From now on, we will differentiate the symbols $\bm{+}, \bm{\times}$ and the usual sum and multiplication of operators only if it may cause a misunderstanding.

We have mentioned in the previous sections that we were interested in defining a vector topology on some subspace of unbounded operators. This moment has finally arrived.

\begin{proposition}
	\label{measurablealgebra}
	$\nalgebra_\tau$ is a complete Hausdorff topological $\ast$-algebra with respect to the topology generated by the system of neighbourhoods of zero $\left\{\nalgebra_\tau \cap D(\varepsilon,\delta)\right\}_{\varepsilon>0, \delta>0}$. Furthermore, $\nalgebra$ is dense in $\nalgebra_\tau$ in this topology. We will denote the balanced absorbing neighbourhood of zero by $N(\varepsilon,\delta)=\nalgebra_\tau \cap D(\varepsilon,\delta)$.
\end{proposition}

\begin{proof}
	It is easy to verify $N(\varepsilon,\delta)$ is in fact balanced and absorbing, furthermore, Proposition \ref{Dsumandproduct} implies that, for every $\varepsilon_1,\varepsilon_2>0$ and $\delta_1, \delta_2>0$, $N\left(\frac{\varepsilon_1}{2},\frac{\delta_1}{2}\right)+N\left(\frac{\varepsilon_1}{2},\frac{\delta_1}{2}\right)=N\left(\varepsilon_1,\delta_1\right)$ and $N\left(\min\{\varepsilon_1,\varepsilon_2\},\min\{\delta_1,\delta_2\}\right)\subset N\left(\varepsilon_1,\delta_1\right)\cap N\left(\varepsilon_2,\delta_2\right)$. Hence, there exists a unique vector topology such that $\left\{N(\varepsilon,\delta)\right\}_{\varepsilon>0, \delta>0}$ is a basis of neighbourhoods of zero.
	
	In order to show that this topology is Hausdorff, let $A,B \in \nalgebra_\tau$ be two distinct operators. For each $\delta>0$, define $$\varepsilon_\delta=\inf\left\{\epsilon\in\mathbb{R}_+\middle| A-B \in N(\epsilon,\delta)\right\}.$$
	Notice that there exists $\tilde{\delta}>0$ such that $\varepsilon_{\tilde{\delta}}>0$, because if we had $\varepsilon_\delta=0$ for every $\delta>0$ there would exist projections $p_\delta^n$ for each $n\in \mathbb{N}$ such that $\|(A-B)p_\delta^n\|<\frac{1}{n}$ and $\tau(1-p_\delta^n)\leq \delta$. Defining $\displaystyle p_\delta=\inf_{n\in\mathbb{N}} p_\delta^n$, we would have $\|(A-B)p_\delta\|=0$ and $\tau(1-p_\delta)\leq \delta$, which implies that $A$ and $B$ coincide on a $\tau$-dense subspace, but this is not possible since $A\neq B$. It is easy to see that $B\notin A+N\left(\frac{\varepsilon_{\tilde{\delta}}}{2},\tilde{\delta}\right)$.
	
	Let us prove that the two neighbourhoods $A+N\left(\frac{\varepsilon_{\tilde{\delta}}}{4},\frac{\tilde{\delta}}{2}\right)$ and $B+N\left(\frac{\varepsilon_{\tilde{\delta}}}{4},\frac{\tilde{\delta}}{2}\right)$ of $A$ and $B$, respectively, are disjoint. In fact, suppose they are not. Then there exist $T_1, T_2 \in N\left(\frac{\varepsilon_{\tilde{\delta}}}{4},\frac{\tilde{\delta}}{2}\right)$ such that $A+T_1=B+T_2 \Rightarrow B=A+T_1-T_2 \in A+N\left(\frac{\varepsilon_{\tilde{\delta}}}{2},\tilde{\delta}\right)$, but this is not possible.
	
	Of course, the vector space operations are continuous, the involution is continuous thanks to Proposition \ref{AinD=>A*inD}. For the product, consider $A,B\in\nalgebra_\tau$ and $AB+N(\varepsilon,\delta)$ a basic neighbourhood of $AB$. There exists $\alpha>\varepsilon$ such that $A \in N\left(\alpha,\frac{\delta}{6}\right)$ and $B\in N\left(\alpha,\frac{\delta}{6}\right)$, then, if $\tilde{A}\in A+N\left(\frac{\varepsilon}{3\alpha},\frac{\delta}{6}\right)$ and $\tilde{B}\in B+N\left(\frac{\varepsilon}{3\alpha},\frac{\delta}{6}\right)$, we have
	$$\begin{aligned}
	AB-\tilde{A}\tilde{B}&=-(A-\tilde{A})(B-\tilde{B})+A(B-\tilde{B})+(A-\tilde{A})B\\
	&\in N\left(\frac{\varepsilon^2}{9\alpha},\frac{\delta}{3}\right)+
	N\left(\frac{\varepsilon}{3},\frac{\delta}{3}\right)+
	N\left(\frac{\varepsilon}{3},\frac{\delta}{3}\right)
	\subset N(\varepsilon,\delta).
	\end{aligned}$$
	
	For the density, we use Proposition \ref{pn->1}. Let $A\in \nalgebra_\tau$ and $(p_n)_{n\in\mathbb{N}}\subset \nalgebra_p$ be an increasing sequence of projections such that $p_n \xrightarrow{SOT} \mathbbm{1}$, $\tau(\mathbbm{1}-p_n)\to 0$ and $\displaystyle \bigcup_{n \in \mathbb{N}}p_n\hilbert\subset \Dom{A}$. Thus $(Ap_n)_{n\in\mathbb{N}}\in\nalgebra$ and $Ap_n\to A$, since the product is continuous and $p_n\xrightarrow[]{\tau} \mathbbm{1}$.
	
	It just remains to show that $\nalgebra_\tau$ is complete. Notice that the topology has a countable basis of neighbourhoods of zero, since $\left\{N\left(\frac{1}{n},\frac{1}{m}\right)\right\}_{n\in\mathbb{N}^\ast,m\in\mathbb{N}^\ast}$ is such a countable basis. This means we just have to prove every Cauchy sequence is convergent.
	
	Let $(A_n)_{n\in\mathbb{N}} \in \nalgebra_\tau$ be a Cauchy sequence. Since $\overline{\nalgebra}^\tau=\nalgebra_\tau$, there exists $(A^\prime_n)_{n\in\mathbb{N}} \subset \nalgebra$ such that $A_n-A^\prime_n\in N\left(\frac{1}{n},\frac{1}{n}\right)$, hence $(A^\prime_n)_{n\in\mathbb{N}}$ is a Cauchy sequence.
	
	For each $k\in \mathbb{N}$ there exists $N_k\in \mathbb{N}$ such that, $\forall n,m\geq N_k$, there exists a projection $q_k$ with $\|(A^\prime_m-A^\prime_n)q_k\|\leq 2^{-k}$ and $\tau(1-q_k)\leq 2^{-k}$.
	
	Define $p=\displaystyle \bigvee_{n\in\mathbb{N}}q_n$. By simply repeating calculation (\ref{increasingtau}), we see that $(p_n)_{n\in\mathbb{N}}\subset\nalgebra_p$, defined by $\displaystyle p_n=\bigwedge_{k\geq n}q_k$, is an increasing sequence of projections such that $\tau(1-p_k)\leq 2^{1-k}$.
	
	Let's now define $\displaystyle \Dom{A}=\bigcup_{k\in\mathbb{N}}p_k\hilbert$. Notice that $x\in \Dom{A}$ implies the existence of $j\in\mathbb{N}$ for which $x\in p_j\hilbert\subset p_l\hilbert$ for all $l\geq j$. Since, for any $k>j$ and $n,m \geq N_k$, we have $$\begin{aligned}
	\|(A^\prime_n-A^\prime_m)x\|&= \|(A^\prime_n-A^\prime_m)p_k x\|\\
	&\leq\|(A^\prime_n-A^\prime_m)p_k\| \| x\|\\
	&\leq \|(A^\prime_n-A^\prime_m)q_k\| \| x\|\\
	& \leq 2^{-k}\|x\|.
	\end{aligned}$$
	Hence $(A^\prime_nx)_{n\geq j} \subset \hilbert$ is a Cauchy sequence. Then, we can define $A:\Dom{A}\to\hilbert$ by 
	$$ Ax=\lim_{j<n\to \infty} A^\prime_n x.$$
	
	Since $\Dom{A}$ is $\tau$-dense and $A\eta \nalgebra$ by construction, it remains to show it is closable and $A_n\xrightarrow{\tau} A$. In order to prove it is closable, notice that $(A^{\prime \ast}_n)_{n\in\mathbb{N}}\in\nalgebra$ is again a Cauchy sequence, so we can repeat the previous construction to obtain $B:\Dom{B}\to\hilbert \in\nalgebra_\tau$ defined by $Bx=\lim A^{\prime \ast}_n x$. Hence, $\forall x \in \Dom{A}$ and $\forall y \in \Dom{B}$,
	$$\ip{Ax}{y}=\lim_{n\to\infty}\ip{A^\prime_nx}{y}	=\lim_{n\to \infty}\ip{x}{A_n^{\prime \ast}y}	=\ip{x}{By},$$
	which means $A\subset B^\ast$. Hence, $A$ is closable and, by Proposition \ref{equivalencetaumeasurability}, $\overline{A}\in \nalgebra_\tau$.
	
	Finally, for the convergence notice that for every $\varepsilon, \delta>0$ there exists $k\in \mathbb{N}$ such that $2^{-k}<\varepsilon$ and $2^{1-k}<\delta$, so, $\tau(1-p_k)<\delta$ and, $\forall n\geq N_k$, $$\begin{aligned}
	\|(\overline{A}-A^\prime_n)p_k\|&=\sup_{\|x\|\leq1}\left\|(A-A^\prime_n)p_k x\right\|\\
	&=\sup _{\|x\|\leq1}\left\|\lim_{m\to\infty} A^\prime_m p_k x -A^\prime_n p_k x\right\|\\
	&=\sup _{\|x\|\leq1}\lim_{m\to\infty}\left\| A^\prime_m p_k x -A^\prime_n p_k x\right\|\\
	&\leq\sup _{\|x\|\leq1}\lim_{m\to\infty}\left\| (A^\prime_m -A^\prime_n) p_k \right\|\|x\|\\
	&\leq\lim_{m\to\infty}\left\| (A^\prime_m -A^\prime_n) p_k \right\|\\
	&\leq 2^{-k}\\
	&<\varepsilon.\\
	\end{aligned}$$
	
\end{proof}

It is interesting to notice that analyticity pervades almost every subject in von Neumann algebras. In this section, we will use the linearity and continuity of the trace, Functional Calculus and Spectral Theory to take advantage of the well-known rigid behaviour of analytic functions to create the aforesaid inequalities. So, now, we will start a digression about the multi-variable Three-Line Theorem, since its proof is not too easy to find in standard literature. Here we followed reference \cite{Araki73}, although  prefer to create a proof based on the Phragm\'en-Lindel\"of theorem.

%### It is a classic theorem, but I wrote this proof alone and have never seen it stated like that.

%Some on the topic can be found in Rudin's Real and Complex Analysis, pag. 256
\begin{theorem}[Phragm\'en-Lindel\"of]\index{theorem! Phragm\'en-Lindel\"of}
	\label{PLT}
	Let $S=\{z\in \mathbb{C} \ | \ a <\Re{z} < b\}$, $a<b \in \mathbb{R}$, and let $f:\overline{S}\to\mathbb{C}$ be a continuous function, analytic on the strip $S$, such that $|f(z)|\leq M$ for all $z \in \partial S$. Suppose there exists a sequence of functions $g_n:\overline{S} \to \mathbb{C}$, analytic in $S$, such that \begin{enumerate}[(i)]
		\item $\displaystyle g_n(z) \xrightarrow{n\to \infty} 1 \quad \forall z \in \overline{S};$\\
		\item $\displaystyle R_n\doteq \max\left\{\sup_{y\in\mathbb{R}}{|g_n(a+\iu y)},\sup_{y\in\mathbb{R}}{|g_n(b+\iu y)}\right\}\xrightarrow{n\to \infty} 1;$\\
		\item $\displaystyle \lim_{|y|\to \infty}\left(\sup_{a\leq x\leq b}{|g_n(x+\iu y)f(x+\iu y)|}\right)=0.$\\
	\end{enumerate}
	Then $|f(z)| \leq M$ for all $z\in \overline{S}$.
\end{theorem}
\begin{proof}
	By hypothesis, there exists $K\in \mathbb{R}$ such that \mbox{$\displaystyle \sup_{a\leq x\leq b}{|g_n(x+\iu y)f(x+\iu y)|}< \frac{M}{2}$} for $|y|\geq K$. Now, by the Maximum Modulus Principle, 
	$$|g_n(z)f(z)| \leq M R_n \quad \forall z\in \{w\in \mathbb{C} \ | \ a \leq\Re{w} \leq b \ \textrm{ and } -K\leq \Im{w}<K\}.$$
	From the two inequalities, it follows that, for every $z\in\overline{S}$
	$$|f(z)|=\lim_{n\to\infty}|g_n(z)f(z)|\leq \lim_{n\to\infty}M R_n =M.$$
	
\end{proof}

\begin{theorem}[Multi-Variable Three-Line Theorem] \index{theorem!Multi-Variable Three-Line} \label{TLT} \ \hspace{\fill} \linebreak	Let $C\subset \mathbb{R}^n$ be convex set, $S=\{z=(x_1+\iu x_1,\ldots,x_n+\iu y_n)\in \mathbb{C}^n \ | \ \Re{z}\in C \}$, and let $f:\overline{S}\to\mathbb{C}$ be a continuous bounded function, analytic on the strip $S$. Then, for $x_1, x_2 \in \overline{S}$ and for all $1\leq t\leq 1$,
	$$\sup_{y\in\mathbb{R}^n}{|f(tx_1+(1-t)x_2+\iu y)|} \leq\left( \sup_{y\in\mathbb{R}^n}{|f(x_1+\iu y)|}\right)^t \left(\sup_{y\in\mathbb{R}^n}{|f(x_2+\iu y)|}\right)^{1-t}.$$
\end{theorem}
\begin{proof}
	First, if $f$ vanishes on $\partial S$, then $f$ vanishes in $S$ thanks to the Edge-of-the-Wedge Theorem and the Identity Theorem. Hence, the inequality holds.
	
	If $f$ doesn't vanishes in $\partial S$, we can define, for $x_1,x_2 \in C$ and $\bar{y}\in \mathbb{R}^n$, we can define $\hat{f}^{\bar{y}}_{x_1,x_2}: \{z\in \mathbb{C} \ | \ 0<\Re{z}<1\} \to \mathbb{C}$ by $$\hat{f}^{\bar{y}}_{x_1,x_2}(z)=f(zx_1+(1-z)x_2 +\iu \bar{y})\left( \sup_{y\in\mathbb{R}^n}{|f(x_1+\iu y)|}\right)^{-z}\left(\sup_{y\in\mathbb{R}^n}{|f(x_2+\iu y)|}\right)^{z-1}.$$
	which is analytic on the strip \mbox{$S_1\doteq\{z\in \mathbb{C} \ | \ 0<\Re{z}<1\}$}, bounded and continuous on $\overline{S_1}$  and satisfies $|\hat{f}_{x_1,x_2}(z)|\leq 1$ for all $z\in \partial S_1$.
	
	By Theorem \ref{PLT}, using the sequence $g_n(z)=e^{\frac{z^2}{n}}$, we have that 
	$$\left|f(zx_1+(1-z)x_2 +\iu \bar{y})\left( \sup_{y\in\mathbb{R}^n}{|f(x_1+\iu y)|}\right)^{-z}\left(\sup_{y\in\mathbb{R}^n}{|f(x_2+\iu y)|}\right)^{z-1}\right|=|\hat{f}^{\bar{y}}_{x_1,x_2}(z)|\leq 1.$$
	Taking the supremum over $\bar{y}$ we have the desired result.
	
	It remains to prove the inequality when $x_1, x_2 \in \overline{C}$. In order to prove that, let $(x_j^n)_{n\in\mathbb{N}} \subset S$ be sequences such that $x_j^n \to x_j$, $j=1,2$.
	
	Applying what we have just showed to each of the functions $\displaystyle f_m(z)= f(z) \exp{\left(\sum_{k=1}^{n}{ \frac{z_k^2}{m}}\right)}$ we obtain $\left|\widehat{(f_m)}^{\bar{y}}_{x_1^n, x_2^n}(z)\right| \leq \exp{\left(\sum_{k=1}^{n}{ \frac{z_k^2}{m}}\right)}$ for every $z\in S_1$, since $x_1^n,x_2^n \in S$ for all $n \in \mathbb{N}$.
	
	Notice that the exponential decay with $\Im{z_k}$ assures that the convergence of $f_m(x_j^n+\iu y)$ to $f_m(x_j+\iu y)$ is uniform in $y$, \ie, $$ \sup_{y\in\mathbb{R}^n}{|f_m(x_j^n+\iu y)|} \xrightarrow{n\to \infty} \sup_{y\in\mathbb{R}^n}{|f_m(x_j+\iu y)|}.$$
	Thus, for every $z\in S_1$,
	$$\left|\widehat{(f_m)}^{\bar{y}}_{x_1, x_2}(z)\right| =\lim_{n\to \infty} \left|\widehat{(f_m)}^{\bar{y}}_{x_1^n, x_2^n}(z)\right|\leq \exp{\left(\sum_{k=1}^{n}{ \frac{z_k^2}{m}}\right)}$$ 
	
	Finally,
	$$\left|\hat{f}^{\bar{y},m}_{x_1, x_2}(z)\right|=\lim_{m\to \infty} \left|\widehat{(f_m)}^{\bar{y}}_{x_1, x_2}(z)\right|\leq \lim_{m\to\infty}\exp{\left(\sum_{k=1}^{n}{ \frac{z_k^2}{m}}\right)}=1 \quad \forall z \in S_1.$$
	
\end{proof}

%See Nelson's notes pag 112 and Xu's notes pag 6
\begin{lemma}
	\label{normtraceinequality}
	Let $\nalgebra$ be a von Neumann algebra, $\tau$ a normal faithful semifinite trace on $\nalgebra$, $A\in \nalgebra$ and $B\in \mathfrak{M}_\tau$. Then
	$$\left|\tau(AB)\right|\leq\tau(|AB|)\leq \|A\|\tau(|B|). $$
\end{lemma}
\begin{proof}
	First, let's prove the lemma in the case $A,B\in \mathfrak{M}_\tau^+$.
	
	Let $D=\sqrt{\|A\|\mathbbm{1}-A}$, then
	$$0\leq \left(DB^\frac{1}{2}\right)^\ast \left(DB^\frac{1}{2}\right)=B^\frac{1}{2}\left(\|A\|\mathbbm{1}-A\right)B^\frac{1}{2}.$$
	
	Now, from the trace positivity, it follows that
	$$0\leq \tau\left(B^\frac{1}{2}\left(\|A\|\mathbbm{1}-A\right)B^\frac{1}{2}\right)=\|A\|\tau(B)-\tau\left(B^\frac{1}{2}A B^\frac{1}{2}\right)=\|A\|\tau(B)-\tau\left(A B\right).$$
	
	For the general statement, let $A=u|A|$ and $B=v|B|$ be the polar decomposition of $A$ and $B$. Using the Cauchy-Schwarz inequality, we obtain
	$$	\begin{aligned}
	\left|\tau(AB)\right|^2&=\left|\tau(u|A|v|B|)\right|^2\\
	&=\left|\tau\left((|B|^\frac{1}{2}u|A|^\frac{1}{2})(|A|^\frac{1}{2}v|B|^\frac{1}{2})\right)\right|^2\\
	&\leq \tau\left(\left(|B|^\frac{1}{2}u|A|^\frac{1}{2}\right)^\ast \left(|B|^\frac{1}{2}u|A|^\frac{1}{2}\right)\right) \tau\left(\left(|A|^\frac{1}{2}v|B|^\frac{1}{2}\right)^\ast \left(|A|^\frac{1}{2}v|B|^\frac{1}{2}\right)\right)\\
	&= \tau\left(|A|^\frac{1}{2}u^\ast |B|u|A|^\frac{1}{2})\right) \tau\left(|B|^\frac{1}{2}v^\ast |A|v|B|^\frac{1}{2}\right)\\
	&= \tau\left(|B|u|A|u^\ast\right) \tau\left(|A|v|B|v^\ast\right)\\
	&= \tau\left(|B||A^\ast|\right) \tau\left(|A||B^\ast|\right).\\
	\end{aligned}$$
	Now, from the first part of the proof that holds for positive operators, we conclude that, for all $A\in\nalgebra$ and $B \in \nalgebra_\tau$,
	$$\left|\tau(AB)\right|^2\leq\tau\left(|B||A^\ast|\right) \tau\left(|A||B^\ast|\right)\leq\|A^\ast\|\tau(|B|)\|A\|\tau(|B^\ast|)=\|A\|^2\tau(|B|)^2.$$
	
\end{proof}

Let us start using what we have just presented to prove the aforementioned important inequalities and then define the noncommutative $L_p$-spaces. We refer to \cite{Ruscai72} for this proof and \cite{Nelson74} for further reading.

%See Ruskai and Nelson references. The Phragm\'en-Lindel\"of Three-Line Theorem can be found in Zigmund's "Trigonometric Series", vol II, pag 93, theorem 1.3.
\begin{theorem}[H\"older's Inequality\index{inequality! H\"older}]
	\label{holder}
	Let $\nalgebra$ be a von Neumann algebra and $\tau$ a normal faithful semifinite trace in $\nalgebra$. Let also $A,B\in\nalgebra$ and $p,\, q>1$ such that $\frac{1}{p}+\frac{1}{q}=1$, then
	$$\tau(|AB|)\leq \tau(|A|^p)^\frac{1}{p}\tau(|B|^q)^\frac{1}{q}.$$
\end{theorem}
\begin{proof}
	First, note that if $\tau\left(|A|^p\right)=0$,  $\tau\left(|A|^p\right)=\infty$, $\tau\left(|B|^q\right)=0$, or $\tau\left(|B|^q\right)=\infty$ the inequality is trivial.
	
	On the other hand, if $0<\tau\left(|A|^p\right),\tau\left(|B|^q\right)<\infty$, we are able to define, for every $n\in\mathbb{N}$, $$|A|_n^p=\frac{1}{\tau\left(|A|^p\right)}\int^{\|A\|}_\frac{1}{n}\lambda^p dE^{|A|}_\lambda \quad \textrm{ and } \quad |B|_n^q=\frac{1}{\tau\left(|B|^q\right)}\int^{\|B\|}_\frac{1}{n} \lambda^q dE^{|B|}_\lambda,$$
	where $\left\{E^{|A|}_\lambda\right\}_{\lambda\in \mathbb{R}^+}$ and $\left\{E^{|B|}_\lambda\right\}_{\lambda\in \mathbb{R}^+}$ are the spectral resolutions of $|A|$ and $|B|$, respectively. These definitions guarantee that $|A|_n^p \to \frac{|A|^p}{\tau\left(|A|^p\right)^{\frac{1}{p}}}$ and $|B|_n^q \to \frac{|B|^q}{\tau\left(|B|^q\right)^{\frac{1}{q}}}$ monotonically in the SOT. In addition, the positivity of the trace gives us that $\tau(|A|^p_n), \tau(|B|^q_n)\leq 1$.

	Let $A=u|A|$, $B=v|B|$ and $AB=w|AB|$ be the respective polar decompositions of these operators.
	Since, for every $\varepsilon>0$, $\sigma\left(\varepsilon\mathbbm{1}+|A|\right),\sigma\left(\varepsilon\mathbbm{1}+|B|\right)\subset [\varepsilon,\infty)$, we can define a function $f_n:\mathbb{C} \to \mathbb{C}$ by
	$$\begin{aligned}
	f_n(z)&=\tau\left(w^\ast u\left(|A|_n^p\right)^z v \left(|B|_n^q\right)^{(1-z)}\right)\\
	&=\tau\left(\left(|B|_n^q\right)^{-z} w^\ast u\left(|A|_n^p\right)^z v |B|_n^q\right).
	\end{aligned}$$
	% A interesting thing to notice, which probably will simplify the proof is that all works in the subalgebra (\nalgebra |B|^q_n), but in this algebra the tracial weight becames a tracial state.
	
	To see that this function is entire analytic, notice that there exists $A_n \in \nalgebra$ such that $$\left(|B|_n^q\right)^{-z}|B|_n w^\ast u\left(|A|_n^p\right)^z v = \sum_{j=1}^{\infty} C_j z^j,$$
	which can be obtained simply as a Taylor series. It follows from Lemma \ref{normtraceinequality} that, for any $\varepsilon>0$, there exists $N\in \mathbb{N}$ large enough such that 
	$$\begin{aligned}
	\left|f_n(z)-\sum_{j=1}^{N} \tau(C_j |B|_n^q ) z^j\right|&=\left|\tau\left(\left(\left(|B|_n^q\right)^{-z} |B|_n w^\ast u\left(|A|_n^p\right)^z v-\sum_{j=1}^{N} C_j z^j\right) |B|_n^q\right)\right|\\
	&\leq\left\|\left(|B|_n^q\right)^{-z} |B|_n w^\ast u\left(|A|_n^p\right)^z v-\sum_{j=1}^{N} C_j z^j\right\|\tau\left( |B|_n^q\right)\\
	&<\varepsilon.\\
	\end{aligned}.$$
	
	These function $f_n$ are also bounded in the strip $\{z\in \mathbb{C} \ | \ 0\leq \Re{z} \leq 1\}$, because 
	$$\begin{aligned}
	|f_n(z)|&\leq \left\|\left(|B|_n^q\right)^{-z} w^\ast u\left(|A|_n^p\right)^z v\right\| \tau\left(|B|_n^q\right)\\
	&\leq \left\||B|_n^q\right\|^{-\Re{z}}\left\||A|_n^p\right\|^{\Re{z}} \tau\left(|B|_n^q\right)\\
	&\leq  
	\max\left\{1,\left\||B|_n^q\right\|^{-1}\right\}\max\left\{1,\left\||A|_n^p\right\|\right\} \tau\left(|B|_n^q\right).\\
	\end{aligned}$$
	
	By the Three-Line Theorem\footnote{There is a confusion concerning the name of this result. It is known as Doetsch's Three-Line Theorem, Hadamard's Three-Line Theorem or the Phragm\'en-Lindel\"of Principle. The confusion occurs because it is a variant of the Three-Circle Theorem due to Hadamard and a consequence of the Phragm\'en-Lindel\"of Maximum Principle, but it was published by Doetsch before Hadamard's result. },
	$$\begin{aligned}
	\left|f_n\left(\frac{1}{p}\right)\right|& \leq \sup_{y\in \mathbb{R}}|f_n(1+ \iu y)|^\frac{1}{p} \sup_{y\in \mathbb{R}}|f_n(0+\iu y)|^\frac{1}{q}\\
	&=\sup_{y\in \mathbb{R}}|\tau\left(w^\ast u\left(|A|_n^p\right)^{1+\iu y} v \left(|B|_n^q\right)^{-\iu y}\right)|^\frac{1}{p} \sup_{y\in \mathbb{R}}|\tau\left(w^\ast u\left(|A|_n^p\right)^{\iu y} v \left(|B|_n^q\right)^{(1-\iu y)}\right)|^\frac{1}{q}\\
	&\leq\sup_{y\in \mathbb{R}}\left|\tau\left(|A|_n^p\right)\right|^\frac{1}{p} \sup_{y\in \mathbb{R}}\left|\tau\left(|B|_n^q\right)\right|^\frac{1}{q}\\
	&=\tau\left(|A|_n^p\right)^\frac{1}{p} \tau\left(|B|_n^q\right)^\frac{1}{q}\\
	\end{aligned} $$
	
	Then
	$$\begin{aligned}
	\tau(|AB|)& =\lim_{n\to \infty}f_n\left(\frac{1}{p}\right)\\
	&\leq \lim_{n\to \infty}\tau\left(|A|_n^p\right)^\frac{1}{p}\tau\left(|B|_n^q\right)^\frac{1}{q}\\
	&=\tau\left(|A|^p\right)^\frac{1}{p} \tau\left(|B|^q\right)^\frac{1}{q}.\\
	\end{aligned}$$
\end{proof}

We can generalise the H\"older inequality using Theorem \ref{TLT} as it follows. Although it is basically a repetition of the argument in \cite{Ruscai72} with minor modifications to use the multi-variable Three-Line Theorem, this proof was unknown by the author.
%###Although it is basic a repetition of the previous proof, the ideia to use the Araki generalizaion of Doetsch's theorem seem to be mine.
\begin{theorem}[H\"older Inequality\index{inequality! H\"older}]
	\label{gholder}
	Let $\nalgebra$ be a von Neumann algebra and $\tau$ a normal faithful semifinite trace in $\nalgebra$. Let also $A_i\in\nalgebra$, $i=1,\dots, k$ and $\displaystyle \sum_{i=1}^{k} p_i>1$ such that $\displaystyle \sum_{i=1}^{k}\frac{1}{p_i}=1$, then
	$$\tau\left(\left|\prod_{i=1}^{k}A_i\right|\right)\leq \prod_{i=1}^{k}\tau(|A_i|^{p_i})^\frac{1}{p_i}.$$
\end{theorem}
\begin{proof}
	Let $A_i=u_i|A_i|$ and $\displaystyle \prod_{i=1}^ {k}A_i=w\left|\prod_{i=1}^{k}A_i\right|$ be the respective polar decompositions of these operators. Following the same steps of the previous proof we can define the analytic function
	
	$$\begin{aligned}
	f_n(z_1,\dots,z_{k-1})&=\tau\left(w^\ast u_1\left(|A_1|_n^{p_1}\right)^{z_1}\dots u_k \left(|A_k|_n^{p_k}\right)^{z_{k-1}}\left(|A_k|_n^{p_k}\right)^{1-\sum_{i=1}^{k-1}z_i}\right)\\
	&=\tau\left(\left(|A_k|_n^{p_k}\right)^{-\sum_{i=1}^{k-1}z_i}w^\ast u_1\left(|A_1|_n^{p_1}\right)^{z_1}\dots u_k \left(|A_k|_n^{p_k}\right)^{z_{k-1}}|A_k|_n^{p_k}\right),\\
	\end{aligned}$$
	where $\displaystyle |A_i|_n^p=\frac{1}{\tau\left(|A_i|^p\right)}\int^{\|A_i\|}_\frac{1}{n}\lambda^p dE^{|A_i|}_\lambda, \quad i=1, \dots, k$.
	
	This function is bounded in the region
	$$B=\left\{(z_1,\dots,z_{k-1}) \in \mathbb{C}^k \ \middle| \  \Re{z_i}\geq 0, \ i=1,\dots k-1, \ \sum_{i=1}^{k-1} \Re{z_i}\leq 1\right\}.$$
	
	Using now Theorem \ref{TLT} we have that
	$$g(y_1,\dots, y_{k-1})=\log\left(\sup_{x_1,\dots,x_{k-1} \in \mathbb{R}}\left|f_n(x_1+\iu y_1,\dots, x_{k-1}+\iu y_{k-1})\right|\right)$$
	is a jointly convex function.
	
	This leads to an equivalent result of the one we had in Theorem \ref{holder}, namely
	$$\begin{aligned}
	\left|f_n\left(\frac{1}{p_1},\dots,\frac{1}{p_{k-1}}\right)\right|& \leq \left(\sup_{y_1,\dots,y_{k-1} \in \mathbb{R}}\left|f_n(1+\iu y_1,\iu y_2,\dots,\iu y_{k-1})\right|\right)^\frac{1}{p_1}\dots\times\\
	&\qquad \times \left(\sup_{y_1,\dots,y_{k-1} \in \mathbb{R}}\left|f_n(\iu y_1,\iu y_2,\dots,1+\iu y_{k-1})\right|\right)^\frac{1}{p_{k-1}}\times\\
	&\qquad \times \left(\sup_{y_1,\dots,y_{k-1} \in \mathbb{R}}\left|f_n(\iu y_1,\iu y_2,\dots,\iu y_{k-1})\right|\right)^\frac{1}{p_{k}}\\
	&\leq \prod_{i=1}^{k}\tau(|A_i|_n^{p_i})^\frac{1}{p_i}.
	\end{aligned} $$
	
	Then
	$$\tau(|AB|) =\lim_{n\to \infty}f_n\left(\frac{1}{p_1},\dots,\frac{1}{p_{k-1}}\right)\leq \lim_{n\to \infty} \prod_{i=1}^{k}\tau(|A_i|_n^{p_i})^\frac{1}{p_i} =\prod_{i=1}^{k}\tau(|A_i|^{p_i})^\frac{1}{p_i}.$$
	
\end{proof}

The reader should keep in mind that H\"older's inequality is a very interesting result to us, since it says something regarding the trace of a product and this is the case in Dyson's series. Nevertheless, it is used in the proof of the Minkowski Inequality which is imperative to define a normed vector space.

\begin{theorem}[Minkowski's Inequality\index{inequality! Minkowski}]
	\label{minkowski}
	Let $\nalgebra$ be a von Neumann algebra, $\tau$ a normal faithful semifinite trace in $\nalgebra$, and $p,\, q>1$ such that $\frac{1}{p}+\frac{1}{q}=1$. Then
	\begin{enumerate}[(i)]
		\item for every $A\in \nalgebra$, $\displaystyle \tau(|A|^p)^\frac{1}{p}=\sup\left\{\left|\tau(AB)\right| \ \middle | \ B\in \nalgebra,  \tau\left(|B|^q\right)\leq 1\right\};$
		\item for every $A,B \in \nalgebra$, $\displaystyle \|A+B\|_p\leq \|A\|_p+\|B\|_p$.
	\end{enumerate}
\end{theorem}

\begin{proof}
	By H\"older's inequality, for all $B\in \nalgebra$ such that $\tau\left(|B|^q\right)\leq 1$,
	$$|\tau(AB)|\leq \tau(|AB|)\leq \tau(|A|^p)^\frac{1}{p}\tau(|B|^p)^\frac{1}{q}\leq \tau(|A|^p)^\frac{1}{p}. $$
	
	On the other hand, if $\tau(|A|^p)<\infty$, let $A=u|A|$ be the polar decomposition of $A$ and define $B=\frac{|A|^{p-1}u^\ast}{\tau\left(|A|^p\right)^\frac{p-1}{p}}$. Then 
	$$|B|^2=B^\ast B=\frac{u|A|^{2p-2}u^\ast}{\tau\left(|A|^p\right)^\frac{2p-2}{p}} \Rightarrow \tau\left(|B|^q\right)=\tau\left(\frac{u|A|^{(p-1)q}u^\ast}{\tau\left(|A|^p\right)^\frac{(p-1)q}{p}}\right)=1.$$
	
	Moreover, $\tau(AB)=\frac{\tau(u|A|^p u^\ast)}{\tau\left(|A|^p\right)^\frac{p-1}{p}}=\tau\left(|A|^p\right)^\frac{1}{p}$.
	
	In case $\tau(|A|^p)=\infty$, we use semifiniteness and normality to take an increasing sequence of operators $(A_n)_{n\in\mathbb{N}} \subset \mathfrak{M}_\tau$ converging to $A$ and apply the recently proved result to this sequence. 
	
\end{proof}

Together, Theorem \ref{gholder} and Theorem \ref{minkowski} provide us with another generalization of H\"older's inequality. This inequality is obvious in the commutative case, but not in the noncommutative one.

\begin{corollary}[H\"older Inequality\index{inequality! H\"older}]
	\label{g2holder}
	Let $\nalgebra$ be a von Neumann algebra and $\tau$ a normal faithful semifinite trace in $\nalgebra$, let also $A,B\in\nalgebra$ and $p,\, q>1$ such that $\frac{1}{p}+\frac{1}{q}=\frac{1}{r}$, then
	$$\tau(|AB|^r)^\frac{1}{r}\leq \tau(|A|^p)^\frac{1}{p}\tau(|B|^q)^\frac{1}{q}.$$
\end{corollary}
\begin{proof}
	Let $s>1$ such that $\frac{1}{r}+\frac{1}{s}=1$, then $\frac{1}{p}+{1}{q}+\left(1-{1}{r}\right)=\frac{1}{p}+{1}{q}+{1}{s}=1$. Hence using $(i)$ in Theorem \ref{minkowski} and Theorem \ref{gholder}, we get 
	$$\begin{aligned}
	\tau(|AB|^r)^{\frac{1}{r}}&=\sup\left\{\left|\tau(ABC)\right| \ \middle | \ C\in \nalgebra,  \tau\left(|C|^s\right)\leq 1\right\}\\
	&\leq \sup\left\{\tau(|A|^p)^\frac{1}{p}\tau(|B|^q)^\frac{1}{q}\tau(|C|^s)^\frac{1}{s} \ \middle | \ C\in \nalgebra,  \tau\left(|C|^s\right)\leq 1\right\}\\
	&\leq\tau(|A|^p)^\frac{1}{p}\tau(|B|^q)^\frac{1}{q}.
	\end{aligned}$$
	
\end{proof}

\begin{definition}\index{noncommutative $L_p$-space! Segal-Dixmier}
	Let $\nalgebra$ be a von Neumann algebra and $\tau$ a normal, faithful and semifinite trace on $\nalgebra$. We define the noncommutative $L_p$-space, denoted by $L_p(\nalgebra,\tau)$, as the completion of
	$$\left\{A\in \nalgebra \ \middle| \ \tau\left(|A|^p\right)<\infty\right\}$$
	with respect to the norm $\displaystyle\|A\|_p=\tau\left(|A|^p\right)^\frac{1}{p}$.
	
	We also set $L_\infty(\nalgebra, \tau)=\nalgebra$ with $\|A\|_\infty=\|A\|$.
	
\end{definition}

Now, it is easy to see that, for $p,\, q\geq1$ H\"older conjugated, the H\"older and Minkowski inequalities can be extended to the whole space $L_p(\nalgebra,\tau)$ through an argument of density  and normality of the trace. With this definition, Theorem \ref{normtraceinequality} and Theorem \ref{holder}, and Theorem \ref{minkowski} can be expressed as
$$\begin{aligned}
\|AB\|_1&\leq \|A\|_p\|B\|_q,\\
\|A+B\|_p &\leq \|A\|_p+\|B\|_p.\\
\end{aligned}$$ 
and this last equality is a triangular inequality for $\|\cdot\|_p$. It is important to notice that faithfulness guarantees $\|A\|_p=0 \Rightarrow A=0$, however semifiniteness was used only at the very end of Theorem \ref{minkowski} and it is completely irrelevant when talking about noncommutative $L_p$-spaces, since the trace is never infinity on these operators.

It is not our intention in this text to discuss this subject, but notice that if $\tau$ is not semifinite, we can define the noncommutative $L_p$ space to a "small" algebra $\overline{\mathfrak{M}_\tau}^{SOT}$.

%Terp Theorem 32; Xu, Non-Commutative L_p-spaces, Duality pag 15; https://matthewhr.files.wordpress.com/2012/09/uniform-convexity-and-reflexivity.pdf -> Milman-Pettis's Theorem
\begin{theorem}
	\label{dualLp}
	Let $p,\, q\geq 1$ such that $\frac{1}{p}+\frac{1}{q}=1$. Then the function below is an isometric isomorphism:
	
\begin{center}
	\begin{tabular}{@{\hskip 2pt}c@{\hskip 2pt}c@{\hskip 2pt}c@{\hskip 2pt}c@{\hskip 2pt}c@{\hskip 2pt}c}
	$\Xi$: $L_p(\nalgebra,\tau)$&$\to$& $L_q(\nalgebra,\tau)^\ast$& & \\
	$A$ &$\mapsto$&$\tau_A:$&$L_q(\nalgebra,\tau)$&$\to $&$\mathbb{C}$\\
	& & &$B$&$\mapsto $&$\tau(AB)$.\\
	\end{tabular}
\end{center}
\end{theorem}
\begin{proof}
	First note that, $\Xi$ is well defined. In fact, it follows by H\"older's inequality that $\tau_A$ belongs to $L_q(\nalgebra,\tau)^\ast$.
	
	From Theorem \ref{minkowski} $(i)$, it becomes obvious that $\Xi$ is a linear isometric injection. It remains to prove surjection.
	
	Let $\phi\in L_q(\nalgebra,\tau)^\ast$. Then the restriction  $\left.\phi\right|_{L_q(\nalgebra,\tau)\cap \nalgebra_+}$ can be extended to a normal semifinite weight on $\nalgebra$. By Theorem \ref{TPTRN2}, there exists a operator $H\eta\nalgebra$ such that
	$\phi(A)=\tau(HA), \ \forall A\in L_q(\nalgebra,\tau)\cap \nalgebra$ and, since $\phi$ is $\|\cdot\|$-continuous and $\overline{L_q(\nalgebra,\tau)\cap \nalgebra}^{\|\cdot\|_q}=L_q(\nalgebra,\tau)$, we have $\phi(A)=\tau(HA)$ for all $A\in L_q(\nalgebra,\tau)$.
	
	Finally,
	$$\begin{aligned}
	\tau(|H|^p)^\frac{1}{p}&=\sup\left\{\left|\tau(HA)\right| \ \middle | \ A\in \nalgebra,  \tau\left(|A|^q\right)\leq 1\right\}\\
	&=\sup\left\{\left|\phi(A)\right| \ \middle | \ A\in \nalgebra,  \tau\left(|A|^q\right)\leq 1\right\}\\
	&=\|\phi\|.
	\end{aligned}$$

\end{proof}

This last result is the famous identification $L_p(\nalgebra,\tau)^\ast=L_q(\nalgebra,\tau)$ where $p,\, q>1$ are H\"older conjugated.

\section{The Haagerup Noncommutative $L_p$-spaces}

The last two sections of this text will be devoted to two different generalizations of $L_p$-spaces for arbitrary von Neumann algebras, one due to Haagerup and another due to Araki and Masuda.

Hitherto we have been as detailed as possible, but henceforth the purpose of the text will be slightly different: sketching how to construct a generalized $L_p$-space. Terp's notes, \cite{terp81}, cover basically all that will be presented here.

We start this chapter noticing that, given a von Neumann algebra $\nalgebra$, the set $Aut(\nalgebra)$ becomes a topological group with respect to the topology generated by the family
$$\left\{U(\alpha;\omega_1,\dots,\omega_n) \ \middle| \ \alpha \in Aut(\nalgebra), \ n \in \mathbb{N}, \ \omega_i\in \nalgebra_\ast, \ i=1,\dots,n \right\},$$
where
$$U(\alpha;\omega_1,\dots,\omega_n)=\left\{\beta\in Aut(\nalgebra) \ \middle| \ \begin{aligned}&\|\omega_i\circ\alpha - \omega_i\circ\beta\|&<1,\\ &\|\omega_i\circ\alpha^{-1}- \omega_i\circ\beta^{-1}\|&<1,
\end{aligned} \quad i=1,\dots,n\right\}.$$

\begin{definition}
	A von Neumann \index{covariant system} covariant system consists of a triplet $(\nalgebra,G,\alpha)$ where $\nalgebra$ is a von Neumann algebra, $G$ is a locally compact group and $\alpha$ a continuous homomorphism of $G$ into $Aut(\nalgebra)$ (\ie, an action of $G$ on $\nalgebra$).  
\end{definition}

Consider the covariant system $(\nalgebra, G, \alpha)$ and let $\pi:\nalgebra\to B(\hilbert)$ be a normal representation. Consider the representations $\pi_\alpha$ of $\nalgebra$ and $\lambda$ of $G$ both in the Hilbert space $L_2(G,\nalgebra)$ defined by
$$\begin{aligned}
\left(\pi_\alpha(A)(f)\right)(t)&=\alpha^{-1}_t(A) f(t), & A\in \nalgebra, \ & \forall f\in L_2(G,\nalgebra), \ & \forall t\in G,\\
\left(\lambda(g)f\right)(t)&=f(g^{-1}t), & g\in G, \ & \forall f\in L_2(G,\nalgebra), \ & \forall t\in G.\\
\end{aligned}$$ 

\begin{definition}
	Let $(\nalgebra,G, \alpha)$ be a covariant system and $\pi:\nalgebra\to B(\hilbert)$ a normal representation. We define the crossed product $\nalgebra \rtimes_\alpha G$ as the von Neumann algebra generated by $\pi_\alpha(\nalgebra)\cup \lambda(G)$.
\end{definition}

Form more details in the theory of crossed products, we refer to \cite{pedersen79} and \cite{Takesaki2003}.

Let's fix a normal faithful semifinite weight $\varphi_0$ and denote by $\{\sigma^{\varphi_0}_t\}_{t\in\mathbb{R}}$ the modular automorphism group related to $\varphi_0$.

Now, we will make a change in notation and denote $\nalgebra\rtimes \{\sigma^{\varphi_0}_t\}_{t\in\mathbb{R}}\doteq\nalgebra\rtimes_\alpha \mathbb{R}$, where $\alpha=\{\sigma^{\varphi_0}_t\}_{t\in\mathbb{R}}$. The reason we are doing this is to put in evidence the modular automorphism group, which is a one-parameter group and it become obvious that the local compact group $G$ in the definition of crossed product must be $\mathbb{R}$. 

The next question we are interested in answering is: how are the weights in the crossed product related to the weights in $\pi_\alpha(\nalgebra)$?

In order to answer this question we need the following definitions:

%Haagerup "Operator valued weights in von Neumann Algebras I"
\begin{definition}
	Let $\nalgebra$ be a von Neumann algebra, the extended positive part $\widehat{\nalgebra}_+$ of this algebras is defined to be the set of all maps $m:\nalgebra_\ast^+\to\overline{\mathbb{R}}_+$ satisfying the following conditions:
	\begin{enumerate}[(i)]
		\item $m(\lambda\phi)=\lambda m(\phi)$, $\forall \phi \in \nalgebra_\ast^+$ and $\forall \lambda\geq0$;
		\item $m(\phi+\psi)=m(\phi)+m(\psi)$, $\forall \phi,\psi \in \nalgebra_\ast^+$;
		\item $m$ is lower semicontinuous.
	\end{enumerate}
\end{definition}

It is easy to see that $\nalgebra_+$ can be seen as a subset of $\hat{\nalgebra}_+$.

%Haagerup "Operator valued weights in von Neumann Algebras I"
\begin{definition}[Operator Valued Weight]
	Let $\nalgebra_1$, $\nalgebra_2$ be von Neumann algebras, $\nalgebra_2\subset\nalgebra_1$, an operator $T:\nalgebra_1^+\to\widehat{\nalgebra}_2^+$ is said to be an operator valued weight if it satisfies the following conditions:
	\begin{enumerate}[(i)]
		\item $T(\lambda A)=\lambda m(A)$, $\forall A \in \nalgebra_1^+$ and $\forall \lambda\geq0$;
		\item $T(A+B)=T(A)+T(B)$, $\forall A,B \in \nalgebra_1^+$;
		\item $T(B^\ast A B)=B^\ast T(A) B$, $\forall A\in\nalgebra_1^+$ and $\forall B\in \nalgebra_2$.
	\end{enumerate} 
	
	In addition, we say that $T$ is normal if $T(A_i)\to T(A)$, for every increasing net $(A_i)_{i\in I}$ such that $A_i\to A$.
	
\end{definition}

Consider the dual action $\theta$ of $\mathbb{R}$ in $\nalgebra\rtimes \{\sigma^{\varphi_0}_t\}_{t\in\mathbb{R}}$, characterized by
$$\begin{aligned}
&\theta_s A=A, &\ A\in \pi_\alpha(\nalgebra),\\
&\theta_s \lambda(t)=e^{-\iu st}\lambda(t), & \ t\in \mathbb{R}.\\
\end{aligned}$$ 
We have the following characterization:
$$\pi_\alpha(\nalgebra)=\left\{A\in \nalgebra\rtimes\{\sigma^{\phi_0}_t\}  \ \middle| \ \theta_t A=A \ \forall t\in\mathbb{R}\right\}. $$

%Haagerup "Operator valued weights in von Neumann Algebras II" Lema 5.2
\begin{lemma}
	\label{Tcrossedproduct}
	Let $\nalgebra$ be a von Neumann algebra. The following properties hold for the operator $T:\left(\nalgebra\rtimes\{\sigma^{\phi_0}_t\}\right)_+\to \nalgebra_+$, given by
	$$TA=\int_{-\infty}^{\infty}\theta_t(A)dt, \quad A\in \left(\nalgebra\rtimes\{\sigma^{\phi_0}_t\}\right)_+ \ ,$$
	and characterized by
	$$(\phi,TA)=\int_{-\infty}^{\infty}\phi\circ\theta_t(A)dt, \quad \forall \phi\in\nalgebra_\ast^+:$$
	\begin{enumerate}[(i)]
		\item $T$ is a normal faithful semifinite operator valued weight;
		\item There exists a unique normal faithful semifinite trace $\tau$ on $\nalgebra \rtimes\{\sigma^{\phi_0}_t\}$ such that $(D\phi\circ T:D\tau)_t=\lambda(t)$ for all $t\in \mathbb{R}$;
		\item The trace $\tau$ satisfies $\tau\circ\theta_t=e^{-t}\tau$ for all $t\in\mathbb{R}$.
	\end{enumerate}
\end{lemma}

The last result was enunciated as a lemma because so it was called in the original paper \cite{haagerup79}, and because it is a preliminary result which will be used to prove the following theorem, giving us the answer to a question stated earlier.

It is quite curious that there exists a trace in the crossed product with such good properties. One may think it is enough to consider the Segal-Dixmier construction for such a trace, but it is not the case and this problem recalls the question aforementioned of identifying the relation between the weighs on the two algebras. 

\begin{theorem}
	Let $\nalgebra_1,\nalgebra_2$ be von Neumann algebras, $\nalgebra_2\subset\nalgebra_1$, and let $\phi$ and $\psi$ be normal faithful semifinite weights in $\nalgebra_1$ and $\nalgebra_2$, respectively. Then, there exists a unique operator valued weight $S$ from $\nalgebra_1$ to $\nalgebra_2$ such that $\phi=\psi\circ S$.
\end{theorem}

Now, for each normal weight $\phi$ on $\pi_\alpha(\nalgebra)$ consider the extension $\widehat{\phi}$ of this weight on $\nalgebra\rtimes\{\sigma^{\phi_0}_t\}_+^\wedge$ given in the following proposition:
\begin{proposition}
	Any normal weight $\phi$ on a von Neumann algebra $\nalgebra$ has a unique extension $\widehat{\phi}$ to $\widehat{\nalgebra}_+$ satisfying:
	\begin{enumerate}[(i)]
		\item $\phi(\lambda m)=\lambda\phi(m)$ for all $\lambda>0$ and for all $m\in\widehat{\nalgebra}_+$;
		\item $\phi(m+n)=\phi(m)+\phi(n)$ for all $m,n \in \widehat{\nalgebra}_+$;
		\item For every increasing net $(m_i)_{i\in I}$ such that $m_i\to m \in \widehat{\nalgebra}_+$, we have $\phi(m_i)\to \phi(m)$.
	\end{enumerate}
\end{proposition}

Using this extension, we can define a normal weight on $\nalgebra\rtimes\{\sigma^{\phi_0}_t\}$ by
$$\tilde{\phi}=\widehat{\phi}\circ T;$$
it is called the dual weight of $\phi$.

\begin{lemma}
	The mapping $\phi\mapsto \tilde{\phi}$ is a bijection of the set of all normal semifinite weights on $\pi_\alpha(\nalgebra)$ onto the set of all normal semifinite weights $\psi$ on $\nalgebra\rtimes\{\sigma^{\phi_0}_t\}$ satisfying
	$$\psi \circ \theta_t=\psi \quad \forall t\in \mathbb{R}. $$
\end{lemma}

Let $\tau$ be the unique trace described in \ref{Tcrossedproduct} $(ii)$ and let us look at the map $H\mapsto \tau(H\ \cdot)$.

\begin{lemma}
	The mapping $H\mapsto\tau(H\ \cdot)$ is a bijection of $\nalgebra\rtimes\{\sigma^{\phi_0}_t\}^\wedge_+$ onto the set of all normal weights on $\nalgebra\rtimes\{\sigma^{\phi_0}_t\}$. In particular, it is a bijection of the positive self-adjoint operators affiliated with $\nalgebra\rtimes\{\sigma^{\phi_0}_t\}$ onto the normal semifinite weights on $\nalgebra\rtimes\{\sigma^{\phi_0}_t\}$. 
\end{lemma}

\begin{definition}
	For each normal weight $\phi$ on $\nalgebra$ we define $H_\phi$ as the unique element of $\widehat{N}_+$ such that $\tilde{\phi}=\tau(H_\phi\ \cdot)$.
\end{definition}

\begin{proposition}
	The mapping $\phi \mapsto H_\phi$ is a bijection of the set of all normal semifinite weights on $\nalgebra$ onto the set of all positive self-adjoint operators $H$ affiliated with  $\nalgebra\rtimes\{\sigma^{\phi_0}_t\}$ satisfying
	$$\theta_t H=e^{-t}H.$$
\end{proposition}

We started this chapter looking at operators affiliated with some von Neumann algebra in order to define measurable operators with respect to a normal faithful semifinite trace. A natural question now is what happens if we look at the $\tau$-measurable operators?

\begin{lemma}
	Let $\phi$ be a normal semifinite weight on $\nalgebra$ and let $\left\{E^{H_\phi}_\lambda\right\}_{\lambda\in\mathbb{R}^+}$ be the spectral resolution of the positive operator $H_{\phi}$. Then
	$$\tau\left(E^{H_\phi}_{\lambda}\right)=\frac{1}{\lambda}\phi(\mathbbm{1}).$$
\end{lemma}

\begin{corollary}
	Let $\phi$ be a normal semifinite weight on $\nalgebra$. Then $H_\phi$ is $\tau$-measurable if $\phi\in\nalgebra_\ast$.
\end{corollary} 

\begin{theorem}
	\label{identificationmeasurable}
	The mapping $\phi \mapsto H_\phi$ extends to a linear bijection of $\nalgebra_\ast$ onto
	$$\left\{H\in\left(\nalgebra\rtimes\{\sigma^{\phi_0}_t\}\right)_\tau \ \middle| \ \theta_t H=e^{-t}H, \ \forall t \in \mathbb{R} \right\}.$$
\end{theorem}

The previous theorem suggests the following definition:

\begin{definition}[Haagerup's $L_p$-spaces]\index{noncommutative $L_p$-space! Haagerup}
	Let $\nalgebra$ be a von Neumann algebra, $\phi_0$ a nornal faithful semifinite weight on $\nalgebra$ and $1\leq p\leq\infty$, we define	
	$$L_p(\nalgebra)\doteq\left\{H\in\left(\nalgebra\rtimes\{\sigma^{\phi_0}_t\}\right)_\tau \ \middle| \ \theta_t H=e^{-\frac{t}{p}}H, \ \forall t \in \mathbb{R} \right\}.$$
\end{definition}  

A first important and non trivial comment is that this construction is independent of the choice of $\phi_0$.

In addition, by Theorem \ref{identificationmeasurable}, $L^1(\nalgebra)=\nalgebra_\ast$ and one can prove that $L^\infty(\nalgebra)=\nalgebra$.

Of course, much more is expected of a vector space to call it a $L_p$-space. Some of them will be enunciated below.

\begin{proposition}
	Let $1\leq p<\infty$ and let $A$ be a closed densely defined operator affiliated with $\nalgebra\rtimes\{\sigma^{\phi_0}_t\}$, $A=u|A|$ its polar decomposition. Then
	$$A\in L_p(\nalgebra)\Leftrightarrow u\in \nalgebra \textrm{ and }|A|^p \in L^1(\nalgebra).$$
\end{proposition}

\begin{definition}
	Consider the function given by
	$tr(H_\phi)=\phi(\mathbbm{1})$, $\phi\in \nalgebra_\ast$. Define
	$$\|A\|_p=tr\left(|A|^p\right)^{\frac{1}{p}}, \ A \in L_p(\nalgebra).$$
	
\end{definition}

\begin{theorem}[H\"older and Minkowski's Inequalities \index{inequality! H\"older} \index{inequality! Minkowski}]
	\label{holderminkowskihaagerup}
	Let $\nalgebra$ be a von Neumann algebra and $p,\, q>1$ such that $\frac{1}{p}+\frac{1}{q}=1$. Then
	\begin{enumerate}[(i)]
		\item $\displaystyle \|AB\|_1\leq \|A\|_p\|B\|_q$, for all $A\in L^p(\nalgebra)$ and for all $B\in L^q(\nalgebra)$;
		\item $\displaystyle\|A\|_p=\sup\left\{|tr(AB)| \ \middle| \ B\in L_q(\nalgebra), \ \|B\|_q\leq1\right\}$;
		\item $\displaystyle \|A+B\|_p\leq \|A\|_p+\|B\|_p$, for all $A,B\in L^p(\nalgebra)$.
	\end{enumerate}
\end{theorem}

\begin{proposition} \, \hspace{10cm}
	\begin{enumerate}[(i)]
		\item The function $\|\cdot\|_p:L_p(\nalgebra)\to\mathbb{R}_+$ is a norm and $\left(L_p(\nalgebra),\|\cdot\|_p\right)$ is a Banach space;
		\item the topology induced on $L_p(\nalgebra)$ as a subspace of $\left(\nalgebra\rtimes\{\sigma^{\phi_0}_t\}\right)_\tau$ coincides with the $\|\cdot\|_p$-topology;
		\item 	let $p,\, q\geq 1$ such that $\frac{1}{p}+\frac{1}{q}=1$, then $L_p(\nalgebra)$ and $L_q(\nalgebra)$ form a dual pair with respect to the bilinear form
		$$\begin{aligned}
		(\cdot,\cdot): &L_p(\nalgebra)\times L_q(\nalgebra)&\to&\hspace{5 mm}\mathbb{C} \\
		&\hspace{9 mm} (A,B)&\mapsto&\hspace{1 mm}tr(AB).\\
		\end{aligned}$$
	\end{enumerate}
\end{proposition}

The results stated above make Haagerup's noncommutative $L_p$-spaces Banach spaces satisfying a duality relation for H\"older conjugated indices spaces.

We finish this section mentioning that, as always, traces come out in the construction and play a central role in the theory.

\section{The Araki-Masuda Noncommutative $L_p$-Spaces }

This section is devoted to present the definition of noncommutative $L_p$-spaces as suggested by Araki and Masuda in \cite{Araki82}, through Tomita-Takesaki modular operator theory. The great advantage of this approach is that its construction is based on the Hilbert space.

\begin{definition}\index{noncommutative $L_p$-space! Araki-Masuda}
Let $\nalgebra$ be a von Neumann algebra and $\Omega$ a cyclic and separating vector.
\begin{enumerate}[(i)]
\item for $1\leq p \leq 2$, define $L_p(\nalgebra,\Omega)$ as the competition of $\left(\hilbert, \|\cdot\|_p^\Omega\right)$, where
\begin{equation}
\label{normLp1}
\|\xi\|_p^\Omega=\inf\left\{\left\|\Delta_{\Phi, \Omega}^{\frac{1}{2}-\frac{1}{p}} \xi\right\| \ \middle| \ \|\Phi\|=1, \ s^\nalgebra(\Phi) \geq s^\nalgebra(\xi) \textrm{ and } \xi \in \Dom{\Delta_{\Phi \Omega}^{\frac{1}{2}-\frac{1}{p}}}\right\};
\end{equation}
\item for $2\leq p \leq \infty$, define
$$L_p(\nalgebra,\Omega)=\left\{\xi \in \bigcap_{\Phi\in \hilbert} \Dom{\Delta_{\Phi, \Omega}^{\frac{1}{2}-\frac{1}{p}}} \ \middle| \ \|\xi\|_p^\Omega<\infty \right\}$$
and the norm on this vector space by
\begin{equation}
\label{normLp2}
\|\xi\|_p^\Omega=\sup\left\{\left\|\Delta_{\Phi, \Omega}^{\frac{1}{2}-\frac{1}{p}}\xi\right\| \ \middle| \ \|\Phi\|=1 \right\}.
\end{equation}

\end{enumerate}
\end{definition}
We are dealing with $\|\cdot\|_p^\Omega$ as if it is a norm, which in fact it is stated in the next results.

\iffalse
****************************Fazer!
\begin{lemma}
Let $2\leq p\leq \infty$ and $\xi\in L_p(\nalgebra,\Omega)$, then
$$\|\xi\|^\Omega_p\|\Omega\|\geq \|\xi\|$$
\end{lemma}
\begin{proof}

\end{proof}\fi

\begin{lemma}
\label{LemmaModularDecreasing}
Let $\nalgebra$ be a von Neumann algebra, $\phi_1, \psi_2$ and $\psi$ weights and suppose that $\phi_1\leq\phi_2$, then $\|\Delta_{\phi_2\psi}\xi\|\leq \|\Delta_{\phi_1\psi}\xi\|$ for any $\xi \in \Dom{\Delta_{\phi_1\psi}}\cap\Dom{\Delta_{\phi_2\psi}}$.
\end{lemma}
\begin{proof}
For every $A \in \mathfrak{N}_{\phi_1}^\ast\cap\mathfrak{N}_{\phi_2}^\ast\cap\mathfrak{N}_{\psi}$ we have
$$\begin{aligned}
\|\Delta_{\phi_2\psi}\eta_\psi(A)\|
&=\|\eta_{\phi_2}(A)\|\\
&=\inf_{N \in N_{\phi_2}}\|A+N\|\\
&\leq\inf_{N \in N_{\phi_1}}\|A+N\|\\
&=\|\eta_{\phi_1}(A)\|\\
&=\|\Delta_{\phi_1\psi}\eta_\psi(A)\|.
\end{aligned}$$

Using now that $\mathfrak{N}_{\phi_1}^\ast\cap\mathfrak{N}_{\phi_2}^\ast\cap\mathfrak{N}_{\psi}$ is dense in $\Dom{\Delta_{\phi_1\psi}}\cap\Dom{\Delta_{\phi_2\psi}}$, the statement follows.

\end{proof}

\begin{proposition}
The function $\|\cdot\|_p^\Omega$ defined in equation \eqref{normLp1} is a norm for each $1\leq p< 2$ and \eqref{normLp2} is a norm for $2\leq p <\infty$. Furthermore, $\left(L_p(\nalgebra,\Omega), \|\cdot\|_p^\Omega\right)$ is a Banach space. 
\end{proposition}

\begin{theorem}
Let $1\leq p, \, p^\prime \leq \infty$ such that $\frac{1}{p}+\frac{1}{p^\prime}=1$. Then $L_p(\nalgebra,\Omega)\times L_{p^\prime}(\nalgebra,\Omega)$ form a dual pair through the continuous sesquilinear form obtained extending
\begin{equation}
\label{dualpair}
L_p(\nalgebra,\Omega)\cap \hilbert\times L_{p^\prime}(\nalgebra,\Omega)\cap \hilbert\ni (\xi,\xi^\prime) \mapsto \ip{\xi}{\xi^\prime}.
\end{equation}
\end{theorem}

In fact, not much seems to be known about the Araki-Masuda $L_p$-Spaces. After the original article, \cite{Araki82}, that first presented the construction of these spaces, Masuda wrote a new one, \cite{masuda83}, that extends the construction to weights.

% flatex input end: [Chapter5/chapter5.tex]
% Non-Commutative
% flatex input: [Chapter6/chapter6.tex]
\chapter{Perturbation of $p$-Continuous KMS States}
\label{chapExtensionPerturb}
\setcounter{section}{1}

The idea of extending Araki's perturbation theory using noncommutative $L_p$-spaces was proposed by C. D. J\"akel and consists in a new approach to the problem. Now we start presenting the main results of this work. All that follows is entirely new.

It is quite clear that one of the key properties used in \cite{Araki73} and \cite{sakai91} to prove the convergence of the Dyson's series, or in \cite{Araki73.2} to prove the convergence of the expansional\footnote{see Chapter \ref{BoundedPert}.}, is that $\|A_1\ldots A_n\|\leq \|A_1\|\ldots \|A_n\|$, which is one of the axioms of Banach algebras. Unfortunately, this property does not hold in noncommutative $L_p$-spaces, in fact, these are not even algebras under the induced multiplication. In particular, we have $\|Q^n\|\leq \|Q\|^n$, but no similar property holds in noncommutative $L_p$-spaces.

We would like to reinforce that by a trace we mean a normal faithful semifinite trace 
in the text of this chapter.

\begin{proposition}
	\label{noconstant}
Let $\nalgebra$ be a von Neumann algebra, $\tau$ be a normal faithful semifinite trace on $\nalgebra$, and $A\in L_1(\nalgebra,\tau)$. There exists $M>0$ such that $\tau\left(|A|^n\right)\leq M^n$ for all $n\in\mathbb{N}$,
if and only if $A \in \nalgebra$.
\end{proposition}
\begin{proof}
	$(\Rightarrow)$ Let's prove the contrapositive.
	Suppose $A$ is unbounded and let \\${\displaystyle |A|=\int_0^\infty \lambda dE_\lambda^{|A|}}$ be the spectral decomposition of $|A|$.
	
	For every $K>M$, $E_{(K,\infty)}$ is non-null, so $\tau(E_{(K,\infty)})>0$. Then,
	$$\tau\left(|A|^n\right)=\int_0^\infty \lambda^n \tau\left(dE_\lambda^{|A|}\right) \geq \int_K^\infty \lambda^n \tau\left(dE_\lambda^{|A|}\right)\geq K^n \tau\left(E_{[K,\infty)}\right). $$
	
	Now, we already know that there exists $N\in \mathbb{N}$ large enough such that, for all $n\geq N$, $M^n< K^n \tau\left(E_{[K,\infty)}\right)$.
	
	$(\Leftarrow)$ The case $A=0$ is trivial. Suppose $A\neq0$ is bounded. Then
	$$\begin{aligned}
	\tau\left(|A|^n\right)&=\tau\left(|A|^{n-1}|A|\right)\\
	&\leq \left\||A|^{n-1}\right\|\tau\left(|A|\right)\\
	&=\left\|A\right\|^n\frac{\tau\left(|A|\right)}{\|A\|}\\
	&\leq \left(\|A\|\max\left\{1,\frac{\tau\left(|A|\right)}{\|A\|}\right\}\right)^n
	\end{aligned}$$
	
\end{proof}

The next definition captures our intentions of having a convergent Dyson's series. In this definition, one subtle  difference is that the exponent cannot be passed out the trace, what is the $C^\ast$-condition for $p=\infty$. On the physical point of view, we do not want the high order terms in perturbation to affect our system too much, at least its integral.

\begin{definition}
	\label{defex} \glsdisp{exp}{\hspace{0pt}}
	Let $\nalgebra$ be a von Neumann algebra, $\tau$ be a normal faithful semifinite trace on $\nalgebra$, $1\leq p\leq\infty$ and $0<\lambda<\infty$. An operator $A \in L_p(\nalgebra,\tau)$ is said to be $(\tau,p,\lambda)$-exponentiable if 
	\begin{equation}
	\label{eq:defexponentiable}
	\sum_{n=1}^\infty\frac{\lambda^n\||A|^n\|_{p}}{n!}< \infty.
	\end{equation}
	Furthermore, an operator $A \in L_p(\nalgebra,\tau)$ is said to be $(\tau,p,\infty)$-exponentiable if 
	\begin{equation}
	\label{eq:defexponentiableinf}
	\sum_{n=1}^\infty\frac{\lambda^n\||A|^n\|_{p}}{n!}< \infty, \qquad \forall \lambda\in\mathbb{R}_+.
	\end{equation}
	We denote $$\ex^\tau_{p,\lambda}=\left\{A\in L_p\left(\nalgebra,\tau\right) \ \middle | \ A \textrm{ is } (\tau,p,\lambda)\textrm{-exponentiable } \right\}.$$
\end{definition}

\iffalse
\begin{definition}
	\label{defex} \gls{exp}
	Let $\nalgebra$ be a von Neumann algebra, $\tau$ be a normal faithful semifinite trace on $\nalgebra$ and $1\leq p<\infty$, an operator $A \in L_p(\nalgebra,\tau)$, $1\leq p<\infty$, is said to be $(\tau,p)$-exponentiable if $$\sum_{n=1}^\infty\frac{1}{n!}\tau\left(|A|^{np}\right)< \infty.$$
	
	We denote $$\ex^\tau_p=\left\{A\in L_p\left(\nalgebra,\tau\right) \ \middle | \ A \textrm{ is } (\tau,p)\textrm{-exponentiable } \right\}.$$
\end{definition}

\begin{remark}
	Define
	$$\begin{aligned}
	N_+=\left\{n\in\mathbb{N} \ \middle | \ \tau\left(\left|A\right|^{pn}\right)>1  \right\}\\
	N_-=\left\{n\in\mathbb{N} \ \middle | \ \tau\left(\left|A\right|^{pn}\right)\leq 1  \right\}\\
	\end{aligned}$$
	
	It is clear that
	\begin{equation}
	\label{eq:calculation2}
	\begin{aligned}
	\sum_{n=1}^{N}\frac{\||A|^n\|_p}{n!}& =\sum_{n=1}^{N}\frac{1}{n!}\tau\left(\left|A\right|^{pn}\right)^\frac{1}{p}\\
	&=\sum_{\substack{n\in N_- \\ n\leq N}}\frac{1}{n!}\tau\left(\left|A\right|^{pn}\right)^\frac{1}{p}+\sum_{\substack{n\in N_+ \\ n\leq N}}\frac{1}{n!}\tau\left(\left|A\right|^{pn}\right)^\frac{1}{p}\\
	&\leq\sum_{n=1}^\infty\frac{1}{n!}+\sum_{n=1}^\infty\frac{1}{n!}\tau\left(\left|A\right|^{pn}\right).
	\end{aligned}
	\end{equation}
	
	Hence it is natural to define $\ex^\tau_\infty=\nalgebra$ since $\|\cdot\|_\infty=\|\cdot\|$ satisfies the $C^\ast$-condition.
\end{remark}
\fi

Some properties can be seen directly from the definition. The first is that, if $\lambda\leq \lambda^\prime$, then $\ex^\tau_{p,\lambda}\subset\ex^\tau_{p,\lambda^\prime}$.
Another very useful property that we will use to simplify our presentation is that $$\ex^\tau_{p,\lambda}=\lambda \ex^\tau_{p,1}=\left\{\lambda A\in L_p(\nalgebra) \ | A\in \ex^\tau_{p,1} \right\}.$$
So, the only special case is $\ex^\tau_{p,\infty}$, for which we have $\displaystyle \ex^\tau_{p,\infty}=\bigcap_{\lambda\in\mathbb{R}_+}\ex^\tau_{p,\lambda}$.
Hence, it is enough to study $\ex^\tau_{p,1}$ and $\ex^\tau_{p,\infty}$.

\begin{notation}
In order to simplify the notation, we will denote $\ex^\tau_{p,1}=\ex^\tau_p$ and call a $(\tau,p,\lambda)$-exponentiable operator just a $(\tau,p)$-exponentiable operator.
\end{notation}

\begin{remark}
	Notice that equation \eqref{eq:defexponentiable} can be written in many forms for $1\leq p< \infty$
	$$\sum_{n=1}^\infty\frac{\||A|^n\|_{p}}{n!}=\sum_{n=1}^\infty\frac{\|A\|_{np}^n}{n!}=\sum_{n=1}^\infty\frac{\tau\left(|A|^{np}\right)^\frac{1}{p}}{n!}< \infty.$$
	
	We prefer equation \eqref{eq:defexponentiable} because it also includes the case $p=\infty$, for which
	$$\sum_{n=1}^{\infty} \frac{\||A|^n\|_{\infty}}{n!}\leq\sum_{n=1}^{\infty} \frac{\|A\|_{\infty}^n}{n!}=e^{\|A\|}-1<\infty.$$
	Hence, we have $\ex^\tau_\infty=\nalgebra$.
	
\end{remark}
%**************\textcolor{red}{Here I have to decide if it's worth to include $|A|^0=s^\nalgebra(|A|)=E_{(0,\infty)}$ in the sum, notice that it means some kind of integrable support. One of the advantages is the next paragraph holds (I'm not sure if it is true without it)}

In order so simplify calculations in our examples and constructions, we prove the following lemma.

\begin{lemma}
Let $\nalgebra$ be a von Neumann algebra and $\tau$ a normal faithful semifinite trace on $\nalgebra$. Then $A\in \ex^\tau_p$ if $$\sum_{n=1}^\infty\frac{1}{n!}\tau\left(|A|^{np}\right)< \infty.$$
\end{lemma}
\begin{proof}
Define
$$\begin{aligned}
N_+=\left\{n\in\mathbb{N} \ \middle | \ \tau\left(\left|A\right|^{pn}\right)>1  \right\},\\
N_-=\left\{n\in\mathbb{N} \ \middle | \ \tau\left(\left|A\right|^{pn}\right)\leq 1  \right\}.\\
\end{aligned}$$

It is clear that
\begin{equation}
\begin{aligned}
\sum_{n=1}^{N}\frac{\|A\|_{np}^n}{n!}& =\sum_{n=1}^{N}\frac{1}{n!}\tau\left(\left|A\right|^{pn}\right)^\frac{1}{p}\\
&=\sum_{\substack{n\in N_- \\ n\leq N}}\frac{1}{n!}\tau\left(\left|A\right|^{pn}\right)^\frac{1}{p}+\sum_{\substack{n\in N_+ \\ n\leq N}}\frac{1}{n!}\tau\left(\left|A\right|^{pn}\right)^\frac{1}{p}\\
&\leq\sum_{n=1}^\infty\frac{1}{n!}+\sum_{n=1}^\infty\frac{1}{n!}\tau\left(\left|A\right|^{pn}\right).
\end{aligned}
\end{equation}
\end{proof}

The next step is to prove the set we have just defined is big, in some sense, in $L_p(\nalgebra,\tau)$.

\begin{proposition}
	$\ex^\tau_p$ and $\ex^\tau_\infty$ are $\|\cdot\|$-dense in $L_p(\nalgebra,\tau)$.
\end{proposition}
\begin{proof}
	It is enough to prove $\ex^\tau_p$ is dense $\|\cdot\|$-dense in $L_p(\nalgebra,\tau)$.  Let $A \in L_p(\nalgebra,\tau)$ be a positive operator and let its spectral decomposition be $\displaystyle A=\int_0^\infty \lambda dE_\lambda$.
	Define $\displaystyle A_m=\int_0^m \lambda dE_\lambda$. Then, for all $n\in\mathbb{N}$,
	$$\begin{aligned}
	\tau\left(\left(A_m^p\right)^n\right)& =\int_0^m \lambda^{pn} \tau(dE_\lambda)\\
	&=\int_0^1 \lambda^{pn} \tau(dE_\lambda)+\int_1^m \lambda^{pn} \tau(dE_\lambda)\\
	&\leq\int_0^1 \lambda^{p} \tau(dE_\lambda)+m^{p(n-1)}\int_1^m \lambda^{p} \tau(dE_\lambda)\\
	& \leq m^{p(n-1)}\int_0^m \lambda^{p} \tau(dE_\lambda)\\
	&=m^{p(n-1)}\tau\left(|A_m|^p\right)
	\end{aligned}.$$
	
	Hence $\left(A_m\right)_{m\in\mathbb{N}}$ is a sequence of $(\tau,p,\infty)$-exponentiable operators and
	$$\tau\left(|A-A_m|^p\right)=\int_m^\infty\lambda^p\tau\left(dE_\lambda\right)\xrightarrow{n\to \infty}0.$$

	For the general case, just remember the polarization identity implies every operator is a linear combination of four positive operators.

\end{proof}

Notice that the previous lemma shows us that $\nalgebra \cap L_p(\nalgebra,\tau)\subset \ex^\tau_{p,\infty}$ and \\ ${\|A^n\|_p\leq \max\{1,\|A\|^{n-1}\|A\|_p\}}$ for $A\geq 0$. It is not difficult to see that the conclusion could also be obtained by Lemma \ref{normtraceinequality} and the well-known result 
$$\overline{\nalgebra \cap L_p(\nalgebra,\tau)}^{\|\cdot\|_p}=L_p(\nalgebra,\tau),$$
which is in fact proved using an argument similar to what we have used above.

This comment raises doubts about the possible ``triviality'' of $\ex^\tau_p$ or $\ex^\tau_{p,\infty}$, I mean, although we have already proved that these sets are big enough to be dense, the set we used to prove density consists of bounded operators. The next example will answer this question.

\pagebreak

\begin{example}
	\label{exfunction}
	Consider a function $f:\mathbb{R}\setminus\{0\}\to\mathbb{R}$ given by
\begin{figure}[H]
	\centering
	\begin{minipage}{0.4\textwidth}
		\vspace{-2cm}
		\centering
		$$f(x)=\begin{cases}m & \textrm{if } \frac{1}{(m+1)!}\leq |x| < \frac{1}{m!}, \, m\in\mathbb{N}\\
		0 & \textrm{if } |x|\geq1.\\ \end{cases}$$
	\end{minipage} \begin{minipage}{0.58\textwidth}
		\centering
		\includegraphics[width=0.95\linewidth, height=0.25\textheight]{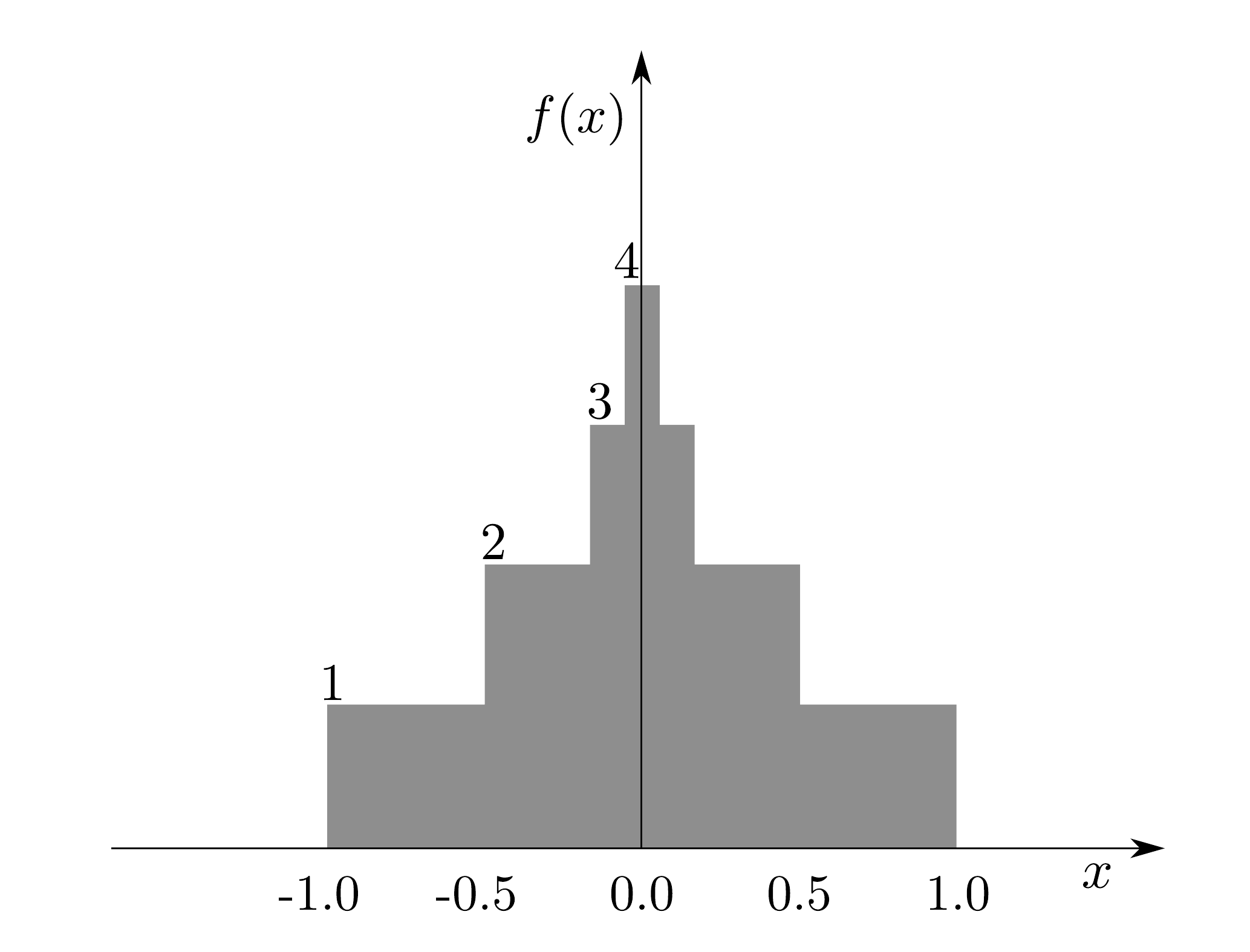}
		\caption{Example of $\left(\int_{K}\cdot dx,1,\infty\right)$-exponentiable operator of $L_1\left(L_\infty(K),\int_{K}\cdot dx\right)$.}
		\label{fig:func}
	\end{minipage}
\end{figure}
	
	This is a positive unbounded integrable function with compact support in $\mathbb{R}$ and, for each $\lambda>0$, 
	$$\begin{aligned}
	\sum_{n=1}^\infty\frac{1}{n!}\int_{\mathbb{R}} \left(\lambda f(x)\right)^n dx &= 2\sum_{n=1}^\infty\frac{1}{n!}\sum_{m=1}^{\infty} \lambda^n m^n \left(\frac{1}{m!}-\frac{1}{(m+1)!}\right)\\
	&=2\sum_{n=1}^\infty\frac{1}{n!}\sum_{m=1}^{\infty} \left(\frac{\lambda^n m^{n+1}}{(m+1)!}\right)\\
	&=\frac{2}{\lambda}\sum_{m=1}^\infty\frac{1}{(m+1)!}\sum_{n=1}^{\infty} \left(\frac{(\lambda m)^{n+1}}{n!}\right)\\
	&=\frac{2}{\lambda}\sum_{m=1}^\infty\frac{1}{(m+1)!}m (e^{\lambda m}-1)\\
	&=\frac{2}{\lambda}\sum_{m=1}^\infty\left(\frac{(m+1)e^{\lambda m}}{(m+1)!}-\frac{1}{e^\lambda}\frac{e^{\lambda(m+1)}}{(m+1)!}-\frac{(m+1)}{(m+1)!}+\frac{1}{(m+1)!}\right)\\
	&=\frac{2}{\lambda}\left(\left(e^{e^\lambda}-1\right)-\left(\frac{e^{e^\lambda}-e^\lambda-1}{e^\lambda}\right)-\left(e-1\right)+\left(e-2\right)\right)\\
	&=\frac{2}{\lambda}\frac{\left(e^{e^\lambda}-1\right)\left(e^\lambda-1\right)}{e^\lambda}.
	\end{aligned}$$
	 Of course, we don't need the exact result and in the forth step of the previous calculation it was already obvious this sum would be less than $e^{e^\lambda}$.
	 
	 For a measurable set $K\in\mathbb{R}\setminus\{0\}$ such that $0$ is an accumulation point of $K$, the restriction of $f$ to $K$ is an example of an unbounded $\left(\int_{K}\cdot dx,1,\infty\right)$-exponentiable operator of $L_1\left(L_\infty(K),\int_{K}\cdot dx\right)$.
	 
	 It is obvious that any integrable function, whose graph is in the shaded region of Figure \ref{fig:func} is also $\left(\int_{K}\cdot dx,1,\infty\right)$-exponentiable.
	 
\end{example}

One could wonder whether $\ex^\tau_p$ is a vector space or not, the following example shows the negative answer is the right one. Another consequence of this examples is that, for $\lambda<\lambda^\prime$, $\ex^\tau_{p,\infty}\subsetneq \ex^\tau_{p,\lambda}\subsetneq \ex^\tau_{p,\lambda^\prime}$. This is very important and non trivial, since it means that $\left\{\ex^\tau_p\right\}_{p\in \overline{R}_+}$ or even  $\left\{\ex^\tau_{p,\lambda}\right\}_{p\in \overline{\mathbb{R}}_+, \lambda\in\overline{\mathbb{R}}_+}$ are, in general, non trivial gradations of $\nalgebra=\ex^\tau_{\infty,\lambda}=\ex^\tau_{\infty,\lambda}$ for every $\lambda\in\overline{\mathbb{R}}_+$.

\begin{example}
	\label{exfunctionbadbehaved}
	Consider a function $f:\mathbb{R}\setminus\{0\}\to\mathbb{R}$ given by
	\begin{figure}[H]
		\centering
		
		\begin{minipage}{0.4\textwidth}
			\centering
			$$f(x)=\begin{cases}m & \textrm{if } (2e)^{-m-1}\leq |x| < (2e)^{-m}, \ m\in\mathbb{N},\\
			0 & \textrm{if } |x|\geq2e.\\ \end{cases}$$
		\end{minipage} %\begin{minipage}{0.55\textwidth}
		%	\centering
		%	\includegraphics[width=1.0\linewidth, height=0.3\textheight]{Figs/function}
%			\caption{Example of $\left(\int_{K}\cdot dx\right)$-exponentiable operator for $L_p\left(L_\infty(K),\int_{K}\cdot dx\right)$.}
%			\label{fig:func}
		%\end{minipage}
	\end{figure}
	
	This is a positive unbounded integrable function with compact support in $\mathbb{R}$ and
	$$\begin{aligned}
	\sum_{n=1}^\infty\frac{1}{n!}\int_{\mathbb{R}} f(x)^n dx &= 2\sum_{n=1}^\infty\frac{1}{n!}\sum_{m=1}^{\infty} m^n \left((2e)^{-m}-(2e)^{-m-1}\right)\\
	&=\frac{4e}{2e-1}\sum_{m=1}^\infty \left(2^{-m}-(2e)^{-m}\right)\\
	&=\frac{4e}{2e-1}\left(1-\frac{1}{2e-1}\right).\\
	\end{aligned}$$
	
	Hence, again we have that for any measurable set $K\in\mathbb{R}\setminus\{0\}$ the restriction of $f$ \mbox{to $K$} is an example of a $\left(\int_{K}\cdot dx\right)$-exponentiable operator for $L_1\left(L_\infty(K),\int_{K}\cdot dx\right)$, but it does not hold for $2f$. In fact, for any $k\in\mathbb{k}$,
	$$\begin{aligned}
	\sum_{n=1}^N\frac{1}{n!}\int_{\mathbb{R}} (2f(x))^{n} dx &=2 \sum_{n=1}^N\frac{1}{n!}\sum_{m=1}^{\infty} (2m)^{n} \left((2e)^{-m}-(2e)^{-m-1}\right)\\
	&=\frac{4e}{2e-1}\sum_{n=1}^N\frac{1}{n!}\sum_{m=1}^{\infty} (2m)^{n} (2e)^{-m}\\
	&>\frac{4e}{2e-1}\sum_{n=1}^N\frac{1}{n!}\sum_{m=k}^{\infty} (2k)^{n} (2e)^{-m}\\
	&=\frac{8e^2}{(2e-1)^2}(2e)^{-k}\sum_{n=1}^N\frac{(2k)^n}{n!}.
\end{aligned}$$

For any $M>0$, there exists $k\in\mathbb{N}$ such that $\left(\frac{e}{2}\right)^k-1>M$. Setting $\varepsilon=(2e)^{-k}$, there exists $N=N(k)>0$ such that $\displaystyle e^{2k}-\sum_{n=1}^N\frac{(2k)^n}{n!}<\varepsilon$. It follows now that
$$\begin{aligned}
(2e)^{-k}\sum_{n=1}^N\frac{(2k)^n}{n!}\geq (2e)^{-k}(e^{2k}-\epsilon)=\left(\frac{e}{2}\right)^k-\varepsilon (2e)^k\geq \left(\frac{e}{2}\right)^k-1>M.
\end{aligned}$$

\end{example}

Although we have presented examples just for $p=1$, it is enough to take the $p$-th root to obtain an example for any $p>1$.

\begin{example}
	In order to construct an example in a noncommutative von Neumann algebra it is sufficient that there exists a monotonic decreasing sequence of projections $(P_n)_{n\in\mathbb{N}}\in\nalgebra_p$ such that $P_n \xrightarrow[]{\|\cdot\|_1} 0$, which is true if there exists any $\tau$-measurable unbounded operator.
		
	In fact, fix $1\leq p$. If there exists such a sequence, we can suppose without loss of generality, by taking a subsequence if necessary, that $\tau(P_n)\leq \frac{1}{(e^n-1)2^n}$. Define the positive unbounded $\tau$-measurable operator $$A=\sum_{n=1}^\infty n^\frac{1}{p}(P_n-P_{n+1}).$$
	
	It follows from the definition that
	$$\sum_{m=1}^\infty \frac{1}{m!}\tau(|A|^{pm})=\sum_{m=1}^\infty \frac{1}{m!}\sum_{n=1}^\infty n^m\tau(P_n-P_{n+1})\leq \sum_{m=1}^\infty \sum_{n=1}^\infty \frac{n^m}{m!} \frac{1}{(e^n-1)2^n}= \sum_{n=1}^\infty \frac{1}{2^n}=1.$$
	Thus $A\in\ex^\tau_p$.
\end{example}

It is not difficult to see, with help of the spectral decomposition, that if an operator is in $L_p(\nalgebra,\tau)\cap L_q(\nalgebra,\tau)$ with $1\leq p < q < \infty$, then it is in $L_r(\nalgebra,\tau)$ for every $p\leq r \leq q$. More than that, it follows by analyticity and the Three-Line Theorem, a special case of the Riesz-Thorin Theorem, that:
\begin{equation}
\label{eq:riesz-thorin}
\begin{aligned}
&\|A\|_r \leq \|A\|_p^{\frac{p}{(q-p)}\left(\frac{q}{r}-1\right)}\|A\|_q^{{\frac{q}{(q-p)}\left(1-\frac{p}{r}\right)}},& \quad \textrm{ if } q<\infty;\\
&\|A\|_r \leq \|A\|_p^\frac{p}{r}\|A\|_\infty^{1-\frac{p}{r}},& \quad \textrm{ if } q=\infty.
\end{aligned}
\end{equation}
 
An analogous property holds for $(\tau,p)$-exponentiable operators:

\begin{proposition}
	Let $1\leq p < q \leq \infty$. Then $\ex^\tau_p\cap\ex^\tau_q\subset \ex^\tau_r$ for every $p\leq r\leq q$.
\end{proposition}
\begin{proof}
	Let $A\in \ex^\tau_p\cap\ex^\tau_q$, in particular $|A|^n\in L_p(\nalgebra,\tau)\cap L_q(\nalgebra,\tau)$ for all $n\in\mathbb{N}$.
	
	Using equation \eqref{eq:riesz-thorin} we get that, for all $p\leq r\leq q$,
	$$\||A|^n\|_r\leq \max\bigl\{\||A|^n\|_p,\||A|^n\|_q\bigr\}\leq \||A|^n\|_p+\||A|^n\|_q ,$$
	$$\begin{aligned}
	\sum_{n=1}^N \frac{\||A|^n\|_r}{n!}&=\sum_{n=1}^N \frac{\||A|^n\|_p+\||A|^n\|_q}{n!}\\
	&=\sum_{n=1}^\infty \frac{\||A|^n\|_p}{n!}+\sum_{n=1}^\infty \frac{\||A|^n\|_q}{n!}\\
	&<\infty.
	\end{aligned}$$	
\end{proof}

Although $\ex^\tau_p$, in general, are not vector spaces, they still have very convenient geometric structure for perturbations.

\begin{proposition}
	\label{exconvex}
	\begin{enumerate}[(i)]
		\item $\ex^\tau_p$ is a balanced and convex set;
		
		\item for every $A\in\ex^\tau_p$ and $B\in \nalgebra$ with $\|B\|\leq 1$, $BA\in \ex^\tau_p$;
		
		\item if $1\leq p, \, q,\, r\leq \infty$ are such that $\frac{1}{p}+\frac{1}{q}=\frac{1}{r}$ and $A,B \in \nalgebra_\tau$,
		$$\sum_{n=1}^\infty\frac{\tau(|A|^{np})}{n!} \quad , \quad \sum_{n=1}^\infty\frac{\tau(|B|^{nq})}{n!}<\infty \ \Rightarrow \ \sum_{n=1}^\infty\frac{\tau(|AB|^{nr})}{n!}<\infty;$$
		
		\item $\ex^\tau_{p,\infty}$ is a subspace of $L_p(\nalgebra)$.
	\end{enumerate}
\end{proposition}
\begin{proof}
	$(i)$ 	It is obvious that, for $A\in\ex^\tau_p$ and $|\lambda|\leq 1$ we have
	$$\begin{aligned}
	\sum_{n=1}^\infty\frac{\|\lambda A\|_{np}^n}{n!} &=\sum_{n=1}^\infty\frac{|\lambda|^n \|A\|_{np}^n}{n!}\leq\sum_{n=1}^\infty\frac{\|A\|_{np}^n}{n!}.\\
	\end{aligned}$$	
	
	Let $A,B\in \ex^\tau_p$ and let $0<\lambda<1$. Then,
	$$\begin{aligned}
	\sum_{n=1}^\infty\frac{\|\lambda A +(1-\lambda)B\|_{np}^n}{n!}&\leq\sum_{n=1}^\infty\frac{1}{n!}\left(\lambda \|A\|_{np}+(1-\lambda)\|B\|_{np}\right)^n\\
	&\leq\sum_{n=1}^\infty\frac{1}{n!}\max\left\{\lambda^n\|A\|_{np}^{n},(1-\lambda)^{n}\|B\|_{np}^{n}\right\}\\
	&\leq\sum_{n=1}^\infty\frac{1}{n!}\left(\lambda^n\|A\|_{np}^{n}+(1-\lambda)^{n}\|B\|_{np}^{n}\right)\\
	&\leq\sum_{n=1}^\infty\frac{1}{n!}\|A\|_{np}^{n}+\sum_{n=1}^\infty\frac{1}{n!}\|B\|_{np}^{n};\\
	\end{aligned}$$
	
	$(ii)$ it follows trivially from $(i)$ in  Theorem \ref{minkowski};
	
	$(iii)$ it follows from Corollary \ref{g2holder} that 
	$$\begin{aligned}
	\sum_{n=1}^N\frac{\tau\left(|AB|^{nr}\right)}{n!}&\leq \sum_{n=1}^N\frac{1}{n!}\tau\left(|A|^{np}\right)^{\frac{r}{p}}\tau\left(|B|^{nq}\right)^{\frac{r}{q}}\\
	&=\sum_{n=1}^N\left(\frac{\tau\left(|A|^{np}\right)}{n!}\right)^\frac{r}{p}\left(\frac{\tau\left(|B|^{nq}\right)}{n!}\right)^\frac{r}{q}\\
	&\leq\left(\sum_{n=1}^N\frac{\tau\left(|A|^{np}\right)}{n!}\right)^\frac{r}{p}\left(\sum_{n=1}^N\frac{\tau\left(|B|^{nq}\right)}{n!}\right)^\frac{r}{q}.\\
	\end{aligned}$$
	
	$(iv)$ Notice that $A\in\ex^\tau_{p,\infty}$ if and only if $\lambda A\in\ex^\tau_p$ for every $\lambda\in\mathbb{R}$. It follows by item $(i)$ that, if $\alpha, \, \beta \in \mathbb{C}$ and $A,B\in A\in\ex^\tau_{p,\infty}$,
	$$\alpha A+\beta B=\left(|\alpha|+|\beta| \right)\left(\frac{|\alpha|}{|\alpha|+|\beta|}\left(\frac{\alpha}{|\alpha|}A\right)+\frac{|\beta|}{|\alpha|+|\beta|}\left(\frac{\beta}{|\beta|}B\right)\right)\in \left(|\alpha|+|\beta| \right)\ex^\tau_p\subset \ex^\tau_{p,\infty}.$$
\end{proof}

The following lemma justifies the choice of the name ``exponentiable'' for such operators.

\begin{lemma}
\label{ExpansionalinLp}
For each $A\in\ex^\tau_{p,\lambda}$ and $B\eta \nalgebra$ positive, define $A(t)=B^{it}AB^{-it}$. Then, for $0\leq t < \lambda$,
$$\mathbbm{1}-Exp_r\left(\int_0^t;A(s) ds\right) \quad \textrm{ and } \quad \mathbbm{1}- Exp_l\left(\int_0^t;A(s) ds\right)\in L_p\left(\nalgebra,\tau\right).$$
\end{lemma}
\begin{proof}
	
	Since $A\in \nalgebra_\tau \Rightarrow A(t)\in \nalgebra_\tau$, Proposition \ref{measurablealgebra} implies that each term in the definition of these operators is in $\nalgebra_\tau$ except for the identity. 
		
	 In addition, using Theorem \ref{gholder} for $p_i=n$, $i=1, \ldots, n$, we have, for every $N>M$, that
	 
	\begin{equation}
	\label{eq:calculation1}
	\begin{aligned}
	\Bigg\|\sum_{n=N}^{M}\int_{0}^{t} dt_1\ldots \int_{0}^{t_{n-1}}dt_{n}&A(t_n)\ldots A(t_1)\Bigg\|_p\\	&\leq\sum_{n=N}^{M}{\int_{0}^{t} dt_1\ldots \int_{0}^{t_{n-1}}dt_{n}\left\|A(t_n)\ldots A(t_1)\right\|_p}\\	
	&=\sum_{n=N}^{M}\int_{0}^{t} dt_1\ldots \int_{0}^{t_{n-1}}dt_{n}\tau\left(\left|A(t_n)\ldots A(t_1)\right|^p\right)^\frac{1}{p}\\
	&=\sum_{n=N}^{M}\int_{0}^{t} dt_1\ldots \int_{0}^{t_{n-1}}dt_{n}\tau\left(\left|A\right|^{pn}\right)^\frac{1}{p}\\
	&=\sum_{n=N}^{M}\frac{t^n}{n!}\tau\left(\left|A\right|^{pn}\right)^\frac{1}{p}\\
	\end{aligned}
	\end{equation}
	which shows simultaneously,\iffalse as seen in equation \eqref{eq:calculation2},\fi \, that each term is in $L_p(\nalgebra,\tau)$ and the partial sum is a $\|\cdot\|_p$-Cauchy sequence. The thesis follows by completeness.
	
	%**********\textcolor{red}{Take care with the identity operator in the definition, this term is not in $L_p$}
\end{proof}

Hitherto, we have defined a set, namely $\ex^\tau_p$, that we assert is the right set to take our perturbation, but the reader should be warned after so many comments about the duality relations between $L_p$-spaces that we will demand some extra ``dual'' property on the original state. This motivates our next definition.

\begin{definition}
	Let $\nalgebra$ be a von Neumann algebra and $\tau$ a faithful normal semifinite trace on $\nalgebra$. We say that a state $\phi$ on $\nalgebra$ is $\|\cdot\|_p$-continuous if it is continuous on $\left(\nalgebra\cap L_p(\nalgebra,\tau),\|\cdot\|_p\right)$.
	
	Of course such a state can be continuously extended to $\left(L_p(\nalgebra,\tau),\|\cdot\|_p\right)$ in a unique way.
\end{definition}

To clarify the previous definition, remember that $\nalgebra\cap L_p(\nalgebra,\tau)$ (or even $\nalgebra\cap L_1(\nalgebra,\tau)$) is $\|\cdot\|_p$-dense in $L_p(\nalgebra,\tau)$.
\pagebreak

\begin{proposition}
	\label{defpcont}
Let $\nalgebra$ be a von Neumann algebra and $\tau$ a faithful normal semifinite trace on $\nalgebra$ and $1\leq p, \, q\leq \infty$ such that $\frac{1}{p}+\frac{1}{q}=1$. A state $\phi$ on $\nalgebra$ is $\|\cdot\|_p$-continuous if and only if there exists $H\in L_q(\nalgebra,\tau)$, $H\eta \mathfrak{M}_{\tau^\phi}$, such that 
$$\phi(A)=\tau_H(A) \quad \forall A\in\nalgebra,$$ in the sense of Lemma \ref{unboundedderivativeweight}.
\end{proposition}
\begin{proof}
	As mentioned in Definition \ref{defpcont}, we can continuously extended $\phi$ to $L_p(\nalgebra,\tau)$. By the dual relation between $L_p(\nalgebra,\tau)$ and $L_q(\nalgebra,\tau)$ stated in Theorem \ref{dualLp}, there exists $H\in L_q(\nalgebra,\tau)$ such that $\phi(A)=\tau(HA)$ for all $A\in L_p(\nalgebra,\tau)$. Such $H$ must be affiliated with $\mathfrak{M}_{\tau^\phi}$ by the same argument as in Theorem \ref{TPTRN}.
	
	In particular, $\phi(A)=\tau(HA)$ for all $A\in \nalgebra\cap L_p(\nalgebra,\tau)$, but $\nalgebra\cap L_p(\nalgebra,\tau)$ is \mbox{WOT-dense} in $\nalgebra$, since the trace is semifinite.
	
	The cases $p=1,\infty$, are analogous and the other part of the equivalence is trivial.
	
\end{proof}

The next two theorems can be seen as the key to guarantee Dyson's series is convergent. It is time to stress how important Araki's multiple-time KMS condition is for the theory. Here it is used with the same purposes of the original Araki's article \cite{Araki73}. Mentioning an interesting connection, this property is also used in the Araki's noncommutative $L_p$-spaces, what makes us believe there is a natural way to extend this result.
 
\begin{theorem}
	\label{TR0}
	Let $\nalgebra$ be a von Neumann algebra, $\phi$ a faithful state in $\nalgebra$, $p,q>1$ such that $\frac{1}{p}+\frac{1}{q}=1$, and $n\in\mathbb{N}$. Let also $\left(\hilbert_\phi,\Phi, \pi_\phi\right)$ be the GNS representation throughout $\phi$, $\tau$ a normal faithful semifinite trace on $B(\hilbert_\phi)$, $Q_i, J_\phi Q_i J_\phi\in L_{2mq}\left(B(\hilbert_\phi),\tau\right)$ such that $\|J_\phi Q_i J_\phi\|_{2mq}=\|Q_i\|_{2mq}$ for all $1\leq i,m\leq n$ and suppose $\phi$ is $\|\cdot\|_p$-continuous. Then, if $\Phi \in\Dom{Q_1}$ and $\Delta_\Phi^{\iu z_{j-1}}Q_{j-1}\ldots \Delta_\Phi^{\iu z_1}Q_1\Phi \in \Dom{Q_j}$ for every $-\frac{1}{2}\leq \Im{z_j} \leq 0$ and for every $2\leq j\leq n$,
	
	$$Q_n\Delta_\Phi^{\iu z_{n-1}}Q_{n-1}\ldots \Delta_\Phi^{\iu z_1}Q_1\Phi \in \Dom{\Delta_\Phi^{\iu z}} \textrm{ for } -\frac{1}{2}\leq\Im{z}\leq0 \textrm{ and }$$
	$$A^n(z_1,\ldots, z_n)\Phi \doteq \Delta_\Phi^{\iu z_n} Q_n \Delta_\Phi^{\iu z_{n-1}}Q_{n-1}\ldots \Delta_\Phi^{\iu z_1}Q_1\Phi$$ is analytic on $S_\frac{1}{2}=\displaystyle \left\{(z_1,\ldots,z_n)\in \mathbb{C}^n \ \middle | \ \Im{z_i}<0, \ 1\leq i\leq n, \textrm{ and } -\frac{1}{2}<\sum_{i=1}^n\Im{z_i}<0\right\}$ and bounded in its closure by
	$$ \left\|A^n(z_1,\ldots, z_n)\Phi\right\|\leq \|H\|_p^\frac{1}{2} \prod_{j=1}^{n}\|Q_i\|_{2nq}.$$
\end{theorem}

\begin{proof}
	Let's proceed by induction on $n$.
	
	For $n=1$, let $Q_1=U|Q_1|$ be the polar decomposition of $Q_1$ and $\displaystyle|Q_1|=\int_0^\infty \lambda dE^{|Q_1|}_\lambda$ the spectral decomposition of $|Q_1|$. Since $\Phi\in\Dom{Q_1}$, $\displaystyle Q_1\Phi=U\lim_{k\to \infty} Q_{1,k}\Phi$, where $\displaystyle Q_{1,k}=\int_0^k \lambda dE^{|Q_1|}_\lambda$.
	Define the following functionals on $\nanalytic\Phi$:
	$$\begin{aligned}
	f_k^z(A\Phi) &\doteq\ip{UQ_{1,k}\Phi}{\Delta_\Phi^{-\iu \bar{z}}A\Phi}_\phi,\\
	f^z(A\Phi)&\doteq\lim_{k\to\infty}f_k^z(A\Phi)=\ip{Q_1\Phi}{\Delta_\Phi^{-\iu \bar{z}}A\Phi}_\phi.\\	
	\end{aligned}$$
	
	Of course, for fixed $A\Phi$, $\bar{f_k}(z)=\overline{f_k^z(A\Phi)}$ is entire analytic and, using Theorem \ref{gholder} we obtain
	\begin{equation}
	\label{calculationX17}
	\begin{aligned}
	\left|\bar{f}(t)\right|&=\lim_{k\to\infty}\left|\bar{f}_k(t)\right|\\
	&=\lim_{k\to\infty}\left|\ip{\Delta_\Phi^{-\iu t}A\Phi}{UQ_{1,k}\Phi}_\phi\right|\\
	&\leq \left\|\Delta_\Phi^{-\iu t}A\Phi\right\| \lim_{k\to\infty} \|UQ_{1,k}\Phi\|\\
	& \leq\left\|A\Phi\right\|\lim_{k\to\infty}\phi\left(Q_{1,k} U^\ast U Q_{1,k} \right)^\frac{1}{2}\\
	&\leq \left\|A\Phi\right\|\lim_{k\to\infty}\tau\left(H^\frac{1}{2}Q_{1,k} U^\ast U Q_{1,k} H^\frac{1}{2}\right)^\frac{1}{2}\\
	& \leq\left\|A\Phi\right\|\|H\|_p^\frac{1}{2}\||Q_1|^2\|_q^\frac{1}{2}\\
	&\leq\left\|A\Phi\right\|\|H\|_p^\frac{1}{2}\|Q_1\|_{2q}, \quad \forall t\in \mathbb{R};\\
	\end{aligned}
	\end{equation}
	\begin{equation}
	\label{calculationX17.1}
	\begin{aligned}	
	\left|\bar{f}\left(t+\frac{1}{2}\iu\right)\right|&=\lim_{k\to\infty}\left|\bar{f}_k\left(t+\frac{1}{2}\iu\right)\right|\\
	&=\lim_{k\to\infty}\ip{\Delta_\Phi^{\frac{1}{2}}\Delta_\Phi^{-\iu t}A\Phi}{UQ_{1,k}\Phi}_\phi\\
	&\leq \left\|\Delta_\Phi^{-\iu t}A\Phi\right\| \lim_{k\to\infty} \|J_\Phi Q_{1,k} U^\ast\Phi\|\\
	&= \left\|A\Phi\right\| \lim_{k\to\infty} \| Q_{1,k} U^\ast\Phi\|\\
	&\leq \left\|A\Phi\right\|\lim_{k\to\infty}\phi\left(U Q_{1,k} Q_{1,k} U^\ast \right)^\frac{1}{2}\\
	&\leq \left\|A\Phi\right\|\lim_{k\to\infty}\tau\left(H^\frac{1}{2}U Q_{1,k}^2 U^\ast H^\frac{1}{2}\right)^\frac{1}{2}\\
	&\leq \left\|A\Phi\right\|\|H\|_p^\frac{1}{2}\||Q_1^\ast|^2\|_q^\frac{1}{2}\\
	&\leq \left\|A\Phi\right\|\|H\|_p^\frac{1}{2}\||Q_1|^2\|_q^\frac{1}{2}\\
	&\leq\left\|A\Phi\right\|\|H\|_p^\frac{1}{2}\|Q_1\|_{2q}, \quad \forall t\in \mathbb{R};\\	
	\end{aligned}
	\end{equation}
	which proves that the functional concerned is bounded for $-\frac{1}{2}\leq\Im{z}\leq0 $ as a consequence of the Maximum Modulus Principle. This bound also proves that if $\bar{f}_k\to \bar{f}$ uniformly for $-\frac{1}{2}\leq\Im{z}\leq0$, then $\bar{f}$ is analytic for $-\frac{1}{2}<\Im{z}<0$ and bounded for $-\frac{1}{2}\leq\Im{z}\leq0 $.
	
	Using first the Hahn-Banach Theorem to obtain an extension, also denoted by $f_z$, to the whole Hilbert space in such a way that $\|f_z\|\leq \|H\|_p^\frac{1}{2}\|Q_1\|_{2q}$, we know by the Riesz Representation Theorem, that there exists a $\Omega(z)\in \hilbert_\phi$ such that $f_z(\cdot)=\ip{\Omega(z)}{\cdot}_\phi$. Since $\nanalytic\Phi$ is dense, $\Omega(z)$ is unique. 
	
	So far we have that $Q_1\Phi\in \Dom{\left(\Delta_\Phi^{-\iu \bar{z}}\right)^\ast}=\Dom{\Delta_\Phi^{\iu z}}$ and $\bar{f}(z)=\ip{A\Phi}{\Delta_\Phi^{\iu z}Q_1\Phi}_\phi$ is analytic on $\left \{z\in\mathbb{C} \ \middle| \ -\frac{1}{2}< \Im{z}< 0\right\}$ and continuous on its closure, for every ${A\in \nanalytic}$.
	
	Since $\overline{\nanalytic\Phi}^{\|\cdot\|}=\overline{\nalgebra\Phi}^{\|\cdot\|}=\hilbert_\phi$, the vector-valued function $A(z)\Phi\doteq\Delta_\Phi^{\iu z}Q_1\Phi$ is weak analytic, hence, strong analytic on $\left \{z\in\mathbb{C} \ \middle| \ -\frac{1}{2}< \Im{z}< 0\right\}$ and
	$$\|A(z)\Phi\|\leq \|H\|_p^\frac{1}{2}\|Q_1\|_{2q} \quad \forall z\in \left \{z\in\mathbb{C} \ \middle| \ -\frac{1}{2}\leq \Im{z}\leq 0\right\}.$$
	
	Suppose now the hypothesis hold for $n\in \mathbb{N}$. We will use the same ideas: let $Q_{n+1}=U|Q_{n+1}|$ be the polar decomposition of $Q_{n+1}$ and $\displaystyle|Q_{n+1}|=\int_0^\infty \lambda dE^{|Q_{n+1}|}_\lambda$ the spectral decomposition of $|Q_{n+1}|$. Since $\Phi\in\Dom{Q_{n+1}}$, $\displaystyle Q_{n+1}\Phi=U\lim_{k\to \infty} Q_{n+1,k}\Phi$.
	
	$$
	f^{(z_1,\ldots, z_{n+1})}(A\Phi) =\ip{Q_{n+1}\Delta_\Phi^{\iu z_n}Q_n\ldots\Delta_\Phi^{\iu z_1}Q_1\Phi}{\Delta_\Phi^{-\iu \bar{z}_{n+1}}A\Phi}_\phi.$$
	
	Since $\bar{f}_k(z_1,\ldots,z_{n+1})\doteq\overline{f_k^{(z_1,\ldots, z_{n+1})}(A\Phi)}$ is an analytic function, it attains its maximum at an extremal point of $S_\frac{1}{2}$ (see \cite{Araki73} Corollary 2.2). Denoting $z_j=x_j+\iu y_j$, $x_j,\, y_j \in \mathbb{R}$ for all $1\leq j\leq n+1$, and repeating the calculations in equations \eqref{calculationX17} and \eqref{calculationX17.1}, first for the extremal points with $\Im{z_j}=0$ for all $1\leq j\leq n+1$, we get
	
	$(i)$ if $\Im{z_i}=0, \ 1\leq i\leq n$,
	$$\begin{aligned}
	\left|\bar{f}_k(z_1,\ldots, z_{n+1})\right|	&=\left|\ip{Q_{n+1}\Delta_\Phi^{\iu x_n}Q_n\ldots\Delta_\Phi^{\iu x_1}Q_1\Phi}{\Delta_\Phi^{-\iu x_{n+1}}A\Phi}_\phi\right|\\
	&\leq\left\|Q_{n+1}\tau^\phi_{x_n}(Q_n)\ldots\tau^\phi_{x_n+\cdots+x_1}(Q_1)\Phi\right\| \left\|\Delta_\Phi^{-\iu x_{n+1}}A\Phi\right\|\\
	&\leq \|A\Phi\| \tau\left(\left| Q_{n+1}\tau^\phi_{x_n}(Q_n)\ldots\tau^\phi_{x_n+\cdots+x_1}(Q_1)H^\frac{1}{2}\right|^2\right)^\frac{1}{2} \\
	&\leq\|A\Phi\| \left\|H^\frac{1}{2}\right\|_{2p}\ \prod_{i=1}^{n+1}\|Q_i\|_{2nq}.
	\end{aligned}$$
	
	$(ii)$ if $\Im{z_i}=0, \ 1\leq i\leq n$, $i\neq k$ and $\Im{z_k}=-\frac{1}{2}$, where $x_i=\Re{z_i}$
	
	$$\begin{aligned}
	&\left|\bar{f}_k(z_1,\ldots, z_{n+1})\right|\\
	&=\left|\ip{Q_{n+1}\Delta_\Phi^{\iu x_n}Q_n\ldots\Delta_\Phi^{\iu x_{k-1}}Q_{k-1}\Delta_\Phi^{\frac{1}{2}}Q_k\Delta_\Phi^{\iu x_1}Q_1\Phi}{\Delta_\Phi^{-\iu x_{n+1}}A\Phi}_\phi\right|\\
	&=\left|\left\langle Q_{n+1}\tau^\phi_{x_n}(Q_n)\ldots\tau^\phi_{x_1+\cdots+x_{k-1}}(Q_{k-1})J_\Phi\tau^\phi_{x_n+\cdots+x_1}(Q_1^\ast)\ldots \tau^\phi_{x_n+\cdots+x_{k}}(Q_{k}^\ast)J_\Phi\Phi,\right. \right.\\
		&\makebox[\textwidth]{$\displaystyle \hfill \left. \left. \Delta_\Phi^{-\iu x_{n+1}}A\Phi\right\rangle_\phi\right|$}\\	
	&\leq \left\|Q_{n+1}\tau^\phi_{x_n}(Q_n)\ldots\tau^\phi_{x_1+\cdots+x_{k-1}}(Q_{k-1})J_\Phi\tau^\phi_{x_n+\cdots+x_1}(Q_1^\ast)\ldots \tau^\phi_{x_n+\cdots+x_{k}}(Q_{k}^\ast)J_\Phi\right\| \|A\Phi\|\\
	&=\|A\Phi\|\times\\
		&\makebox[\textwidth]{$\displaystyle \hfill \times \tau\left(\left|Q_{n+1}\tau^\phi_{x_n}(Q_n)\ldots\tau^\phi_{x_1+\cdots+x_{k-1}}(Q_{k-1})J_\Phi\tau^\phi_{x_n+\cdots+x_1}(Q_1^\ast)\ldots \tau^\phi_{x_n+\cdots+x_{k}}(Q_{k}^\ast)J_\Phi H^\frac{1}{2}\right|^2\right)^\frac{1}{2}$}\\
	&=\|A\Phi\| \left\| Q_{n+1}\tau^\phi_{x_n}(Q_n)\ldots\tau^\phi_{x_1+\cdots+x_{k-1}}(Q_{k-1})J_\Phi\tau^\phi_{x_n+\cdots+x_1}(Q_1^\ast)\ldots \tau^\phi_{x_n+\cdots+x_{k}}(Q_{k}^\ast)J_\Phi H^\frac{1}{2}\right\|_2\\
	&\leq \|A\Phi\| \left\|H\right\|_p^\frac{1}{2} \prod_{i=1}^{k-1}\left\|\tau^\phi_{x_n+\cdots+x_i}(Q_i)\right\|_{2nq} \ \prod_{i=k}^{n+1}\left\|J_\phi\tau^\phi_{x_n+\cdots+x_i}(Q_i)J_\phi\right\|_{2nq}\\
	&= \|A\Phi\| \|H\|_p^\frac{1}{2} \prod_{i=1}^{n+1}\left\|Q_i\right\|_{2nq}\\
\end{aligned}$$
%	*************** lembrar da conta com $\|A\|$!!!
\end{proof}

The previous result depends a lot on the possibility of ``extending'' the trace, that is originally defined only in the algebra, to the algebra generated by $\nalgebra\cup\nalgebra^\prime$. One may try to define
$$\tau(J_\Phi |A| J_\Phi B)=\tau(|A|)\tau(|B|),$$
but it immediately fails, in general, because the application of this formula with either $A=\mathbbm{1}$ or $B=\mathbbm{1}$, since  $\tau(\mathbbm{1})=\infty$.

In order to relax the condition on the possibility of having a trace in all the $GNS$-represented algebra, we have to demand more regularity on the perturbation. The next theorem shows almost the same result as the previous one, with a little more restricted perturbation.

\begin{theorem}
\label{TR1}
Let $\nalgebra\subset B(H)$ be a von Neumann algebra, $\tau$ a normal faithful semifinite trace on $\nalgebra$, $\phi(\cdot)=\ip{\Phi}{\cdot\Phi}$ a state on $\nalgebra$ and $n\in\mathbb{N}$. Let also $n\in\mathbb{N}$, $p, \, q\geq1$ with $\frac{1}{p}+\frac{1}{q}=1$, \mbox{$Q_i\in L_{4mq}\left(\nalgebra,\tau\right)$} for all $1\leq i,m\leq n$ and suppose $\phi$ is $\|\cdot\|_p$-continuous.

Then, if $\Phi \in\Dom{Q_1}$ and $\Delta_\Phi^{\iu z_{j-1}}Q_{j-1}\ldots \Delta_\Phi^{\iu z_1}Q_1\Phi \in \Dom{Q_j}$ for every $-\frac{1}{2}\leq \Im{z_j} \leq 0$ and for every $2\leq j\leq n$,
$Q_n\Delta_\Phi^{\iu z_n}Q_{n-1}\ldots \Delta_\Phi^{\iu z_1}Q_1\Phi \in \Dom{\Delta_\Phi^{\iu z}}$ for $-\frac{1}{2}\leq\Im{z}\leq 0$ and $$A^n(z_1,\ldots, z_n)\Phi \doteq Q\Delta_\Phi^{\iu z_n}Q_n\ldots \Delta_\Phi^{\iu z_1}Q_1\Phi$$ is analytic on $S_\frac{1}{2}=\displaystyle \left\{(z_1,\ldots,z_n)\in \mathbb{C}^n \ \middle | \ \Im{z_i}<0, \ 1\leq i\leq n, \textrm{ and } -\frac{1}{2}<\sum_{i=1}^n\Im{z_i}<0\right\}$ and bounded in its closure by
$$ \left\|A^n(z_1,\ldots, z_n)\Phi\right\| \leq \|H\|_p^\frac{1}{2}\max_{0\leq l\leq n-1}\left\{\underbrace{\left(\prod_{j=1}^{l}\|Q_j\|_{4lq}\right)}_{=1 \textrm{ if } l=0}\left(\prod_{j=l+1}^{n}\|Q_j\|_{4(n-l)q}\right)\right\}\\.$$
\end{theorem}
\begin{proof}
	Let's proceed by induction on $n$.
	
	For $n=1$, let $Q_1=U|Q_1|$ be the polar decomposition of $Q_1$ and $\displaystyle|Q_1|=\int_0^\infty \lambda dE^{|Q_1|}_\lambda$ the spectral decomposition of $|Q_1|$. Since $\Phi\in\Dom{Q_1}$, $\displaystyle Q_1\Phi=U\lim_{k\to \infty} Q_{1,k}\Phi$, where $\displaystyle Q_{1,k}=\int_0^k \lambda dE^{|Q_1|}_\lambda$.
	 Define the following functionals on $\nanalytic\Phi$
	$$\begin{aligned}
	f_k^z(A\Phi) &\doteq\ip{UQ_{1,k}\Phi}{\Delta_\Phi^{-\iu \bar{z}}A\Phi}_\phi,\\
	f^z(A\Phi)&\doteq\lim_{k\to\infty}f_k^z(A\Phi)=\ip{Q_1\Phi}{\Delta_\Phi^{-\iu \bar{z}}A\Phi}_\phi.\\	
	\end{aligned}$$
	
	Of course, for fixed $A\Phi$, $\bar{f_k}(z)=\overline{f_k^z(A\Phi)}$ is entire analytic and, in the two lines of extremal points, we have
	\begin{equation}
	\label{calculationX17.2}
	\begin{aligned}
	\left|\bar{f}(t)\right|&=\lim_{k\to\infty}\left|\bar{f}_k(t)\right|\\
	&=\lim_{k\to\infty}\left|\ip{\Delta_\Phi^{-\iu t}A\Phi}{UQ_{1,k}\Phi}_\phi\right|\\
	&\leq \left\|\Delta_\Phi^{-\iu t}A\Phi\right\| \lim_{k\to\infty} \|UQ_{1,k}\Phi\|\\
	& \leq\left\|A\Phi\right\|\lim_{k\to\infty}\phi\left(Q_{1,k} U^\ast U Q_{1,k} \right)^\frac{1}{2}\\
	&\leq \left\|A\Phi\right\|\lim_{k\to\infty}\tau\left(H^\frac{1}{2}Q_{1,k} U^\ast U Q_{1,k} H^\frac{1}{2}\right)^\frac{1}{2}\\
	& \leq\left\|A\Phi\right\|\|H\|_p^\frac{1}{2}\||Q_1|^2\|_q^\frac{1}{2}\\
	&\leq\left\|A\Phi\right\|\|H\|_p^\frac{1}{2}\|Q_1\|_{2q} \quad \forall t\in \mathbb{R};\\	
	\left|\bar{f}\left(t+\frac{1}{2}\iu\right)\right|&=\lim_{k\to\infty}\left|\bar{f}_k\left(t+\frac{1}{2}\iu\right)\right|\\
	&=\lim_{k\to\infty}\ip{\Delta_\Phi^{\frac{1}{2}}\Delta_\Phi^{-\iu t}A\Phi}{UQ_{1,k}\Phi}_\phi\\
	&\leq \left\|\Delta_\Phi^{-\iu t}A\Phi\right\| \lim_{k\to\infty} \|J_\Phi Q_{1,k} U^\ast\Phi\|\\
	&= \left\|A\Phi\right\| \lim_{k\to\infty} \| Q_{1,k} U^\ast\Phi\|\\
	&\leq \left\|A\Phi\right\|\lim_{k\to\infty}\phi\left(U Q_{1,k} Q_{1,k} U^\ast \right)^\frac{1}{2}\\
	&\leq \left\|A\Phi\right\|\lim_{k\to\infty}\tau\left(H^\frac{1}{2}U Q_{1,k}^2 U^\ast H^\frac{1}{2}\right)^\frac{1}{2}\\
	&\leq \left\|A\Phi\right\|\|H\|_p^\frac{1}{2}\||Q_1^\ast|^2\|_q^\frac{1}{2}\\
	&\leq \left\|A\Phi\right\|\|H\|_p^\frac{1}{2}\||Q_1|^2\|_q^\frac{1}{2}\\
	&\leq\left\|A\Phi\right\|\|H\|_p^\frac{1}{2}\|Q_1\|_{2q} \quad \forall t\in \mathbb{R};\\	
	\end{aligned}
	\end{equation}
	which proves that the functional concerned is bounded for $-\frac{1}{2}\leq\Im{z}\leq0 $ because of the Maximum Modulus Principle. This bound also proves that if $\bar{f}_k\to \bar{f}$ uniformly for $-\frac{1}{2}\leq\Im{z}\leq0$, then $\bar{f}$ is analytic for $-\frac{1}{2}<\Im{z}<0$ and bounded for $-\frac{1}{2}\leq\Im{z}\leq0 $.
	
	Using first Hahn-Banach Theorem to obtain an extension, also denoted by $f_z$, to the whole Hilbert space in such a way that $\|f_z\|\leq \|H\|_p^\frac{1}{2}\|Q_1\|_{2q}$, we know by the Riesz's Representation Theorem, that there exists a $\Omega(z)\in \hilbert_\phi$ such that $f_z(\cdot)=\ip{\Omega(z)}{\cdot}_\phi$. Since $\nanalytic\Phi$ is dense, $\Omega(z)$ is unique. 

	So far we have that $Q_1\Phi\in \Dom{\left(\Delta_\Phi^{-\iu \bar{z}}\right)^\ast}=\Dom{\Delta_\Phi^{\iu z}}$ and $\bar{f}(z)=\ip{A\Phi}{\Delta_\Phi^{\iu z}Q_1\Phi}_\phi$ is analytic on $\left \{z\in\mathbb{C} \ \middle| \ -\frac{1}{2}< \Im{z}< 0\right\}$ and continuous on its closure, for every ${A\in \nanalytic}$.
	
	Since $\overline{\nanalytic\Phi}^{\|\cdot\|}=\overline{\nalgebra\Phi}^{\|\cdot\|}=\hilbert_\phi$, the vector-valued function $A(z)\Phi\doteq\Delta_\Phi^{\iu z}Q_1\Phi$ is weakly analytic, hence, strongly analytic on $\left \{z\in\mathbb{C} \ \middle| \ -\frac{1}{2}< \Im{z}< 0\right\}$ and
	$$\|A(z)\Phi\|\leq \|H\|_p^\frac{1}{2}\|Q_1\|_{2q} \quad \forall z\in \left \{z\in\mathbb{C} \ \middle| \ -\frac{1}{2}\leq \Im{z}\leq 0\right\}.$$
	
	Suppose now the hypothesis hold for $n\in \mathbb{N}$. We will use the same ideas: we can define the sequence $\displaystyle  Q^{k_{i}}_i=U_i \int_0^{k_i} \lambda dE^{|Q_{i}|}_\lambda$, where $Q_{i}=U_i|Q_i|$ is the polar decomposition of $Q_{i}$ for ever $i\leq i\leq n+1$. Define
	
	$$
	f^{(z_1,\ldots, z_{n+1})}_{k_1,\ldots,k_{n+1}}(A\Phi) =\ip{Q_{n+1}\Delta_\Phi^{\iu z_n}Q_n\ldots\Delta_\Phi^{\iu z_1}Q_1\Phi}{\Delta_\Phi^{-\iu \bar{z}_{n+1}}A\Phi}_\phi.
	$$
	
	For now, we will omit the superscript index on the operators to not overload the notation.
	
	Since $\bar{f}(z_1,\ldots,z_n)\doteq\overline{f^{(z_1,\ldots, z_{n+1})}_{k_1,\ldots,k_{n+1}}(A\Phi)}$ is an analytic function, it attains its maximum at an extremal point of $S$ (see \cite{Araki73} Corollary 2.2). Denoting $z_j=x_j+\iu y_j$ and repeating the calculations in equation \eqref{calculationX17.2}
	
	\noindent $(i)$ for the extremal points with $\Im{z_j}=0$ for all $1\leq j\leq n+1$ we get
	\begin{equation}
	\label{calculationX18}
	\begin{aligned}
	&\left|\bar{f}(x_1,\ldots, x_{n+1})\right|\\
	&=\left|\ip{Q_{n+1}\Delta_\Phi^{\iu x_n}Q_n\ldots\Delta_\Phi^{\iu x_1}Q_1\Phi}{\Delta_\Phi^{-\iu x_{n+1}}A\Phi}_\phi\right|\\
	&\leq \left\|\Delta_\Phi^{-\iu t}A\Phi\right\| \left\|Q_{n+1}\Delta_\Phi^{\iu x_n}Q_n\ldots\Delta_\Phi^{\iu x_1}Q_1\Phi\right\|\\
	&\leq \left\|A\Phi\right\| \ip{Q_{n+1}\tau^\phi_{x_n}(Q_n)\ldots\tau^\phi_{x_n+\cdots+x_1}(Q_1)\Phi}{Q_{n+1}\tau^\phi_{x_n}(Q_n)\ldots\tau^\phi_{x_n+\cdots+x_1}(Q_1)\Phi}^\frac{1}{2}\\
	& \leq\left\|A\Phi\right\|\left\|\tau^\phi_{x_n+\cdots+x_1}(Q_1^\ast)\ldots\tau^\phi_{x_n}(Q_n^\ast)Q_{n+1}^\ast Q_{n+1}\tau^\phi_{x_n}(Q_n)\ldots\tau^\phi_{x_n+\cdots+x_1}(Q_1)\Phi\right\|^\frac{1}{2}\underbrace{\|\Phi\|^\frac{1}{2}}_{=1}\\
	& \leq\left\|A\Phi\right\|\left(\tau\left(\left|\tau^\phi_{x_n+\cdots+x_1}(Q_1^\ast)\ldots\tau^\phi_{x_n}(Q_n^\ast)Q_{n+1}^\ast Q_{n+1}\tau^\phi_{x_n}(Q_n)\ldots\tau^\phi_{x_n+\cdots+x_1}(Q_1)H^\frac{1}{2}\right|^2\right)^\frac{1}{2}\right)^\frac{1}{2}\\
	& \leq\left\|A\Phi\right\|\left(\|H^\frac{1}{2}\|_{2p} \prod_{j=1}^{n+1}\|Q_j\|_{4(n+1)q} \ \prod_{j=1}^{n+1}\|Q_j^\ast\|_{4(n+1)q}\right)^\frac{1}{2}\\
	& \leq\left\|A\Phi\right\|\|H\|_p^\frac{1}{2} \prod_{j=1}^{n+1}\|Q_j\|_{4(n+1)q},  \quad \forall (x_1,\ldots,x_{n+1})\in \mathbb{R}^{n+1}. \\
	\end{aligned}
	\end{equation}
	
	\noindent$(ii)$ now for $\Im{z_j}=0$ for all $i\neq l$ and $\Im{z_l}=-\frac{1}{2}$ where $l\neq n+1$
	\begin{equation}
	\label{calculationX19}
	\begin{aligned}
	&\left|\bar{f}\left(x_1,\ldots, x_l-\frac{1}{2}\iu ,\ldots, x_{n+1}\right)\right|\\
	&=\left|\ip{Q_{n+1}\Delta_\Phi^{\iu x_n}Q_n\ldots\Delta_\Phi^{\frac{1}{2}}\Delta_\Phi^{\iu x_l}Q_l\ldots\Delta_\Phi^{\iu x_1}Q_1\Phi}{\Delta_\Phi^{-\iu x_{n+1}}A\Phi}_\phi\right|\\
	&=\left|\left\langle \tau^\phi_{x_{n+1}}(Q_{n+1})\ldots \tau^\phi_{x_{n+1}+\cdots+x_{l-1}}(Q_{l+1})\Delta_\Phi^{\frac{1}{2}}\tau^\phi_{x_{n+1}+\cdots+x_l}(Q_l)\ldots\tau^\phi_{x_{n+1}+\cdots+x_1}(Q_1)\Phi, \right. \right.\\
		&\makebox[\textwidth]{$ \hfill \displaystyle \left. \left. \Delta_\Phi^{-\iu x_{n+1}}A\Phi\right\rangle_\phi\right|$}\\
	&=\left|\left\langle\tau^\phi_{x_{n+1}}(Q_{n+1})\ldots \tau^\phi_{x_{n+1}+\cdots+x_{l+1}}(Q_{l+1}) J_\Phi\tau^\phi_{x_{n+1}+\cdots+x_1}(Q_1^\ast)\ldots\tau^\phi_{x_{n+1}+\cdots+x_l}(Q_l^\ast)\Phi,\right.\right.\\
		&\makebox[\textwidth]{$ \hfill \displaystyle \left. \left. \Delta_\Phi^{-\iu x_{n+1}}A\Phi\right\rangle_\phi\right|$}\\
	&=\left\|A\Phi\right\|\left\|\tau^\phi_{x_{n+1}}(Q_{n+1})\ldots \tau^\phi_{x_{n+1}+\cdots+x_{l+1}}(Q_{l+1})J_\Phi\tau^\phi_{x_{n+1}+\cdots+x_1}(Q_1^\ast)\ldots\tau^\phi_{x_{n+1}+\cdots+x_l}(Q_l^\ast)\Phi\right\|\\
	&=\left\|A\Phi\right\|\left\langle\tau^\phi_{x_{n+1}}(Q_{n+1})\ldots \tau^\phi_{x_{n+1}+\cdots+x_{l+1}}(Q_{l+1})J_\Phi\tau^\phi_{x_{n+1}+\cdots+x_1}(Q_1^\ast)\ldots\tau^\phi_{x_{n+1}+\cdots+x_l}(Q_l^\ast)\Phi, \right.\\
		&\makebox[\textwidth]{$\displaystyle \hfill \left. \tau^\phi_{x_{n+1}}(Q_{n+1})\ldots \tau^\phi_{x_{n+1}+\cdots+x_{l+1}}(Q_{l+1})J_\Phi\tau^\phi_{x_{n+1}+\cdots+x_1}(Q_1^\ast)\ldots\tau^\phi_{x_{n+1}+\cdots+x_l}(Q_l^\ast)J_\Phi\Phi\right\rangle_\phi^\frac{1}{2}$}\\
	&=\left\|A\Phi\right\|\left\langle J_\Phi\tau^\phi_{x_{n+1}+\cdots+x_l}(Q_l)\ldots\tau^\phi_{x_{n+1}+\cdots+x_1}(Q_1)J_\Phi^2\tau^\phi_{x_{n+1}+\cdots+x_1}(Q_1^\ast)\ldots\tau^\phi_{x_{n+1}+\cdots+x_l}(Q_l^\ast)\Phi ,\right.\\
		&\makebox[\textwidth]{$\displaystyle \hfill \left.\tau^\phi_{x_{n+1}}(Q_{n+1})\ldots\tau^\phi_{x_{n+1}+\cdots+x_{l+1}}(Q_{l+1})\tau^\phi_{x_{n+1}+\cdots+x_{l+1}}(Q_{l+1}^\ast)\ldots\tau^\phi_{x_{n+1}}(Q_{n+1}^\ast)\Phi\right\rangle_\phi^\frac{1}{2}$}\\
	&\leq \left\|A\Phi\right\|\left\|\tau^\phi_{x_{n+1}+\cdots+x_l}(Q_l)\ldots\tau^\phi_{x_{n+1}+\cdots+x_1}(Q_1)\tau^\phi_{x_{n+1}+\cdots+x_1}(Q_1^\ast)\ldots\tau^\phi_{x_{n+1}+\cdots+x_l}(Q_l^\ast)\Phi\right\|^\frac{1}{2}\times\\
		&\makebox[\textwidth]{$\displaystyle \hfill \times \left\|\tau^\phi_{x_{n+1}}(Q_{n+1})\ldots\tau^\phi_{x_{n+1}}(Q_{l+1})\tau^\phi_{x_{n+1}}(Q_{l+1}^\ast)\ldots\tau^\phi_{x_{n+1}}(Q_{n+1}^\ast)\Phi\right\|^\frac{1}{2}$}\\
	& \leq\left\|A\Phi\right\|\left(\|H\|_p \prod_{j=1}^{l}\||Q_j|^4\|_{lq}\right)^\frac{1}{4}\left(\|H\|_p \prod_{j=l+1}^{n+1}\||Q_j|^4\|_{(n-l+1)q}\right)^\frac{1}{4}\\
	& \leq\left\|A\Phi\right\|\|H\|_p^\frac{1}{2}\left(\prod_{j=1}^{l}\|Q_j\|_{4lq}\right)\left(\prod_{j=l+1}^{n+1}\|Q_j\|_{4(n+1-l)q}\right) \quad \forall (x_1,\ldots,x_{n+1})\in \mathbb{R}^{n+1}.\\
	\end{aligned}
	\end{equation}
	
	\noindent$(iii)$ finally, for $\Im{z_j}=0$ for all $i\neq n+1$ and $\Im{z_{n+1}}=-\frac{1}{2},$
	
	\begin{equation}\begin{aligned}
	\label{calculationX20}
	\left|\bar{f}\left(x_1,\ldots, x_l,\ldots,x_n, x_{n+1}-\frac{1}{2}\iu \right)\right|
	&=\left|\ip{Q_{n+1}\Delta_\Phi^{\iu x_n}Q_n\ldots\Delta_\Phi^{\iu x_1}Q_1\Phi}{\Delta_\Phi^{-\iu x_{n+1}}\Delta_\Phi^{\frac{1}{2}}A\Phi}_\phi\right|\\
	&=\left|\ip{\tau^\phi_{x_{n+1}}(Q_{n+1})\ldots\tau^\phi_{x_{n+1}+\cdots+x_1}(Q_1)\Phi}{J_\Phi A^\ast\Phi}_\phi\right|\\
	&=\left|\ip{J_\Phi A^\ast J_\Phi \Phi}{\tau^\phi_{x_{n+1}}(Q_{n+1}^\ast)\ldots\tau^\phi_{x_{n+1}+\cdots+x_1}(Q_1^\ast)\Phi}_\phi\right|\\
	&\leq\left\|A\Phi\right\| \left\|\tau^\phi_{x_{n+1}}(Q_{n+1}^\ast)\ldots\tau^\phi_{x_{n+1}+\cdots+x_1}(Q_1^\ast)\Phi\right\|\\
	&\leq\left\|A\Phi\right\|\|H\|_p^\frac{1}{2} \prod_{j=1}^{n+1}\|Q_j\|_{4(n+1)q}, \forall (x_1,\ldots,x_{n+1})\in \mathbb{R}^{n+1},\\
	\end{aligned}\end{equation}
where the last line follows by the same calculation done in equations \eqref{calculationX18} and \eqref{calculationX19}.

The last step is to remember that we omitted the superscripts and notice that
$$\begin{aligned}
\lim_{k_{n+1}\to \infty}\ldots\lim_{k_1\to \infty} f^{(z_1,\ldots, z_{n+1})}_{k_1,\ldots,k_{n+1}}(A\Phi)& =\ip{Q_{n+1}^{k_{n+1}}\Delta_\Phi^{\iu z_n}Q_n\ldots\Delta_\Phi^{\iu z_1}Q_1^{k_1}\Phi}{\Delta_\Phi^{-\iu \bar{z}_{n+1}}A\Phi}_\phi\\
&=\ip{Q_{n+1}\Delta_\Phi^{\iu z_n}Q_n\ldots\Delta_\Phi^{\iu z_1}Q_1\Phi}{\Delta_\Phi^{-\iu \bar{z}_{n+1}}A\Phi}_\phi,\\
\end{aligned}$$
so $\bar{f}(z_1,\ldots, z_{n+1})=\overline{\ip{Q_{n+1}\Delta_\Phi^{\iu z_n}Q_n\ldots\Delta_\Phi^{\iu z_1}Q_1\Phi}{\Delta_\Phi^{-\iu \bar{z}_{n+1}}A\Phi}_\phi}$ is the limit of a sequence of analytic functions uniformly bounded, thus, analytic on $$S_\frac{1}{2}=\displaystyle \left\{(z_1,\ldots,z_{n+1})\in \mathbb{C}^{n+1} \ \middle | \ \Im{z_i}<0, \ 1\leq i\leq n+1, \textrm{ and } -\frac{1}{2}<\sum_{i=1}^{n+1}\Im{z_i}<0\right\}$$ and bounded on its closure, as desired.

As we saw in equation \eqref{calculationX20} the term $\Delta_\Phi^{z_{n+1}}$ does not interfere with the conclusion for $-\frac{1}{2}<\Im{z_{n+1}}<0$ and, by the very same argument used above to obtain a continuous linear extension of $\bar{f}(z_1,\ldots, z_{n+1})$, it follows that $Q_n\Delta_\Phi^{\iu z_n}Q_{n-1}\ldots \Delta_\Phi^{\iu z_1}Q_1\Phi \in \Dom{\Delta_\Phi^{\iu z}}$ for $-\frac{1}{2}<\Im{z}<0$.

\end{proof}
\begin{remark}
In contrast to what we did in equation \eqref{calculationX18}, for $\Im{z_i}=0$ for all $1\leq i\leq n+1$, it holds that
\begin{equation}
\begin{aligned}
\left|\bar{f}_k(z_1,\ldots, z_{n+1})\right|&=\lim_{k\to\infty}\left|\bar{f}_k(z_1,\ldots, z_{n+1})(t)\right|\\
&=\lim_{k\to\infty}\left|\ip{Q_{n+1,k}\Delta_\Phi^{\iu x_n}Q_n\ldots\Delta_\Phi^{\iu x_1}Q_1\Phi}{\Delta_\Phi^{-\iu x_{n+1}}A\Phi}_\phi\right|\\
&\leq \left\|\Delta_\Phi^{-\iu t}A\Phi\right\| \lim_{k\to\infty} \left\|UQ_{n+1,k}\Delta_\Phi^{\iu x_n}Q_n\ldots\Delta_\Phi^{\iu x_1}Q_1\Phi\right\|\\
&\leq \left\|\Delta_\Phi^{-\iu t}A\Phi\right\| \lim_{k\to\infty} \left\|UQ_{n+1,k}\tau^\phi_{x_n}(Q_n)\ldots\tau^\phi_{x_1+\cdots+x_n}(Q_1)\Phi\right\|\\
& \leq\left\|A\Phi\right\|\lim_{k\to\infty}\phi\left(UQ_{n+1,k}\tau^\phi_{x_n}(Q_n)\ldots\tau^\phi_{x_1+\cdots+x_n}(Q_1)\right)^\frac{1}{2}\\
&\leq \left\|A\Phi\right\|\lim_{k\to\infty}\tau\left(H^\frac{1}{2} \tau^\phi_{x_1+\cdots+x_n}(Q_1^\ast)\ldots \tau^\phi_{x_n}(Q_n^\ast)Q_{n+1,k} U^\ast \right. \\ 
&\hspace{3.5cm} \left. UQ_{n+1,k}\tau^\phi_{x_n}(Q_n)\ldots\tau^\phi_{x_1+\cdots+x_n}(Q_1) H^\frac{1}{2}\right)^\frac{1}{2}\\
& \leq\left\|A\Phi\right\|\|H\|_p^\frac{1}{2} \lim_{k\to\infty}\left(\tau\left(|Q_{n+1,k}|^{2nq}\right)^\frac{1}{nq}\prod_{j=1}^{n}\tau\left(|Q_i|^{2nq}\right)^{\frac{1}{nq}}\right)^\frac{1}{2}\\
& \leq\left\|A\Phi\right\|\|H\|_p^\frac{1}{2} \prod_{j=1}^{n+1}\|Q_i\|_{2nq} \quad \forall t\in \mathbb{R}.\\
\end{aligned}
\end{equation}
\end{remark}

\begin{corollary}
	\label{CR1}
	Let $\nalgebra\subset B(\hilbert)$ be a von Neumann algebra, $\tau$ a normal faithful semifinite trace on $\nalgebra$ and $\phi(\cdot)=\ip{\Phi}{\cdot\Phi}$ a state on $\nalgebra$. Let also $\frac{1}{p}+\frac{1}{q}=1$, $\lambda\in\overline{\mathbb{R}}_+$, $Q\in \ex^\tau_{4q,\lambda}$ and suppose $\phi$ is $\|\cdot\|_p$-continuous. Then, if $\Phi \in\Dom{Q}$,
	$$\Phi(Q)\doteq\sum_{n=0}^\infty \int dt_1\ldots \int_{S_n}dt_{n} \left(\frac{t}{2\lambda}\right)^n\Delta_\Psi^{ t_n}Q\Delta_\Psi^{ t_{n-1}}Q\ldots \Delta_\Psi^{t_1}Q\Phi,$$ 
	where 	 $S_n=\displaystyle \left\{(t_1,\ldots,t_n)\in \mathbb{R}^n \ \middle | \ t_i>0, \ 1\leq i\leq n, \textrm{ and } 0\leq\sum_{i=1}^n t_i\leq\frac{1}{2}\right\}$, is absolutely and uniformly convergent for $t<\lambda$.
\end{corollary}
\begin{proof}
	By Theorem \ref{TR1} $\Delta_\Psi^{\iu z_n}Q\Delta_\Psi^{\iu z_{n-1}}Q\ldots \Delta_\Psi^{\iu z_1}Q A\Phi$ is well defined and
	$$\begin{aligned}
	\big\|\Delta_\Psi^{\iu z_n}Q\Delta_\Psi^{\iu z_{n-1}}&Q\ldots \Delta_\Psi^{\iu z_1}QA\Phi\big\|\\
	&\leq \left\|A\Phi\right\|\|H\|_p^\frac{1}{2}\max_{0\leq l\leq n}\left\{\underbrace{\left(\prod_{j=1}^{l}\|Q_j\|_{4lq}\right)}_{=1 \textrm{ if } l=0}\left(\prod_{j=l+1}^{n}\|Q_j\|_{4(n-l)q}\right)\right\}\\
	&=  \left\|A\Phi\right\|\|H\|_p^\frac{1}{2}\max_{0\leq l\leq n}\left\{\underbrace{\|Q\|_{4lq}^{l}}_{=1 \textrm{ if } l=0}\|Q\|_{4(n-l)q}^{n-l}\right\}\\
	&=  \left\|A\Phi\right\|\|H\|_p^\frac{1}{2}\max_{0\leq l\leq \lfloor \frac{n}{2}\rfloor}\left\{\underbrace{\|Q\|_{4lq}^{l-1}}_{=1 \textrm{ if } l=0}\|Q\|_{4(n-l)q}^{n-l}\right\}\\
	&= \left\|A\Phi\right\|\|H\|_p^\frac{1}{2}\max_{0\leq l\leq \lfloor \frac{n}{2}\rfloor}\left\{\underbrace{\tau\left(|Q|^{4lq}\right)^\frac{l}{4lq}}_{=1 \textrm{ if } l=0}\tau\left(|Q|^{4(n-l)q}\right)^\frac{n-l}{4(n-l)q}\right\}.
	\end{aligned}$$
	
	Notice that, for $\displaystyle Q_m=\int_0^m \lambda dE_\lambda^{|Q|}$, we have that $$f_m(z)\doteq\tau\left(Q_m^{4 z q}\right)^\frac{1}{4q}\tau\left(Q_m^{4(n-z)q}\right)^\frac{1}{4q}$$
	is an analytic function in the region $\{z\in \mathbb{C} \ | \ 1\leq\Re{z}\leq n-1\}$, hence its modulus in the region mentioned is assumed when $\Re{z}=1$ or $\Re{z}=n-1$ by the Maximum Modulus Principle. In these cases, we have
	$$	|f_m(1+\iu t)|= |f_m(n-1+\iu t)|=\|Q_m\|_{4q}\|Q_m\|_{4(n-1)q}^{n-1}.$$
	
	As usual, taking the limit $m\to \infty$ we obtain
	$$\tau\left(Q^{4 z q}\right)^\frac{1}{4q}\tau\left(Q^{4(n-z)q}\right)^\frac{1}{4q}\leq \|Q\|_{4q}\|Q\|_{4(n-1)q}^{n-1} \quad \forall z\in\{w\in\mathbb{C} \ | \ 1\leq \Re{w} \leq n-1 \}.$$
	
	Finally, the series
	$$\sum_{n=0}^\infty \int_{0}^{t} dt_1\ldots \int_{0}^{t_n}dt_{n} \Delta_\Psi^{ t_n}Q\Delta_\Psi^{ t_{n-1}}Q\ldots \Delta_\Psi^{ t_1}QA\Phi$$ is $\|\cdot\|$-convergent. In fact, considering first the case $Q\in \ex^\tau_{4q,\lambda}$ with $0<\lambda<\infty$, there exists $N\in\mathbb{N}$ such that, for all $ k,l >N$, $\displaystyle \lambda \sum_{n=k}^l\frac{\lambda^{n-1}\|Q\|_{4(n-1)q}^{n-1}}{(n-1)!}<\frac{\epsilon}{2}$ and $\displaystyle \frac{\|Q\|_{4q}}{N}< 1$, thus
	
	$$\begin{aligned}
	\Bigg \|\sum_{n=k}^l \int_{0}^{t} dt_1\ldots \int_{0}^{t_{n-1}}dt_{n} \Delta_\Psi^{\iu t_n}Q&\Delta_\Psi^{\iu t_{n-1}}Q\ldots \Delta_\Psi^{\iu t_1}QA\Phi\Bigg\|\\
	&\leq \sum_{n=k}^l \int_{0}^{t} dt_1\ldots \int_{0}^{t_{n-1}}dt_{n} \left \|\Delta_\Psi^{\iu t_n}Q\Delta_\Psi^{\iu t_{n-1}}Q\ldots \Delta_\Psi^{\iu t_1}QA\Phi\right\|\\
	&\leq \sum_{n=k}^l\frac{t^n\max\left\{\|Q\|_{4q}\|Q\|_{4(n-1)q}^{n-1},\|Q\|_{4nq}^{n}\right\}}{n!}\\
	&\leq \sum_{n=k}^l\frac{t^n\left(\|Q\|_{4q}\|Q\|_{4(n-1)q}^{n-1}+\|Q\|_{4nq}^{n}\right)}{n!}\\
	&\leq  \sum_{n=k}^l \frac{\|Q\|_{4q}}{n}\frac{t^n\|Q\|_{4(n-1)q}^{n-1}}{(n-1)!}+\frac{t^n\|Q\|_{4nq}^{n}}{n!}\\
	&<\epsilon.
	\end{aligned}$$
	
For the case $\lambda=\infty$, just remember that $\displaystyle \ex^\tau_{p,\infty}=\bigcap_{\lambda\in\overline{ \mathbb{R}}_+}\ex^\tau_{p,\lambda}$.
\end{proof}

Notice that, for the KMS condition, the interesting case is the case $\lambda=\frac{1}{2}$ as one can see in the definition of the expansional, in Lemma \ref{lemma5.3} or even in the original article \cite{Araki73}. More explicit, we can express the relation between to states in terms os the relative hamiltonian, $Q$, by
$$\Psi=\sum_{n=0}^{\infty}{(-1)^n\int_{0}^{\frac{1}{2}} dt_1\ldots \int_{0}^{t_{n-1}}dt_{n}\Delta^{t_n}_\Phi Q\Delta^{t_{n-1}-t_n}_\Phi\ldots \Delta^{t_1-t_2}_\Phi Q\Delta^{t_1}_\Phi}\Phi.$$

Another comment in this direction is about the interpretation of $\lambda$ in $\ex^\tau_{p,\lambda}$. As one can check in \cite{sakai91}, in the case of bounded operators for which our result gives an extension, we have equation \eqref{Duhamel}
$$e^{t(A+B)}e^{-tB}\xi= Exp_l\left(\int_0^t;A(s) ds\right)\xi \, , \qquad \xi \in G(B),$$
where $G(A)$ are the set of geometric vectors\footnote{The author discovered this concept in Sakai's book \cite{sakai91} a little after discovering the first version of the theorems in this chapter.} with respect to $B$, which are defined as the vectors with the property that there exists a positive constant $M_\xi$ such that $\|B^n\xi\|\leq M_\xi^n \|\xi\|$. The similarity with our approach is notorious, nevertheless they are not nearly the same concept, as one can notice by Proposition \ref{noconstant}.

It is also possible to obtain a complex version of this formula, namely
$$e^{z(A+B)}e^{-zB}\xi=\sum_{n=0}^{\infty}{\int_{0}^{1} dt_1\ldots \int_{0}^{t_{n-1}}dt_{n} z^n A(t_1)\ldots A(t_n)}\xi  \, , \qquad \xi \in G(B),$$
which evidences more our interests in the case $\lambda=\frac{1}{2}$.

We already know that the modular automorphism groups of $\phi$ and $\psi$ can be related as follows:
$$\begin{aligned}
u^{\phi\psi}_t&=Exp_r\left(\int_0^t;-\iu \tau^\psi_s(Q) ds\right)\\
\hat{u}^{\phi\psi}_t&=Exp_l\left(\int_0^t;\iu \tau^\psi_s(h) ds\right).\\
\end{aligned}$$
$$\begin{aligned}
\left(u^{\phi\psi}_t\right)^\ast&=\hat{u}^{\phi\psi}_t\\
u^{\phi\psi}_t \hat{u}^{\phi\psi}_t&=\hat{u}^{\phi\psi}_t u^{\phi\psi}_t=\mathbbm{1}\\
u^{\phi\psi}_t \tau^\psi_t(A)&=\tau^\phi_t(A) \hat{u}^{\phi\psi}_t \, , \qquad A\in \nalgebra.
\end{aligned}$$
Thus, Lemma \ref{ExpansionalinLp} says that this operator is in $L_{4q}(\nalgebra,\tau)$ if $Q$ is $(\tau,p,\lambda)$-exponentiable. Moreover, the KMS vector representation of the $\|\cdot\|_p$-continuous state is in the domain of $u^{\phi\psi}_t$ for $t<\lambda$.

We would add at this point that the author does not believe these results prove any kind of stability of KMS states. This belief is based on the necessity of adding a ``dual'' continuity property on the state, hence, all the stability we proved seems to be a consequence of that continuity. An important exception is the stability of the domain of the Modular Operator, which allows us to extend the multiple-time KMS condition to unbounded operators, proved in Theorem \ref{TR0} and \ref{TR1}.

\begin{proposition}
	Let $(Q_n)_{n\in\mathbb{N}} \subset\ex_{4q,\lambda}^\tau$ be a sequence such that $Q_n\xrightarrow{\|\cdot\|_{4mq}}Q\in\ex_{4q,\lambda}^\tau$, $\|Q_n\|_{4mq}\leq \|Q\|_{4mq}$ and $\|Q-Q_n\|_{4mq} \leq M$ for all $m\in\mathbb{N}$. In addition, suppose that, for each fixed $n\in\mathbb{N}$,
	$\Phi \in\Dom{Q_1}$ and $\Delta_\Phi^{\iu z_{j-1}}Q_{j-1}\ldots \Delta_\Phi^{\iu z_1}Q_1\Phi \in \Dom{Q_j}$ for every $\frac{1}{2}\leq z_j \leq 0$ and for every $2\leq j\leq n$, where $Q_j$ can be either $Q_n$ or $Q$. Then
	$$Exp_{l,r}\left(\int_0^t;Q_n(s) ds\right)\Phi\to Exp_{l,r}\left(\int_0^t;Q(s) ds\right)\Phi \ , \quad t<\lambda.$$
\end{proposition}
\begin{proof}
	First, notice that $Q$ and $Q_n$, $n\in\mathbb{N}$, are densely defined closed operators. Furthermore, they have a common core because of the increasing hypothesis and the properties of $\tau$-dense subsets.
	
	Define 
	$$A^m_j(z_1,\ldots, z_n)\Phi=\Delta_\Phi^{\iu z_m} Q \ldots\Delta_\Phi^{\iu z_{j-1}}Q\Delta_\Phi^{\iu z_{j}}(Q-Q_n)\Delta_\Phi^{\iu z_{j+1}}Q_n\ldots \Delta_\Phi^{\iu z_1}Q_n\Phi$$.
	
	Using a telescopic sum argument, we have, for $m>1$,
	$$\begin{aligned}
	\left\|\Delta_\Phi^{\iu z_m} Q\ldots \Delta_\Phi^{\iu z_1}Q\Phi-\Delta_\Phi^{\iu z_n} Q_n\ldots \Delta_\Phi^{\iu z_1}Q_n\Phi\right\|\\
	&\hspace{-4.2cm}=\left\|\sum_{j=1}^m A^m_j(z_1,\ldots, z_n)\Phi \right\|\\
	&\hspace{-4.2cm}\leq \sum_{j=1}^m \left\| A^m_j(z_1,\ldots, z_n)\Phi \right\|\\
	&\hspace{-4.2cm}\leq \sum_{j=1}^m \|H\|_p^\frac{1}{2}\max_{0\leq l\leq m-1}\left\{\|Q\|_{4lq}^l\|Q-Q_n\|_{4(m-l)q}\|Q_n\|_{4(m-l)q}^{m-l-1}\right\}\\
	&\hspace{-4.2cm}\leq m \, \|H\|_p^\frac{1}{2}\max_{0\leq l\leq m-1}\left\{\|Q\|_{4lq}^l\|Q\|_{4(m-l)q}^{m-l-1}\|Q-Q_n\|_{4(m-l)q}\right\}\\
	&\hspace{-4.2cm}= m \, \|H\|_p^\frac{1}{2}\max_{0\leq l\leq m-1}\left\{\|Q\|_{4lq}^l\|Q\|_{4(m-l)q}^{m-l}\frac{\|Q-Q_n\|_{4(m-l)q}}{\|Q\|_{4(m-l)q}}\right\}.\\
	\end{aligned}$$
	
	Applying equation \eqref{eq:riesz-thorin} to the inequality above we get
	
	$$\begin{aligned}
	\left\|\Delta_\Phi^{\iu z_m} Q\ldots \Delta_\Phi^{\iu z_1}Q\Phi-\Delta_\Phi^{\iu z_n} Q_n\ldots \Delta_\Phi^{\iu z_1}Q_n\Phi\right\|\\
	&\hspace{-3.9cm} \leq m \, \|H\|_p^\frac{1}{2}\|Q\|_{4q}\|Q\|_{4(m-1)q}^{m-1}\max_{0\leq l\leq m-1}\left\{\frac{\|Q-Q_n\|_{4(m-l)q}}{\|Q\|_{4(m-l)q}}\right\}.\\
	\end{aligned}$$
	Hence,
	\begin{equation}
	\label{calcx34}
	\begin{aligned}
	\left\|Exp_{l,r}\left(\int_0^t;Q(s) ds\right)\Phi-Exp_{l,r}\left(\int_0^t;Q_n(s) ds\right)\Phi\right\|\\
	&\hspace{-7cm}\leq \|H\|_p^\frac{1}{2}\|Q-Q_n\|_{4q}+\\
	&\hspace{-7cm} +\sum_{m=2}^{\infty}\frac{t^m}{m!} m \, \|H\|_p^\frac{1}{2}\|Q\|_{4q}\|Q\|_{4(m-1)q}^{m-1}\max_{0\leq l\leq m-1}\left\{\frac{\|Q-Q_n\|_{4(m-l)q}}{\|Q\|_{4(m-l)q}}\right\}\\
	&\hspace{-7cm} = \|H\|_p^\frac{1}{2}\|Q-Q_n\|_{4q}+ \\
	& \hspace{-7cm} + \|H\|_p^\frac{1}{2}\|Q\|_{4q} \, t \sum_{m=2}^{\infty}\frac{t^{m-1}}{(m-1)!} \|Q\|_{4(m-1)q}^{m-1}\max_{0\leq l\leq m-1}\left\{\frac{\|Q-Q_n\|_{4(m-l)q}}{\|Q\|_{4(m-l)q}}\right\}.\\
	\end{aligned}
	\end{equation}
	
	Finally, let $\epsilon>0$ be given. Since $Q\in \ex_{4q,\lambda}^\tau$, there exists $m_0\in\mathbb{N}$ such that, for all $m\leq m_0$, 
	$$\sum_{m=M}^{\infty}\frac{t^{m-1}}{(m-1)!} \|Q\|_{4(m-1)q}^{m-1}<\frac{\epsilon}{3M}.$$
	By hypothesis, there also exists $n_0\in\mathbb{N}$ such that 
	$$\frac{\|Q-Q_n\|_{4mq}}{\max\left\{\|Q\|_{4nq},1\right\}}<\epsilon \left[3 \lambda \|H\|_p^\frac{1}{2}  \sum_{m=2}^{\infty}\frac{t^{m-1}}{(m-1)!} \|Q\|_{4(m-1)q}^{m-1}\right]^{-1} \, , \quad \forall n\geq n_0. $$
	
	It follows from equation \eqref{calcx34} that
	$$\left\|Exp_{l,r}\left(\int_0^t;Q(s) ds\right)\Phi-Exp_{l,r}\left(\int_0^t;Q_n(s) ds\right)\Phi\right\|<\epsilon \, , \quad \forall n\geq n_0.$$
\end{proof}

One of the consequences of the previous proposition is that the sequence of Araki's perturbations obtained by the upper cut in the spectral decomposition of the modulus of a $\ex_{4q,\lambda}^\tau$-perturbation converges to the perturbation described in Corollary \ref{CR1}.
% flatex input end: [Chapter6/chapter6.tex]
% My proofs 
\renewcommand\thechapter{C}
% flatex input: [Conclusion/conclusion.tex]

\chapter{Conclusions and Perspectives}

% ************************* Conclusion *************************

We have already mentioned in the introduction that our goal was to extend the theory of perturbations of KMS states to some class of unbounded perturbations using noncommutative $L_p$ spaces. We achieved our goal defining $\ex^\tau_{p,\lambda}$ and proving that this is a set that includes some unbounded operators whose elements give rise to convergent Dyson's series related to a $\|\cdot\|_p$-continuous state. This is Corollary \ref{CR1}.

As important as the Corollary \ref{CR1} are the theorems that make it possible. Theorem \ref{TR0} and \ref{TR1} provide a good understanding of the domain of the Modular Operator with respect to a $\|\cdot\|_p$-continuous state, as well as important bounds for the multiple-time KMS condition, which is well defined and analytic thanks to those theorems.

The unsurprising, but valuable result that Araki's perturbation theory is a special case of our results, namely, the case $p=1$ and $q=\infty$ in the aforementioned results, shows that we made a gradation of an important theorem with many physical applications. 

No less important, our review on the theory of weights and Noncommutative $L_p$-Spaces yielded some interesting results, for example, Theorem \ref{gholder} which is a generalization of H\"older's inequality that is the core to the results in Chapter \ref{chapExtensionPerturb} and many other small results. In fact, ``review'' is not the appropriate word here, since most of the topic is not  developed in the standard literature.  

We finish this work with a huge amount of unanswered interesting questions, some of which, not answered due to lack of time. Working on some of them is still in our plans.

The first one has been raised on the very begining of this work. We would like to understand the relation between the unbounded perturbations described using noncommutative $L_p$ spaces and that of  J. Derezinski, V. Jaksic and C.-A. Pillet presented in \cite{Derezinski03}. This task seems quite difficult to be completely achieved, since we don't even know examples of perturbations as in \cite{Derezinski03}, but our theory may lead to an example, which would be a great accomplishment.

Another question that has been in our minds from the very beginning of this thesis is how to extend our results to the Araki's noncommutaive $L_p$-spaces. In fact, it was our first idea to use this description, since the multiple-time KMS condition seems natural in this description, as well as, suitable duality and H\"older inequality.

One of the questions that appears during the development concerns Lemma \ref{ExpansionalinLp}. It says that the expansional has $\tau$-dense domain. We have shown in Theorem \ref{TR1} some interesting vectors that lie in the domain of the expansional. But, what are the other vector that lies in the domain of the expansional, if there are any?

Thanks to Prof. Jean-Bernard Bru, whom I immensely thank for that, the author got in touch with an interesting article, \cite{nittis2016}, published a short time ago that shows applications of noncommutative $L_p$-spaces in Linear Response Theory. In this article, derivations are presented in the context of noncommutative $L_p$-spaces and applied to some physical problems. At that time those perturbations were unknown by the author, but Sakai's version of bounded perturbations using derivations was well known. The natural question is if there exists an approach to our results using derivations.

Physical applications of the theory developed in Chapter \ref{chapExtensionPerturb} are very desirable too. As the article \cite{nittis2016} suggests, linear response is a good candidate for that purpose. In addition, find states with physical significance, \eg quasi-free states, that are $\|\cdot\|_p$-continuous is fundamental for applications. 

We also didn't have time to write down some results we believe we already have. This is the case of an analogous of Corollary \ref{PerturbationsareDense} showing that the vector states obtained by perturbations constitute a dense set.
% flatex input end: [Conclusion/conclusion.tex]

% My proofs 
\renewcommand\thechapter{S}

\renewcommand{\thechapter}{\arabic{chapter}}

% ****************************** Back Matter *****************************
% Backmatter should be commented out, if you are using appendices after References
%\backmatter

% ***************************** Appendices ****************************

\begin{appendices} % Using appendices environment for more functunality

% flatex input: [Appendix1/appendix1.tex]
% ******************************* Thesis Appendix A ****************************
\chapter{Complementary Results in Functional Analysis} 
\label{AppCRFA}

\begin{definition}[Direct Sum]
\label{DS}
Let $\left(X_\alpha\right)_{\alpha}$ be a family of Banach spaces, the direct sum of this family is defined as follows
$$\bigoplus_\alpha X_\alpha=\left\{(x_\alpha)_\alpha \in \prod_{\alpha}X_\alpha \ \middle| \ \sum_{\alpha} \|x_\alpha\| < \infty \right\}$$
\end{definition}

\begin{remark}
A sum of non-negative numbers over an arbitrary set of index is defined as the supremum of the sum for all finite number of terms. It is well know that a sum such as in the previous definition does make sense only if it vanishes in a cocountable number of index.

It is not difficult to see this direct sum is the closure of the usual direct sum.
\end{remark}

\section{Hahn-Banach Theorem}

We will refer to the next result as Hahn-Banach Theorem, or more specifically, as Geometric Hahn-Banach Theorem but it was proved long after S. Banach's works or even H. Hahn did the generalization known nowadays as (Analytical) Hahn-Banach Theorem.

\begin{theorem}[Hahn-Banach] \index{theorem! Hahn-Banach}
	\label{realHBT}
	Let $V$ be a vector space over $\mathbb{R}$, let $p:V\to \mathbb{R}$ be a sublinear functional and let $f:V_0 \to \mathbb{R}$ be a linear functional defined on a subspace $V_0\subset V$ such that $f(x) \leq p(x) \ \forall x \in V_0$. Then, there exists a linear functional $\tilde{f}:V \to \mathbb{R}$ de $f$ such that $ \tilde{f}_{|_{V_0}}=f$ and $\tilde{f}(x) \leq p(x) \ \forall x \in V$.
\end{theorem}

We will not present the proof of the previous theorem since it is a standard result that can be found in almost any Functional Analysis's book, such as \cite{Kreyszig89}.

Another standard result found in the same book was obtained by F. Bohnenblust and A. Sobczyk in 1938 as a generalization of an idea by F. Murray in a paper of 1936.

\begin{theorem}[Hahn-Banach for Normed Spaces]\index{theorem! Hahn-Banach}
	\label{normedHBT}
	Let $(V,||.||)$ be a normed vector space and let $f:V_0 \to \mathbb{K}$ be a bounded linear functional defined in a subspace $V_0\subset V$. Then, there exists a bounded linear functional $\tilde{f}:V \to \mathbb{K}$ such that $ \tilde{f}_{|_{V_0}}=f$ and  $$\sup_{\substack{x\in V\\ ||x||\leq1}}{|\tilde{f}(x)|}=\sup_{\substack{x\in V_0\\ ||x||\leq 1}}{|f(x)|}.$$
\end{theorem}

\begin{theorem}[de Mazur-Dieudonn\'e]\index{theorem! Mazur Dieudonn\'e}
\label{TMD}
Let $V$ be a topological vector space, $M$ a subspace of $V$ and $A\subset V$ an open convex subset with $A\cap M=\emptyset$, then there exists a maximal closed subspace $H$ of $V$ disjoint of $A$ and containing $M$.
\end{theorem}

\begin{proof}
First, suppose $V$ is a topological vector space over $\mathbb{R}$.

Let $a\in A$, then, $A-a$ is an open and convex subset which contains the origin, thus the Minkowski functional $\rho_{A-a}$ is a continuous sublinear functional satisfying 
\begin{equation}
\label{mbola}
A-a = \{x\in V| \rho_{A-a}(x)<1\} \textrm{ e } A = \{x\in V | \rho_{A-a}(x-a)<1\}
\end{equation}

So, since $M\cap A =\emptyset$, it follows that $\rho_{A-a}(x-a)\geq 1 \ \forall x \in M$.

Define now $N=\langle M \cup \{a\} \rangle$ and let $\phi: N \to \mathbb{R}$ be given by $\phi(x-\lambda a)=\lambda$ which is clearly linear and:

If $x\in M$ and $\lambda>0$ then $\displaystyle \phi(x-\lambda a)=\lambda \leq \lambda \rho_{A-a}\left(\frac{x}{\lambda} -a\right)=\rho_{A-a}(x-\lambda a)$.

If $x\in M$ and $\lambda<0$ then $\phi(x-\lambda a)=\lambda <0 \leq \rho_{A-a}(x-\lambda a)$.

This means $\phi$ is dominated by $\rho_{A-a}$ in $N$ and, using the Hahn-Banach Theorem, we obtain an extension $\tilde{\phi}$ of $\phi$ such that $\tilde{\phi}(x)\leq \rho_{A-a}(x) \ \forall x \in V$, in particular, $\tilde{\phi}$ is continuous, due to this and equation \ref{mbola}, we must have $\tilde{\phi}(x)<1 \ \forall x \in A-a$.

Define the maximal subspace $H=\ker{(\tilde{\phi})}$, which is closed since $\tilde{\phi}$ is continuous, and contain $M=\ker{(\phi)}$. Moreover,
$$x\in H \Rightarrow 0=\tilde{\phi}(x)=\tilde{\phi}(x-a) + \tilde{\phi}(a)=\tilde{\phi}(x-a)+\phi(a) \leq \rho_{A-a}(x-a)-1 $$
It means, $\displaystyle \rho_{A-a}(x-a)\geq 1 \ \forall x\in H \Rightarrow H\cap A=\emptyset$ by \ref{mbola}.

For the complex topological vector space case, $M_\mathbb{R} \subset V_\mathbb{R}$ is a subspace, and by the previous proof, there exists a maximal closed subspace $H$ of $V_\mathbb{R}$ disjoint of $A$ such that $M \subset H$. Identify the maximal closed subspace with the kernel of a continuous functional $\phi \in V_\mathbb{R}^*$ such that $H=\ker{(\phi)}$ and consider $\tilde{H} = H\cap \iu H$, which is also disjoint of $A$, and the functional $\tilde{\phi} \in V^*$ given by $\tilde{\phi}(x)=\phi(x)-\iu \phi(\iu x)$.
$$\ker{(\tilde{\phi})}=\{x \in V \ | \ \phi(x)=0 \textrm{ e } \phi(\iu x)=0\}=H\cap iH$$ and thus $\tilde{H}$ is a maximal closed  subspace of $V$.
$$M=\iu M \subset \iu H \textrm{ e } M \subset H \Rightarrow M \subset H\cap \iu H.$$

\end{proof}

\begin{corollary}
\label{TMD2}
Let $V$ be a topological vector space, $M$ an affine linear manifold of $V$ and $A\subset X$ an open convex subset disjoint of $M$, then there exists a closed hyperplane $H$ of $V$ disjoint of $A$ and containing $M$.
\end{corollary}

\begin{corollary}
\label{CTMD}
Let $V$ be a locally convex space, $M$ a closed affine linear manifold of $V$ and $K\subset X$ a compact convex subset which if disjoint of $M$, then there exists a closed hyperplane $H$ of $V$ which is disjoint of $K$ and contain $M$ .
adon\end{corollary}
\begin{proof} Let $U$ be a neighbourhood of $0$ such that $(K+U)\cap M =\emptyset$. Since $V$ is locally convex, we can assume $U$ is a convex set, hence $K+U$ is an open convex (not empty) subset of $V$ which is disjoint of $M$. By Corollary \ref{TMD2}, there exists a hyperplane $H$ disjoint of $K$ and containing $M$.

\end{proof}

\begin{corollary}
\label{C2TMD}
Let $V$ be a locally convex space and $M \subset V$ a subspace. Then $x\in \overline{M}$ if, and only if, $x^*(x)=0$ for all $x^* \in V^*$ which vanish in $M$.
\end{corollary}
\begin{proof}

($\Rightarrow$) Obvious.

($\Leftarrow$) Of course $\overline{M}$ is a closed subspace of $V$, if $x \notin \overline{M}$ then we fall back in the conditions of Corollary \ref{CTMD} since $\{x\}$ is a compact convex set which does not intercept $M$. We conclude that there exists $x^* \in V^*$ such that $M\subset \ker{(x^*)}$ e $\{x\}\cap \ker{(x^*)}=\emptyset \Rightarrow x^*(x)\neq 0.$

\end{proof}

\begin{corollary}
	\label{C3TMD}
	Let $V$ be a locally convex space and $x\in V$, if $x^\ast(x)=0$ for all $x^* \in V^*$ then $x=0$.
\end{corollary}
\begin{proof}
Of course, the set of functionals vanishing in the subspace $M=\{0\}$ is $V^\ast$. By Corollary \ref{C2TMD} we conclude that $x\in \overline{M}=\{0\}$.

\end{proof}

\section{Krein-Milman Theorem}\index{theorem! Krein-Milman}
\label{AppKMT}

\begin{definition}[Face]

Let $V$ a topological vector space and $C$ an convex subset, a non-empty closed and convex set $F\subset C$ is said to be an \emph{extremal set} or a \emph{face} of $C$ if given $x,y \in C$ and $\lambda \in (0,1)$ with $\lambda x +(1-\lambda)y \in F$ then $ x,y \in F$.

\end{definition}

That is, a face is a set such that if it contains any point in the interior of a straight segment, then it contains the whole segment.

\begin{definition}
\label{ExtDef}
Let $V$ be an topological vector space and $C$ a convex subset. An extremal point in $C$ is a point $x\in C$ such that $\{x\}$ is a face of $C$.

We denote by $\mathcal{E}(C)= \{x \in C  \ | \ x \textrm{ is an extremal point of C}\}$. \glsdisp{extr}{\hspace{0pt}}

\end{definition}

\begin{proposition}
Let $V$ be a topological vector space and $C$ a convex subset, the following conditions are equivalent:
\begin{enumerate}[(i)]
	\item $x \in \mathcal{E}(C)$;
	\item $x=\lambda y +(1-\lambda) z$, with $y,z \in C$ and $\lambda \in (0,1) \Rightarrow x=y=z$; 
	\item $C\setminus\{x\}$ is convex;
\end{enumerate}
\end{proposition}
\begin{proof}

$(i)\Leftrightarrow (ii)$ It follows from definition.

$(ii)\Leftrightarrow (iii)$ Let $y,z \in C\setminus\{x\}$ and $\lambda \in (0,1)$, then follows from convexity of $C$ that $\lambda y +(1-\lambda) z \in C$. On the other hand the extremicity of $x$ implies $\lambda y +(1-\lambda) z \neq x$ and then $C\setminus\{x\}$ is convex.

Let $y,z \in C$  and $\lambda \in (0,1)$. Suppose by absurd that $y\neq x$, so we would have $z\neq x$ and hence $x=\lambda y +(1-\lambda) z \in C\setminus\{x\}$ since it is convex, a contradiction, so we must have $x=y=z$.

\end{proof}

\begin{proposition}
\label{exex}
Let $V$ be a locally convex space and $C\subset V$ a non-empty compact convex subset, then $\mathcal{E}(C)\neq\emptyset$.
\end{proposition}
\begin{proof}
Denote by $\mathcal{F}$ the family of all faces of $C$ partially ordered by $\leq$, where $F_1\leq F_2$ if, and only if, $F_2\subset F_1$. Of course such family is non-empty since $C\in\mathcal{F}$.

Let $\mathcal{F}_0\subset\mathcal{F}$ be a chain. Since $\mathcal{F}_0$ is a totally ordered set and its elements are closed, $\mathcal{F}_0$ has the finite intersection property and it follows from compactness of $C$ that $\displaystyle \bigcap_{F\in \mathcal{F}_0}F \neq \emptyset$. It is easy to see that $\displaystyle \bigcap_{F\in \mathcal{F}_0}F$ is a face of $C$, since intersections preserve the closed convex and extreme properties of a set, furthermore, it is clearly an upper bound of $\mathcal{F}_0$. It follows then by Zorn's lemma that there exists $\tilde{F} \subset \mathcal{F}$ maximal.

Let us show $\tilde{F}$ is a unitary set. In order to do that, suppose it is not true, that is, take $x,y \in \tilde{F}$ with $x\neq y$. Consider then the compact and convex set $D=\{x\}$ and the linear affine manifold $M=\{0\}+y$, by Theorem \ref{TMD2} there exists a hyperplane $H$ which contains $M$ and is disjoint of $D$.

Let $f \in V^*$ and $c \in \mathbb{K}$ such that $H=f^{-1}(\{c\})$ and define 
$$ \tilde{F}_0=\left\{x\in \tilde{F} \ \middle | \ f(x)=\inf_{y \in \tilde{F}}{f(x)}\right\}.$$

Let's show that $\tilde{F}_0$ is a face of $\tilde{F}$ and consequently a face of $C$, in fact, suppose $x=\lambda y + (1-\lambda) z$ with $x\in \tilde{F}_0$, $y,z \in \tilde{F}$ and $\lambda \in (0,1)$, then
$$\inf_{y \in \tilde{F}}{f(x)}=f(\lambda y + (1-\lambda) z)=\lambda f(y) + (1-\lambda) f(z) \leq \inf_{y \in \tilde{F}}{f(x)} \Rightarrow f(x)=f(y)=f(z)$$
It follows from the definition that $y,z \in \tilde{F}_0$ and thus it is a face. Note now that we cannot have $x$ and $y$ simultaneously as elements of $\tilde{F}_0$, so $\tilde{F} \leq \tilde{F}_0$ and this contradicts the maximality of $\tilde{F}$. Therefore $\tilde{F}$ is unitary and this guarantees the existence of maximal points of $C$.

\end{proof}

\begin{theorem}[Krein-Milman]\index{theorem! Krein-Milman}
\label{TKM}
Let $V$ be a locally convex space and $C$ a compact convex subset, then $C=\cchull{\mathcal{E}(C)}$.
\end{theorem}
\begin{proof}
By Proposition \ref{exex}, $\mathcal{E}(C)\neq \emptyset$. Suppose that $C\setminus \cchull{\mathcal{E}(C)}\neq \emptyset$ and take $x\in C\setminus \cchull{\mathcal{E}(C)}$, note that $\cchull{\mathcal{E}(C)} \subset C$ is compact since it is a closed set contained in a compact one. By theorem \ref{TMD2} there exists a closed hyperplane $H$ containing $\{x\}$ and disjoint of $\cchull{\mathcal{E}(C)}$.

Let $f \in V^*$ and $c \in \mathbb{K}$ such that $H=f^{-1}(\{c\})$, we can assume without loss of generality that $f(x)=c < f(y) \ \forall y \in \cchull{\mathcal{E}(C)}$, and take $\displaystyle F=\left\{x\in C \ \middle | \ f(x)=\inf_{y \in C}{f(x)}\right\}$. We have already seem in the proof of Proposition \ref{exex} that $F$ is a proper face of $C$ and that $\mathcal{E}(F)\neq \emptyset$, furthermore, if $y\in \mathcal{E}(F)$, $\displaystyle f(y)=\inf_{z \in C}{f(z)} \leq f(x) = c < f(w) \ \forall w \in \cchull{\mathcal{E}(C)}$, hence $y \notin \cchull{\mathcal{E}(C)}$. On the other hand we must have $\mathcal{E}(F)\subset \mathcal{E}(C)$. This leads us to a contradiction, so we conclude that $C=\cchull{\mathcal{E}(C)}$.

\end{proof}

\begin{lemma}
% Sakai pag 10 1.6.1 
\label{extremealidentity}
Let $\calgebra$ be a $C^\ast$-algebra and $\mathcal{S}$ be its unit sphere, then $\mathcal{S}$ has an extremal point if, and only if, $\calgebra$ has an identity.
\end{lemma}
\begin{proof}
$(\Rightarrow)$ If $\calgebra$ has an identity $\mathbbm{1}$, we can write it as $\mathbbm{1}=\frac{A+B}{2}$ with $A,B \in \mathcal{S}$. It follows that $\mathbbm{1}= \frac{\tilde{A}+\tilde{B}}{2}$ with $\tilde{A}=\frac{A+A^\ast}{2}$ and $\tilde{B}=\frac{B+B^\ast}{2}$. Now, since $\tilde{A}=2\mathbbm{1}-\tilde{B}$, they are self-adjoint elements (they commute) and by Theorem \ref{ST}, $\tilde{A},\tilde{B}\geq \mathbbm{1}$. But $\tilde{A},\tilde{B} \in \mathcal{S}$, hence $\tilde{A}\leq \mathbbm{1}$ and $\mathbbm{1}\leq\tilde{B}$, thus we conclude that $\tilde{A}=\tilde{B}=\mathbbm{1}$.

Returning to the definition of the operators, $A=2\mathbbm{1}-A^\ast$ and it follows that $A$ is a normal operator such that $2A=AA^\ast+A^\ast A$, hence positive. By the above argument $A=\mathbbm{1}$ and the using the analogous argument for $B$, it follows that $\mathbbm{1}$ is an extremal point.

$(\Leftarrow)$ 
Suppose now $A \in \mathcal{E}(\mathcal{S})$. Of course $\sigma(A^\ast A),\sigma(AA^\ast)\subset {0,1}$, otherwise it is easy (again by \ref{ST}) to construct a positive operator $B\subset \calgebra$ such that $\|B\| \leq 1$, $\|A \pm B\| =1$, in particular, $A^\ast A$ and $AA^\ast$ are projections.

% Let $h=A^\ast A +A A^\ast$ and $\calgebra_1$ the smalest $C^\ast$-subalgebra such that $h\in\calgebra_1$.

Now, let $B\in \{C-C A^\ast A-A^\ast A C A^\ast A+A^\ast A C A^\ast A  \ | \ C \in \calgebra\}$ such that $\|B\|\leq1$. A straight forward calculation using that $A^\ast A$ is a projection shows that $B^\ast A A^\ast B=0$, thus $\| A^\ast B\|=\|B^\ast A(B^\ast A)^\ast\|^{\frac{1}{2}}=0 \Rightarrow B^\ast A=A^\ast B=0$ and $A^\ast A B^\ast B=0$.
So, we must have 
\begin{equation}\label{eq2}\begin{aligned}
\|A\pm B\|&
=\|(A^\ast\pm B^\ast)(A\pm B)\|^{\frac{1}{2}}\\
&=\|A^\ast A\pm(A^\ast B+B^\ast A)-B^\ast B\|^{\frac{1}{2}}\\
&=\|A^\ast A+B^\ast B\|^{\frac{1}{2}}=\max\{\|A\|,\|B\|\}\\
&\leq 1.\end{aligned} \end{equation}

From equation \ref{eq2} we conclude, since $A$ is an extremal point,
$$
A=\frac{1}{2}\frac{A+B}{\|A+B\|}+\frac{1}{2}\frac{A-B}{\|A-B\|} \Rightarrow B=0 \Rightarrow$$
$$\{C-C A^\ast A-A A^\ast C +A A^\ast C A^\ast A  \ | \ C \in \calgebra\}=\{0\}
$$

Now, define $h=A^\ast A+AA^\ast$, and suppose it does not have an inverse, that means, via the identification in \ref{TGN} and Theorem \ref{TI}, there exists a positive operator $B \in \calgebra$ with $\|B\|=1$ and $hB=0$. But then $$\|AB\|=\|B A^\ast\|=\|B A^\ast A B\|^{\frac{1}{2}}\leq \|B h B\|=0.$$

Doing the analogous estimation to $\|BA\|$ we conclude
$$\| B-B A^\ast A-A A^\ast B +A A^\ast B A^\ast A\|=\|B\|= 1.$$

Since it is a contradiction we must have that $hh^{-1}$ is an identity.

\end{proof}

\begin{definition}
	\label{defAppId}
	Let $\calgebra$ be a $C^\ast$-algebra. A net $(E_i)_{i\in I}\subset\calgebra$ is said to be an approximate identity \index{approximate identity} if, for all $A\in\calgebra$,
	$$\lim_{i\in I}\|AE_i-A\|=0.$$
\end{definition}

In the previous definition we defined a right approximate identity, but in a $C^\ast$-algebra, right and left approximate identity are the same.

The next result is a original proof of a well known result in the theory of $C^\ast$-algebras. After write this proof and consider to publish it, the author become aware that the same technique was used to prove it in Takesaki's book \cite{Takesaki2003}.
%### Essa demonstração tambem é minha, mas, depois de fazê-la, descobri que a mesma técnica é usada no livro do Takesaki. Eu pensei nela enquanto lia o livro do Sakai e vi uma proposição que dizia que se a bola em uma C^\ast-algebra tem ponto extremo, então ela tem identidade (Sakai pag 10 1.6.1 \label{extremealidentity})
\begin{theorem}[Segal] \index{theorem! Segal}
\label{ExisAppId}
Every $C^\ast$-algebra $\calgebra$ contains a positive
approximate identity.
\end{theorem}
\begin{proof}

First we recall Theorem \ref{TGN}, to reduce to the case of a subspace of $B(\hilbert)$.

Let $\mathcal{P}=\{A\in\calgebra \ | \ A\geq0\}$. Then $\mathcal{P}$ is a convex pointed cone, since by the Banach-Alaoglu Theorem the unit ball is weak-operator compact, $K=\overline{\mathcal{P}\cap B_{B(\hilbert)}}^{WOT}=\overline{\mathcal{P}}^{WOT}\cap\overline{\mathcal{S}_{B(\hilbert)}}^{WOT}$ must be so, thus $K=\overline{conv}^{WOT}{\left(\mathcal{E}(K)\right)}$ due to Theorem \ref{TKM}. On the other hand, according to Lemma \ref{extremealidentity} there must exists an identity in the weak-operator closure of $\mathcal{S}_{B(\hilbert)}$. Let $(E_i)_{i\in I} \subset \mathcal{S}_{B(\hilbert)}$ be a net, which is convergent to this identity denoted by $\mathbbm{1}$. This means, for all $x, y\in \hilbert$ we have $\ip{x-E_i x}{y}\rightarrow 0$.

Using Theorem \ref{sqrt}, for each index $\alpha$ there exists a unique positive operator $\sqrt{\mathbbm{1}-E_i}$ such that $\sqrt{\mathbbm{1}-E_i}^2=\mathbbm{1}-E_i$. It follows that
$$\begin{aligned}
\|x-E_i x\|^4
&=\ip{(\mathbbm{1}-E_i)^{\frac{1}{2}}x}{(\mathbbm{1}-E_i)^{\frac{3}{2}}x}\\
&\leq \|(\mathbbm{1}-E_i)^{\frac{1}{2}}x\| \|(\mathbbm{1}-E_i)^{\frac{3}{2}}x\|\\
&\leq \|E_i\|^{3} \|x\|^2 \ip{x-E_i x}{x}\\
&\leq \|x\|^2 \ip{x-E_i x}{x}\\
\end{aligned}$$
and from this we conclude, for each fixed $x\in \hilbert$, $\|E_i x -x\|\rightarrow 0$, $E_i \in \mathcal{P}\cap \mathcal{S}_\calgebra$.
%Notes on Operator Algebras, Roe - lemma 1.38
\end{proof}

\end{appendices}

% ****************************** Symbols ********************************
\glsaddallunused %to add all entries to the glossary
\printglossaries % If glossaries is present

% ****************************** Index ********************************

\printthesisindex % If index is present

% ***************************** Bibliography ***************************
\begin{spacing}{0.9}

% To use the conventional natbib style referencing
% Bibliography style previews: http://nodonn.tipido.net/bibstyle.php
% Reference styles: http://sites.stat.psu.edu/~surajit/present/bib.htm

%\bibliographystyle{apalike}
%\bibliographystyle{plainnat} % use this to have URLs listed in References

%\cleardoublepage
%\bibliography{References/references} % Path to your References.bib file

% If you would like to use BibLaTeX for your references, pass `custombib' as
% an option in the document class. The location of 'reference.bib' should be
% specified in the preamble.tex file in the custombib section.
% Comment out the lines related to natbib above and uncomment the following line.

\printbibliography[heading=bibintoc, title={References}]

\end{spacing}

\end{document}